\documentclass[12pt,letterpaper]{report}

\usepackage{amsmath,amsthm,amsfonts,amssymb,amscd}
\usepackage{enumerate} 
\usepackage{pifont} 

\usepackage[font=footnotesize]{caption}
\usepackage{graphicx}
\usepackage{subfig}
\graphicspath{{Figures/}}

\usepackage{tikz}
\usetikzlibrary{shapes,arrows,fit}
\usetikzlibrary{calc}

\theoremstyle{plain}
\newtheorem{theorem}{Theorem}[section]
\newtheorem*{theorem*}{Theorem}

\theoremstyle{definition}
\newtheorem{definition}[theorem]{Definition}
\newtheorem*{definition*}{Definition}

\theoremstyle{remark}

\newtheorem*{remark*}{Remark}

\usepackage[left=1.5in,top=1in,right=1in,bottom=1in]{geometry} 
\usepackage{setspace} 
\usepackage{fancyhdr} 

\newcommand{\approvalspace}{\vspace{.8in plus 3em minus 3em}}
\newlength{\approvallength}
\newcommand{\titleType}{}

\usepackage{etoolbox}
\patchcmd{\chapter}{\thispagestyle{plain}}{}{}{}

\usepackage[noadjust]{cite} 

\usepackage[breaklinks]{hyperref}
\usepackage{footnotebackref} 
\usepackage{breakurl} 
\usepackage{doi} 

\makeatletter
\g@addto@macro{\UrlBreaks}{\UrlOrds}
\newcommand{\labeltarget}[1]{\Hy@raisedlink{\hypertarget{#1}{}}}
\makeatother

\hypersetup{
    colorlinks,
    linkcolor={blue!80!black},
    citecolor={blue!80!black},
    urlcolor={blue!80!black},
    filecolor={blue!80!black}
}

\newcommand{\thesisspacing}{\doublespacing} 

\usepackage{graphicx}

\newcommand\defeq{\stackrel{\smash{\scriptscriptstyle\mathrm{def}}}{=}}
\usepackage{makecell}

\renewcommand\theadfont{\bfseries}

\usepackage{array}
\usepackage{booktabs}
\usepackage{multirow}
\newcommand{\head}[1]{\textnormal{\textbf{#1}}}
\newcommand{\normal}[1]{\multicolumn{1}{l}{#1}}

\usepackage{tikz}
\usetikzlibrary{backgrounds,fit}
\usetikzlibrary{patterns,calc}
\tikzset{Sm/.style={rectangle,fill=white,draw,minimum size=0.5cm,inner sep=0pt},}
\tikzset{Im/.style={rectangle,fill=black,draw,minimum size=0.5cm,inner sep=0pt},}
\tikzset{Sw/.style={circle,fill=white,draw,minimum size=0.5cm,inner sep=0pt},}
\tikzset{Iw/.style={circle,fill=black,draw,minimum size=0.5cm,inner sep=0pt},}
\usetikzlibrary{decorations.markings} 
\usetikzlibrary{decorations.text} 
\usetikzlibrary{positioning} 
\usetikzlibrary{backgrounds} 
\usetikzlibrary{fit} 
\usetikzlibrary{calc} 
\usetikzlibrary{arrows,shapes}

\usepackage{dot2texi}
\usepackage[ruled,vlined]{algorithm2e}
\usepackage{enumitem}
\usepackage{fixltx2e}
\usepackage{algpseudocode}
\MakeRobust{\Call}

\def\ErdosGallai{{Erd\H{o}s$-$Gallai }}
\def\ErdosRenyi{{Erd\H{o}s-R\'{e}nyi }}
\usepackage{amsfonts,amsmath,amsthm,amsxtra,amssymb}
\def\ED/{$\text{ED}_{\text{max}}$}
\def\ND/{$\text{ND}_{\text{max}}$}
\def\SND/{$\text{S}_{\text{max}}\text{ND}_{\text{min}}$}

\usepackage{tikz}
\usetikzlibrary{positioning,backgrounds,snakes,fit}
\usetikzlibrary{patterns,calc}
\tikzset{main node/.style={circle,fill=blue!20,draw,minimum size=.5cm,inner sep=0pt},}
\tikzset{Sm/.style={rectangle,fill=white,draw,minimum size=0.5cm,inner sep=0pt},}
\tikzset{Im/.style={rectangle,fill=black,draw,minimum size=0.5cm,inner sep=0pt},}
\tikzset{Sw/.style={circle,fill=white,draw,minimum size=0.5cm,inner sep=0pt},}
\tikzset{Iw/.style={circle,fill=black,draw,minimum size=0.5cm,inner sep=0pt},}
\usetikzlibrary{decorations.markings} 
\usetikzlibrary{decorations.text} 
\usetikzlibrary{positioning} 
\usetikzlibrary{backgrounds} 
\usetikzlibrary{fit} 
\usetikzlibrary{calc} 
\usetikzlibrary{arrows,shapes}
\usepackage{dot2texi}

\makeatletter
\newcommand*{\rom}[1]{\expandafter\@slowromancap\romannumeral #1@}
\makeatother
\usepackage{bm}

\usepackage{float}
\usepackage{pgfplots}

\newif\ifnotesw \noteswtrue

\usepackage{titlesec}
\usepackage{rotating}
\titleformat{\paragraph}
{\normalfont\normalsize\bfseries}{\theparagraph}{1em}{}
\titlespacing*{\paragraph}
{0pt}{3.25ex plus 1ex minus .2ex}{1.5ex plus .2ex}

\def\be{\begin{equation}}
\def\ee{\end{equation}}

\usepackage{dcolumn}
\newcolumntype{d}[1]{D{.}{.}{4}}
\newcommand{\mlc}[1]{\multicolumn{1}{c}{#1}}

\def\so { \textbf{SocNet }} 
\def\se { \textbf{SexNet }}

\usepackage{tabstackengine}
\setstackgap{L}{1.2\normalbaselineskip}
\setstacktabbedgap{.4em}
\fixTABwidth{T}

\begin{document}

\pagenumbering{gobble}

\newpage
\singlespacing
\renewcommand{\titleType}{an abstract} 
\newpage
  \thispagestyle{empty}
    {\sc
      \begin{center}
      \large{
        Mathematical Models for Predicting and Mitigating the Spread of Chlamydia Sexually Transmitted Infection\\
        \vspace{.5in}
        \titleType \\
        submitted on the \today \\
        to the department of mathematics\\
        of the school of science and engineering of\\
        tulane university\\
        in partial fulfillment of the requirements\\
        for the degree of\\
        doctor of philosophy\\
        by\\
        \vspace{.3in} \approvalspace
        \parbox[c]{\approvallength}{\hrulefill \\ 
          \makebox[\approvallength]{Asma Azizi Boroojeni}}
        }
    \end{center}
    \approvalspace \hfill approved: \makebox[\approvallength]{\ }
    
    \hfill \parbox{\approvallength}{ \hrulefill \\ James Mac Hyman \\ chairman}
    \approvalspace
    
     \hfill
      \parbox{\approvallength}{ \hrulefill \\ Patricia Kissinger}
      \approvalspace

     \hfill
      \parbox{\approvallength}{ \hrulefill \\ Lisa Fauci}
      \approvalspace
    
     \hfill
      \parbox{\approvallength}{ \hrulefill \\ Kun Zhao}
      \approvalspace

     \hfill
      \parbox{\approvallength}{ \hrulefill \\ Scott McKinley}
      \approvalspace
  }

\thesisspacing

\chapter*{Abstract} 
Chlamydia trachomatis
(Ct) is the most common  bacterial sexually transmitted infection (STI) in the United States and is major cause of infertility, pelvic inflammatory disease, and ectopic pregnancy among women.
Despite decades of screening women for Ct, rates continue to increase among them in high prevalent areas such as New Orleans.  A pilot study in New Orleans found
approximately $11\%$ of  $14-24$ year old of African Americans (AAs)
were infected with Ct. 
 Our goal is to mathematically model the impact of
different interventions for  AA men resident in New Orleans  on the
general rate of Ct among women resident at the same region.
We create and analyze  mathematical models such as multi-risk and continuous-risk compartmental models and  agent-based network model to first help understand the spread of Ct and second evaluate and estimate  behavioral and biomedical interventions including condom-use, screening, partner notification, social friend notification, and rescreening.
Our compartmental models predict the Ct prevalence is a function of the number of partners for a person, and quantify how this distribution changes as a function of condom-use. We also  observe that although increased Ct screening and  rescreening, and treating partners of infected people will reduce the prevalence, these mitigations alone are not sufficient to control the epidemic. A combination of both sexual partner and social friend notification is  needed to mitigate Ct.

\newpage
\singlespacing
\renewcommand{\titleType}{a dissertation}
\newpage
  \thispagestyle{empty}
    {\sc
      \begin{center}
      \large{
        Mathematical Models for Predicting and Mitigating the Spread of Chlamydia Sexually Transmitted Infection\\
        \vspace{.5in}
        \titleType \\
        submitted on the \today \\
        to the department of mathematics\\
        of the school of science and engineering of\\
        tulane university\\
        in partial fulfillment of the requirements\\
        for the degree of\\
        doctor of philosophy\\
        by\\
        \vspace{.3in} \approvalspace
        \parbox[c]{\approvallength}{\hrulefill \\ 
          \makebox[\approvallength]{Asma Azizi Boroojeni}}
        }
    \end{center}
    \approvalspace \hfill approved: \makebox[\approvallength]{\ }
    
    \hfill \parbox{\approvallength}{ \hrulefill \\ James Mac Hyman \\ chairman}
    \approvalspace
    
     \hfill
      \parbox{\approvallength}{ \hrulefill \\ Patricia Kissinger}
      \approvalspace

     \hfill
      \parbox{\approvallength}{ \hrulefill \\ Lisa Fauci}
      \approvalspace
    
     \hfill
      \parbox{\approvallength}{ \hrulefill \\ Kun Zhao}
      \approvalspace

     \hfill
      \parbox{\approvallength}{ \hrulefill \\ Scott McKinley}
      \approvalspace
  }

\thesisspacing

\newpage
~\vspace{1in}

\newpage
\pagenumbering{roman}
\setcounter{page}{2} 
\addcontentsline{toc}{chapter}{Acknowledgement} 
\chapter*{Acknowledgement}

Foremost, \textbf{Thanks to merciful God} for all the countless gifts dedicated me and to my parents for their love and for  supporting me spiritually throughout my PhD  and my life.

 I would like to express my heartiest gratitude to my advisor Prof. \textbf{James Mac Hyman} for his  sincere guidance, his continuous support of my Ph.D study and research,  his patience, motivation, enthusiasm, and immense knowledge.  
 Mac you are more than   an advisor for me: with your way of teaching I grew up both personally and profesionally. You are the very best advisor any lucky student can have. 
 
I also would like to thank my thesis committee members: Profs. Patricia Kissinger, Lisa Fauci, Scott Mcckinley and  Kun Zhao, for their encouragement, insightful comments, and hard questions. My sincere thanks also goes to Norine Schmidth, Dr. Charles  Stocker, and Dr. Martin , for leading me working on this project projects.

And an special thank to  all my friends and fellows specially to Martha Dryer,  Jeremy Dewar, Zhuolin Qu, Justin Davis, and Ling Xue and  all those, too many to name.

This work was supported by the endowment for the Evelyn and John G.
Phillips Distinguished Chair in Mathematics at Tulane University, grants from the
National Institutes of Health National Institute of Child Health and Human
Development, NCHID/NIAID ( R01HD086794),  Office of Adolescent Health, OAH
(TP2AH000013), the National Institute of General Medical Sciences program for
Models of Infectious Disease Agent Study (U01GM097658), and by grants from the National Science Foundation (DMS1263374).

\newpage 
\addcontentsline{toc}{chapter}{List of Tables} 
\listoftables

\newpage 
\addcontentsline{toc}{chapter}{List of Figures} 
\setcounter{lofdepth}{2} 
\listoffigures

\newpage
\tableofcontents

\newpage
\pagenumbering{arabic}
\pagestyle{myheadings} 
\chapter{Chlamydia Trachomatis}\label{ct}
Chlamydia trachomatis (Ct)  is an infection  from the family of Sexually Transmitted Infections (STIs),   which are  transmitted through sexual acts and they can be caused by  bacteria or viruses.  Here, sexual act means any  type of sexual intercourse including oral, vaginal, and anal sex.
STIs  are one of the most common causes of illness and even death worldwide, and therefore, are a major public health issue. These infections exert a high emotional toll on suffered individuals, as well as an economic burden on public health system.  The World Bank estimated that among women aged $15-44$ year old, STIs (excluding HIV) are the second most common causes of healthy life lost after maternal morbidity \cite{adler1998abc}.  

There are more than $20$ different STIs including chlamydia, gonorrhea, syphilis, herpes,  viral hepatitis, and HIV  affecting men and women of all backgrounds and economic levels. However, because of lack of an effective notification system in many countries and also lack of symptom in most of these STIs, the size of the global burden of STIs is uncertain. 
In this Chapter we are going to review a background and feature of  the most prevalent bacterial  STI i.e Ct.

\section{History of Ct}
Chlamydia  is an infection caused by a kind of bacteria called  Chlamydia trachomatis (Ct) that is passed during sexual act, usually thorough vaginal and anal intercourse. Ct was first discovered in $1907$ by german parasitologist Stanislaus von Prowazek.  Genus part of the name, Chlalmydia, comes from the Greek word chlamys, which means cloak and the species part of the name, trachomatis is also Greek and means rough or harsh \cite{dimitrakov2011chlamydia}.
 
 Most of Ct cases do not show any symptom, for that reason it sometimes called \textit{Silent infection}, for the cases showing symptom, its  symptoms are similar to some other infections, therefore, it was not recognized as a sexually transmitted disease till $1980$.
 
 Today Ct is the most common and the most spread bacterial STI in the world. In $1997$ there were $537,904$ reported diagnoses, however, by $2009$ the annual total had more than doubled to $1,244,180$. 
  In the United States over $2.8$ million cases of Ct are reported each year \cite{torrone2014prevalence}. Based on Center of Disease Control and Prevention (CDC) report, about three million American women and men become infected with Ct every year. Spreading of Ct  among African Americans, AAs, was eight times bigger than whites and rates among American Indians/Alaska Natives and Hispanics are also higher than among whites \cite{cdc}. The Figure  (\ref{fig:rate1}) shows the rate of reported infected cases by gender in years $1994\textnormal{-}2014$ and Figure (\ref{fig:rate2}) is reported cases by region in $2014$ in the United States. 
\begin{figure}[H]
\centering
\includegraphics[width=0.7\textwidth]{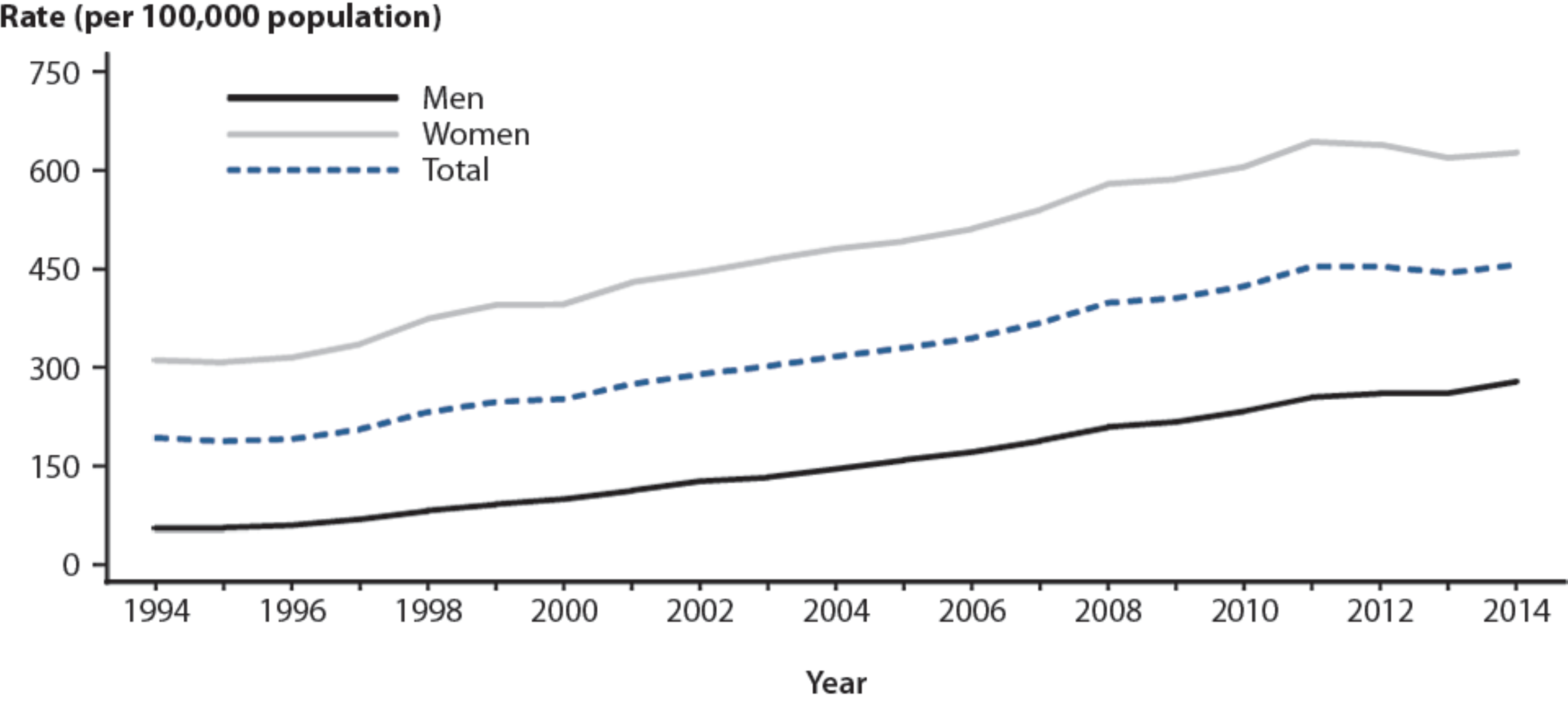}
\caption[\textbf{Ct vs gender}]{Rates of Reported Cases by gender in United States in years $1994-2014$ \cite{cdc}.}
\label{fig:rate1}
\end{figure}
\begin{figure}[H]
\centering
\includegraphics[width=0.7\textwidth]{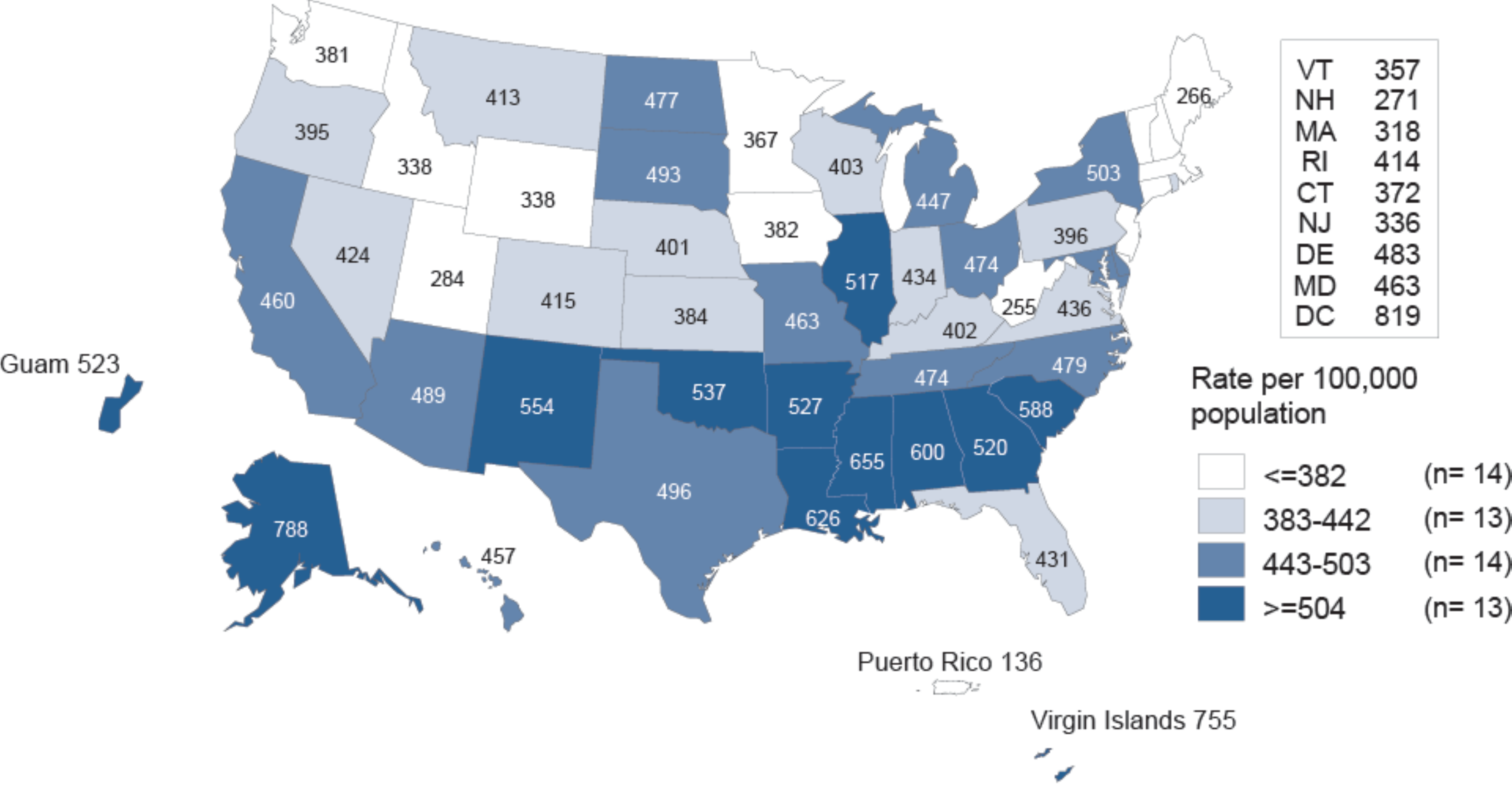}
\caption[\textbf{Ct vs region}]{Rates of Reported Cases by region in United States in  $2014$ \cite{cdc}.}
\label{fig:rate2}
\end{figure}

Ct mostly affects  young people, individuals ages $15-25$ years old. CDC estimates that adolescent and young adults, people  ages $15-25$ years old,  make up around  one quarter of the sexually active population, but account for $67\%$ of the Ct infections that occur in the United States \cite{cdc}.  In $2013$, the rate among $15-19$ year old people was $1852.1$ cases per $100000$ and the rate among $20-25$ year old people was $2451.6$ cases per $100000$. 
As shown in Figure (\ref{age1}), among women, the highest age-specific rates of reported Ct in $2013$ were among those aged $15-19$ years ($2941.0$ cases per $100000$ women) and $20-25$ years ($3651.1$ cases per $100000$ women) \cite{cdc}.

\begin{figure}[h]
\centering
\includegraphics[width=\textwidth]{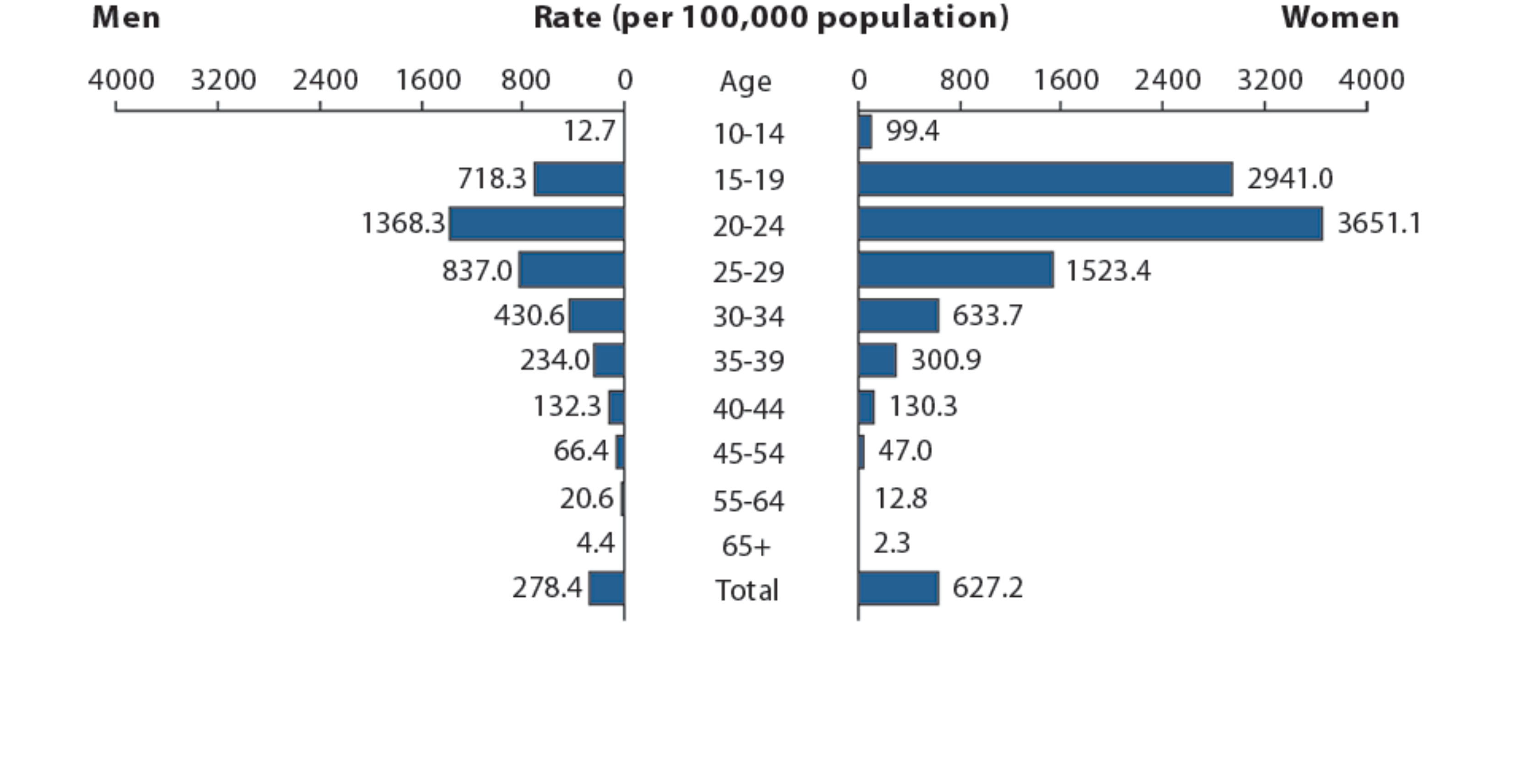}
\caption[\textbf{Ct vs age}]{ Ct Rates of Reported Cases by Age Group and Sex, United States, 2016 \cite{cdc}.}
\label{age1}
\end{figure}

\begin{figure}
\centering
\includegraphics[scale=2]{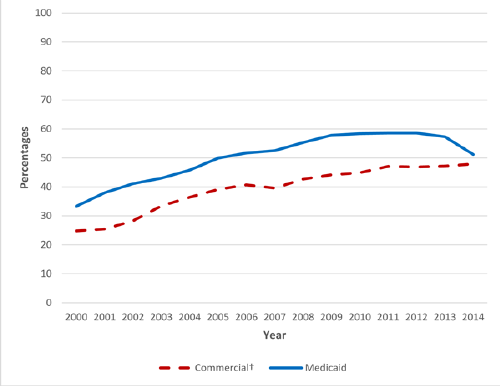}
\caption[\textbf{Screened females for Ct}]{Percentage of sexually active female aged $16$--−$24$ years  who were screened for Ct infection, by health plan type and year, United States, $2000$--‐$2014$}
\label{screen}
\end{figure}

\section{Symptoms and Causes }
Ct is known as a silent infection because most of the infected people are asymptomatic and lack abnormal physical examination findings: about $70\%-95\%$  of women and $90\%$ of men with Ct have no  symptoms \cite{farley2003asymptomatic,korenromp2002proportion}.
However, in the case infected people show symptoms, they are different for women and men. 

Infected women with Ct may experience  abdominal pain, abnormal vaginal discharge, bleeding between menstrual periods, low-grade fever, painful intercourse,
 pain or a burning feeling while urinating, swelling inside the vagina or around the anus, the urge to urinate more than usual,  vaginal bleeding after intercourse, and
  yellowish discharge from the cervix \cite{cdc}. Infected men  may experience  pain or a burning feeling while urinating,
 pus or watery or milky discharge from the penis,
 swollen or tender testicles, and swelling around the anus \cite{cdc}.

Ct infections are associated with a spectrum of clinical diseases,
urethritis and including epididymitis
among men, and cervicitis, salpingitis,
and acute urethral syndrome among
women  \cite{quinn1996epidemiologic}. 
Although at the early stage of Ct the damages go unnoticed, but Ct can lead to serious health problems, that is, because Ct is silent infection, it can sometimes  cause other diseases.

 Ct is a major cause of
infertility, pelvic inflammatory disease (PID), and ectopic pregnancy among women with estimated annual cost exceeds five billion dollars \cite{cohen1998sexually,  datta2012chlamydia,gottlieb2010introduction,gottlieb2010summary,hillis1996screening,lan1995chlamydia,pearlman1992review,westrom1995effect,westrom1993sexually},  
 and has been associated with increased HIV acquisition
and transmission  \cite{cohen1998sexually,gottlieb2010introduction,lan1995chlamydia,pearlman1992review,golden2003partner,hillis1996screening,niccolai2011behavioural,ward2010contribution,westrom1995effect,westrom1993sexually}. Untreated, an estimated $16\%$ of,
women with Ct will develop PID \cite{cohen1998sexually}, and $6\%$ will have tubal infertility \cite{lan1995chlamydia}.   
 In pregnant women, untreated Ct has been associated with pre-term delivery, as well as ophthalmia neonatorum (conjunctivitis) and pneumonia in the newborn \cite{rours2011chlamydia}.

\section{Control and Prevention}

The rate of spread of Ct in a population is determined 
by three factors \cite{eng1997hidden}:
\begin{enumerate}
\item the probability of acquiring
the infection by susceptible individuals, i.e the efficiency of transmission ($\beta$),
\item the rate of exposure of susceptible persons to infected partners ($c$),  and
\item  the length of time that  persons are infected and are able to transmit infection ($\tau$).
\end{enumerate}
There is an important concept in epidemiology- called basic reproduction number and shown by $\mathcal{R}_0$-  which states on average how many infections result from one infected person in a wholly susceptible population. For the very simple model this value is calculated as $\mathcal{R}_0=\beta c\tau$. If this value is greater than one, then  Ct can increase in the community. But, if it is less than one, then the rate of spread
of the Ct will die out. The pattern of spread for when $\mathcal{R}_0\leq 1$ and $\mathcal{R}_0> 1$ is shown Figure (\ref{patern_of_r0}). 

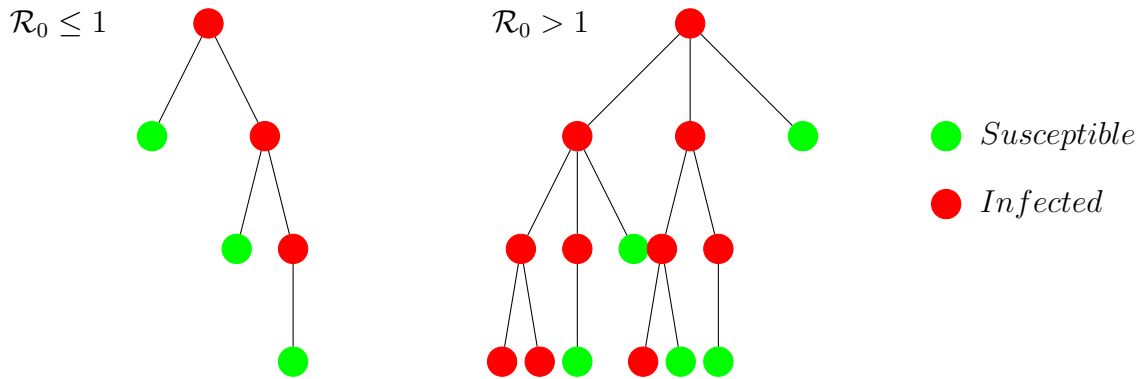
\begin{figure}[H]
\begin{tikzpicture}[level/.style={sibling distance=15mm/#1, level distance=15mm}
]
\node [circle,draw,fill, red] (z){}
  child {node [circle,draw,fill, red] (a) {}
    child {node [circle,draw,fill, red] (b) {}
        child {node [circle,draw,fill, red] (d) {}}
        child {node [circle,draw,fill, red] (e) {}}
               }
    child {node [circle,draw,fill, red] (g) {}
      child {node [circle,draw,fill, green] (x) {}}
          }
     child {node [circle,draw,fill,green] (g) {}
          }
  }
    child {node [circle,draw,fill, red] (a) {}
    child {node [circle,draw,fill, red] (b) {}
        child {node [circle,draw,fill, red] (d) {}}
        child {node [circle,draw,fill, green] (e) {}}
    }
    child {node [circle,draw,fill, red] (g) {}
    child {node [circle,draw,fill, green] (f) {}}
    }
  }
  child {node [circle,draw,fill, green] (j) {}
}
;
\path (a) -- (j) ;
\path (b) -- (g);
\path (d) -- (e) ;
  child [grow=down] {
    node (y) {}
    edge from parent[draw=none]
  };
 \node [left= 1cm of z] (SnW)  {$\mathcal{R}_0> 1$};
  \node [left= 6cm of z, circle,draw,fill, red] (0){}
    child {node [circle,draw,fill, green] (1) {}
}
  child {node [circle,draw,fill, red] (2) {}
       child {node [circle,draw,fill, green] (3) {}
}
 child {node [circle,draw,fill, red] (4) {}
      child {node [circle,draw,fill, green] (5) {}
}
}
}
  ;
\node [left= 1cm of 0] (SnW)  {$\mathcal{R}_0\leq 1$};
\node [right= 1.5cm of j, circle,draw,fill, green] (S)  {};
\node [right=1mm of S] (l1)  {$Susceptible$};
\node [below= .5cm of S, circle,draw,fill, red] (I)  {};
\node [right=1mm of I] (l2)  {$Infected$};
\end{tikzpicture}
\caption[\textbf{Concept of} $\mathcal{R}_0$]{Pattern of Spreading of infection for the cases $\mathcal{R}_0\leq 1$ and $\mathcal{R}_0 > 1$.}
\label{patern_of_r0}
\end{figure}
Our  main goal of interventions is keeping $\mathcal{R}_0$ less than $1$. Therefore, we can prevent the spread of Ct within a population by reducing the rate of exposure to Ct, reducing the efficiency of transmission, or shortening the duration of infectiousness for Ct \cite{eng1997hidden}. Targeting each of these factors by individuals or committees to control the epidemic of Ct ends up with different strategies to take \cite{eng1997hidden}. 

For individual level, the only safe way to prevent Ct is to abstain from sexual act with others \cite{cdc}. However, this way is not realistic and applicable. People can reduce their risk of catching or transmitting the infection  by changing their behaviors such as  using condoms during every sexual act, reducing the number of concurrent sex partners, and undergoing regular screenings.

For population level, the strategies to control Ct with emphasis on different components depends on the local pattern and distribution of Ct in the community and economical condition of the community. There are several principles to apply: prevention can be aimed at uninfected people in the community to prevent them from acquiring infection (reducing to exposure and transmission) or at infected people to prevent the  transmission of the infection to their sexual partners (reducing infection period) \cite{adler1998abc,rogstad2011abc}. 
In this subsection, we explain each of the principles which was taken from \cite{adler1998abc}.

\subsection{Behavioral approach: reducing to exposure and transmission  }
A \textbf{behavioral intervention} is a set of interventions  encouraged by public health to individuals for implementing in order to reduce Ct transmission. Individual, group, and community-level
behavioral interventions seek to directly change so-called
behavioral determinants of risky behaviors,
such as sexual and drug use knowledge, attitudes, beliefs,
perceptions of risk, barriers, social norms, motivation to
change, behavioral intentions, self-efficacy (confidence)
and a variety of skills (e.g., partner negotiation skills,
correct condom use skills) as a route to behavior change \cite{adler1998abc}.

A \textbf{sexual behavior}  is commonly defined as behavior that effects one's risk of contracting Ct and generally STIs.
 Because sexual activity is typically initiated
in adolescence or early adulthood and
because that period for many young people is
characterized by greater amounts of experimentation,
partner change, and risk taking
than in later years, research programs with a
focus on the behaviors of adolescents and
young adults are of particular importance \cite{turner1998adolescent}. Aral \cite{aral1993sexual} reviewed the sexual and other behaviors that place
individuals at a high risk of exposure to Ct. These behaviors are:
\begin{enumerate}
\item Initiation of sexual intercourse at an early age, because  adolescents are biologically more susceptible to Ct than adults.
\item Taking lots of concurrent partners: the greater the number of partners an individual has, the greater is the risk of exposure to any STI, because this behavior increases the chance of having an infectious partner.
\item  Having sex with a partner who is likely to have had
many partners.
\item Increased frequency of intercourse: the greater is the frequency of intercourse with an infected partner, the greater are the chances of transmission. 
\item Lack of circumcision of male partner: women with male partners who are circumcised are at lower risk of exposure compared to those with uncircumcised partners.
\item Lack of barrier contraceptive use such as condoms \cite{aral1993sexual}.
\end{enumerate}

\textbf{Condoms}, if  used correctly and consistently during every sexual intercourse, are
 the most effective method of preventing exposure to Ct \cite{cdc}. Condoms are also highly effective against bacterial and viral STIs including HIV
infection, however, failure to use a condom correctly and consistently, rather than potential
defects of the condom itself, is considered to be the major barrier to condom
effectiveness \cite{centers1993update}. Data show that condom-use has increased in the United States in the last few decades: six in ten high school students in the United States, who are sexually active, reported they used condoms at their most recent sexual intercourse. Condom-use among this group increased from $46\%$ in $1991$, to $63\%$ in $2003$, and was $59\%$ in $2013$ \cite{WinNT}. Reece et al. \cite{reece2010condom} also studied  rates of condom-use among sexually active individuals in the United States population. Based on their result, adolescents reported condom-use during $79.1\%$ of the past $10$ vaginal intercourse events.

\subsection{Biomedical approach: reducing infection period }
The goal of \textbf{Biomedical interventions}  is to reduce the risk of
infected individuals  transmitting  infection to
their partners \cite{adler1998abc}. These approaches entail encouraging health seeking behavior and increasing
screening and appropriate treatment of symptomatic and
asymptomatic people and  tracing, screening, and treating sexual partners of infected people, and presumptive treatment of people at high risk of infection \cite{adler1998abc}. 
Historically,  most of the Ct programs aim to  reduce infection period by treating infected people and their partners through screening and partner notification.

\emph{\textbf{Screening}}:
early diagnosis and treatment of Ct are valuable and inexpensive, because if Ct is well controlled other serious long term sequelae can be prevented \cite{eng1997hidden}. In the United States specialized STI clinics provide screening and treatment for people with symptoms of, or who feel they are at risk of, STIs  \cite{eng1997hidden}. 

\emph{\textbf{Partner Notification}}:
partner notification has been a component of STI programs in the United States for many years \cite{rothenberg1990strategies}, and has continued to be supported through current federally funded STI programs. 
For the infections which the incubation period  is long like syphilis, partner notification is able to break the chain of transmission  by identifying source of infection and their partners and or by identifying and treating partners exposed to infection \cite{cotton1994cdc}. For STIs with short incubation period like gonorrhea and Ct the rational of partner notification has to be modified \cite{cotton1994cdc}. For example, emphasis can be placed on locating
asymptomatic infected  partners of symptomatic or asymptomatic screened individuals and on providing early treatment to prevent complications \cite{cotton1994cdc}. Partner notification  prevents transmission to 
partners and directly benefit the exposed individual by preventing
symptomatic infection and is considered to be a strategy that benefits the  partner of individual
index patient and the community and even index patients themselves because treatment of their partner causes that they do not  become reinfected \cite{adler1998abc}. 

There are several approaches of implementing partner notification: 
one of the widespread techniques in partner notification is \textit{provider referral}. This method  relies on intensive interviews with
patients about their sexual histories and partners, followed by active outreach by public health staff to identify and locate partners to ensure that they are
examined and treated \cite{eng1997hidden}. Although labor intensive and costly, provider referral is still carried out within most public health programs for some selected 
 STIs including Ct \cite{golden2003partner}.
 
In  \textit{Patient referral} the  patients themselves notify  their partners about Ct exposure which is time effective for public health staff.
 
 If index patient undertakes to notify partners themselves in a given time frame \textit{Contract referral} will be used, i.e if the partners are not notified in this period, the  health adviser will attempt to notify them with the patient’s
consent \cite{adler1998abc}. 

When the partner is notified through one of the above methods then he/she may be given medications without testing,  \textit{partner treatment}. This practice, although widespread,  has several disadvantages for preventing infection including the small but real risk of adverse drug reactions in unseen patients, the inability to screen the partner for other STIs, and the lack of opportunity to examine and counsel the partner \cite{adler1998abc}.
Partner may follow test first and then follow medication if infected, \textit{partner screening}. 
The following diagram explains the process in partner notification.
\begin{figure}[h]
\centering
\par\medskip
\hspace*{-1.5cm}
\begin{tikzpicture}[
    grow=right,
    level 1/.style={sibling distance=4.5cm,level distance=5.2cm},
    level 2/.style={sibling distance=2cm, level distance=5.4cm},
    edge from parent/.style={very thick,draw=blue!40!black!60,
        shorten >=5pt, shorten <=5pt},
    edge from parent path={(\tikzparentnode.east) -- (\tikzchildnode.west)},
    kant/.style={text width=4.5cm, text centered, sloped},
    every node/.style={text ragged, inner sep=2mm},
    punkt/.style={rectangle, rounded corners, shade, top color=white,
    bottom color=blue!50!black!20, draw=blue!40!black!60, very
    thick }
    ]
    
\node[punkt, text width=9em] {Partner Notification}
     child { 
     node[punkt, text width=4em] {Contract Referral}
     child {
        node[punkt] [rectangle split, rectangle split, rectangle split parts=1,text ragged] {Partner Screening}
        edge from parent
             node [kant, above] {} node [kant, below] {} 
    }
    child {
        node[punkt] [rectangle split, rectangle split, rectangle split parts=1,text ragged] {Partner Treatment}
        edge from parent
                         node [kant, above] {} 
    }
     edge from parent
             node [kant, above,] {} node [kant, below] {}
     }
    child {
        node[punkt, text width=4em] {Patient Referral}
        child {
        node[punkt] [rectangle split, rectangle split, rectangle split parts=1,text ragged] {Partner Screening}
        edge from parent
             node [kant, above] {} node [kant, below] {} 
    }
    child {
        node[punkt] [rectangle split, rectangle split, rectangle split parts=1,text ragged] {Partner Treatment}
        edge from parent
                         node [kant, above] {} 
    }
    edge from parent
                         node [ above] { } node [ below] {}
             }
        child {
        node[punkt, text width=4em] {Provider Referral}
        child {
        node[punkt] [rectangle split, rectangle split, rectangle split parts=1,text ragged] {Partner Screening}
        edge from parent
             node [kant, above] {} node [kant, below] {} 
    }
    child {
        node[punkt] [rectangle split, rectangle split, rectangle split parts=1,text ragged] {Partner Treatment}
        edge from parent
                         node [kant, above] {} 
    }
                    edge from parent{
                node [kant, above] {} node [kant, below] {} }
               };
\end{tikzpicture}
\caption[\textbf{Components of partner notification}]{Components of Partner Notification.}
\label{chart}
\end{figure}
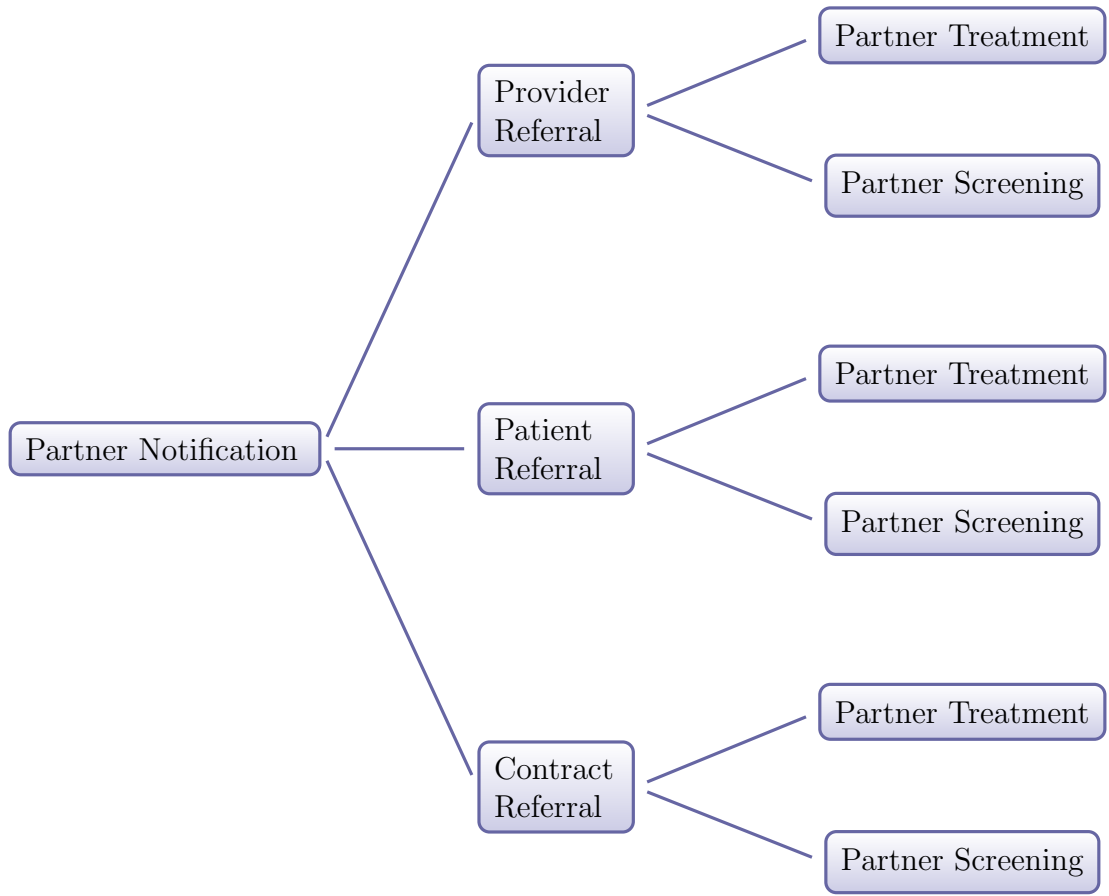
\section{Mathematical Approaches to Control Ct}
Mathematical models create frameworks for understanding underling epidemiology of diseases and how they are correlated to the social structure of the infected population \cite{eaton2010youth, del2005effects, del2007mixing, hyman1997behavior, hyman1997disease, hyman1999differential, hyman2001initialization, hyman2003modeling, hyman1988using, hyman1989effect, hyman1994risk}.  Transmission-based models can help the medical/scientific community to understand and to anticipate the spread of diseases in different populations and help them to evaluate the potential effectiveness of different approaches for bringing the epidemic under control. Therefore, primary goal of mathematical modeling effort is to create a detailed model that can be used understand the spread of infection and predict the impact of mitigation efforts to reduce the prevalence. In this Section we review some of the most important and recent mathematical models for Ct, and then we outline our project on mathematical model and mitigation efforts on controlling Ct in New Orleans.
    
 The SEIRS Ct transmission model  developed by Althaus et al. \cite{althaus2010transmission} captures the most essential transitions through an infection with Ct to assess the impact of Ct infection screening programs. 
Using sensitivity analysis they identified the time to recovery from infection and the duration of the asymptomatic period as the two most important model parameters governing the disease prevalence. 
Longer recovery time diminishes the effect of screening, however  longer duration of the asymptomatic period results in a more pronounced impact of  program.  They also used their model to improve the estimates for the duration of the asymptomatic period by reanalyzing previously published data on persistence of Ct in asymptomatically infected women.  This model did not divide the population into separate risk groups and assumed that all men and women had the same number of partners.

 Clarke et al.   \cite{clarke2012exploring} investigated how control plans can affect observable quantities and demonstrated that partner positivity (the probability that the partner of an infected person is infected) is insensitive to changes in screening coverage or partner notification efficiency.
They also evaluated the  cost-effectiveness of increasing  partner notification versus screening and concluded that partner notification along with screening is the most cost-effective mitigation approach.

Kretzschmar et al.  and Turner et al. \cite{kretzschmar1996modeling, kretzschmar2001comparative} 
evaluated different screening  and  partner referral methodologies 
in controlling Ct.   They compared the 
RIVM model  to evaluate the effectiveness of opportunistic  Ct  screening program in the Netherlands \cite{kretzschmar2001comparative}; 
the ClaSS model to evaluate proactive, register-based  Ct  screening using home sampling in the UK \cite{low2007epidemiological}; and 
the HPA model to evaluate opportunistic national  Ct  screening program in UK \cite{turner2006modelling}.

A selective sexual mixing STI model was developed by Hyman et al. \cite{hyman1997disease} to capture the 
heterogenous mixing among people with different number of partners. 
 This model is well-described by
Del Valle et al. \cite{del2013mathematical} to investigate the impact of different mixing assumptions on spread of infectious diseases and  how sensitivity analysis can be used to prioritize different possible mitigation efforts.

 
 Our goal in this work is  to create, to analyze, and to extend mathematical models for understanding and predicting the spread of Ct. These models can help guide public health workers
improve the effectiveness of intervention strategies for mitigating the impact of this infection. 
Our  main focus will be to help optimize the interventions in an ongoing program for
reducing the prevalence of Ct in the New Orleans adolescent and young adult AAs.
We will design several compartment and agent-based network models within different Chapters:\\
Chapter \ref{multi} is about a multi-risk compartment model for Ct. In Chapter \ref{continuous} we will extend this model to a continuous-risk compartment model. And finally in Chapter \ref{abm} we will provide a next generation of  agent-based network models for the spread of Ct in New Orleans, and will test and will compare different mitigation method to control its epidemic. 

In all proposed models the parameters will be  estimated within a reasonable level of accuracy in order for results to give  qualitative and quantitative understanding of how Ct is spreading \cite{hyman1997disease}.  We  will use local sensitivity analysis to identify the relative importance of the model parameters and numerical examples to illustrate how we can prioritize mitigation strategies based on their predicted effectiveness. 
\chapter{Multi-risk  Compartment Model}\label{multi}

In this chapter we develop and analyze our first and simplest model  that is a multi-risk compartmental   model that can be used to help understand the spread of Ct and to quantify the relative effectiveness of different mitigation efforts. 
Our model is closely related to the deterministic population-based model developed by Clarke et al.   \cite{clarke2012exploring} to explore the short-term impacts of increasing screening and  partner notification. Also, it  is related to the STI models for the spread of the HIV/AIDS virus in a heterosexual network \cite{hyman1988using, hyman1989effect}. 

The number of partners a person has (his/her  risk), and the number of partners that their partners have (his/her partner's risk) both affect the spread of Ct.  That is, different  assumptions about the distribution of risk behavior of the population will result in different disease forecasts.   
We use  the selective sexual mixing STI model developed by Hyman et al. \cite{hyman1997disease} to capture the 
heterogenous mixing among people with different numbers of partners.

Although age, ethnicity, economic statues,  and the spatial location of the individuals all influence the assortative mixing of sexual acts,  the risk of contracting Ct is primarily a function of the number of partners a person has, the number of acts per partner, the probability that a partner is infected, and the use of prophylactics (e.g. condoms).

In our ordinary differential equation model (ODE), we consider defining the risk categories based on  the number of partners a person has.
The relative importance of the number of partners and the number of acts per partner on the spread of an STI depends on the disease infectiousness.  Ct is a very infectious disease and the probability of transmission per sexual act  from an infected person to uninfected one is high;  one act with an infected person is enough to catch the infection. Therefore, the number of people a person infects  depends mostly  upon the number of partners he/she has.

In this chapter we first formulate the mathematical model, then  we derive the basic reproduction number, $\mathcal{R}_0$, for two main risk groups (high-risk and low-risk) for men and women.   
We then use sensitivity analysis of $\mathcal{R}_0$  and the equilibrium points with respect to the model parameters 
to study how the heterogeneous mixing affects spread of Ct \cite{azizi2016multi}.

\section{Ct Transmission Model Overview}\label{Model}
In modeling the spread of Ct,  the population is divided into the susceptible sexually active population (S), the exposed infected, but not infectious population (E) and the infectious population (I). Once a person has recovered from Ct infection, they are again susceptible to infection.  Therefore, the models all have a S$\rightarrow$E$\rightarrow$I$\rightarrow$S (SEIS) structure, or a SIS structure if the exposed state is combined with the infectious state.

Because the exposed (infected, but not infectious) time period is short compared to time in the infectious stage, we 
do not include a exposed stage in our model.  
We divide men and women into $n$ risk groups based on the number of partners an individual has in a year.  
This SIS model can be written as the system of $2n$ ordinary differential equations: 
\begin{equation}\label{multi_model}
\begin{split}
\frac{d S_k}{d t}&=\mu (N_k-S_k)-\lambda_kS_k+\gamma_k I_k,\\
\frac{d I_k}{d t}&=\lambda_kS_k-\gamma_k I_k-\mu I_k,\\
\end{split}
\end{equation}
where $k=1, \cdots, n$ denotes men  with risk from $1$ to $n$, and $k=n+1, \cdots, 2n$ denotes women  with risk from $1$ to $n$.  
The migration rate, $\mu$, determines the rate at which people enter and leave the population,  $N_k = S_k+I_k$ is the total population  of group $k$, 
$\lambda_k$ is  the rate at which a susceptible person in risk group $k$ is being infected, and $\gamma_k$ is the rate that a person recovers  either through screening, or natural recovery.  

We model a population of $15-25$ year-old individuals and assume that 
the primary mechanism for migration is by aging into, and out of, the population, where migration rate $\mu = 0.003 = [(25-15)~ \mbox{years}]^{-1}$, with the assumption  that death is negligible compared to the rate that 
people enter and leave the modeled population. We assume that, in the absence of infection, 
equilibrium population $N^o_k$ for each risk group of men and women is given, and that 
everyone aging into the model population enters as a susceptible person.

The rate  the infected population is treated, $\gamma_k$, depends upon the sex of the person and  their risk level.  The treatment can be initiated when infection is identified through screening, partner notification, or a medical check-up.
Most infected people  are asymptomatic and, when a significant fraction of a population is infected, then screening has been found to be a cost-effective approach to identify, and treat,  infected people.

Natural recovery rate, $\gamma^n_k$, is determined by assuming an exponential distribution for the average time to recovery $1/\gamma^n_k$, and screening recovery rate, $\gamma^s_k$, is determined by assuming a log normal distribution for the average time to recovery through screening $1/\gamma^s_k$.
We also define the probability that an infected individual is screened and treated each day, $\sigma_d^k$, in terms of the fraction of the population that will be screened at least once within a year as $\sigma_y^k$.  That is, 
\begin{equation}
\sigma_d^k=1 - (1-\sigma_y^k)^{1/365}.
\end{equation}

\subsection{Transmission rate}

We will derive the disease transmission rate for the heterosexual case where 
a susceptible person in group $k$ can be infected by someone of the opposite sex in any of the infected groups $j$.

The force of infection, $\lambda_k$,  is the rate that people in risk group $k$ are infected through sexual acts.   
We define $\lambda_k$ as the sum of the rate of disease transmission from
each infected group, $I_{j}$,  to the susceptible group, $S_k$:
\begin{equation} \label{infrate}
\lambda_k = \sum_{j=1}^{n}  \lambda_{kj}.
\end{equation}
The rate of disease transmission from the
infected people $I_{j}$ in group $j$ to the susceptible individuals $S_{k}$ in group $k$, $\lambda_{kj}$, 
 is defined as the product of three factors:
 \begin{align*}
\lambda_{kj} &= \left( \begin{array}{c}\mbox{Number of partners} \\
                                       \mbox{a susceptible in group $k$} \\
                                       \mbox{has with someone in} \\
                                        \mbox{group $j$ per unit time}
                       \end{array}
                \right)
                \left( \begin{array}{c}\mbox{ Probability of}\\
                                       \mbox{disease transmission} \\
                                       \mbox{ per partner}
                       \end{array}
                \right)
                \left( \begin{array}{c}\mbox{Probability that} \\
                                       \mbox{partner in  group $j$ } \\
                                       \mbox{ is infectious}
                       \end{array}
                \right) \\ 
                &= \hspace{2cm}p_{kj}\hspace{5cm}\beta_{k}\hspace{4cm}P_I(t,j)~,
\end{align*}

where
\begin{itemize}
\item $p_{kj}$ is the number of sexual partners per unit time that each individual in group $k$ has with
someone in group $j$, and 
\item $\beta_{k}$ is the probability of disease transmission per partner for a susceptible person  in group $k$, and 
\item $P_I(t,j)$ is the probability of that the person in group $j$ is infected. 
\end{itemize}

For this last factor, we assume that the partners in group $j$ are all equally likely to be infected.  That is, 
the probability  the person in group $j$ is infected  is the same as the fraction of the people in 
group $j$ that are infected, $\frac{I_{j}}{N_j}$.

\subsection{Partnership formation} 
The extent that Ct  spreads through a population is sensitive to 
the  heterogenous mixing (partnership selection) among the different risk groups. 
The model approximates the  mixing  through  mixing probabilities $p_{kj}$  
that define how many partners a typical person in group $k$ has with someone in group $j$. 
These mixing functions must  dynamically change to 
account for variations in the size of the groups 
\cite{hyman1997behavior, del2013mathematical, busenberg1991general}.

The force of infection, $\lambda_k$,  depends on how many partners 
people in group $k$ have, the number 
of acts they have per partner, and the probability that their partners are infected.  
The mixing is biased since people who only have a few sexual partners (low-risk) 
typically have partners who are also at low risk. 

We define the model parameters so that someone in group $k$ has, on average, 
 $p_{kj}$ partners who are in group $j$ per day.  
 Therefore, the total number of partnerships per day  between people in group 
 $k$ and group $j$ is  $p_{kj}N_k = p_{jk}N_j$.
Since each partnership 
may have more than one act, we define $a_{k}$ as the average number of sexual acts per partner for people in group $k$.  

To determine $p_{kj}$, we use a heterogeneous mixing algorithm developed 
in   \cite{hyman1997behavior}.  
This approach starts by defining 
$\bar p_k$ as the \emph{desired} number of partnerships someone in group 
$k$ wishes to have per unit time.   Because there may not be sufficient  
available partners for everyone to have their desired number of partners, 
the actual number of partners could be different. 

We define the \emph{proportional} partnership (mixing) as the desired fraction $\rho_{kj}$ 
of these partnerships that a person in group $k$ wants to have 
with someone in group $j$.  That is, a person in group $k$ wants to have 
 an average of $\rho_{kj} \bar p_k$ partnerships per unit time with someone in group $j$.
Unfortunately, there is no guarantee that the total number of desired partnerships 
that people in group $k$ want to have with people in group $j$ will be the same 
as the total number of desired partnerships 
that people in group $j$ want to have with people in group $k$.  That is, in general 
$\rho_{kj} \bar p_k N_k \ne  \rho_{jk} \bar p_j N_j$, and this must be reconciled. 

Since not everyone can have their desired number of partners distributed 
exactly as they wish, the different heterogenous mixing algorithms represent different  
compromises to resolve  these conflicts. 
All of the heterogenous mixing algorithms maintain the  detailed balance for 
mixing where the total number of partnerships for people in group $k$ with 
people in group $j$ is the same as the total number of partnerships that 
people in group $j$ have with people in group $k$.  
In our model, we use the heterogenous mixing algorithm based on the algorithm 
described in  \cite{hyman1997behavior,  del2013mathematical} 
to determine $p_{kj}$.

The population in group $k$ desires $\rho_{kj} \bar p_k N_k$ partners from group $j$, 
and the population in group $j$ desires $\rho_{jk} \bar p_j N_j$ partners from group $k$.  
As a compromise, we  set the total number of partners the people in group 
$k$  have with people in group $j$, and vice versa, to be the harmonic mean 
\begin{equation}
p_{kj} N_k= p_{jk}N_j= \frac{2(\rho_{kj} \bar p_k N_k) (\rho_{jk} \bar p_j N_j)}{(\rho_{kj} \bar p_kN_k)+ (\rho_{jk} \bar p_j N_j)}.
\end{equation} 
Other possibilities include the geometric mean or minimum of 
$(\rho_{kj} \bar p_kN_k)$ and  $(\rho_{jk} \bar p_j N_j)$.
All of these averages satisfy the balance condition 
to have the property 
that if $\rho_{jk} =0$ then $p_{kj} =p_{jk}=0$, where if one group refuses to have a 
partnership with another group, then this partnership does not happen. 
In our model, we use the  harmonic mean and define 
\begin{equation}
p_{kj} = \frac{1 }{ N_k} \frac{2(\rho_{kj} \bar p_kN_j) (\rho_{jk} \bar p_j N_k)}{(\rho_{kj} \bar p_kN_k)+ (\rho_{jk} \bar p_j N_j)}~.\end{equation} 
Hence, $p_k= \sum_{j}p_{kj}$ is the actual average number of partners someone in group $k$ has per day. 

Note that this approach is only appropriate if the desired number of partners between any two groups is in close agreement, 
that is,  $\rho_{kj} \bar p_k N_k \approx \rho_{jk} \bar p_j N_j$.  
This is because, the approach assumes that if the partners are not available from the desired group, then the individuals 
will not change their preferences to seek partners in other risk groups.  The model can be extended to 
handle these situations where the people adjust their desires to be in closer alignment with the availability of partners 
through a simple iterative algorithm.  However, 
we avoid this complication in our simulations and initialize the populations so the groups desires are close to the 
availability of partnerships. 

\subsection{Probability of transmission per partner}
The probability of a susceptible person catches infection from their infected partner depends upon the number of sexual acts between the people. 
We allow the number of acts per partner for a person in group $k$, $a_k$, to depend upon the number of his/her actual partners  and  his/her total number of acts per unit time,
$A_k=a_k p_k$, where $A_k$ is total number of acts per unit time.
  
 The probability of transmission per act, $ \beta$, can be used to define  the probability that a susceptible person will not be infected by a single act with an infected person, $1- \beta$.  
Therefore, the probability of someone in group $k$ not being infected after $a_{k}$ acts with an 
infected person  is  $(1-\beta)^{a_{k}}$. 
Hence, the  probability of being infected per partner is   \cite{hyman2003modeling}
\begin{equation} \label{betaij}
\beta_{k} = 1-(1-\beta)^{a_{k}}.
\end{equation}

\begin{figure}[htp] \label{ fig:cir}
\begin{center}
  \includegraphics[width=4in]{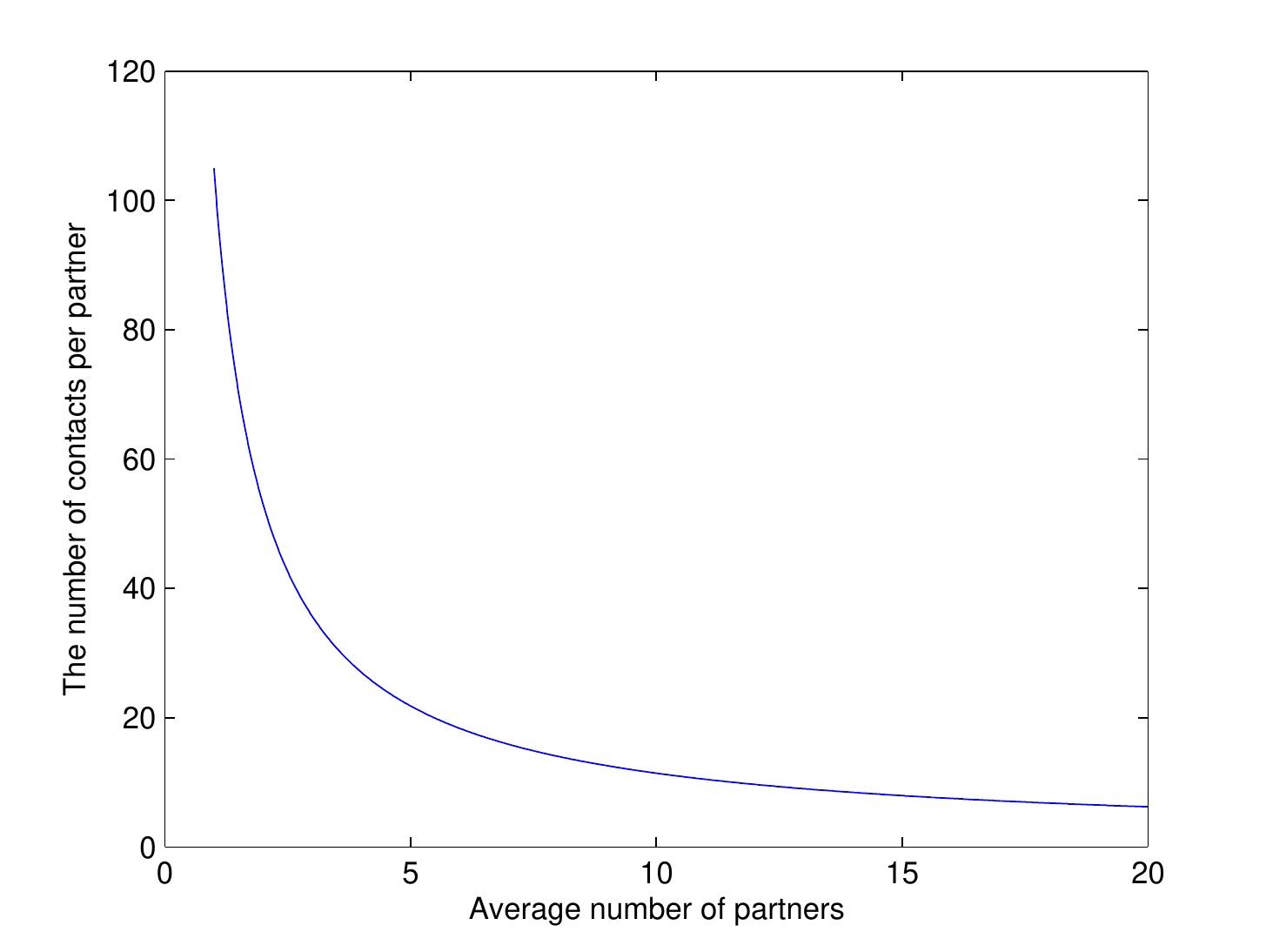}\\
  \caption[\textbf{The number of acts per partnership per unit time vs number of partners}]{The number of acts per partnership per unit time (year), $a_k$,  is a decreasing function of $p_{k}$, that 
  is, people with more partners have fewer acts per partner per unit time than people with fewer partners. If $A_k$ is total number of acts
  for a person in group $k$, $a_k$ is defined as $a_k=\frac{A_k}{p_k}=\frac{A_k}{\sum_j p_{kj}}.$}
  \end{center}
\end{figure}

\begin{table}[htbp]
\centering
\resizebox{\columnwidth}{!}{
\begin{tabular}{ lp{9cm}lll }
\toprule[1.5pt]
 \normal{\head{Parameter}} & \normal{\head{Description}}
  & \normal{\head{Unit}} & \head{Baseline}\\
  \cmidrule(lr){1-3}\cmidrule(l){3-4}
  N& Total population. & people & $20000$\\
 $N_1$ & High-risk men population. & people &  $2400$\\ 
   $N_2$ & Low-risk men population. & people &  $4600$\\
 $N_3$ & High-risk women population. & people &  $2500$\\
 $N_4$ & Low-risk women population. & people &  $10,500$\\
$\mu$ & Migration rate. & 1/days& $0.00$ \\
  $\beta$ & Probability of transmission per act. & 1/act & $0.11$\\
        $1/\gamma_i^n~ (i=1,2,3,4)$& Average time to recover without treatment.   & days&  $365$   \\
        $1/\gamma_i^s~ (i=1,2,3,4)$& Average time to recover with treatment.   & days&  $7$   \\
        $A_1$& Total number of acts per time for a high-risk man.  & 1/days&  $0.14$   \\
        $A_2$& Total number of acts per time for a low-risk man.   & 1/days&  $0.07$   \\
        $A_3$& Total number of acts per time for a high-risk woman.   & 1/days&  $0.07$   \\
        $A_4$& Total number of acts per time for a low-risk woman.   & 1/days&  $0.04$   \\
        $\bar p_{1}$ & Desired number of partners per time for a high-risk man.& people/days&  $0.14$ \\
         $\bar p_{2}$ & Desired number of partners per time for a low-risk man.& people/days&  $0.03$   \\
          $\bar p_{3}$ & Desired number of partners per time for a high-risk woman.& people/days&  $0.06$\\
           $\bar p_{4}$ & Desired number of partners per time for a low-risk woman. & people/days&  $0.02$ \\
            $\bar \rho_{13}$ & Desired fraction of high-risk partner for a high-risk man.& --&  $0.75$ \\
            $\bar \rho_{24}$ & Desired fraction of low-risk partner for a low-risk man.& --&  $0.80$ \\
            $\bar \rho_{31}$ & Desired fraction of high-risk partner for a high-risk woman.& --&  $0.75$ \\
            $\bar \rho_{42}$ & Desired fraction of low-risk partner for a low-risk woman.& --&  $0.80$ \\
            $\sigma_y^i~ (i=1,2,3,4)$ & Fraction of people in group $i$ randomly screened per year. & -- & $0.00$ \\
\bottomrule[1.5pt]
\end{tabular}}
\caption[\textbf{Variables and parameters of the  model}]{Table for the variables and parameters in the  model (\ref{multi_model}). 
All the parameter values are based on our assumption and are not estimated or taken from literature.}
\label{tab:params_table}
\end{table}

\section{Basic Reproduction Number}

The basic  reproduction number, $\mathcal{R}_0$, is the number of new infections introduced if a newly infected person is introduced into a population at the ($S_k=N_k$, $I_k=0$). 
We will derive the basic reproduction number $\mathcal{R}_0$  using the next generation approach  \cite{van2002reproduction, heffernan2005perspectives} 
for situation  with two risk levels for men and women labeled: $1=$ high-risk men, $2=$ low-risk men, $3=$ high-risk women and $4=$ low-risk women.
We have  differential equations 
\begin{equation} 
\frac{d I_k}{d t}=\sum_j \alpha_{kj} I_j  -   \tau_k^{-1} I_k,
\label{model2I}
\end{equation}
where $\tau_k=1/(\gamma_k+\mu)$ is the average time that an infected person stays in the $k$-th infection compartment.
The force from infection, $\alpha_{kj}=\beta_k q_{kj} {S_k}$, is the rate (per day) that a typical infected person in group $j$ infects a susceptible one in group $k$.  Here the factor $q_{kj}$ is defined as $q_{kj}=p_{kj}/{N_j}$ and because we had $p_{kj}N_k=p_{jk}N_j$, therefore, $p_{kj}/{N_j}=p_{jk}/{N_k}$, and that means $q_{kj}=q_{jk}$  is 
 the fraction of people in group $j$ someone in group $k$ has as a partner. 
The  Equation (\ref{model2I}) for the infected populations,  $I=(I_1, I_2, I_3, I_4)^T$,  can 
be written as a matrix equation for the rate of production of new infections, $F$, 
minus the removal rate of individuals from that population class, $V$,  
 \begin{equation}\label{nextgen}
\frac{dI}{dt} =  {\bf F} I-  {\bf V} I,
\end{equation} 
where  the $kl$-th element of the matrix \textbf{F} is ${\bf F}_{kl}=\alpha_{kl}$, and $V$ is diagonal matrix  $V_{kk}=\tau_k^{-1}$, for $k,l=1,...,4$.  
At DFE, $\alpha_{kj}=\beta_k q_{kj} {N_k} $ and the  Jacobian matrices, ${\bf J}_F $ and ${\bf J}_V^{-1}$, of $FI$ and $VI$ are 
\begin{equation*}
J_F= \left( \begin{array}{cccc}
0 & 0 & \alpha_{13}& \alpha_{14}\\
0 & 0 & \alpha_{23} & \alpha_{23} \\
\alpha_{31} & \alpha_{32} & 0  & 0\\
\alpha_{41} & \alpha_{42} & 0 & 0
 \end{array} \right), ~~~
J_V^{-1}= \left( \begin{array}{cccc}
\tau_1 & 0 & 0 & 0\\
0 & \tau_2 & 0 & 0\\
0 & 0 & \tau_3 & 0 \\
0 & 0 &  0 &\tau_4 
\end{array} \right).
\end{equation*}
We define $\mathcal{R}_0$ as spectral radius of $J_F J_V^{-1}$ or (equivalently) $J_V^{-1}J_F$,
         \begin{equation}
      J_F J_V^{-1}=\left( 
         \begin{matrix}
           0 & 0 & \alpha_{13} \tau_3& \alpha_{14} \tau_4  \\
           0 & 0 & \alpha_{23} \tau_3& \alpha_{24}  \tau_4\\
           \alpha_{31} \tau_1 &   \alpha_{32}  \tau_2 & 0 & 0 \\
            \alpha_{41}  \tau_1 &  \alpha_{42} \tau_2 & 0& 0\\
         \end{matrix}
         \right),~~~
        J_V^{-1}J_F=\left( 
         \begin{matrix}
           0 & 0 & \alpha_{13} \tau_1& \alpha_{14} \tau_1  \\
           0 & 0 & \alpha_{23} \tau_2& \alpha_{24} \tau_2\\
           \alpha_{31} \tau_3 & \alpha_{32} \tau_3 & 0 & 0 \\
           \alpha_{41} \tau_4 & \alpha_{42} \tau_4 & 0& 0\\
         \end{matrix}
         \right).
         \end{equation}
Note that $jk$-th element of $J_V^{-1}J_F $ is  $\mathcal{R}_0^{j \rightarrow k}=\alpha_{jk} \tau_j$.  
Based on a result of  Sylvester's inertia theorem \cite{sylvester1852xix}\footnote{Let $K$ be a hermittian matrix. We define $e^+(K)$ as the number of positive eigenvalues, $e^-(K)$ as the number of negative eigenvalues, and $e^0(K)$ as the number of zero eigenvalues. Inertia of $K$ is a tuple $(e^+(K), e^-(K)), e^0(K)$. If $A$ is an invertible matrix then Sylvester inertia theorem states:
	$inertia(K)=inertia(A^{-1}KA).$ },  if a matrix $K$ can be factored into the product of a diagonal  positive definite matrix $A$ and a symmetric matrix $B$, then eigenvalues of  $K$ are the same as eigenvalues of  $A^{\frac{1}{2}}  B  A^{\frac{1}{2}}$.
To apply this result, we  rewrite $J_V^{-1}J_F=AB$ where $A$ is a diagonal  positive definite matrix and $B$ is symmetric,
        \begin{equation}
        A=\left( 
                \begin{matrix}
                  \beta_1N_1 \tau_1  & 0 & 0& 0  \\
                   0 & \beta_{2} N_2 \tau_2  & 0& 0 \\
                 0 & 0 & \beta_{3} N_3 \tau_3 & 0 \\
                  0 & 0 & 0& \beta_{4}  N_4 \tau_4\\
                \end{matrix} 
                \right),~~~
                B= \left( 
                         \begin{matrix}
                           0 & 0 &q_{13}&q_{14}  \\
                           0 & 0 & q_{23}& q_{24} \\
                           q_{31} & q_{32} & 0 & 0 \\
                           q_{41} & q_{42} & 0& 0\\
                         \end{matrix}
                         \right).
                            \end{equation}
   Therefore, eigenvalues of  $J_V^{-1}J_F$ are the same as eigenvalues of symmetric  block anti-diagonal generation matrix  
  \small
 
   \begin{align} 
   & A^{\frac{1}{2}}  B  A^{\frac{1}{2}} = \left( \begin{matrix}
                                  0 & 0 & \sqrt{\alpha_{13} \tau_1  \alpha_{31}  \tau_3 }&
                                  \sqrt{\alpha_{14} \tau_1 \alpha_{41} \tau_4} \\
                                         0 & 0 & \sqrt{\alpha_{23} \tau_2 \alpha_{32} \tau_3 }&
                                         \sqrt{\alpha_{24} \tau_2 \alpha_{42} \tau_4} \\
                                  \sqrt{\alpha_{31} \tau_3 \alpha_{13} \tau_1} & 
                                  \sqrt{\alpha_{32} \tau_3 \alpha_{23} \tau_2 } & 0 & 0 \\
                                  \sqrt{\alpha_{41} \tau_4 \alpha_{14} \tau_1 } & 
                                  \ \sqrt{\alpha_{42} \tau_4 \alpha_{42} \tau_2} & 0& 0\\
                                \end{matrix} \right)   
                                = \left( 
  	\begin{matrix}
  	0_{4\times 4} & M   \\
  	M^T & 0_{4\times 4} 
  	\end{matrix}
  	\right),
  \label{matrix1}\end{align}    
  where $M^T$ is transpose of $M$ and
	\begin{equation}
              M =
               \left( 
                 \begin{matrix}
                                   \sqrt{\alpha_{13} \tau_1  \alpha_{31}  \tau_3 }&
                                  \sqrt{\alpha_{14} \tau_1 \alpha_{41} \tau_4} \\
                                  \sqrt{\alpha_{23} \tau_2 \alpha_{32} \tau_3 }&
                                   \sqrt{\alpha_{24} \tau_2 \alpha_{42} \tau_4} \\
                 \end{matrix}
                 \right)~~~ =~~~\left( 
                 \begin{matrix}
                                  r_{13}&
                                   r_{14} \\
                                  r_{23}&
                                    r_{24} \\
                 \end{matrix}
                 \right),
              \end{equation}
where  $r_{jk} =  \sqrt{\mathbb{R}_0^{j \rightarrow {k}}\mathbb{R}_0^{{k}\rightarrow j}}$ is the geometric average of group j-to-group $k$ and group $k$-to-group j reproduction numbers.
The basic reproduction number  is spectral radius of $J_V^{-1}J_F$, therefore, 
$\mathcal{R}_0=\rho({A^{\frac{1}{2}}  B  A^{\frac{1}{2}}})=\sqrt{\rho(M^TM)}$ where 
 $$
 M^TM=
 \left( \begin{matrix}                               
r_{13}^2 + r_{23}^2& r_{13}r_{14} + r_{23}r_{24}\\
 r_{13}r_{14} + r_{23}r_{24}&     r_{14}^2 + r_{24}^2
\end{matrix}   \right),
 $$
 and 
   \begin{equation} \label{R0}
   \mathcal{R}_0=\frac{1}{2}((r_{13}^2+r_{23}^2+r_{14}^2+r_{24}^2)
   +\sqrt{(r_{13}^2+r_{23}^2-r_{14}^2-r_{24}^2)^2+4(r_{13}r_{14}+r_{23}r_{24})^2 }).
   \end{equation}


\section{Sensitivity Analysis}
We  use sensitivity analysis  to  quantify the change in model output quantities of interest (QOI), such as the basic reproduction number $\mathcal{R}_0$ and endemic equilibrium point,  due to variations in  the model input  parameters of interest (POI), such as the average time to recovery after infection \cite{arriola2007being, arriola2009sensitivity, manore2014comparing}.

Consider the situation where the baseline value of the input POI is $p_{b}$ and generates the baseline output QOI $q_b=q(p_b)$. 
Sensitivity analysis is used to address what happens if $p_b$ is changed by the fraction $\theta_p$, $p_{new}=p_b(1+\theta_p)$. Our goal is to find resulting fractional change in the output variable $q_{new}=q_b(1+\theta^q_p)$.
That is, the normalized  sensitivity index measures the relative change in the input variable  $p$, with respect to the output variable $q$
and  can be estimated by the Taylor series
\begin{equation}
q_{new}=q(p_b+\theta_p p_b) \approx q_b+\theta_p p_b \frac{\partial q}{\partial p}\biggr\vert_{p=p_b}
=q_b(1+\theta^q_p)~.
\end{equation}
We define the normalized  sensitivity index as
\begin{equation}\label{E:defsenind}
\mathbb{S}_p^q := \frac{p_b}{q_b} \frac{\partial q}{\partial p}\biggr\vert_{p=p_b} = \frac{\theta^q_p}{\theta_p}.
\end{equation}
That is, if the input $p$ is changed by $\theta_p$ percent, then the output $q$ will change by $\theta^q_p=\mathbb{S}_p^q \theta_p$ percent. 
The sign of $\mathbb{S}_p^q$  determines the direction of changes, increasing (for positive $\mathbb{S}_p^q$) and decreasing (for negative $\mathbb{S}_p^q$). Note that this local  sensitivity index is valid only in a small neighborhood of the baseline values. 

\subsection{Sensitivity indices of $\mathcal{R}_0$} 

The ability of Ct to become established in a population and its early growth rate is characterized by $\mathcal{R}_0$, Equation (\ref{R0}).
Sensitivity analysis of $\mathcal{R}_0$ can quantify the relative importance of the  different social 
and epidemiological parameters in reducing the ability of the STI to become established in a new population.  

Table (\ref{table:sens})  of the sensitivity indices of $\mathcal{R}_0$ shows that it is most sensitive to the 
probability of transmission per act $\beta$ with $\mathbb{S}^{\mathcal{R}_0}_{\beta} = 1.95.$
That is, if the probability of infection per act decreases- say by increasing the condom-use - by $15\%$ 
then $\theta_{ \beta} =-0.15$, then $\mathcal{R}_0$ will decrease by $30\%$ from $4.01$ to $2.84$:
$$
{\mathcal{R}_0}_{new}= \mathcal{R}_0 ( 1 + \theta_{\beta} \mathbb{S}^{\mathcal{R}_0}_{{\beta}} ) =4.01( 1 - 0.15 \times 1.95) = 2.84.
$$
That is, sensitivity analysis can quantify the amount of behavior change that would be needed to keep an epidemic from becoming established in a new population. 

A negative sensitivity index indicates that $\mathcal{R}_0$ is a decreasing function of correspondent parameter, while the positive ones show $\mathcal{R}_0$ increases when the parameter increases.  
The second most important model parameters for the early growth rate are the recovery rates of the high-risk men and women $\gamma_1$ and $\gamma_3$.   
Since, $\mathbb{S}^{\mathcal{R}_0}_{\gamma_1}= \mathbb{S}^{\mathcal{R}_0}_{\gamma_3}=-0.95$, a $10\%$ increase in the screening rate 
would result in a $9.5\%$ decrease in $\mathcal{R}_0$, which this supports the need to actively screen both men and women for Ct infection. 

The number of acts for high-risk men, $A_1$, is also an important parameter for controlling the early growth of the Ct. Because local-sensitivity analysis is valid in a small neighborhood of the baseline case, sometimes it is useful to plot the change in the QOI over a wide range of possible values.  The sensitivity index is then the slope of the response curve at the baseline  values. 
The Figure  (\ref{fig:sensitive}) shows how $\mathcal{R}_0$ changes as  parameters $\beta$ and $A_1$ are varied over a broad range.  

\begin{table}[H]
  \centering
\begin{tabular}{@{}c d{4} d{4} d{4} d{4} d{4} d{4} @{}}
    \toprule
   
    & \multicolumn{6}{c}{\textbf{Sensitivity index of $\mathcal{R}_0$ for all parameters in the model}}
    \\
    \cmidrule(rl){2-5} \cmidrule(rl){5-7}

    & \mlc{Parameter p} & \mlc{Baseline } & \mlc{$\mathbb{S}^{\mathcal{R}_0}_p$}
    & \mlc{Parameter p} & \mlc{Baseline } & \mlc{$\mathbb{S}^{\mathcal{R}_0}_p$}\\
    \midrule
     & {\beta}  & 0.11   & 1.95   & \bar{p}_1   & 0.14    & 0.05 \\
     & A_1  & 0.14  & 0.76   &  \bar{p}_2 & 0.03  & -0.06\\
     & A_2 & 0.07 & 0.15 &  \bar{p}_3 & 0.06& 0.24\\
     & A_3   & 0.07   & 0.71   & \bar{p}_4  &  0.02  & -0.07\\
     & A_4 & 0.04& 0.22   & \gamma_1& 0.003 & -0.95\\
     & \bar{\rho}_{13} & 0.75 & 0.35 &\gamma_2& 0.003  & -0.48\\
      & \bar{\rho}_{31} & 0.75 & 0.58 & \gamma_3 & 0.003  & -0.95\\
       & \bar{\rho}_{24}& 0.8 & 0.04 & \gamma_4 & 0.003  & -0.48\\
        & \bar{\rho}_{42} & 0.8 & 0.12 & \mu & 0.00  & -0.09\\
   \bottomrule

  \end{tabular}
  \caption[\textbf{Sensitivity index of} $\pmb{\mathcal{R}_0}$ \textbf{with respect to parameters of the model}]{The sensitivity index of $\mathcal{R}_0$ with respect to parameters of the model at the baseline parameter values where $\mathcal{R}_0=4.01$. The most sensitive parameter is the  probability of transmission per act, $\beta$, followed by the recovery (screening) rates of the high-risk men and women $\gamma_1$ and $\gamma_3$.
}
\label{table:sens}
\end{table}

\begin{figure}[htp] 
\begin{center}
  \includegraphics[width=90mm]{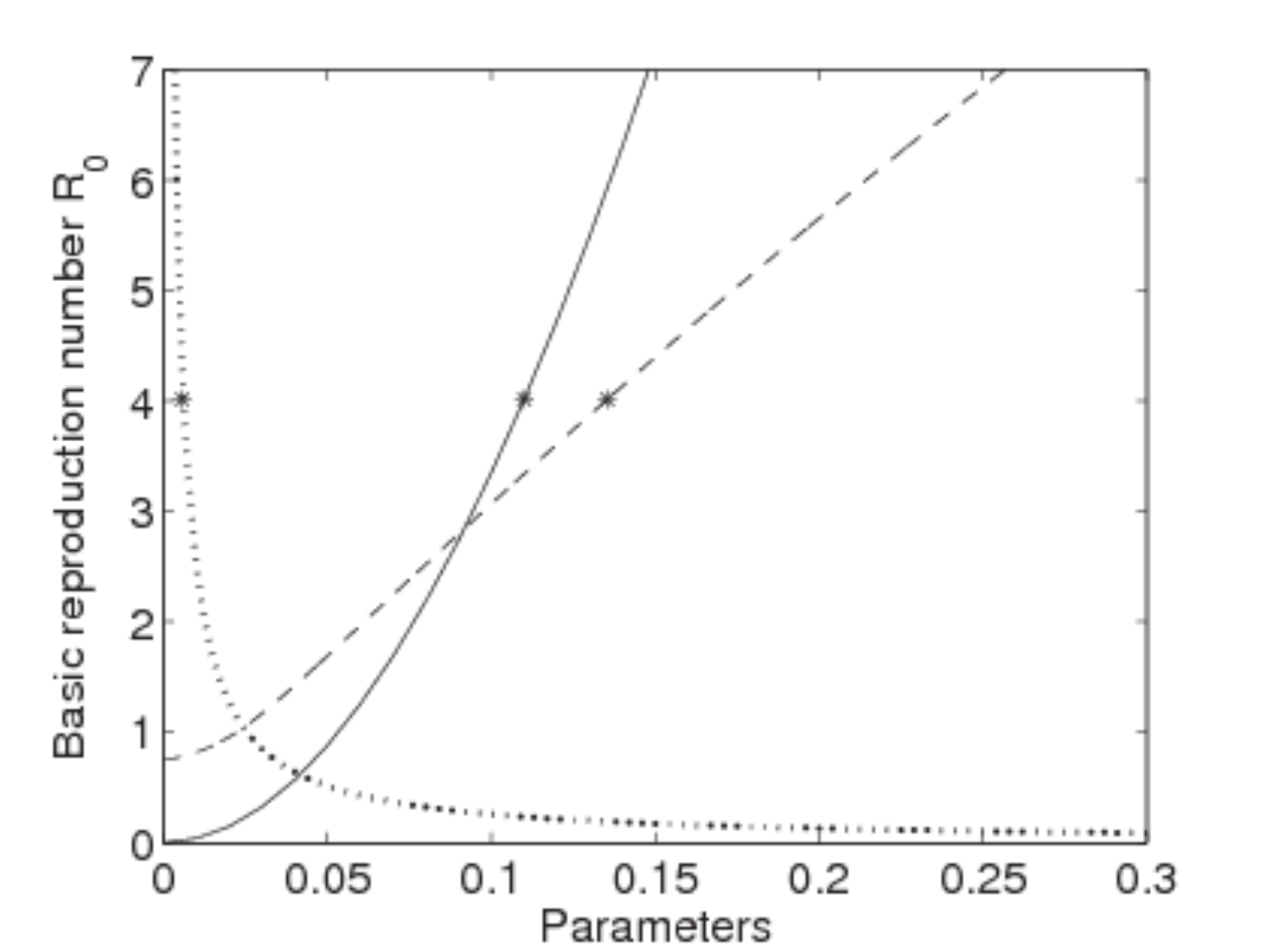}\\
 \caption[\textbf{Sensitivity of basic reproduction number}]{ The sensitivity of $\mathcal{R}_0$ with respect to $\beta$ (solid line), and $A_1$ (dashed line). 
The sensitivity index is then the slope of the response curve at the baseline  values, indicated by $*$.   
The response is approximately linear near the baseline case and, therefore, the local sensitivity analysis is actually valid over a  broad range of parameters.}
  
  \label{fig:sensitive}
  \end{center}
 
\end{figure}
  
\subsection{Sensitivity indices of  endemic equilibriums}

The current Ct epidemic is  established in many cities, therefore,  to evaluate the relative impact of the model 
parameters in bringing it under control requires that the sensitivity analysis be preformed 
about the current state of the system, the steady-state endemic equilibrium. 
We will investigate the impact the mitigation efforts on the relative change in the number of infected people 
as a function of the relative change in the model parameters.  This is best done in terms of 
the nondimensional variables defined by dividing each variable by the 
steady-state zero-infection equilibrium total population for that sex.  
That is, $i_m=I_m/N^o_m$, $i_w=I_w/N^o_w$,
$n_1=N_1/N^o_m$,  $n_2=N_2/N^o_m$,  $n_3=N_3/N^o_w$,  and $n_4=N_4/N^o_w$, 
where,  $N^o_m=N^o_1+N^o_2$, $N^o_w=N^o_3+N^o_4.$ The Table (\ref{table:sie}) shows that the  sensitivity indices for endemic (steady-state) equilibrium infected populations, $i_j$, 
 as a function of the model parameters.  Note that the magnitudes (relative importance) of sensitivity indices have 
the same order as they did for $\mathcal{R}_0$, although the magnitudes are different. 
\begin{table}[H]
  \centering
\begin{tabular}{@{}c d{4} d{4} d{4} d{4} d{4} d{4} @{}}
    \toprule
   
    & \multicolumn{6}{c}{\textbf{Sensitivity of equilibriums for all parameters in the model}}
    \\
    \cmidrule(rl){2-5} \cmidrule(rl){5-7}

    & \mlc{Parameter p} & \mlc{Baseline } &i_1& i_2 &i_3 & i_4\\
    \midrule
     & \beta  & 0.11   & 2.61  &4.05&2.7& 4.02\\
     & A_1  & 0.14  & 1.7   &0.58&0.71&0.66\\
     & A_2 & 0.07 & 0.31 & 1.5 &0.33&1\\
     & A_3   & 0.07   & 0.65   & 0.56&1.13&0.53\\
     & A_4 & 0.04 & 0.42   &1.18&0.37&1.61\\
     & \bar{\rho}_{13} & 0.75&0.36 & -0.21&  0.24  & -0.36 \\
      & \bar{\rho}_{31} & 0.75 &0.41 & -0.15 & 0.56 & 0.03 \\
       & \bar{\rho}_{24}& 0.8 & 0.02& -0.13 & 0.08 & -0.10\\
        & \bar{\rho}_{42} & 0.8 &  0.15&-0.21 & 0.08 & -0.33\\
        &  \bar{p}_1   & 0.14    &0.04 & 0.17 &  0.08 &  0.23 \\
     &  \bar{p}_2 & 0.03  & 0.27 & 0.31 & 0.21 & 0.26 \\
     &  \bar{p}_3 & 0.06&  0.27 & 0.31 & 0.21 & 0.26 \\
     &  \bar{p}_4  &  0.02  &  -0.05 & 0.03 & -0.03 & 0.04 \\
     &  \gamma_1& 0.003  & -1.10&    -0.60 &  -0.73&  -0.68\\
     & \gamma_2& 0.003  & -0.35 & -1.66&  -0.37&  -1.11\\
      &  \gamma_3 & 0.003  &-0.66 &  -0.56 &  -1.13& -0.53 \\
       & \gamma_4 & 0.003  &-0.47 & -1.31 &  -0.41 &  -1.78\\
        & \mu & 0.00  &  -0.13& -0.20 & -0.13 &-0.20\\
   \bottomrule

  \end{tabular}
  \caption[\textbf{Sensitivity index of endemic point with respect to parameters of the model}]{Local sensitivity indices of the endemic equilibrium points. At this baseline $\mathcal{R}_0\ge 1$, so this endemic point is a solution of model at steady state. }
\label{table:sie}
\end{table} 
 The prevalence of infection, $i_j$,  is most sensitive to probability of transmission per act $\beta$, i.e increasing $\beta$ increases $i_j$s more than other parameters. Then $A_j$s and $\gamma_j$s  have the second most effect on $i_j$s in positive and negative direction, correspondingly.

 Prevalence in high-risk men, $i_1$, is sensitive to the total number of acts for the high-risk men $A_1$ and $\gamma_1$ more than the other $A_j$s and $\gamma_j$s for $j\neq 1$. Prevalence in high-risk women, $i_3$, is also sensitive to $A_3$ and $\gamma_3$ more than the other $A_j$s and $\gamma_j$s for $j\neq 3$. It means when high-risk people increase their number of acts, regardless of  what others do, the fraction of infected people between high-risk people increases, because they have many partners.  On the other hand, when infection period for high-risk people increases, the prevalence in high-risk population increases.  

For low-risk men, the prevalence, $i_2$,  has the same sensitivity to $A_2$ and $A_4$. It means when  low-risk people increase their act, the prevalence in low-risk men increases, and we have the same story for low-risk women. It  is reasonable, because  low-risk people do not have  many partner, therefore, more acts for them and their partners plays an important role. 
Prevalence in low-risk men, $i_2$,  has also the same sensitivity to $\gamma_2$ and $\gamma_4$. It means when we decrease $\gamma_2$ and $\gamma_4$-infected  people in low-risk men stay in infection category for a longer time and also  infected women in low-risk group stay in infection category for a longer time- we see  increment in the value of  $i_2$ more than the other parameters. There is a similar  analysis for low-risk group $i_4$: low-risk group $i_4$ is sensitive to $\gamma_2$ and $\gamma_4$ with the same magnitude and more than  the other $\gamma_j$s.

Another interesting result is that the endemic equilibrium points are more sensitive, than $\mathcal{R}_0$, to most of the parameters. This result says,  controlling parameters to have a low fraction of  infected population is easier than adjusting the parameters to have smaller  $\mathcal{R}_0$.

 \section{Screening Scenarios}
 The goal of this Section is to study the impact of  different screening strategies  on prevalence of Ct among different groups. In  all the simulations, the parameters are fixed with the baseline values given in Table (\ref{tab:params_table}), unless specifically defined otherwise. 
In the first simulation, we assume that fraction of people who can be screened each year, $\sigma_y$, is limited by a budget, or other factors.     
We also assume that if an infected person is screened for Ct, then there is a $100\%$  probability that infection  will be detected. 

We will compare the fraction of the population that is infected  as a function of the screening rate  $\sigma_y^k$ people from different subgroups $k$.
We will also optimize the $\sigma_y$, for a fixed budget, that will minimize the fraction of infected people at steady state.  
That is, if $(i^*_1,i^*_2,i^*_3,i^*_4)$ are the fraction  of infected people at steady state, we find the optimal screening rates that  solve the optimization problem:
  \begin{equation*}
\begin{aligned}
& \underset{\sigma_y^k}{\text{minimize}}
& & \sum_{j=1}^4N_ji^*_j(\sigma_y^1,\sigma_y^2,\sigma_y^3,\sigma_y^4),\\\
& \text{subject to}
& & \sum_{j=1}^{4} N_j\sigma_y^j=N\sigma_y=0.2N=400,
\end{aligned}
\end{equation*} 
 where $N_ji^*_j$ is the number of infected people in group j and in steady state. The Figure. (\ref{fig:senario}) shows  the result for six different scenarios defined in Table (\ref{table:R0}).
\begin{table}[htp]
\centering
\begin{tabular}{llll}
\hline
\cline{1-4}
\textbf{Scenario} & $\mathcal{R}_0$&\textbf{Scenario} & $\mathcal{R}_0$\\
\hline
 (1) No screening & $4.01$& (2) Screen high-risk men& $0.86$\\
 \hline
(3) Screen high-risk women&  $0.91$& (4) Screen low-risk men& $1.26$   \\
 \hline
(5) Screen low-risk women&  $1.22$ &(6) Optimized Screening& $0.12$\\
\hline
\cline{1-4}
\end{tabular}
\caption[\textbf{ $\pmb{\mathcal{R}_0}$ for different scenarios}]{Basic reproduction number, $\mathcal{R}_0$, for different scenarios with respect to parameters of the model at the baseline parameter values. Implementing optimized screening decreases $\mathcal{R}_0$ to the order of $-1$.}
\label{table:R0}
\end{table}

\begin{figure}[htp]
\centering
\subfloat[\textbf{Different scenarios}][]{\label{fig:scenario_a}\includegraphics[width=0.45\textwidth]{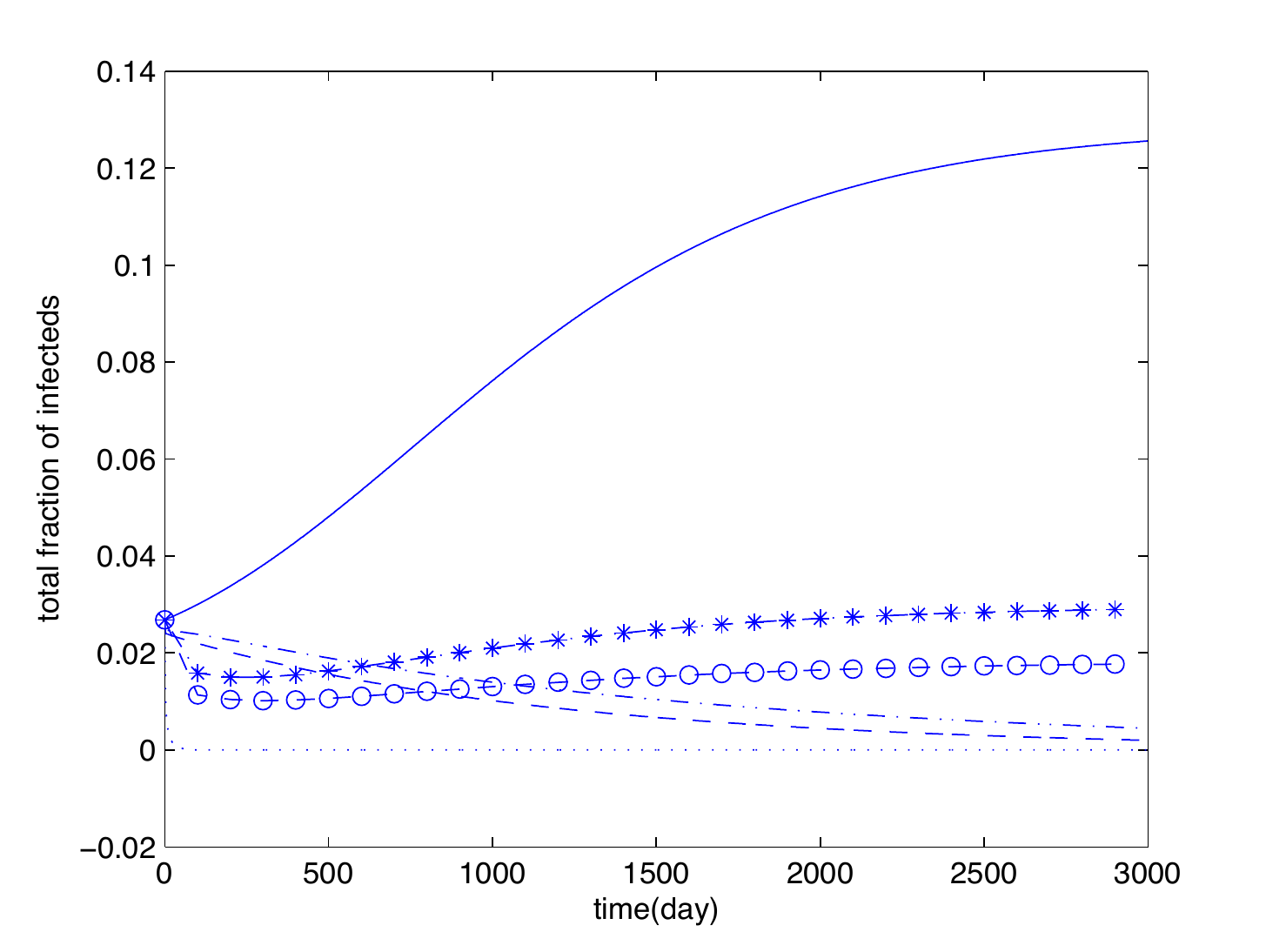}}
\subfloat[\textbf{Optimized scenario}][]{\label{fig:scenario_b}\includegraphics[width=0.45\textwidth]{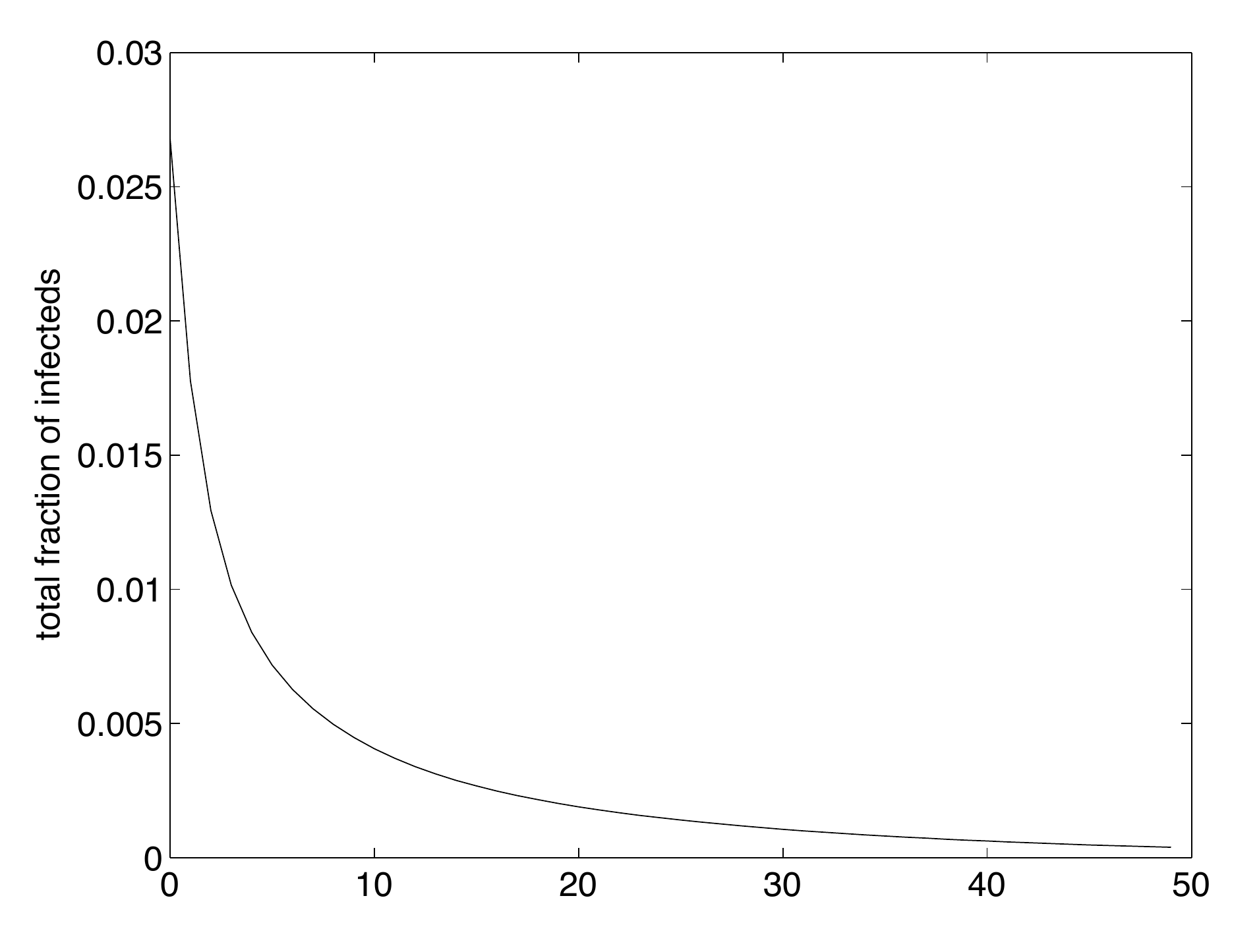}}
 \caption[\textbf{Prevalence of Ct for different screening scenarios}]{\ref{fig:scenario_a}: the fraction of infected people after implementing different scenarios: no screening (solid line), screen $\sigma_yN$ people per year for: high-risk men(dash line), high-risk women(dash-dot line), low-risk men (dash-star), low-risk women (dash-circle), and optimized screening(dotted line). \ref{fig:scenario_b}: zooms on the optimized screening, by optimized screening the infection dies out very fast.}
            \label{fig:senario}
\end{figure} 

We observe that in case of  no screening the epidemic goes up to its original endemic equilibrium point.  
The effectiveness of screening is seen by the dramatic reduction in the fraction of infected people.  However, between all scenarios screening high-risk people and optimized screening cause that epidemic dies out and for  optimal choice it dies out much faster than the other two  cases. 

We also list the value of $\mathcal{R}_0$ for different scenarios in Table (\ref{table:R0}).  The Figure (\ref{fig:senario}) and  Table (\ref{table:R0}) show that  $\mathcal{R}_0>1$ implies a persistent infection, though not a macroscopic outbreak in screening cases,  and when $\mathcal{R}_0 <1$ epidemic goes to DFE. Also, for the optimized scenario, which  its $\mathcal{R}_0$ is the lowest one, the epidemic dies out faster than the other scenarios.
Therefore, optimized screening was the most effective scenario among all six scenarios. 

To push our understanding of  the effects of optimized screening further,  we  do sensitivity analysis of equilibrium points with respect to screening rates at their optimized values.
In this case sensitivity index become a matrix like:
\[ \mathbb{S}= \left( \begin{array}{cccc}
i^*_1& 0 & 0& 0\\
0 & i^*_2& 0& 0\\
0 & 0& i^*_3& 0\\
0 & 0& 0& i^*_4
\end{array} \right)^{-1}
\times J_{i^*}(\sigma_y)\times
\left( \begin{array}{ccc}
\sigma_y^1& 0 &0\\
0 & \sigma_y^2&0\\
0 & 0&\sigma_y^3
\end{array} \right),
\]
where $J_{i^*}(\rho_y)$ is jacobian matrix. Each column $k$ of $\mathbb{S}$ represents sensitivity index of equilibrium points with respect to  screening rate $\sigma_y^k$. Therefore, $(k,j)th$ element of $\mathbb{S}$ is sensitivity index of $i^*_k$ with respect to $\sigma_y^j$. Table (\ref{table:sa_mat})  lists the elements of this matrix: all the values in table are negative, it means there is a inverse pattern between equilibrium points and screening rate: when we increase screening rates the fraction -therefore, the number- of infected people at steady state will decrease. 
Among all, $i^*_1$ is the most sensitive one, it means changing screening rates affects the  high-risk men more than the others. 
\begin{table}[htp]
\centering
\begin{tabular}{lllll}
\hline
\cline{1-5}
 & $i^*_1=0.00$ & $i^*_2=0.00$ & $i^*_3=0.00$ & $i^*_4=0.00$  \\
\hline
$\sigma_y^1=0.06$& $ -0.13$ & $-0.00 $& $-0.01$ & $-0.13$ \\
 \hline
$\sigma_y^2=0.20$& $-0.56$ & $-0.37$ & $-0.27$ & $-0.41$  \\
\hline
 $\sigma_y^3=0.05$& $-0.08$ & $-0.09$ & $-0.09$ & $-0.09$  \\
\hline
\cline{1-5}
\end{tabular}
\caption[\textbf{Sensitivity indices of equilibrium points with respect to screening rates }]{Sensitivity indices of equilibrium points with respect to screening rates at optimized baseline values. The most sensitive output parameter is fraction of infected high-risk men.}
\label{table:sa_mat}
\end{table}

\section{Discussion and Conclusion}
In this Chapter we created a multi-risk  heterosexual SIS transmission model for the spread of Ct with biased mixing partnership 
selection to investigate the impact that screening for the disease can have in controlling its spread. 
We derived the threshold conditions for the early spread of the disease and defined the basic reproductive number, 
$\mathcal{R}_0$, using the next generation matrix approach.  The analysis of $\mathcal{R}_0$ identified a new 
approach to reduce the size of the next generation matrix for a heterosexual Ct model with $n$ risk groups from 
an $2n \times 2n$ nonsymmetric sparse matrix to an $n \times n$ symmetric full matrix. This approach can be used 
in similar heterosexual STI models to greatly simplify the threshold analysis.

We used the  sensitivity analysis of $\mathcal{R}_0$ and endemic 
equilibrium steady-state solutions to quantify the relative effectiveness of different intervention strategies in mitigating 
the disease.  The analysis identified the probability of transmission per act (related to condom-use) is 
the most sensitive parameter in controlling the epidemic.  The second most effective control mechanism 
was the screening, and treating infections, of both  high-risk men and women.  Currently, most mitigation programs 
only target screening high-risk women. The model indicates that it is equally important to identify infections and treat 
high-risk men.  We confirmed that in the model the higher-risk groups are driving the epidemic and that 
$\mathcal{R}_0$ is most sensitive to the behavior of these higher-risk people.

We implemented different screening scenarios consist of screening only high-risk men, only high-risk women, 
only low-risk men, only low-risk women, and optimized screening. 
We then solved  for an optimal screening strategy for infection mitigation when there are limited resources
and then determined the best screening approach to minimize the endemic steady state infection prevalence, optimized screening.
Not surprisingly, we found that this same strategy also  minimizes $\mathcal{R}_0$.
In the next Chapter, we generalize our multi-risk model to continuous-risk, when individuals can take as many number of partners as they want and we focus on impact of condom-use to control the epidemic of Ct.

\chapter{Continuous Risk-based  Model}\label{continuous}

In this chapter, we extend the multi-risk group model in Chapter \ref{multi}  to a continuous risk-based transmission  model  that can be used to understand the spread of Ct in the adolescents and  young adult population. The model predicts the impact of people having different number of  concurrent partners or using prophylactics, such as condoms, on the rate that infection spreads.

We use this risk-based integro-differential  model \cite{azizi2017risk} to study the  impact of  variations in number of partners, mixing patterns in selecting partners, and condom-use to determine optimal Ct prevention policies.  
We study how the number of partners that a person has, and how often they use condoms, will affect the spread of Ct. 
Here,  the risk is defined  based on  the number of partners a person has per year.  The distribution of risk behavior for a population, such as the fraction of the population having multiple partners and also, the number of partners that their partners have (their partner's risk)  affects the spread of Ct and must be accounted for in the model. 
Our model accounts for a broad range of risk behavior, defined as the number of partners per year, that is captured as a continuous variable.
This model could also be used to include separate core high-risk groups, such as sex workers.   However, in the young adult population being modeled, sex-workers are not believed to be a major factor in the spread of highly infectious STIs like Ct.

The risk of contracting Ct is  primarily a function of a person's risk,  the probability that a partner is infected, and the use of prophylactics (e.g. condoms).
We use the selective  mixing   model developed by Busenberg et al. \cite{busenberg1991general} to capture the heterogenous mixing among  people with different number of partners. Our model is closely related to the models for the spread of the HIV/AIDS   in heterosexual networks \cite{hyman1988using,hyman1989effect}  that distribute the population based on their risk, such as the number of partners  \cite{hyman2001initialization,hyman2003modeling,hyman1988using,hyman1989effect}. 

We design the model with a complete explanation of its variables and parameters. For the parameters, we used two different data sources for population distribution and amount of condom-use by people with different risks. We use local sensitivity analysis to identify the relative importance of condom-use  and illustrate how this analysis can be used to prioritize individual-level  behavioral strategies based on their predicted effectiveness.

\section{Ct Transmission Model Overview}

We model a population of $15$-$25$ year-old sexually active individuals and assume that the primary  mechanism for migration is by aging into, and out of, the population. 
 We assume a closed steady-state  population $N(r) = S(t,r) + I(t,r)$  of people with  risk $r\in [r_0,r_{\infty}]$ is divided into $S(t,r)$, and $I(t,r)$, where $S(t,r)$ ($I(t,r)$) is the number of susceptible (infected) people with risk $r$ at time t. 
The susceptible population becomes infected at  the rate of $\lambda$ per year, and infected population recovers  with constant rate $\gamma$
to again become susceptible. 
 We assume both susceptible and infected  people leave the population at the migration rate $\mu$ per day and that people maintain the same risk $r$ while in the modeled population. 
Our integro-differential equation model for the spread of Ct is
\begin{equation}\label{con-model}
\begin{split}
\frac{\partial S(t,r)}{\partial t}&=\mu (N(r)-S(t,r))-\lambda(t,r)S(t,r)+\gamma I(t,r),\\
\frac{\partial I(t,r)}{\partial t}&=\lambda(t,r)S(t,r)-\gamma I(t,r)-\mu I(t,r),\\
S(0,r)&=S_0(r),~~~~I(0,r)=N(r)-S_0(r),
\end{split}
\end{equation}
where initial distributions of the susceptible and infected population are given at time $t=0$.
Note that this model does not distinguish between men and women and is appropriate for homosexual STIs or infections when the distribution of risk  and infection incidence in men and women is approximately the same.  This also requires that the probability of transmitting the infection from an infected man to a susceptible woman is approximately the same as the probability of transmission  from an infected woman to a susceptible man. This is a reasonable assumption for some STIs including Ct.  In the absence of symmetry in the transmission parameters or in the risk behavior in men and women, then the model would need to be extended to a two-sex bipartite model. 

We model a population of $15$-$25$ year-old sexually active individuals and assume that individuals enter and leave the modeled population only through aging, that is, migration rate is defined as $\mu = [(25-15)~\text{years}]^{-1}= 1/(10$ years)$=1/(3650$ days).  We also assume  that everyone aging into the population is susceptible to infection, and that people do not change their risk while in the modeled population. 
To properly account for changes in risk behavior as the population ages would require adding an additional variable (age) and is beyond the scope of this model.
The risk behavior is distributed in a way that number of people with risk $r$ decreases as risk $r$ increases, that is, there are fewer individuals with many partners. We also assume that there is an exponential distribution for the rate the infected population with an average infection period  $1/\gamma$ days.  

\subsection{Transmission rate }
The force of infection, or transmission rate, $\lambda(t,r)$, for susceptible person with risk $r$ at time $t$, is the rate that susceptible people with risk $r$ become infected through sexual act.   The mixing among people with different risks  determines if a susceptible person with  risk $r$ can be infected by someone infected with risk  $r'$.
We define $\lambda(t,r)$ as the integral of the rate of
disease transmission at time $t$ from each infected person with risk $r'$, $I(t,r')$, to the susceptible one by
\begin{eqnarray}
\lambda(t,r)=\int_{r_0}^{r_\infty}\tilde{\lambda}(t,r,r')dr'.
\end{eqnarray}
The rate of disease transmission from the
infected persons  with risk $r'$ to the susceptible individuals with risk $r$, $\tilde{\lambda}(t,r,r')$, 
 is defined as the product of three factors:
 \begin{align*}
\tilde{\lambda}(t,r,r') &= \left( \begin{array}{c}\mbox{Number of $r'-$risk } \\
                                       \mbox{ partners of susceptible  } \\
                                       \mbox{with risk r, per year}
                       \end{array}
                \right)
                \left( \begin{array}{c}\mbox{ Probability of}\\
                                       \mbox{disease transmission} \\
                                       \mbox{ per partner}
                       \end{array}
                \right)
                \left( \begin{array}{c}\mbox{Probability that} \\
                                       \mbox{partner with risk $r'$} \\
                                       \mbox{ is infected}
                       \end{array}
                \right) \\ 
                &= \hspace{2cm}p(r,r')\hspace{4cm}\beta(r,r')\hspace{3cm}P_I(t,r')~~,
\end{align*}
where
\begin{itemize}
\item $p(r,r')$ is the partnership mixing function defined as the number of sexual partners per day that a person with risk $r$   has with a person with risk $r'$, and 
\item $\beta(r,r')$ is the probability of disease transmission per partner to a susceptible person  with risk $r$ from their infected partner with risk  $r'$,  and 
\item $P_I(t,r')$  is the probability  that a person of risk $r'$ is infected.  Here we assume that there is random mixing among individuals with the same risk, $P_I(t,r')=\frac{I(t,r')}{N(r')}$. 
\end{itemize}

The Figure (\ref{chart}) shows a  diagram of components of the transmission rate $\lambda$, and in the following sections, all  components will be explained.

\begin{figure}[h]
\centering
\par\medskip
\begin{tikzpicture}[
    grow=right,
    level 1/.style={sibling distance=4.5cm,level distance=5.2cm},
    level 2/.style={sibling distance=2cm, level distance=5.4cm},
    edge from parent/.style={very thick,draw=blue!40!black!60,
        shorten >=5pt, shorten <=5pt},
    edge from parent path={(\tikzparentnode.east) -- (\tikzchildnode.west)},
    kant/.style={text width=4cm, text centered, sloped},
    every node/.style={text ragged, inner sep=2mm},
    punkt/.style={rectangle, rounded corners, shade, top color=white,
    bottom color=blue!50!black!20, draw=blue!40!black!60, very
    thick }
    ]
\node[punkt, text width=11em] {$\lambda(t,r)=\int_{r_0}^{r_\infty}\tilde{\lambda}(t,r,r')dr'$};
\end{tikzpicture}
\hspace*{-1.5cm}
\begin{tikzpicture}[
    grow=right,
    level 1/.style={sibling distance=4.5cm,level distance=5.2cm},
    level 2/.style={sibling distance=2cm, level distance=5.4cm},
    edge from parent/.style={very thick,draw=blue!40!black!60,
        shorten >=5pt, shorten <=5pt},
    edge from parent path={(\tikzparentnode.east) -- (\tikzchildnode.west)},
    kant/.style={text width=4cm, text centered, sloped},
    every node/.style={text ragged, inner sep=2mm},
    punkt/.style={rectangle, rounded corners, shade, top color=white,
    bottom color=blue!50!black!20, draw=blue!40!black!60, very
    thick }
    ]
    
\node[punkt, text width=4em] {$\tilde{\lambda}(t,r,r')$}
     child { 
     node[punkt, text width=3em] {$P_I(t,r')$}
     edge from parent
             node [kant, above,] {probability that} node [kant, below] {partner is infected}
     }
    child {
        node[punkt, text width=3em] {$\beta(r,r')$}
        child {
        node[punkt] [rectangle split, rectangle split, rectangle split parts=1,text ragged] {$C(r,r')$}
        child {
        node[punkt] [rectangle split, rectangle split, rectangle split parts=1,text ragged] {$c(r)$}
        edge from parent
                         node [kant, above] {\% of acts per} 
                         node [kant, below] {\small partner with condom} 
    }
        edge from parent
             node [kant, above] {~~~~~\% of total acts } node [kant, below] {with condom} 
    }
    child {
        node[punkt] [rectangle split, rectangle split, rectangle split parts=1,text ragged] {$A(r,r')$}
        child {
        node[punkt] [rectangle split, rectangle split, rectangle split parts=1,text ragged] {$a(r)$}
        edge from parent
                         node [kant, above] { \# of acts}
                         node [kant, below] { per partner}
    }
        edge from parent
                         node [kant, above] {total \# of acts} 
    }
    edge from parent
                         node [ above] {probability of  } node [ below] {transmission}
             }
        child {
        node[punkt, text width=3em] {$p(r,r')$}
        child {
        node[punkt, text width=3em] {$\rho(r,r')$}
        child {
           node [punkt,rectangle split, rectangle split,
            rectangle split parts=1] {
              \textbf{$\rho_{bm}(r,r')$} }
            edge from parent
             node[above, kant] {biased mixing}
                     }
        child {
          node [punkt,rectangle split, rectangle split,
           rectangle split parts=1] {
              \textbf{$\rho_{rm}(r,r')$}
                       }
                 edge from parent
              node[above, kant] {\small ~~~~~~random mixing}
                   }
        child {
            node [punkt,rectangle split, rectangle split,
            rectangle split parts=1] {
                \textbf{$\epsilon$}
                        }
                edge from parent
              node[above, kant] {preference level}
                    }
        edge from parent
                        node [ above] {mixing function}}
                    edge from parent{
                node [kant, above] {\# of partners} node [kant, below] {with risk $r'$} }
               };
\end{tikzpicture}
\caption[\textbf{Components of the transmission rate} $\pmb{\lambda(t,r)}$ ]{Components of the transmission rate  {$\lambda(t,r)$} for a susceptible person with risk {$r$}.}
\label{chart}
\end{figure}
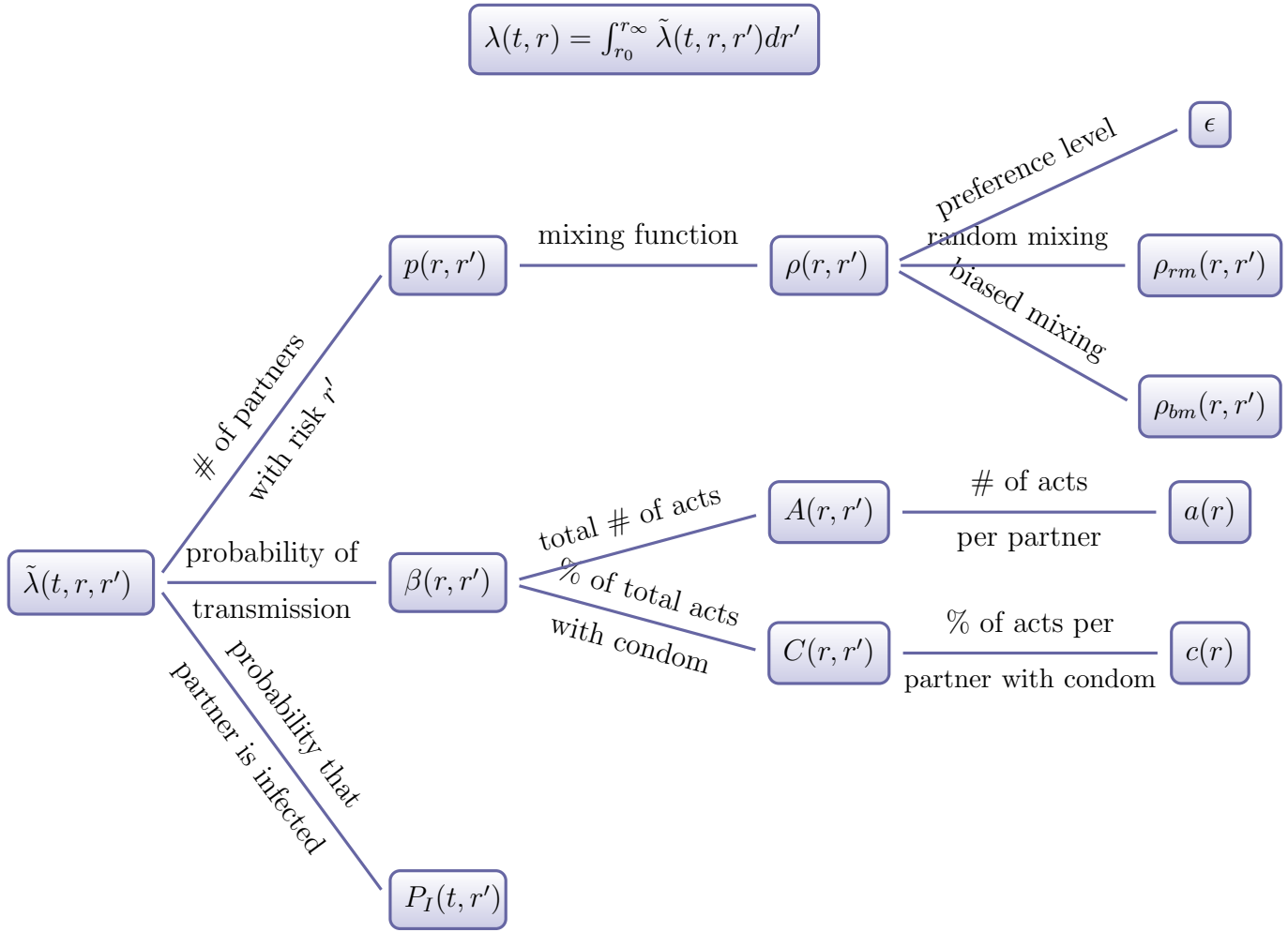

\subsection{Partnership formation }
The mixing distribution,  $\rho(r,r')$, captures the mixing between people of different risks. This distribution is  defined as the fraction  of  partners of a person with risk $r$  who have risk $r'$. 
The distribution function  $\rho(r,r')$ is the expected distribution of partners and is typically estimated based on inaccurate survey data or other assumptions.  It cannot be as the actual mixing function $p(r,r')$ since it usually will  not satisfy the balance condition:
$$N(r)p(r,r')= N(r')p(r',r), $$
that means the total number of people with risk $r$ with partners of risk $r'$ must be equal to the total number of people with risk $r'$ with partners of risk $r$.  

We assume that $\rho(r,r')$ is a linear combination of randomly selected partners, with the random mixing distribution $\rho_{rm}(r,r')$, and partners based on their preference, with the biased mixing distribution $\rho_{bm}(r,r')$. These mixing distribution functions $\rho_{rm}$ and $\rho_{bm}$  are normalized to have unit integral.
Feng et al. \cite{feng2016evaluating} used a similar model to account for multi-level mixing of people within a specified group and among the general population. 

\underline{\textbf{Random mixing distribution}}:
When the mixing is random (sometimes called proportional mixing), then individuals with risk $r$ 
 do not show any preference  for their partners based on risk. 
The random mixing function for the probability that a person of risk $r$ picks a partner with risk $r'$ is defined by the ratio of  total number of partners for all people with risk $r'$, $r' N(r')$,  to  total number of partnerships, $\int_{r_0}^{r_\infty}uN(u)du$.  Thus, the random mixing distribution 
\begin{equation} \label{randommixing}
\rho_{rm}(r,r')=\frac{r'N(r')}{\int_{r_0}^{r_\infty}uN(u)du}~~,
\end{equation}
is independent of the risk $r$ of the person seeking a partnership. 

\underline{\textbf{Biased mixing distribution}}:
In our biased (associative or preferential) mixing model, we assume homophily (love of the same) where people with risk $r$ prefer to have partners with similar risk. 
We also assume that  people at high risk  have partners with a broader range of risk than people at low risk.
That is, the standard deviation, $\sigma(r)$,  for the distribution of risk of partners of a person with risk $r$ is an increasing function of $r$.   This is in agreement with the study by 
Lescano et al. \cite{lescano2006condom} that observed the partners 
of people with many partners are mostly casual partners with few acts (sexual acts) per partnership. 
They also observed that the partners of people with few partners are more often longer term partners  with more acts per partnership.

We  define the  biased mixing distribution $\rho_{bm}(r,r')$ for the probability that a person with risk $r$ prefers to have a partner with risk $r'$ from the range $r' \in [r-\sigma(r), r+\sigma(r)]$ as
\begin{equation} \label{biasedmixing}
\rho_{bm}(r,r') = \left\{
        \begin{array}{ll}
            \frac{-|r'-r|+\sigma(r)}{\sigma(r)^2} & \quad  |r'-r| \leq \sigma(r) \\
            0 & \quad elsewhere, 
        \end{array}
    \right.
\end{equation}
which satisfies the condition $\int_{-\infty}^{\infty}\rho_{bm}(r,r')dr'=1.$ The Figure (\ref{delta}) shows how the biased function $\rho_{bm}$ is wider for the higher risk groups. 
\begin{figure}[htp]
\centering
\includegraphics[scale=.5]{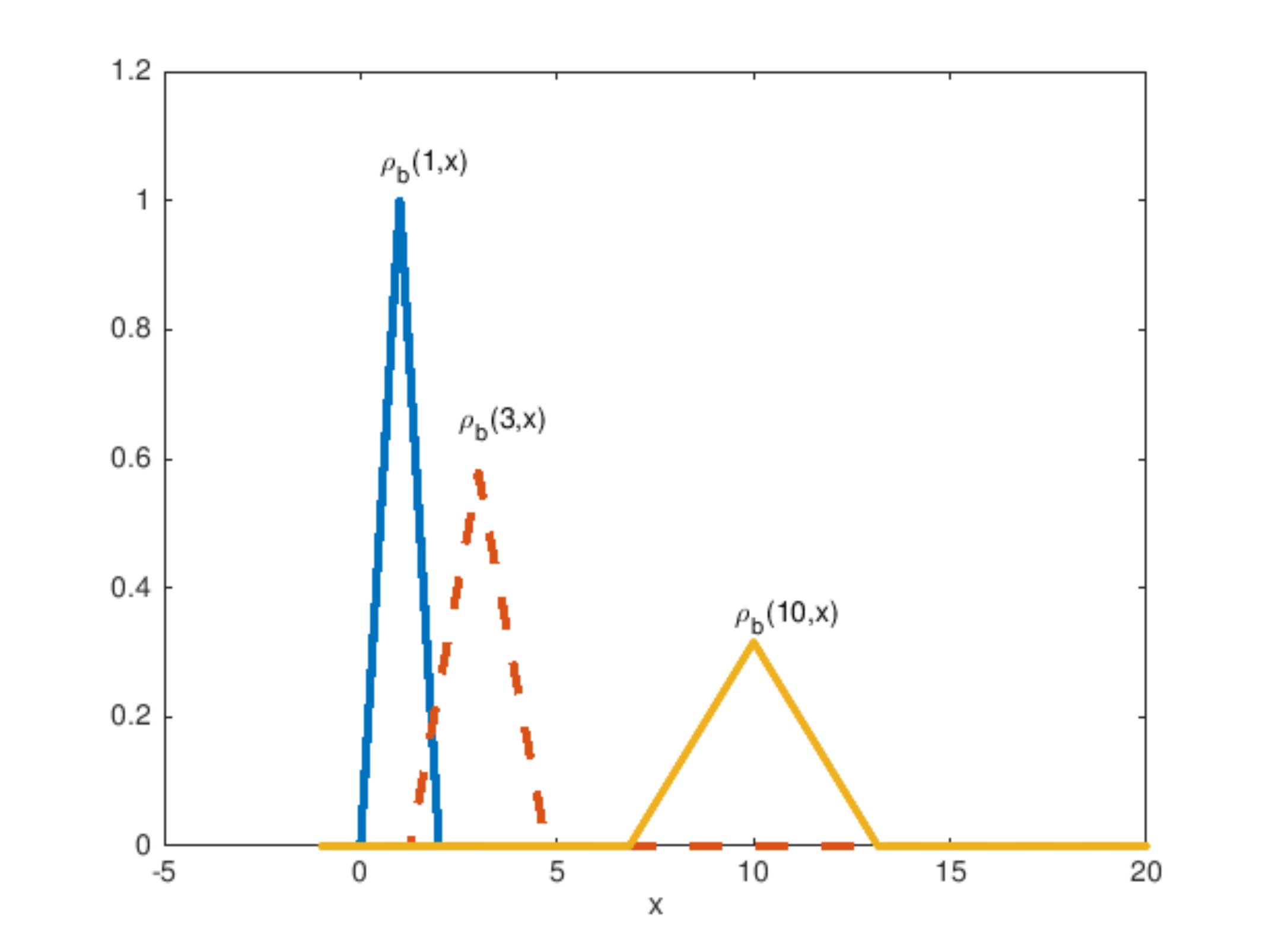}
\caption[\textbf{Biased mixing function}]{Plot of the triangle (hat) biased mixing function $\rho_{bm}(r,x)$ for $r=1,3,10$. As the risk $r$ increases, the mixing function becomes fatter and shorter to capture the effect that partners of higher-risk people have a broader range of risk than that of lower-risk people. This is similar to the mixing function used by Hyman et al. \cite{hyman1988using}.
}
\label{delta}
\end{figure}

\underline{\textbf{Combination of random and biased mixing distributions}}:
We assume people choose some of their partners based on their preference (biased mixing) and that they have other partners chosen randomly from the whole population (random mixing). 
We define the preference level $\epsilon$ as fraction of partners of a person with risk $r$ are selected preferentially and the rest are selected randomly, then we can express the mixing distribution as a convex combination of $\rho_{rm}$ and $\rho_{bm}$:
\begin{eqnarray}\label{E:risk}
\rho(r,r')=\epsilon\rho_{bm}(r,r')+(1-\epsilon)\rho_{rm}(r,r').
\end{eqnarray}
When $\epsilon=0$ the mixing is random, and when $\epsilon=1$ it is purely biased mixing. 
Otherwise, a person with risk $r$ chooses an $\epsilon$ fraction of his/her partners with a hat distribution of people with risk $r'\in[r-\sigma(r),r+\sigma(r)]$, and chooses the other partners randomly from all risk value groups. 
 
\underline{\textbf{Partnership mixing function}}: 
The partnership function $p(r,r')$ is the number of partners a person with risk $r$ has with someone of risk $r'$ per year.  
A person with risk $r$ wants to have $r\rho(r,r')$ partners with risk $r'$, therefore, all individuals with risk $r$ want to  have $r\rho(r,r')N(r)$ partners with risk $r'$. On the other hand, all individuals with risk $r'$ want to  have $r'\rho(r',r)N(r')$ partners with risk $r$. 
The balance condition states that if  people of risk $r$ have $P(r,r')$ partners with risk $r'$, then the people with risk $r'$ must have $P(r',r)=P(r,r')$ partners with risk $r$. 
Therefore, we define \textbf{actual} number of partnership between people with risk $r$ and people with risk $r'$ as harmonic average of $r\rho(r,r')N(r)$ and $r'\rho(r',r)N(r')$:
\begin{eqnarray}
P(r,r')\defeq2\frac{r\rho(r,r')N(r)\times r'\rho(r',r)N(r')}{r\rho(r,r')N(r)+r'\rho(r',r)N(r')}.
\end{eqnarray}
The distribution $P(r,r')$ is a compromise for the actual number of partnerships between all people with risk $r$ and all people with risk $r'$. Therefore, the actual number of  partners  that a  person with risk $r$ has with people of risk $r'$ is 
\begin{eqnarray}
p(r,r')=\frac{P(r,r')}{N(r)}.
\end{eqnarray}
\textit{Remark:} Harmonic average of two values is closer to the smaller one and this compromise weights the decision on forming a sexual partnership towards the person who is less interested in making partnership.

 \subsection{Probability of transmission per partner}
The probability  per partner, $\beta(r,r')$,  that a susceptible person of risk $r$ becomes infected by an infected partner of risk $r'$ depends upon the number of acts (sexual acts) between the two risk groups, $A(r,r')$, and how often condoms are used in their acts, $C(r,r')$.  

\underline{\textbf{Sexual acts per partnership between risk groups}}:\label{contacts}
We define $A(r,r')$ as the total number of sexual acts per person per day between a person with risk $r$ and a partner with risk $r'$.  Since there must be the same as the number of sexual acts between person of risk $r'$ with partner of risk $r$, the balance condition, $A(r',r)=A(r,r')$ must hold.
Suppose a person with risk $r$  desires to have, on average,  $a(r)$ sexual acts per partner per day.   We assume that $a(r)$ is a decreasing function of $r$:
\begin{equation}
a(r)=\frac{A}{r}, 
\end{equation}
where $A$ is the total number of sexual acts  per day.
Because the number of desired sexual acts per partner for people of risk $r$ is not necessarily the same as the number of desired sexual acts per partnership for people of risk $r'$, $a(r) \ne a(r')$, then there must be a compromise for the balance condition to hold. 
We define the actual number $A(r,r')$ of sexual acts per person between the people in risk groups $r$ and $r'$ as
\begin{equation}\label{E:Ar}
 A(r,r')\defeq 2\frac{a(r)a(r')}{a(r)+a(r')}  ~~.
\end{equation}
The Equation (\ref{E:Ar}) satisfies the balance condition, and when there is a conflict, the harmonic average results in the actual number of sexual acts to be closer to the smaller number desired by the two individuals. 

\underline{\textbf{Condom-use as a function of risk}}: 
We assume that person with risk $r$ desires to use a male-latex condom in $c(r)$ fraction of their sexual acts. We acknowledge that increased condom-use might have an effect on the risk behavior, however, this is not  investigated in this work. We assume  that higher-risk people are more likely to use condoms than the lower-risk people \cite{beadnell2005condom, lescano2006condom}.  Therefore, we define $c(r)$ as increasing function of $r$.  We observed that the function
\begin{eqnarray}\label{E:cr}
c(r)\defeq \alpha \left(\frac{r}{c_0+  r}\right),
\end{eqnarray}
is a good approximation to survey data and interpolates between the case where people have no partners (hence no condom-use), $\lim_{r\to0} c(r)=c(0)=0$, and the limit where people have many partners and use condoms $\alpha=\lim_{r\to\infty} c(r)$ fraction of acts.

In  Section \ref{contacts}, we described how the actual number of sexual acts between people of different risk groups had to be compromised to satisfy a balance condition.  The same is true for condom-use.  We define the actual fraction of times that a person of risk $r$ uses a condom when having sex with a person of risk $r'$ as $C(r,r')=C(r',r)$ and define this by an appropriate average of $c(r)$ and $c(r')$. The average will depend if the preference (final decision) is closer to the desired condom-use of the person who prefers to use condoms fewer times, or the person who prefers to use condoms more often.

\underline{Preference to low condom-use}:
In this case, we assume that a person who is less likely to use condom is more likely to convince the other not to use condom.  We approximate this situation for partners with risk $r$ and $r'$ to use a condom in 
\begin{eqnarray}\label{har}
C_l(r,r')=\frac{2c(r)c(r')}{c(r)+c(r')},
\end{eqnarray}
fraction of their acts.

\underline{Preference to high condom-use}:
In this case, a person who is more likely to use condom is more probable to convince the other one to use condom.
We approximate this situation by taking the harmonic average of the fraction of acts people do not use condom ($1-c(r)$ and $1-c(r')$), and therefore, they use condoms in   
\begin{eqnarray}\label{1-har}
C_h(r,r')=1-\frac{2(1-c(r))(1-c(r'))}{2-c(r)-c(r')}~~,
\end{eqnarray}
fraction of their acts. 

\underline{\textbf{The probability of transmission with condom-use}}:
We define $\beta_{nc}$ and $\beta_c$ as the  probabilities of transmission per act for not using and using a condom, and we assume these probabilities are gender-independent, because unlike the heterosexual transmission of HIV/AIDS, the probability of highly infectious STIs (like chlamydia and gonorrhea) transmission  from an infected man to a woman is approximately the same as from an infected woman to a man \cite{althaus2011transmission,quinn1996epidemiologic, turner2006developing}.  
If the condom is $90\%$ effective in preventing the infection from being transmitted, then probability of transmission when using a  condom-use is  $\beta_c=0.1\beta_{nc}$.  

To determine the probability of a susceptible person with risk $r$  being infected by their infected partner with risk $r'$ depends on the number of acts, $A(r,r')$, and how often they use condoms. 
If someone uses a  condom in $C(r,r')$ fraction of acts, then they have a total of $C(r,r')A(r,r')$ acts with condoms and $(1-C(r,r'))A(r,r')$ acts without  condom per unit time.  The person with risk $r$ does not catch infection from their partner during a condom act with probability $(1-\beta_c)^{C(r,r')A(r,r')}$, and for when not using a condom this probability  is $(1-\beta_{nc})^{(1-C(r,r'))A(r,r')}$. Combining these,  the probability of a susceptible being infected after one act by infected partner with risk $r'$ is
\begin{eqnarray}
\beta(r,r')=1-(1-\beta_c)^{C(r,r')A(r,r')}(1-\beta_{nc})^{(1-C(r,r'))A(r,r')}.
\end{eqnarray}

\section{Parameter Estimation}\label{p_est}
The model parameters in Table (\ref{table:parameter}), the distribution of risk in the population, and the condom-use were estimated from recent studies on sexual behavior.

\begin{table}[htbp]
\centering
\resizebox{\columnwidth}{!}{
\begin{tabular}{ lp{9cm}llll }
\toprule[1.5pt]
 \normal{\head{Parameter}} & \normal{\head{Description}}
  & \normal{\head{Unit}} & \head{Baseline}& \head{Ref.}\\
  \cmidrule(lr){1-3}\cmidrule(l){3-5}
\centering$\int N(r) dr$ & Total population.  &   people & $10000$  &  Assumed \\
\centering\vspace{.2mm}$A$& Total (max) number of acts  per time. &   1/day   & $0.57$  & Assumed \\
 $1/\gamma$& Average time to recover without treatment.   & days&  $365$   & \cite{kretzschmar2001comparative}\\
$\mu$ & Migration rate. & 1/days& $0.00$& Assumed\\

  $\beta_{nc}$ & Probability of transmission per no-condom  act. & 1/act & $0.11$& \cite{kretzschmar2001comparative}\\
  \centering\vspace{.7mm}$\alpha$ &  Fraction of acts condom used by risky people.&--&  $0.70$ & Estimated\\
$\beta_{c}$ & Probability of transmission per condom  act. & 1/act & $0.01$& Assumed\\       
\vspace{.4mm}\centering $r_0(r_\infty)$& Minimum(maximum) number of partners per time. & people/days & $0.00$$($0.14$)$ & Assumed  \\
\centering$\epsilon$ & Preference level. & -- & $0.60$ & Assumed \\
        \bottomrule[1.5pt]
\end{tabular}}
\caption[\textbf{Model parameters}]{Model parameters: parameter values are chosen for all simulations unless indicated otherwise.}
\label{table:parameter}
\end{table}

\subsection{Population distribution}
A sample of  $616$ people  ages $15$-$25$ years old resident in Orleans Parish were asked about their number of concurrent partners \cite{kissinger2014expedited}\footnote{We will explain these data in  Chapter. (\ref{abm})}. This data was in agreement with other recent studies \cite{colgate1989risk}, that show that the partner distortion often follows an inverse cubic power law, $N(r) \propto r^{-3}$, for $r > r_0$.  The value of $N(r)$ is chosen for the function to agree with the total population size, $\int N(r) dr$, being modeled.

\subsection{Condom-use}
  The distribution of risk and condom-use  were estimated based on surveys for the 
sexually active adolescents and young adult populations \cite{beadnell2005condom,lescano2006condom,reece2010condom}. 
Reece et al. \cite{reece2010condom}  studied  rates of condom-use among sexually active individuals in the U.S. population and observed that adolescents reported condom-use during $79.1\%$ of the past $10$ vaginal intercourse events.
Similar studies \cite{WinNT} in sexually active  high school students in the U.S. reported that during $1991$, $46\%$,  during $2003$, $63\%$, and  in $2013$, $59\%$  of the students used condoms at their most recent sexual intercourse. 

Beadnell et al. \cite{beadnell2005condom}  surveyed $8-12$th grade students in a large urban northwest school district  annually for seven years. They observed that the younger students were  more likely to use condoms and also the students with more partners were more likely to use condoms:  the students with many partners used condoms, on average,  in  $68\%$ of their sexual acts, while the students with  few partners used condoms in $49\%$ of their sexual acts. 
The condom-use function, Equation  \ref{E:cr},  is in close agreement with their observations (Table (\ref{table:data})) with the parameters $\alpha=0.69$ and $c_0=1.35$:
\begin{equation}\label{cdata}
c(r)=0.69\frac{r}{1.35+ r}~~.
\end{equation}
A simple check shows this function is in close agreement with the survey data: $c(2)=0.41$, $c(4.7)=0.50$, $c(5.4)=0.58$, and $c(7.4)=0.60$.

\section{Numerical Simulations}

\begin{table}[tb]
\centering
\hspace*{-1.0cm}
\scalebox{0.8}{
\begin{tabular}{l|l|l|l|l|l}
\toprule
 \specialrule{.1em}{.1em}{.1em} 
 \multicolumn{2}{l|}{\hspace{1cm}{$15-16$ years old}}&\multicolumn{2}{l|}{\hspace{1cm}{$16-17$ years old}}& \multicolumn{2}{l}{\hspace{1cm}{$17-18$ years old}}\\
\cline{1-6}
 {Risk}~ $r$&{Fraction of condom-use}&{Risk}~ $r$&{Fraction of condom-use}&{Risk}~ $r$&{Fraction of condom-use}\\
\hline
~~~$1.2$&\hspace{1cm}$0.42$&~~~$1$&\hspace{1cm}$0.42$&~~~$1$&\hspace{1cm}$0.39$\\
~~~$5.4$&\hspace{1cm}$0.58$&~~~$2$&\hspace{1cm}$0.42$&~~~$2$&\hspace{1cm}$0.38$\\
&&~~~$7.4$&\hspace{1cm}$0.60$&~~~$4.7$&\hspace{1cm}$0.50$\\
\hline
\end{tabular}}
    \caption[\textbf{Condom-use fraction based on age and risk}]{The average fraction of condom-use by high school students with different risks and different ages, the result of survey conducted in a large urban northwest high school \cite{beadnell2005condom}.}
  \label{table:data}%
\end{table}

Because the equations are homogeneous in the total population, our results scale with the total population size.  We display our numerical simulations in terms of the nondimensional variables defined by dividing each variable by the steady-state zero-infection equilibrium the total population of individuals with the risk $r$, $N(r)$. That is, we present the numerical simulations in terms of the fraction of the population at risk $r$, i.e susceptible $s(t,r)\defeq\frac{S(t,r)}{N(r)}$ or infected $i(t,r)\defeq\frac{I(t,r)}{N(r)}$.   We define $I^*(r)$ as the number of and $i^*(r)$ as the fraction of the population that is infected at the endemic steady state.
In the numerical simulations, all the parameters are fixed with the baseline values given in Table (\ref{table:parameter}), unless specifically defined otherwise. 

\subsection{Basic reproduction number $\mathcal{R}_0$}
  When the population is distributed as a function of risk, then it is possible to define a basic reproduction number for each value of risk, or a single  $\mathcal{R}_0$ for the entire population based on the  dominant eigenvalue of next generation operator.  Using a single $\mathcal{R}_0$ is useful when studying the impact that changes in the biased mixing and condom-use parameters have on the early growth of an epidemic.  

We follow Diekmann et al. \cite{diekmann1990definition} and 
define  $\mathcal{R}_0$ as the spectral radius of the next generation operator defined as
\begin{equation}
K(r)=S(t,r)\int_{r_0}^{r_\infty}\tau {p(r,r')\beta(r,r')}I(t,r')dr',
\end{equation}
where $\tau = \frac{1}{\mu+\gamma}$ is average time that a person is infected and  $\tau {p(r,r')\beta(r,r')}$
is the expected number of people with risk $r$ will be infected by  a single infected person with risk $r'$. 
Thus, the next generation operator, $K(r)$,  is number of secondary cases for over all the infected people with risk $r'$, $I(t,r')$, and is found by integrating over all possible risk groups.  That is,  $K(r)$ is the number of secondary cases with risk $r$ that arises from all the infected people $I(t,r')$. 
The basic reproduction number $\mathcal{R}_0$ is the dominant eigenvalue of $K(r)$.

We first partition our integro-differential equation model (\ref{con-model}) into subdomains for different risk groups, 
$[r_0,r_\infty]=\cup_{i=1}^{n}[r_{i-1},r_i]$, where $r_n=r_\infty$ and define  the populations on for each risk group as $I_i(t)=\int_{r_{i-1}}^{r_i} I(t,r') dr'$ and $S_i(t)=\int_{r_{i-1}}^{r_i} S(t,r)dr$.
The equations can then be expressed as
\begin{align*} \label{partmodel}
\frac{\partial S_i(t)}{\partial t}=&\mu (N_i-S_i(t))-\lambda_i(t)S_i(t)+\gamma I_i(t),\\
\frac{\partial I_i(t)}{\partial t}=&\lambda_i(t)S_i(t)-\gamma I_i(t)-\mu I_i(t),
\end{align*}
where $\lambda_i=\sum_j \int_{r_{i-1}}^{r_i} p(r_i,r_j)\beta(r_i,r_j)I_jdr_j$.

We divide the equations by $N(r)$ and approximate the next generation operator $K(r)$ with the $n$-by-$n$ next generation matrix $\mathcal{K}$ based on assuming the populations are approximately constant within each risk group and that the population is at the zero-infection equilibrium, $s_i=1$. The entries of $\mathcal{K}$  are defined by
\begin{equation}
k_{ij}=\int_{r_{i-1}}^{r_i}  \tau p(r_i,r_j)\beta(r_i,r_j) dr_j~~.
\end{equation}
The basic reproduction number $\mathcal{R}_0$  defined as the dominant eigenvalue of $\mathcal{K}$, is calculated numerically. 

The Figure (\ref{r0})  illustrates how $\mathcal{R}_0$ increases as the amount of biased mixing $\epsilon$ increases. 
When a new infection is introduced into the population, if there is even a slight amount of random mixing, someone in the high-risk population will quickly become infected \cite{hyman1988using}. Once this happens, then if the mixing is highly-biased (large $\epsilon$)  these  infected high-risk people will infect other high-risk people and the epidemic will grow rapidly (large $\mathcal{R}_0$).  If the mixing is close to random mixing (small $\epsilon$), then many of the secondary infections from the early high-risk infected people will have low-risk and the epidemic will grow slower (smaller $\mathcal{R}_0$).
The extreme sensitivity of $\mathcal{R}_0$ to  $\alpha$ also is an indication of the importance of educating high-risk individuals in consistent condom-use to prevent infecting others, and the need of the low-risk population in using condoms to protect themselves from infection.

\begin{figure}[t]
\centering
\includegraphics[width=.35\paperwidth]{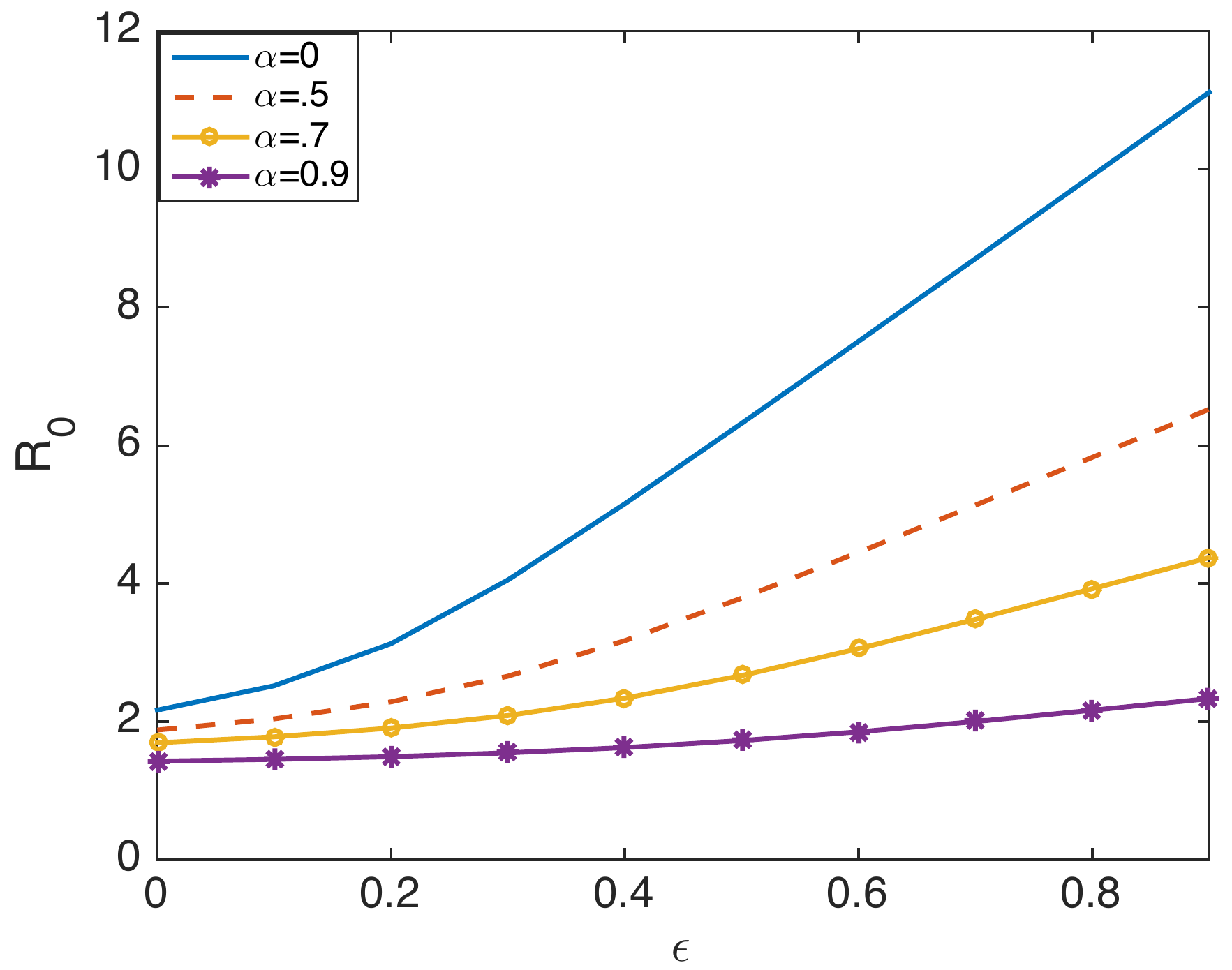}
\caption[$\pmb{\mathcal{R}_0}$ \textbf{vs preference level} $\pmb{\epsilon}$]{Basic reproduction $\mathcal{R}_0$ versus preference level $\epsilon$ for different condom-uses: the impact of $\alpha$ on $\mathcal{R}_0$ depends on mixing, for more biased mixing $\alpha$ has more impact on preventing the infection, however for less biased mixing, the impact of $\alpha$ decreases. As $\alpha$ decreases $\mathcal{R}_0$ increases much faster at bigger $\epsilon$s than smaller ones.}
\label{r0}
\end{figure}

\subsection{Endemic equilibrium}
The fraction of the population that are infected at the endemic equilibrium infection, {$i^*$}, depends upon the distribution of risk, $N(r)$, the mixing between people of different risk behaviors, as measured by $\epsilon$ in Equation (\ref{E:risk}), and the fraction of the acts condom used by high-risk people, as measured by $\alpha$ in Equation (\ref{E:cr}).  

The Figures \ref{fig:eps_a} - \ref{fig:eps_f} plot  the endemic infection distribution as a function of  risk $r$ for
\begin{enumerate}
\item Random mixing where $90\%$ of the partners are chosen randomly form the population, i.e $\epsilon=0.1$,
\item Balanced mixing  where all but $60\%$ of the partners have similar risk behavior, i.e $\epsilon=0.6$, and
\item Highly biased mixing  where all but $90\%$  of the partners have similar risk behavior, i.e $\epsilon=0.9$.
\end{enumerate}    
For all values of risk, the fraction of infected population  at steady state, $i^*$, decreases as  condom-use, $\alpha$, increases.  
In the Figures \ref{fig:eps_a}, \ref{fig:eps_c}, and \ref{fig:eps_e}, the $\alpha$ axis is between $\alpha=1$ where the high-risk population uses condoms all the time, to $\alpha=0$ where condoms are never used.  The fitted value $\alpha=0.69$ agrees with Beadnell et al. \cite{beadnell2005condom} studies. 
For low condom-use (small values of $\alpha$), $i^*$  increases with $r$ indicating that a higher percentage of the high-risk people are infected than the low-risk people.   
For most values of condom-use, $\alpha< 0.95$, having more partners (increase one's risk $r$), increases the likelihood of being infected.
\begin{figure}[htp]
\centering
\subfloat[$\pmb{3D~ \epsilon=0.1}$][]{\label{fig:eps_a}\includegraphics[width=0.45\textwidth]{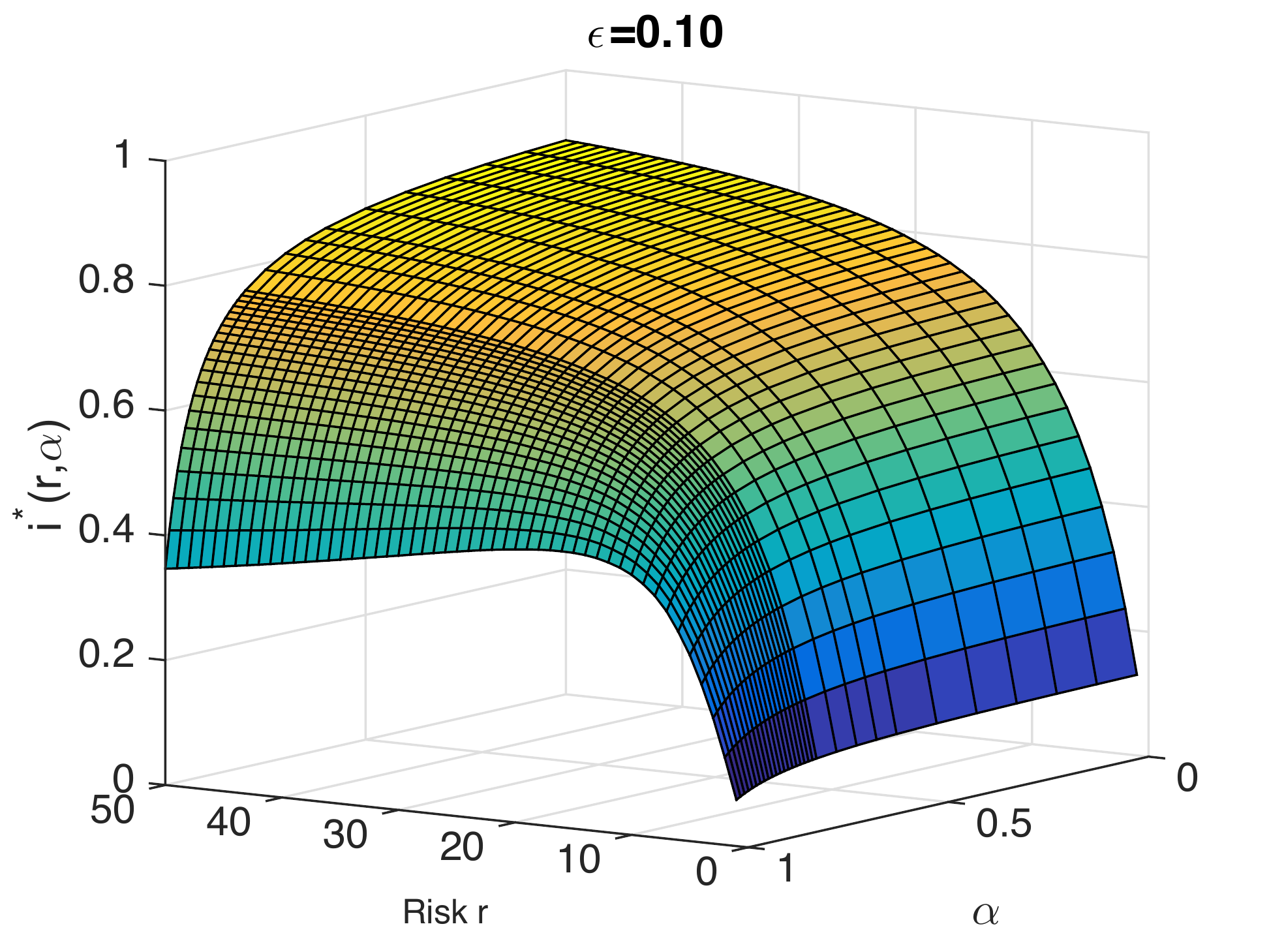}}
\subfloat[$\pmb{2D~\epsilon=0.1}$][]{\label{fig:eps_b}\includegraphics[width=0.45\textwidth]{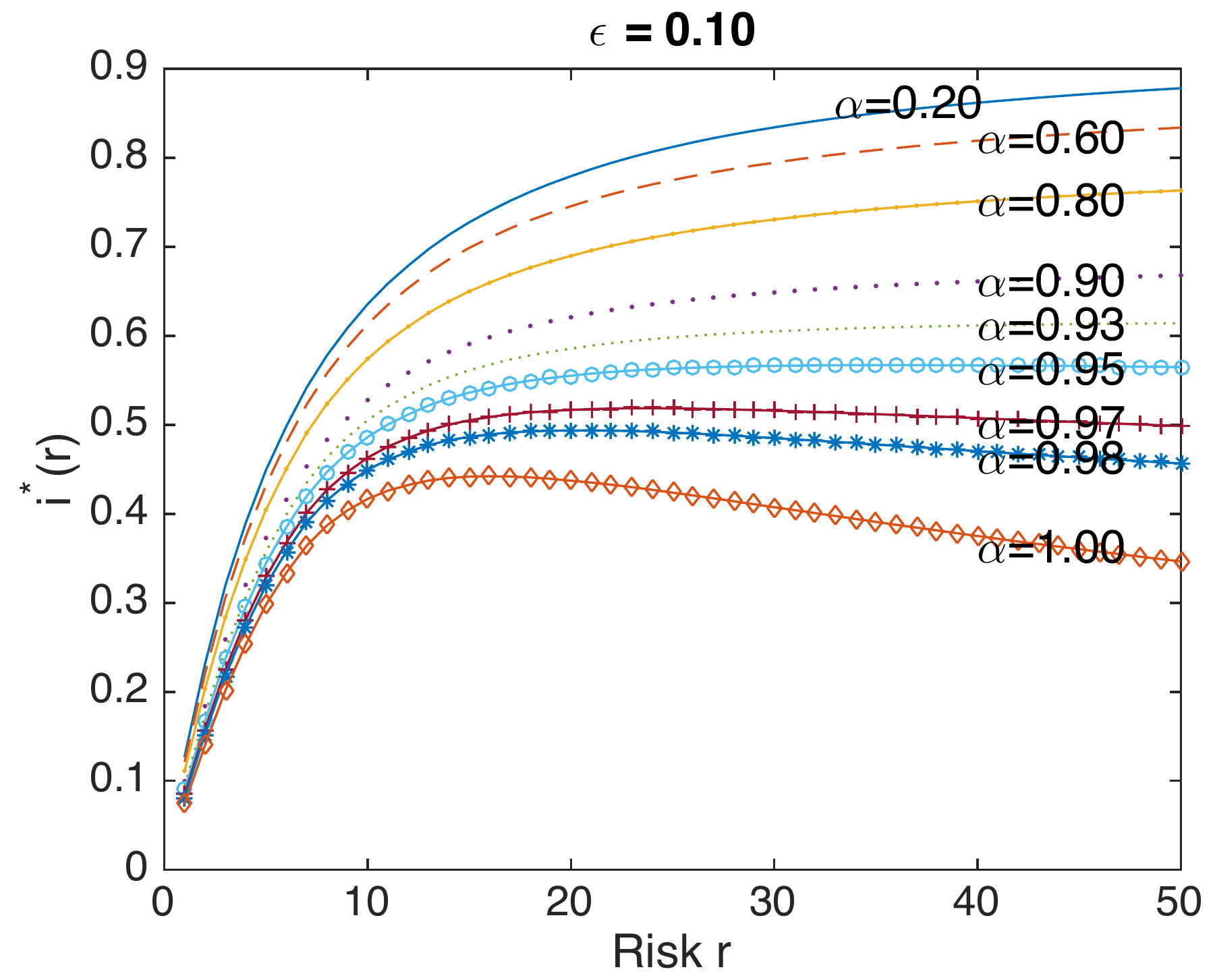}}\\
\subfloat[$\pmb{3D~\epsilon=0.6}$][]{\label{fig:eps_c}\includegraphics[width=0.45\textwidth]{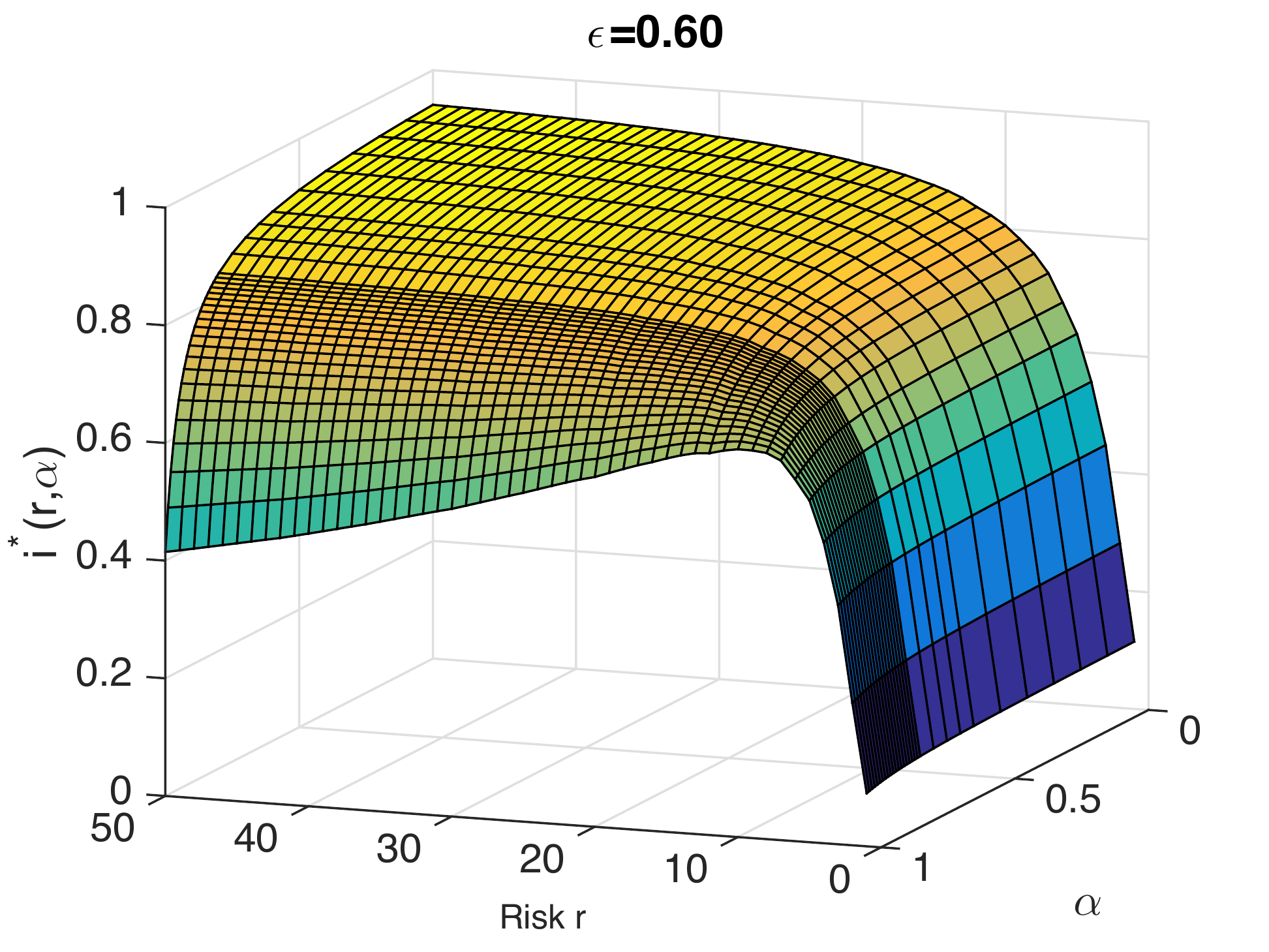}}
\subfloat[$\pmb{2D~\epsilon=0.6}$][]{\label{fig:eps_d}\includegraphics[width=0.45\textwidth]{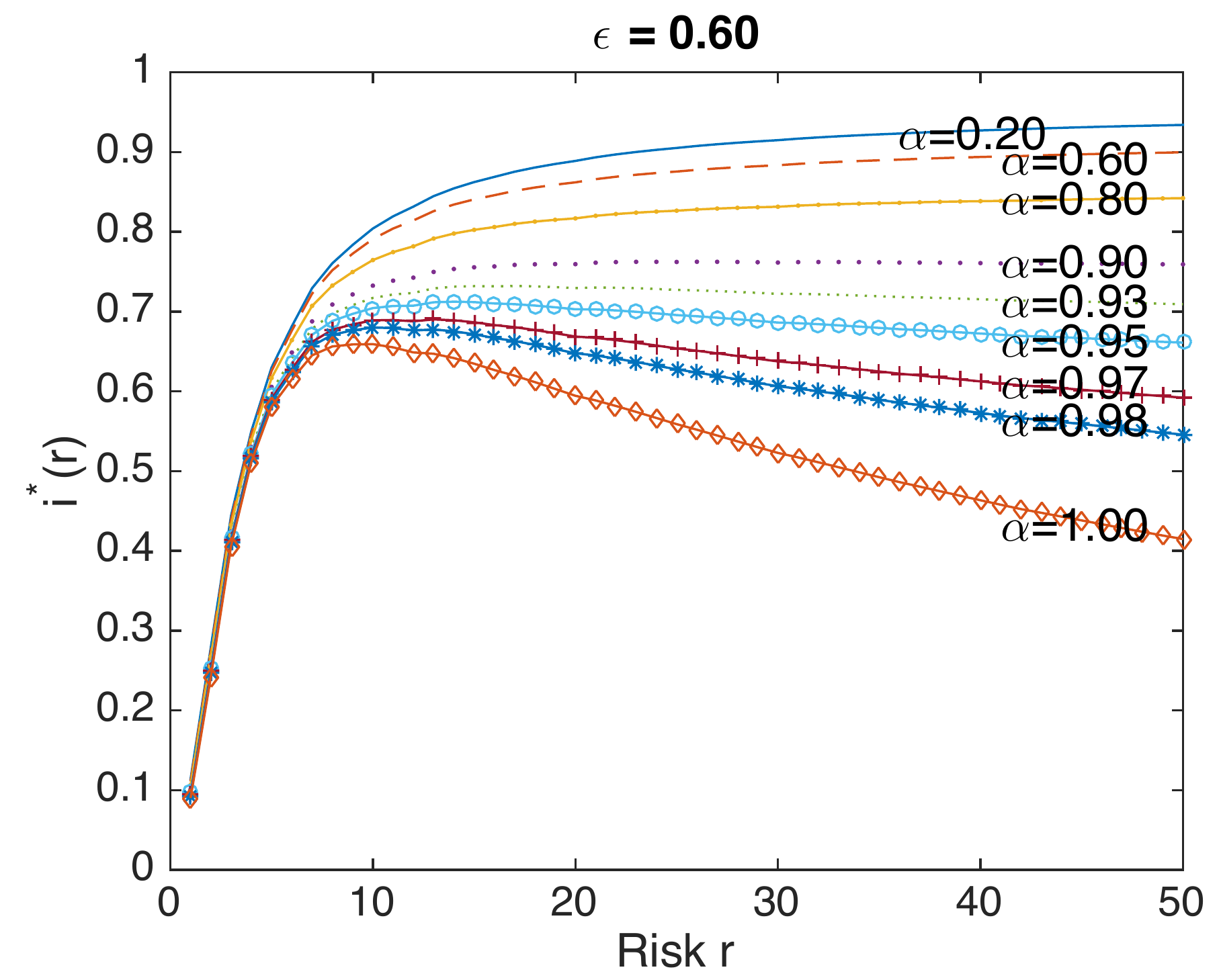}}\\
\subfloat[$\pmb{3D~\epsilon=0.9}$][]{\label{fig:eps_e}\includegraphics[width=0.45\textwidth]{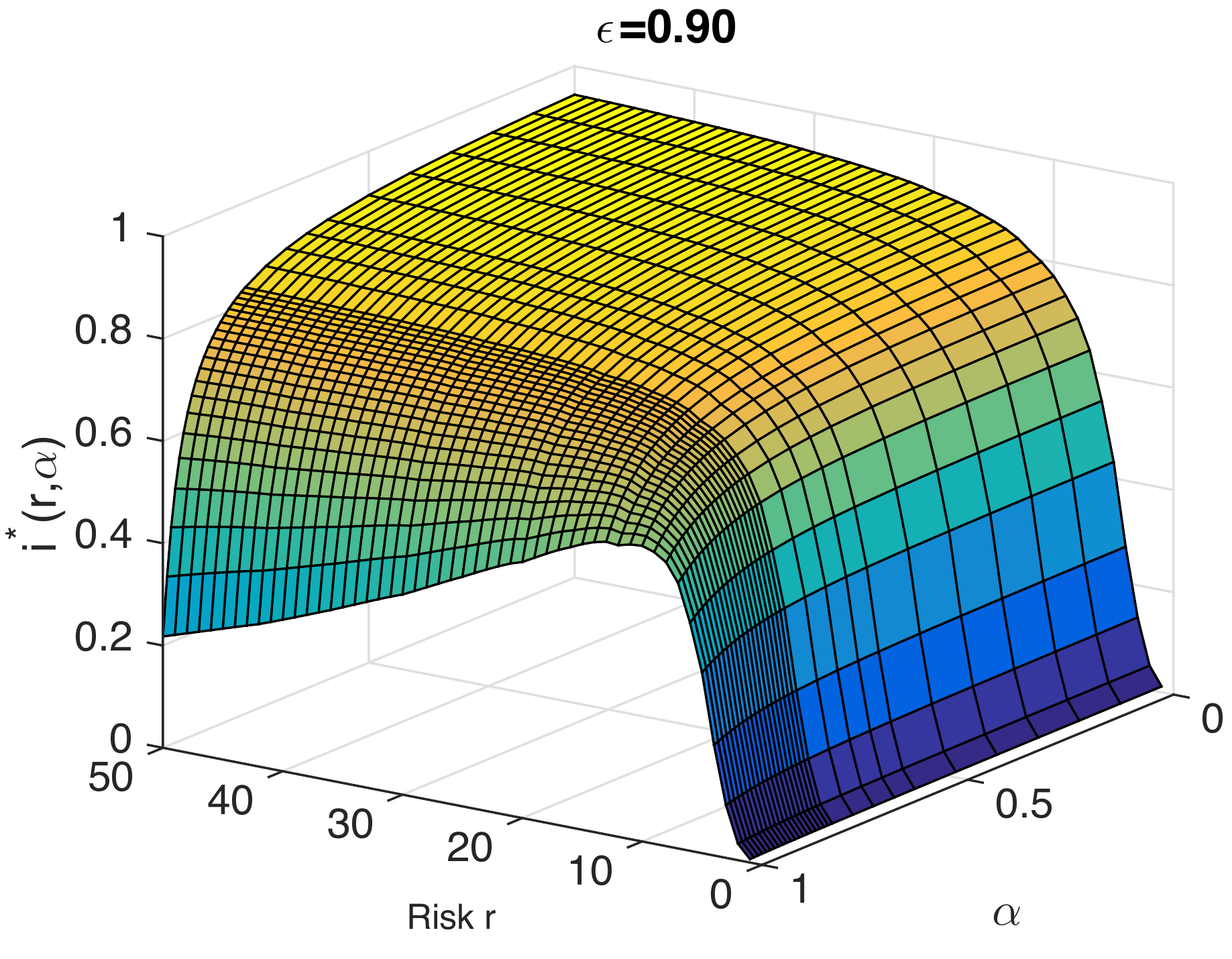}}
\subfloat[$\pmb{2D~\epsilon=0.9}$][]{\label{fig:eps_f}\includegraphics[width=0.45\textwidth]{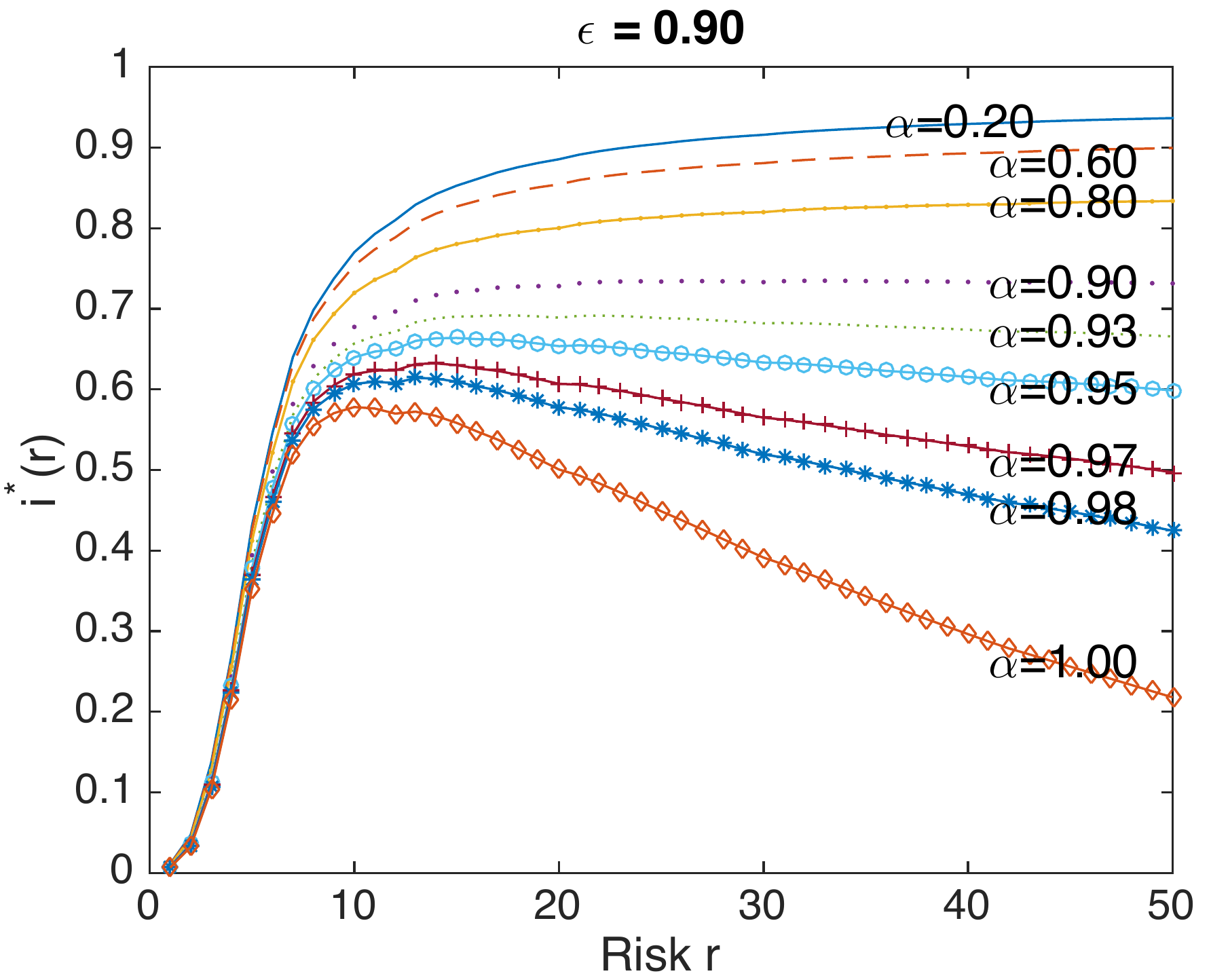}}
\caption[\textbf{Infected fraction vs risk r and preference level} $\pmb\epsilon$]{Surface plots of fraction of the infected population, $i^*(r,\alpha)$, at steady state versus $r$ and $\alpha$, for preference levels \subref{fig:eps_a} $\epsilon=0.1$, \subref{fig:eps_c} $\epsilon=0.6$, and \subref{fig:eps_e} $\epsilon=0.9$, and slices of the 3D surfaces versus $r$, $i^*(r)$, for different $\alpha$ values and preference levels  \subref{fig:eps_b} $\epsilon=0.1$,  \subref{fig:eps_d}, $\epsilon=0.6$, and  \subref{fig:eps_f} $\epsilon=0.9$: when $\alpha<0.95$,  the  $i^*$ increases with risk $r$,  when  high-risk people use condoms most of the time, $\alpha>0.95$, then $i^*$  decreases  in the higher-risk groups as a function of $r$.}\label{fig:eps}
\end{figure}

When the high-risk people use condoms most of the time, $\alpha \geq 0.95$, while the lower-risk population only uses condoms occasionally, this trend is reversed.   This effect is strongest when the mixing is highly biased ($\epsilon=0.9$) i.e when most of a person's partners have very similar risk.  We note that although this is mathematically consistent with our model, it is in an unrealistic parameter range for the population. 

To quantify the effectiveness of condoms at reducing the prevalence, in  Figure (\ref{fig:sens}) we show fraction of the total infected  population as a function of 
$\alpha$, $i^*_T=\int I^*(r)dr/ \int N(r)dr$ 
for different preference levels $\epsilon$. 
There is a threshold for $\alpha$ to drops the epidemic down, and  this threshold increases as mixing level $\epsilon$ increases. For example  when level of mixing is $\epsilon=0.1$ (Random mixing), to drop the prevalence drastically, $\alpha$ needs to be around $70\%$, however, for when $\epsilon=0.6$ (Balanced mixing) this threshold is $\alpha=0.9$, but for $\epsilon =0.9$ (Highly biased mixing) threshold disappears  which means condom-use by high-risk  individuals does not have impact on controlling the prevalence. The reason is when people mix more randomly, then high-risk  people have many partners with different risks, therefore, using more condom by them save this many partners with different risks, however, when mixing tends to be more biased, $\epsilon=0.9$, most of the partners of high-risk  people are themselves high-risk, which this case this group does not take heavy toll on the prevalence, no mater what fraction of their acts they use condom.

\begin{figure}[tb]
\centering
\vspace{10pt}
 \includegraphics[width=3in]{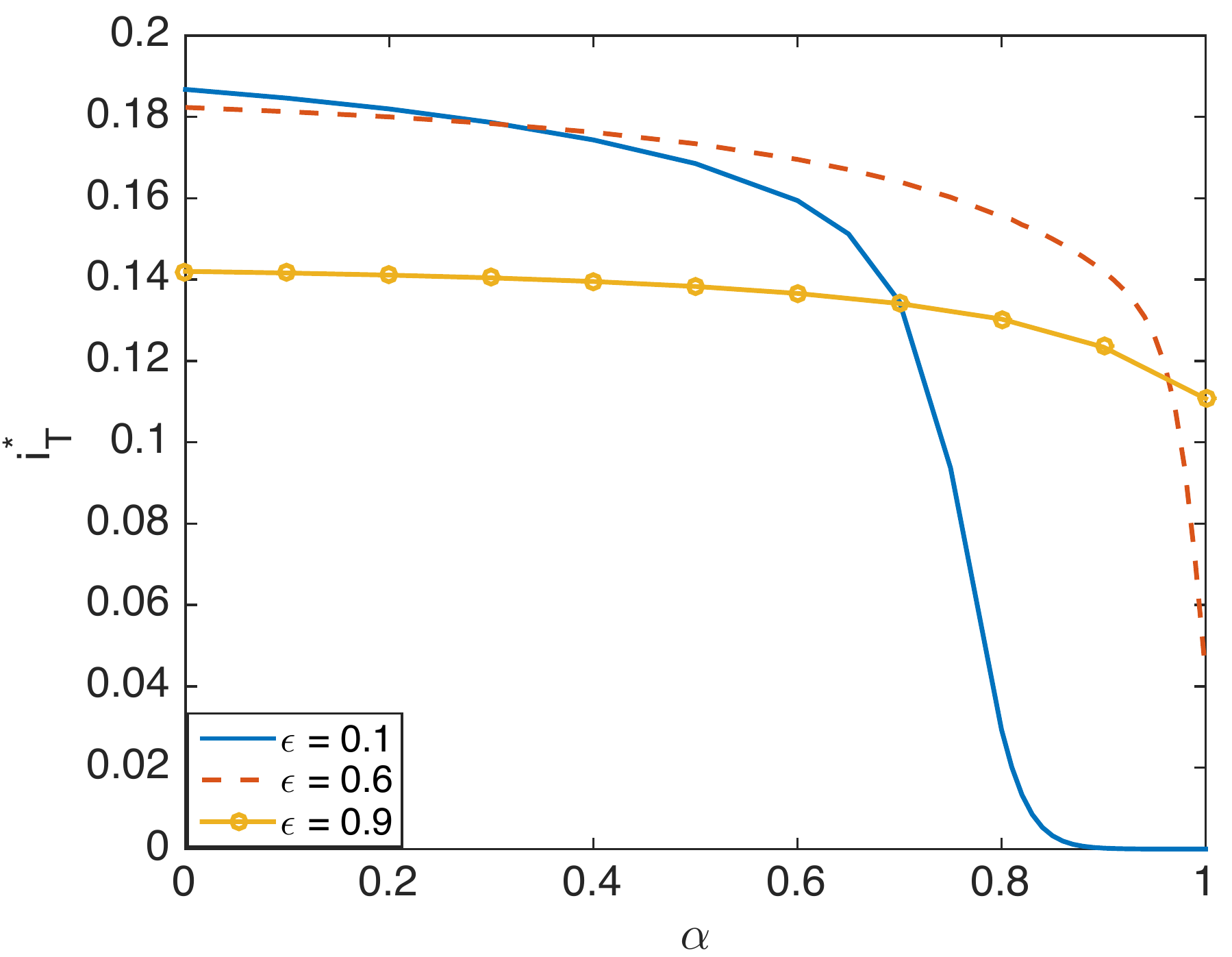}
\caption[\textbf{Total infected fraction vs condom-use}]{ Total fraction of the  population that is infection $i^*_T=\int I^*(r)dr/ \int N(r)dr$  decreases as condom-use $\alpha$ increases for random mixing $\epsilon=0.1$, combined mixing $\epsilon=0.6$, and highly biased mixing  $\epsilon=0.9$ in partnership selection. 
Note that when people tend to pick partners randomly,  $\epsilon=0.1$, and the population uses condoms most of the time,  $\alpha > 0.8$, then condom-use is an effective way to control the epidemic.  
  }\label{fig:sens}
\end{figure}

\subsection{Condom-use scenarios}
We compare three condom-use scenarios--explained in Table (\ref{table:condom_scenario})-- to quantify their impact on reducing the prevalence of the STI at the endemic equilibrium. 

\begin{table}[htbp]
\centering
\begin{tabular}{p{2cm}p{12cm}}
\hline
\cline{1-2}
\textbf{Scenario} & \textbf{Description} \\
\hline
 \vspace{.1mm}\textbf{NCU} & No Condom-Use: the unrealistic case where condoms are never used is included as a reference case. \\
 \hline
 \vspace{.3mm} \textbf{SCU} &Some Condom User: the population is divided into condom users and non-users where  in each risk group, $\hat c$ fraction of $N(r)$ of the people use condom all the time, while $(1-\hat c)$ fraction of them never use a condom.\\
 \hline
  \vspace{.1mm}\textbf{FCU} & Fraction Condom Users: everyone uses a condom with probability $\bar c$ in each act, that is,  $c(r)=\bar c$ is constant.\\
 \hline
 
    \vspace{.2mm}\textbf{RCU} &  Risk-based Condom-Use: the condom-use is a function of risk based on the function $c(r, \alpha)$ in Equation (\ref{E:cr}) and the scaling parameter  $\alpha$ is chosen so the average condom-use  $<c(r,\alpha)> = \bar c.$ \\
\hline
\cline{1-2}
\end{tabular}
\caption[\textbf{Condom-use scenarios}]{Different condom-use scenarios. }
\label{table:condom_scenario}
\end{table}

To study the influence of different scenarios on the total prevalence, we recorded prevalence at time $t$ for each scenario. In Figure (\ref{scenarios}), the prevalence for all scenarios are shown as a function of time $t$ for $\hat c=\bar c=0.37$. When condoms are never used (NCU),  the prevalence tends to $i(t) \rightarrow i_T^*=0.18$.  The prevalence is reduced the most for SCU when $\hat c=37\%$ of population uses condoms all  times. In this case, we observe a reduction of $7\%$ of prevalence at steady state. The reason is that condom-use comes by act, and when $\hat c=37\%$ of population use condoms in all their acts, then $37\%$ of population are rarely infected. On the other hand, for scenario {FCU}, i.e when all people use condom $\bar c=0.37$ of the  acts, the reduction of prevalence is very weak, almost  $0.5\%$, and this is because the model is applied for Ct as a highly infectious STI, that is the chance of catching or transmitting the infection by one act is high, therefore, even if all people use condom partially, there is a high chance of infection transmission in the acts which condom is not used.

In the scenario RCU, i.e using Equation (\ref{E:cr}) as a condom-use function for when $\bar c=0.37$ which results $\alpha=0.75$, the prevalence at steady state reduces by $2\%$. In this scenario, people on average use a condom in $37\%$ of their acts, however, high-risk people are more likely to us a condom. 
As we observe, for this scenario, the growth of infection is slower than the other scenarios and it takes more time (around $10$ years) to reach steady state. This is because, high-risk  individuals, who are mostly responsible of spreading infection, use condom more and then transmit or catch infection less than the other scenarios, therefore, it takes time for them to transmit or catch infection.
\begin{figure}[tb]
\centering
\includegraphics[scale=.5]{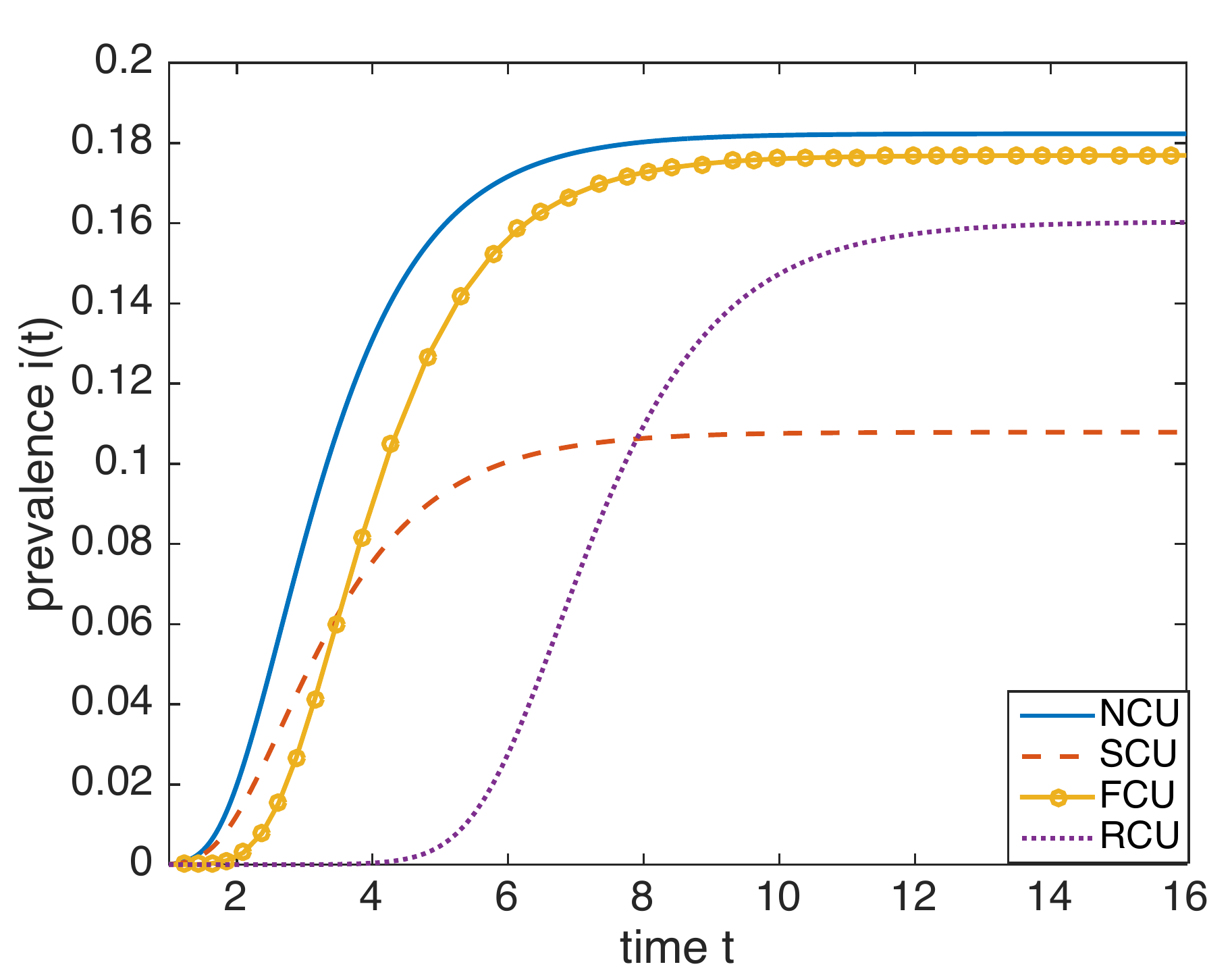}
\caption[\textbf{Ct prevalence vs time for different condom-use scenarios}]{The prevalence of Ct as a function of time for different Scenarios: NCU=no condom-use, SCU=sometime condom user, FCU=fraction condom user, RCU = risk-based condom-use where $\hat c = \bar c=0.37$ and $\alpha=0.74$ and $\epsilon=0.8$. }
\label{scenarios}
\end{figure}

\section{Discussion and Conclusions}
In this chapter we created a continuous-risk SIS transmission model for the spread of Ct with biased mixing partnership selection to investigate the impact that condoms  can have in controlling their spread. The model incorporates functions describing mixing patterns as well as condom-use by individuals based on their risk. 
The  mixing between people of different risks was modeled as a combination of random mixing and biased mixing, where people prefer partners of similar risk \cite{hyman1989effect}. 
 Our model includes the observed correlation between condom-use and the number of  partners among adolescents and young adults \cite{beadnell2005condom,kabiru2009correlates,lescano2006condom} where people with higher number of partners are more likely to use condoms. Based on that, we fitted an increasing function of risk for condom-use to the information provided in \cite{beadnell2005condom}.
 We assumed  that people with more partners (higher risk $r$)  were less picky about the risk of their partners than people with fewer partners. We modeled this  increased acceptance of the risk of the partners by increasing the standard deviation of risk of the partners as the square root of risk.  

The endemic infection equilibrium is more  sensitive to the rate that the people with bigger risk $r$ -where there are fewer acts per partnership- use condoms than it is for people with smaller risk $r$ -where there are more acts per partnership. 
When the probability  of infection is high for a single act, as it is in our simulations, then the number of people an infected person infects is more correlated to the number of partners that he/she has unprotected sex with, than the number of acts they have.    
Our model assumes that people with fewer partners have more acts per partnership than people with more partners.  The risk of infection is high for a single act where condoms are not used, then even failing to use condoms a few times in a partnership is enough to pass on the infection.   
That is, the model indicates increasing the fraction of times that  people with many partners use condoms could be an effective strategy in mitigating Ct.

The simulations quantified  the rate that Ct spreads through a population based on different distributions of condom-use as a function of the population risk. We estimated the impact of condom-use by higher risk individuals  on the distribution of endemic equilibrium. We found that  for almost  all amount of  condom-use, having more partner   increases the likelihood of being infected,  the infection prevalence is greatest in the higher risk populations and it is always a good mitigation strategy to increase condom-use in these populations to mitigate an epidemic. This effect is  stronger for when people select most of their partners preferentially.  

 We also observed that the total prevalence does on drop drastically unless the mixing tends to more random and  high-risk individuals  use condom  in at least $70\%$ of their acts. However, when the mixing tends more toward  biased mixing, prevalence at steady state looses its sensitivity to condom-use. 
Our simulations, also,  demonstrate that when level of biased mixing is low, then it is also an effective mitigation strategy to increase condom-use in the the lower risk populations, as shown in Figure (\ref{fig:sens}).

We derived the  basic reproduction number $\mathcal{R}_0$  using the  next generation approach \cite{diekmann1990definition} and used simulations to show the early growth of the epidemic depends on mixing pattern and condom-use.  For very biased mixing, when people pick their partners to have similar risk, condoms are an effective approach to mitigate the spread of Ct. However, when the population  mixed more randomly, then condom-use is less effective in controlling the epidemic. 

The model investigates the role of the risk-structure and importance of homophily in the mixing between people with different risk on the spread of the epidemic.  
We formulated this simplified model because it is easier to analyze and can provide insight into the dynamics of the more complex models that also account situations where these assumptions do not hold.

The current model does not distinguish between men and women.   In heterosexual populations, this approximation  is only appropriate when 
the mixing between men and women is symmetric and the infection prevalence is approximately the same in both men and women.  
We are extending the model to a heterosexual mixing model, similar to our previous model in Chapter. \ref{multi} where we only included two risk groups .
The heterosexual model can be used to more closely match partnership studies that show, on average, a sexually active man will have more partners 
than the sexually active women in the adolescents and young adult population.  It can also be used to study the relative effectiveness of increasing the screening for men, women, or both sexes for Ct when there are limited resources. 

We recognize that a more realistic approach is needed for guiding public health policy.  This realistic model would track behavior change and mixing based on a person's age. For example, when an individual is infected and treated, then they are more likely to change their behavior to prevent being infected again. Behavior change is an important assumption which could be added in this model by including risk-based partial derivative terms in the model. 
This extension would make the model significantly more complex and would not be a good  model as using an agent-based model that can follow the infection status of each individual.

The analysis and simulations of our continuous-risk model has led us in creating a more appropriate model for studying the impact of screening, partner notification, partner treatment, condom-use, and behavior change in controlling the spread of Ct.  In  Chapter. \ref{abm} we will formulate a stochastic Monte Carlo - Markov Chain (MCMC) agent-based bipartite disease-transmission network-model where the men and women are the network nodes and sexual acts are represented by edges between the nodes.  
The network captures the distributions for number of partners that men and women have, and the correlations between the number of partners that a person has and the number or partners their partners have.  These partnership distributions, and the transmission parameters, are based on survey data for the $15-25$ year-old  AA community in New Orleans.  

Unlike the continuous-risk model, the network model can track an individual's behavior change, such as  condom-use  after being treated for infection, the affect of aging on number of partners a person has, or the differences in condom-use between primary and casual partners.  
\chapter{Network Model}\label{abm}

Up to now we have introduced   an ordinary differential equation  Ct transmission model  that captures the most essential transitions through an infection with Ct to assess the impact of Ct infection screening programs, Chapter \ref{multi}.
We also have provided a selective sexual mixing hybrid differential/integral equation Ct model  to capture the heterogeneous mixing among people with different number of partners, Chapter \ref{continuous}. 
An alternative, and more realistic model, is to represent the sexual network by a graph where each individual within a population is a node.  The connecting edges between the nodes denote sexual relationships that could lead to the transmission of infection.
These sexual mixing networks can capture the heterogeneity of whom an infected person can, or cannot, infect.

In this Chapter we create an agent-based heterosexual network model of Ct transmission to evaluate potential intervention strategies for reducing the Ct prevalence in urban cities, such as New Orleans \cite{azizi2018using}.
We construct a  network  model that mimics the heterosexual behavior obtained from a sexual behavior survey of the young adult AA population in New Orleans and model Ct transmission as a discrete time Monte Carlo stochastic event on this network.  
The model is initialized to agree with the current New Orleans Ct prevalence.
We use sensitivity analysis to quantify the effectiveness of different prevention and intervention scenarios, including screening,  partner notification -which includes partner treatment, and partner screening (contact tracing)- and social friend notification, and rescreening \cite{hyman1997disease,azizi2018using}. This model structure allows the sexual partnership dynamics, such as partner concurrency, sexual histories of each person, and complex sexual networks, to be governed at the individual level.

In this Chapter we first review the data used for generating network, then we explain how we generated our networks, and the last Section would be transmission model on the networks and testing different interventions.

\section{The New Orleans Sexual Activity Survey Data}
Two types of studies were conduced to estimate the Ct infection among local people and  to assess the effectiveness of biomedical and behavioral intervention programs  in the general heterosexual population reside in New Orleans.

A community-based pilot study,  called "Check-it", was performed   among African American, AA, men ages $15-25$ years old in a typical three-month period. The overall $n=202$ men participant were asked about their age, number partners in the past three months, history of Ct test results, as well as living habits.
 Meanwhile, their partners information  have been collected by asking questions referring to the status of each relationship such as, the partner's age, strength of relationship, first and last time of intercourse, and the possibility that their partner have intercourse with others \cite{kissinger2014check}. 
 
An internet pregnancy/STI prevention study, called  "You Geaux Girl", was conduced among AA  women ages $18-21$ years old. A total $n=414$ participants have  been asked the similar questions. Additionally,  more partners information such as total number of sexual act for each relationship and  number of partners their partner might have were asked from the participants \cite{green2014influence}.

\subsection{Distributions for  number of partners}
The first and most vital information for constructing network is distribution of number of partners for men and women. Therefore, individuals were asked about how many partners they have had during the past three months. The Figure (\ref{fig:number of partners})
shows the result  for both men and women participants.
\begin{figure}[H] 
\centering
\includegraphics[width=13cm, height=6cm]{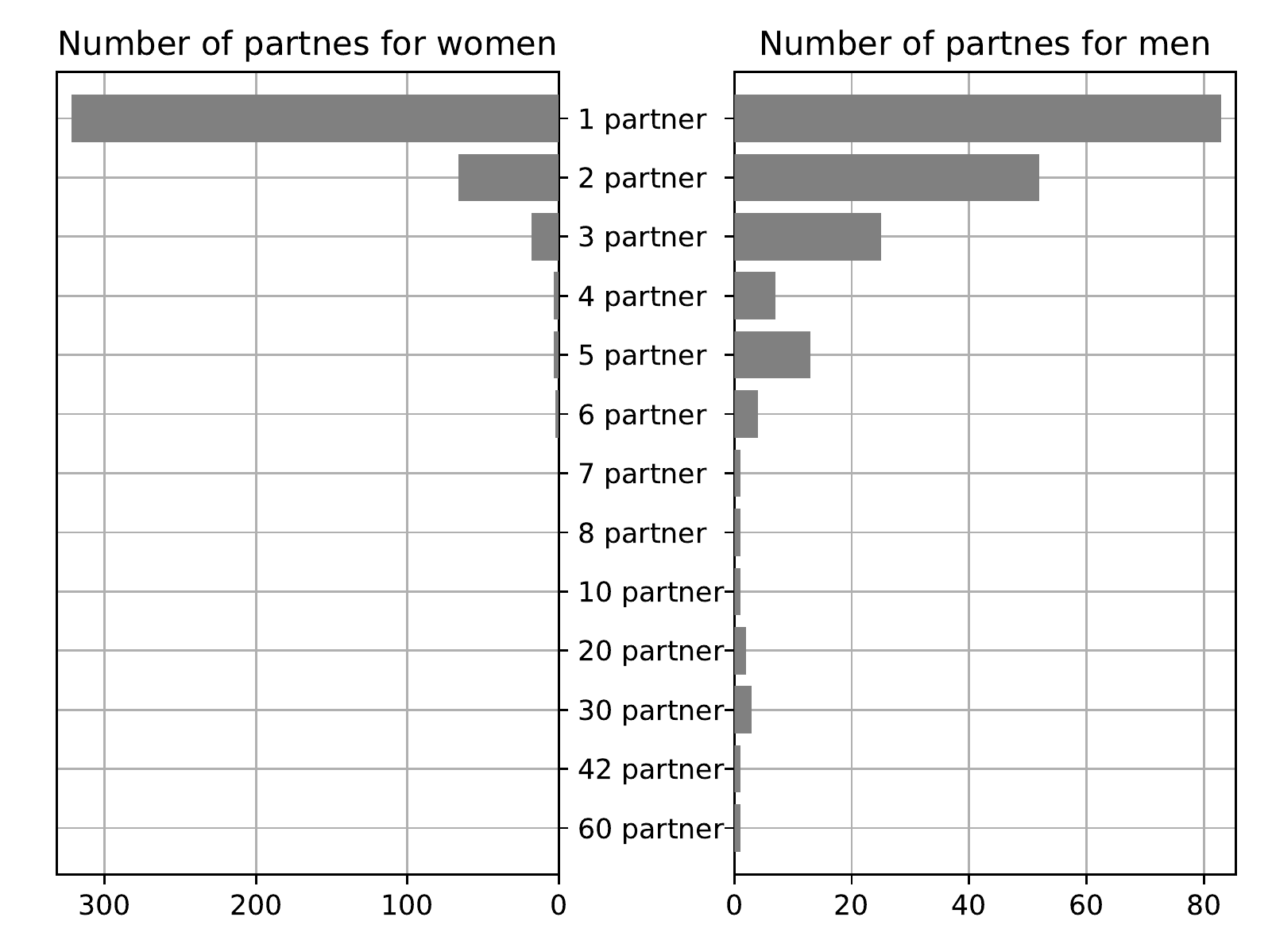}
\caption[\textbf{Bar plot for number of partners}]{Bar plot of number of partners for men and women for the last three months.}
\label{fig:number of partners}
\end{figure}

We fitted a Generalized Pareto distribution to  the data and fill the gaps in data.
The Figure (\ref{fig:part_men})  shows the probability of having $x$ partners for  men and women and also their  fitted distributions.

\begin{figure}[htp]
\centering
\subfloat[\textbf{Number of partners for men}][]{\label{fig:m}\includegraphics[width=0.5\textwidth]{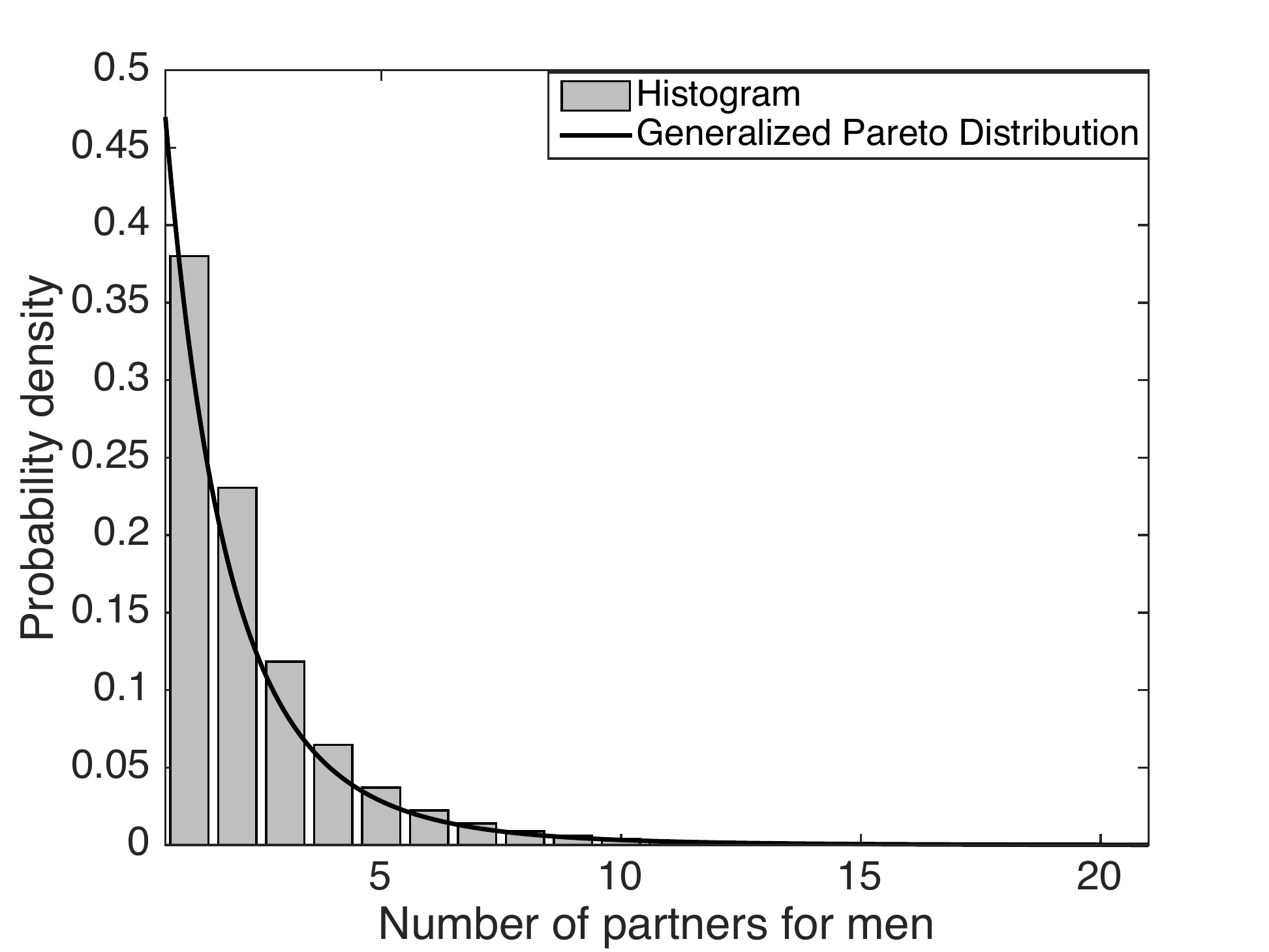}}
\subfloat[\textbf{Number of partners for women}][]{\label{fig:w}\includegraphics[width=0.5\textwidth]{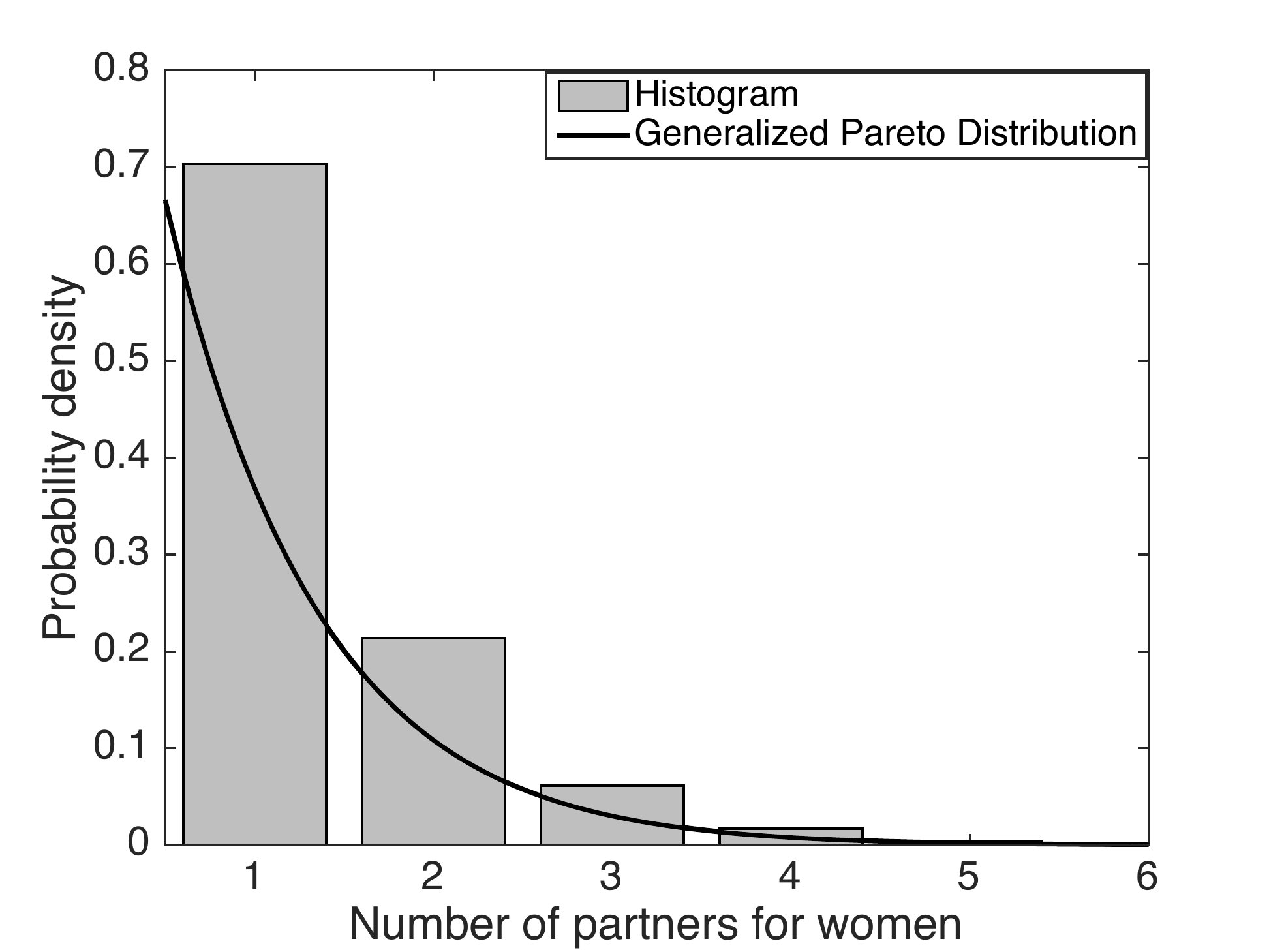}}\\
\caption[\textbf{Bar plot of }]{Reference number of partners for men (\ref{fig:m}) and women 
(\ref{fig:w}) participant with the probability of having x partners.}\label{fig:part_men}
\end{figure}

To generate a consistent  network, the  number of men and women should be selected properly: if  $N_m$ men in the population have an average $\bar p_m$ partnerships with women, and  $N_w$ women have an average $\bar p_w$ partnerships with men, then  for  consistency, we should have a total of $N_m\bar p_m = N_w \bar p_w$ partnerships in the population. In other word the number of men and women in the network should follow consistency condition  $$\frac{N_m}{N_w}=\frac{\bar{p}_w}{\bar{p}_m}=\rho.$$

Using the fitted distribution shown in Figure (\ref{fig:part_men}) we can find the data $f_i$  as the fraction of individuals having $i$ partners for both men and women and also expected values
$\bar{p}_m$ and $\bar{p}_w$, and then we can find the fraction of active men and women in  population with arbitrary size.

\subsection{Distribution for number of partners of partners}
 In the "You Geaux Girl" survey,  women 
participants  were asked about the number of partners their partners have. The Table (\ref{tab:degdeg1}) is the extracted result  from survey: each element $(i,j)$ in the Table is the number of  men with $j$ partners which women with $i$ partners have.
\begin{table}
\begin {center}
\begin{tabular}{c c c c c c c c c c c c c c c c c c c}
\hline
& & \multicolumn{15}{c}{Number of Partners for  Partners of  Women Participants}\\  
\hline
& & 1 & 2  & 3  & 4  & 5 & 6 & 7 &8  & 9 &10 &11 &12 &13 &16 &21\\  
\hline
Number & 1 & 197 & 56 & 28 & 9  & 2 & 3 & 0 & 1 & 0 & 0 & 0 & 0 & 0 & 1 &0\\
of & 2 & 46  & 34 & 23 & 3  & 2 & 5 & 2 & 0 & 1 & 1 & 0 & 0 & 1 & 0 &0\\
Partners & 3 & 16  & 11 & 3  & 10 & 2 & 1 & 1 & 0 & 0 & 0 & 0 & 0 & 0 & 0 &1\\ 
for & 4 & 4   & 2  & 4  & 0  & 0 & 1 & 0 & 0 & 0 & 0 & 1 & 0 & 0 & 0 &0\\ 
Women & 5 & 0   & 2  & 3  & 0  & 0 & 4 & 0 & 0 & 0 & 0 & 0 & 1 & 0 & 0 &0\\ 
Participants & 6 & 3   & 0  & 5  & 1  & 0 & 0 & 1 & 0 & 0 & 0 & 0 & 1& 0 & 1 &0\\ 
\hline
\end{tabular}
\captionof{table}[\textbf{Joint-degree table from data }]{Joint-degree Table of sample data for women participant in the last three months.}
\label{tab:degdeg1}
\end{center}
\end{table}
 In network  terminology each element $(i,j)$ counts the number of edges between women with degree $i$ and men with degree $j$, therefore, it is trivial to see that for a consistent Table, if men with degree $j$ were asked about the number of their partners with degree $i$, all of them would provide the same answer $(i,j)$.  However, this is not always true for sample data, thus  we have to implement this consistency condition for joint-degree Table for the population. To do that, at first we generate the joint-degree Table for the population, and then we make it consistent. 
 
The Table (\ref{tab:degdeg1}) is an sparse table, for most of women we have no high-degree partners, which  is not correct,  because from degree distribution for men in the previous subsection  we have men with more than $10$ partners. This can happen because of lack of data or information. Therefore, we have to fill the incorrect zeros in the table. In order to do that, we divide each row of table to sum of its elements, that is, we divide row $i$ by $\sum_{j}(i,j)$. The new elements are fraction of edges between women with  $i$ partners and men with  $j$ partners. Then  we fit these fractions to Generalized Pareto distribution to find probability of an edge between a woman with $i$ partners and a man with $j$ partners, Table (\ref{tab:prob-digdig}).
 \begin{table}[htp]
 \begin {center}
\scalebox{0.95}{\begin{tabular}{c c c c c  c}
\hline
&  \multicolumn{5}{c}{Probability  of Partner's Partner for Women Participants}\\  
\hline
& & 1 & 2  & ...  &n\\  
\hline
Probability  & 1 & $p_{11}$ & $p_{12}$ & ... & $p_{1n}$\\

of Partners & 2 & $p_{21}$  & $p_{22}$ & ... & $p_{2n}$\\

for  & $\vdots$ & $\vdots$  & $\vdots$ & $\vdots$  & $\vdots$\\ 

Women & m & $p_{m1}$   & $p_{m2}$  & ...   & $p_{mn}$\\ 
\hline
\end{tabular}}
\caption[\textbf{Joint-degree probability}]{Joint-degree probability for women participant in the last three months.}
\label{tab:prob-digdig}
\end{center}
\end{table}

Now our goal is to find the number of edges between women with i partners and men with  j partners: we have $f^w_iN_w$ women with i partners, where $f_i^w$ is Generalized Pareto probability of having $i$ partners for women and was computed in the previous subsection, and $N_w$ is the population size of women. These women make total $if^w_iN_w$ edges which $p_{ij}$ fraction of these edges are between them and  men with j partners. Therefore, we will have total $if^w_iN_wp_{ij}$ edges between women with i partners and men with j partners. The value $if^w_iN_wp_{ij}$ defines the  $(i,j)th$
 element of this  fitted  joint-degree Table. 

If men were asked about this question, we should have consistent table, but it is not consistent, therefore,  we are using constrained optimization to refine elements $p_{ij}$ to make the table consistent. 
In fact, we solve minimization problem
  \begin{equation*}
\begin{aligned}
& \underset{}{\text{minimize}}
& & \|P-\tilde{P}\|^2,\\
& \text{subject to}
& & \sum_{j}{if^w_iN_w\tilde{p}_{ij}}=if_i^wN_w,~~i=1,...,m,\\ & & and &~~~~~~\sum_{i}{if^w_iN_w\tilde{p}_{ij}}= j f_j^mN_m,~~j=1,...,n.
\end{aligned}
\end{equation*} 
where $P=(p_{ij})$, $\tilde{P}=(\tilde{p}_{ij})$, $f_j^m$ is Generalized Pareto probability of having $j$ partners for men, and $N_m$ is population size of men. By definition of $\rho=\frac{N_m}{N_w}$ we can simplify the above optimization problem to  
  \begin{equation*}
\begin{aligned}
& \underset{}{\text{minimize}}
& & \|P-\tilde{P}\|^2,\\
& \text{subject to}
& & \sum_{j}{\tilde{p}_{ij}}=1,~~i=1,...,m,\\ & & and &~~~~~~\sum_{i}{if^w_i\tilde{p}_{ij}}= \rho j f_j^m,~~j=1,...,n.
\end{aligned}
\end{equation*} Using least square method  we come up with the new Table of consistent probability of existing  edge between women with $i$ partners and men with $j$ partners, Table (\ref{tab:degdeg2}).
\begin{table}[htp]
\centering
\resizebox{.79\columnwidth}{!}{\begin{tabular}{ccccccc}
\toprule[1.5pt]
\diaghead{\theadfont Diag ColumnmnHead II}{\textbf{Degree}\\ \textbf{of men}}{\textbf{Degree of} \\\textbf{ women}}
& $1$ & $2$   & $3$ & $4$   & $5$   & $6$   \\[0.50ex]
\hline
 $1$ & $0.3775$   & $0$    & $0$     & $0$  &    $0$  &  $0$  \\[0.50ex]
   \hline
 $2$ & $0.1880$   & $0.0205$    & $0$     & $0$  &    $0$  &  $0$  \\[0.50ex]
   \hline
  $3$ & $0.0988$   & $0.0396$    & $0.0162$     & $0.0033$  &    $0$  &  $0$  \\[0.50ex]
   \hline 
   $4 $ & $0.0833$   & $0.0734$    & $0.0285$     & $0.0087$  &    $0.0034$  &  $0.0010$  \\[0.50ex]
   \hline
   $5 $ & $0.0313$   & $0.0264$    & $0.0162$     & $0.0083$  &    $0.0037$  &  $0.0014$  \\[0.50ex]
   \hline 
   $6$ & $0.0207$   & $0.0209$    & $0.0134$     & $0.0066$  &    $0.0025$  &  $0.0003$  \\[0.50ex]
   \hline  $7 $ & $0.0150$   & $0.0167$    & $0.0107$     & $0.0046$  &    $0.0008$  &  $0$  \\[0.50ex]
   \hline 
  $8 $ & $0.0116$   & $0.0135$    & $0.0082$     & $0.0026$  &    $0$  &  $0$  \\[0.50ex]
   \hline
  $9 $ & $0.0093$   & $0.0110$    & $0.0061$     & $0.0008$  &    $0$  &  $0$   \\[0.50ex]
   \hline
  $10 $ & $0.0075$   & $0.0090$    & $0.0043$     & $0$  &    $0$  &  $0$   \\[0.50ex]
   \hline
   $11 $ & $0.0061$   & $0.0074$    & $0.0027$     & $0$  &    $0$  &  $0$ \\[0.50ex]
   \hline
   $12 $ & $0.0050$   & $0.0062$    & $0.0014$     & $0$  &    $0$  &  $0$  \\[0.50ex]
   \hline
   $13 $ & $0.0042$   & $0.0053$    & $0.0005$     & $0$  &    $0$  &  $0$  \\[0.50ex]
   \hline
 $14 $ & $0.0034$   & $0.0045$    & $0$     & $0$  &    $0$  &  $0$  \\[0.50ex]
   \hline
   $15 $ & $0.0027$   & $0.0037$    & $0$     & $0$  &    $0$  &  $0$  \\[0.50ex]
   \hline
   $16 $ & $0.0021$   & $0.0031$    & $0$     & $0$  &    $0$  &  $0$  \\[0.50ex]
   \hline   $17 $ & $0.0016$   & $0.0027$    & $0$     & $0$  &    $0$  &  $0$  \\[0.50ex]
   \hline
  $18 $ & $0.0012$   & $0.0023$    & $0$     & $0$  &    $0$  &  $0$  \\[0.50ex]
   \hline 
   $19 $ & $0.0009$   & $0.0020$    & $0$     & $0$  &    $0$  &  $0$  \\[0.50ex]
   \hline 
   $20 $ & $0.0007$   & $0.0017$    & $0$     & $0$  &    $0$  &  $0$  \\[0.50ex]
   \hline 
   $21$ & $0.0059$   & $0.0070$    & $0.020$     & $0$  &    $0$  &  $0$  \\[0.50ex]
\bottomrule[1.5pt]
\end{tabular}}
\caption[\textbf{Joint-degree probability distribution}]{Joint-degree probability distribution for heterosexual partnerships used in the computer simulations.  Men and women are assumed to have fewer than $21$ and $6$ partners respectively.  The entry in the $i^{th}$ row and $j^{th}$ column is the fraction of partnership (edges) between  men who have i partners (degree i) and women who  have j partners (degree j).}
\label{tab:degdeg2}
\end{table}
\subsection{The number of sexual acts per partner}
To better understand and predict the  spread of Ct, we  need to know the number of acts per unit time between two typical partners. In the "You Geaux Girl" survey data women were asked about their total number of acts per partner during the last three months. Each one reported two numbers which refer to their number of acts per primary and casual partners.

The Table (\ref{tab:p-c})  shows the number of  women with $k$ partners who engaged in  $n$ sexual acts with their primary and casual partners in the last three months. 
 \begin{table}[H]
\centering 
\scalebox{.8}{\begin{tabular}{c rrrrrrrrrrrr} 
\toprule[1.5pt]
Number of acts & 1 & 2 & 3&4&5&6&7&8&9&10&11 \\ 
\hline
$\#$ of degree $1$ women   
&  53(23) & 27(6) & 24(0)&19(0)&18(3)&10(1)&8(0)&12(0)&5(0)&20(0)&68(0) \\ 
$\#$ of degree $2$ women    
&  18(25) & 7(10) & 8(2)&3(7)&4(5)&1(1)&2(1)&3(0)&2(0)&3(0)&9(3) \\
$\#$ of degree $3$ women  
&  2(13) & 1(8) & 3(1)&1(2)&2(0)&4(1)&1(1)&1(0)&1(0)&2(1)&1(0) \\
$\#$ of degree $4$ women    &  1(5) & 2(1) & 0(1)&0(1)&0(0)&0(0)&0(1)&0(0)&0(0)&0(0)&0(0)\\
$\#$ of degree $5$ women   &  1(4) & 2(0) & 0(0)&0(0)&0(1)&0(0)&0(0)&0(0)&0(0)&1(1)&0(0)\\
$\#$ of degree $6$ women  &  1(3) & 2(1) & 0(1)&0(1)&0(0)&0(0)&0(0)&0(0)&0(0)&0(0)&0(0) \\
\bottomrule[1.5pt]
\end{tabular}}
\caption[\textbf{The number of sexual acts based on number of partners for women}]{Element $(k,n)$  in the table is the number of women with $k$ partners who have $n$ sexual act per primary (casual) partners within three months.}
\label{tab:p-c}
\end{table}
Using this Tables, we can find average number of acts per primary (casual) partner for women with  $k$ partners as 
 $$a^{k}_{wp}=\frac{\sum_{i=1}^{11} i(k,i)}{\sum_{i=1}^{11}(k,i)}~~(a^{k}_{wc}=\frac{\sum_{i=1}^{11} i(k,i)}{\sum_{i=1}^{11}(k,i)}).$$
Our goal is to find one common value for  sexual acts for women with  $k$ partners defined as a weighted average of $a_{wp}^k$ and $a_{wc}^k$ i.e

\be \label{E:womancontactss}
a_w^k=\alpha a_{wp}^k+(1-\alpha)a_{wc}^k,
\ee
where $\alpha$ is fraction of primary partners for women with degree $k$.

If men were asked the same question about the number of their acts with different partners, then we would have similar information: for a man with  $k'$ partners, we can find $a^{k'}_{m}$ with the same manner. 
Then, if a $k-degree$ woman  is partner of a $k'-degree$ man, so the numbers $a^{k}_{w}$ and $a^{k'}_{m}$ should be the same, but usually this does not happen, therefore,  we use the average idea:  our compromise, or resolution, function will be an average of $a^{k}_{w}$ and $a^{k'}_{m}$.  As in resolving conflicts in the partnership data, we choose the harmonic average that weights the smaller number more:
\be\label{E:totalcontactss}
a_{kk'}=\frac{2a^{k}_{w}a^{k'}_{m}}{a^{k}_{w}+a^{k'}_{m}}. 
\ee
  By these definitions we have a consistent definition for our weighted network or symmetric adjacency matrix of network.
However, we do not have any information about the number of acts reported by men. 
Therefore, we only use one-sided sexual act number: for any edge that its relevant woman node has degree $k$, we put weight $a_w^k$ defined in $\eqref{E:womancontactss}$ on the correspondent edge. The Table (\ref{act_value}) shows the number of acts per day per partner for a women with different number of partners.

 \begin{table}[H]
\centering 
\scalebox{0.9}{\begin{tabular}{ccccccc} 
\toprule[1.5pt]
Number $i$ of partners for women & 1 & 2 & 3&4&5&6\\ 
\hline
Number of act per partners per day $a(i)$ & $0.1104$ & $0.0563$ & $0.0442$&$0.0241$&$0.0503$&$0.0222$\\ 
\bottomrule[1.5pt]
\end{tabular}}
\caption[\textbf{Number of acts per partner per unit time}]{The average number of sexual act per partne per day for women with different number of partners.}
\label{act_value}

\end{table}

\subsection{Age distribution}
Another factor in generating a sexual network is the age of individuals. Therefore, in Check-it survey men  participants were asked about their age. The Figure (\ref{men_age}) shows the frequencies of each age in the survey data. 
\begin{figure}[H]
\centering
\includegraphics[scale=.45]{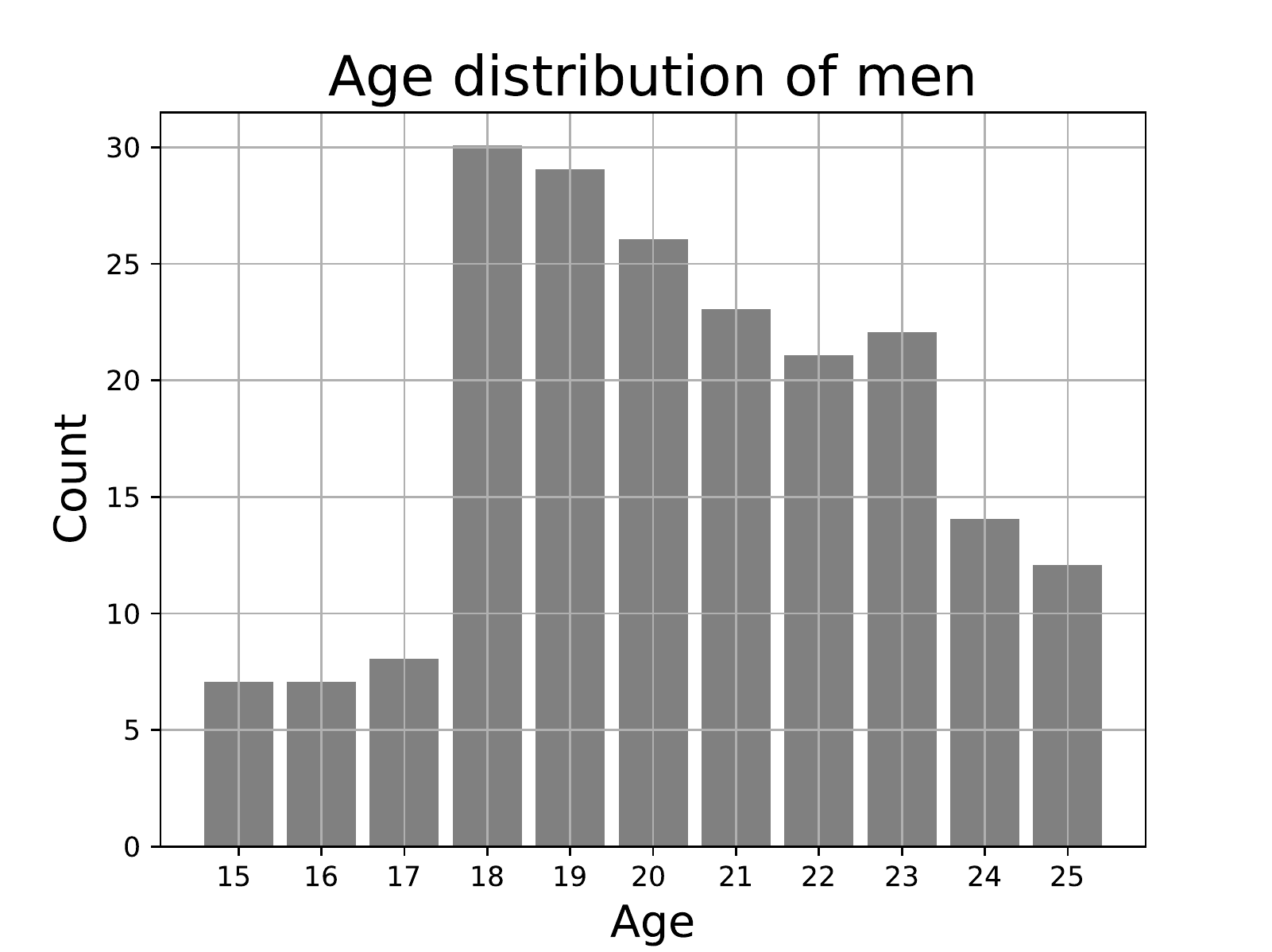}
\caption[\textbf{ Bar plot of  age distribution of men}]{Bar plot of age distribution for men.}
\label{men_age}
\end{figure}
Because in Check-it survey, population under study was  men ages  $15-25$ (people younger than $15$ or older than $25$ years old were not eligible in taking survey)  we did not fit any distribution to this data, instead, we used  empirical distribution to define probability distribution function for age of men. The Table (\ref{tab:age_pdf}) shows this function.
\begin{table}[H]
\centering 
\scalebox{.8}{\begin{tabular}{p{3.5cm} rrrrrrrrrrrrr} 
\toprule[2pt] 
Age of men & 15 & 16 & 17&18&19&20&21&22&23&24&25 \\ \hline
pdf of  age for men
& $0.035$ & $0.035$ & $0.04$& $0.15$ &$0.15$& $0.13$ &$0.115$&$0.105$& $0.11$&$0.07$&$0.06$\\ 
\bottomrule[1.5pt]
\end{tabular}}
\caption[\textbf{Probability distribution for men's age}]{Probability distribution for men's age.}
\label{tab:age_pdf}
\end{table}
Men also were asked about their partners's age.  The Figure \ref{fig:pm} is the box plot of partners age versus men's age, and Figure \ref{fig:ppm} is the scatter plot of primary partners age versus men's age. These plots show, age is one factor in  selecting primary partners  for men: there is a strong linear trend with slope $1$ between men's age and the average age of  their primary partners. However,  there is no strong age bias for selecting casual partners, though partners age is close to men's age.

\begin{figure}[htp]
\centering
\subfloat[\textbf{For partners of men}][]{\label{fig:pm}\includegraphics[width=0.5\textwidth]{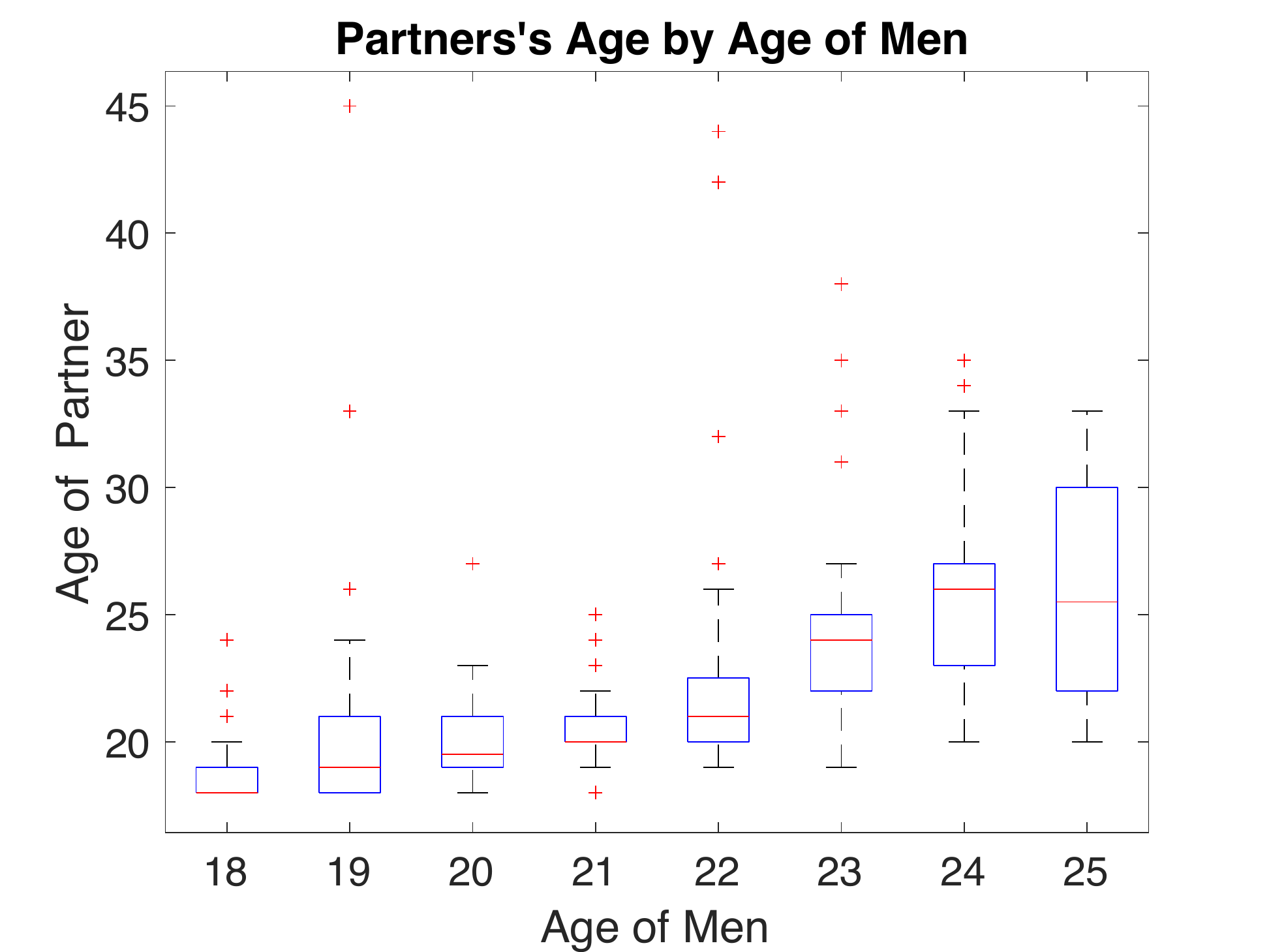}}
\subfloat[\textbf{For primary partners of men}][]{\label{fig:ppm}\includegraphics[width=0.5\textwidth]{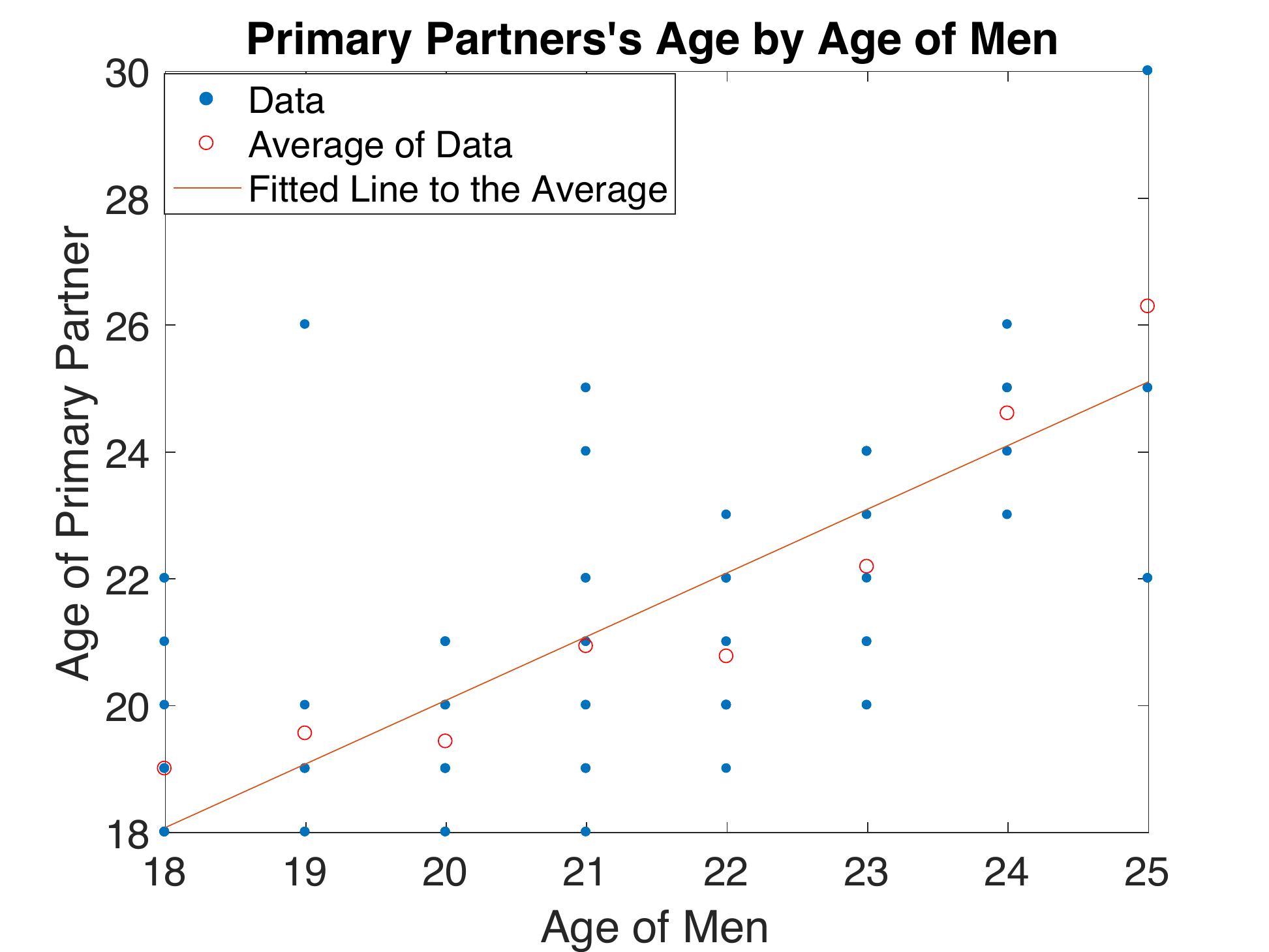}}\\
\caption[\textbf{Bar plot of age distribution}]{ Box plot of age for  all partners of men with different age: partners age is close to men's age \ref{fig:pm}. Scatter plot of age for primary partners: there is perfect linear correlation between men's age and average age of their primary partners, that is primary partners of men  have the same age as them \ref{fig:ppm}.}\label{age_women}
\end{figure}

\section{The New Orleans Social  Activity Data} 
For social activity of people in New Orleans
 we used the synthetic data produced using Simfrastructure \cite{eubank2010detail,eubank2008synthetic} which is  a high-performance, service-oriented, agent-based modeling and simulation system for representing and analyzing interdependent infrastructures.  This data is shown as two tables: first includes people information including their age and gender and second is a contact file such as\footnote{We should mention we did not use  Simdemic software to generate Table. (\ref{cfile}), but providing information about New Orleans population demographics we asked Network Dynamics and Simulation and Science Laboratory (NDSSL) to generate social contact file for synthetic population of New Orleans \cite{eubank2008synthetic}.}
  \begin{table}[H]
\centering 
\scalebox{.8}{\begin{tabular}{c rrrr} 
\toprule[1.5pt]
\textbf{PID} & \textbf{FID}& \textbf{A}& \textbf{T} \\ 
\hline
43722 & 16981& H &$0.3$ \\ 
 & 11462& W& $0.2$ \\ 
 & 37790& Sh& $0.01$ \\ 
  & \vdots & \vdots & \vdots \\
  \hline
51981 & 23476& Sc&$0.16$ \\ 
 & 18462& O& $0.02$ \\ 
 & 10790& H& $0.5$ \\ 
  & \vdots & \vdots & \vdots \\ \hline
 \vdots & \vdots & \vdots & \vdots \\ 
\bottomrule[2pt]
\end{tabular}}
\caption[\textbf{The Contact file of social network}]{This table is another representation of social contact network of 130,000 synthetic people reside in New Orleans.}
\label{cfile}
\end{table}
In Table (\ref{cfile}) \textbf{PID} is personal ID and \textbf{FID} is their social friend ID, \textbf{A} is activity in which PID meet FID, and \textbf{T} is fraction of  time in day that two friends meet with each other through activity A. We have $5$ different activities including H as home, W as work, Sc as school, Sh as shopping, and O as others. For example person 43722 lives with person 16981 at the same home for $8$ hours  a day.
\section{Network Generation}
Our goal of this Section is to generate a sexual network that represents the sexual activity of adolescent and young adult AAs reside in New Orleans. In this heterosexual network, each person is a distinct identity represented by a node in the network. Each node \textbf{i} in the network -representing a person- is denoted by the index \textbf{i}, and each edge $\textbf{ij}$ represents sexual partnership between two nodes \textbf{i} and \textbf{j}.
The network is weighted, where the weight $0<w_\textbf{ij}\leq 1$ for edge $\textbf{ij}$ is the probability that there will be a sexual act between two partners \textbf{i} and \textbf{j} on an average day.  
In the model, each day and through a stochastic process, the edge  $\textbf{ij}$ will exist (turn on) with probability $w_\textbf{ij}$, the probability that two nodes \textbf{i} and \textbf{j} have sexual act in that day, or not exist (turn off)  with probability of $1-w_\textbf{ij}$, which is equivalent to not having an  edge (sexual act) between partners \textbf{i} and \textbf{j} on that day.
To generate such a network we have sexual activity and social activity data and the tools (algorithms) defined in Appendices. (\ref{ap1}) and (\ref{ap2}).

The survey results explained in the previous Sections  were used to 
construct the joint-degree distribution or $BJD$ matrix (explained in appendix \ref{ap1}) of a heterosexual network of individuals in New Orleans. For a population $P=15,000$ sexually active adolescent and young adult men and their women partners ($5250$ men and $9750$ women) residing in New Orleans we have

\[ 
    BJD=\begin{pmatrix} 
1663 & 1588     &  1225 & 896    & \dotsm & 14\\
474 &  452      & 350 & 255    & \dotsm & 4\\
\multicolumn{6}{c}{$\vdots$}     \\
3 & \dotsm  &  2 & 1   & 1      & 0
    \end{pmatrix}. 
\]
The dimension of this $BJD$ matrix is $6\times21$, that is, the maximum number of partners for women is $6$ and for men is $21$.
We used this $BJD$ matrix and algorithm explained in Appendices (\ref{ap1}) and (\ref{ap2}) to construct the bipartite heterosexual network. 
These generated  networks statistically agree with the distribution for the number of partners men and women have had in the past three months, and the distribution for the number of partners of their partners (the joint-degree distribution). 
\subsection{Analysis of generated networks}\label{NA2}
We generated $150$  sexual networks  of $15000$ people, $30$ of them $20\%$, $30$ others $40\%$, $30$ $60\%$, $30$  $80\%$, and  the last $30$  are  $100\%$ subgraph of social network. We then compared some descriptive measures of  this ensemble of random networks such as  $Sg$, $Nc$, and  $Rc$  that were not imposed when generating the networks. These measurements are defined in Appendix (\ref{ap1}).

First, we evaluated and compared the $Sg$ for giant component and bi-component (the first and second biggest connected components of network) for each group of the networks. The Figure (\ref{ssg}) shows the box plot of these sizes: there is an increment in the size of giant components when people select most of their  sexual partners from their social friends. Because in that case, sexually active people are more tight together within the social contact network. However, there is not a significant difference in the size of giant bi-component.
\begin{figure}[htp]
\centering
\includegraphics[width=.7\paperwidth]{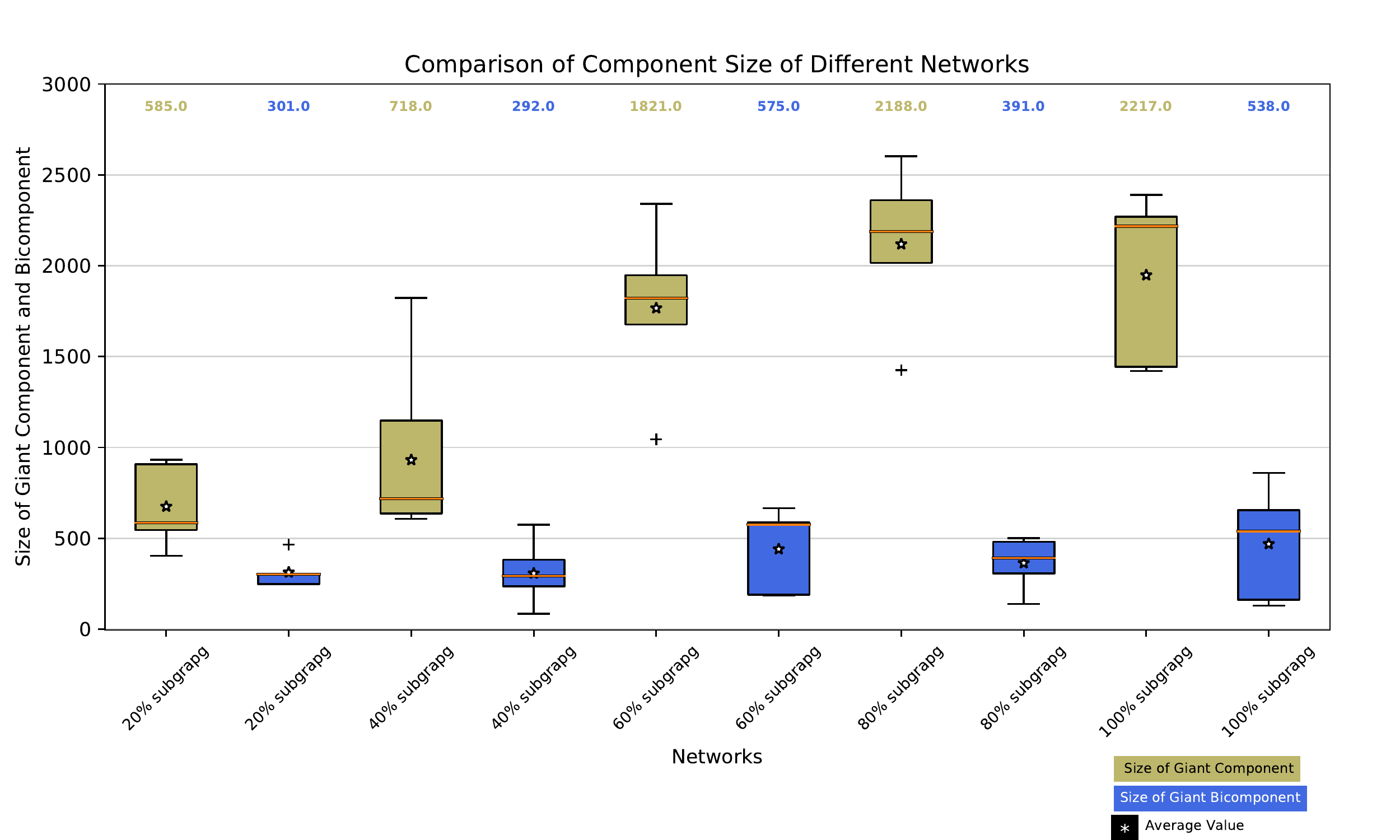}
\caption[\textbf{Box plot for size of giant component and bi-component}]{Box plot representing size of giant component and bi-component for each group of networks: the size of giant component becomes bigger when being subgraph of social network becomes stronger, however, social network does not have impact on the size of giant bi-component. }
\label{ssg}
\end{figure}
\begin{figure}[htp]
\centering
\includegraphics[width=.7\paperwidth]{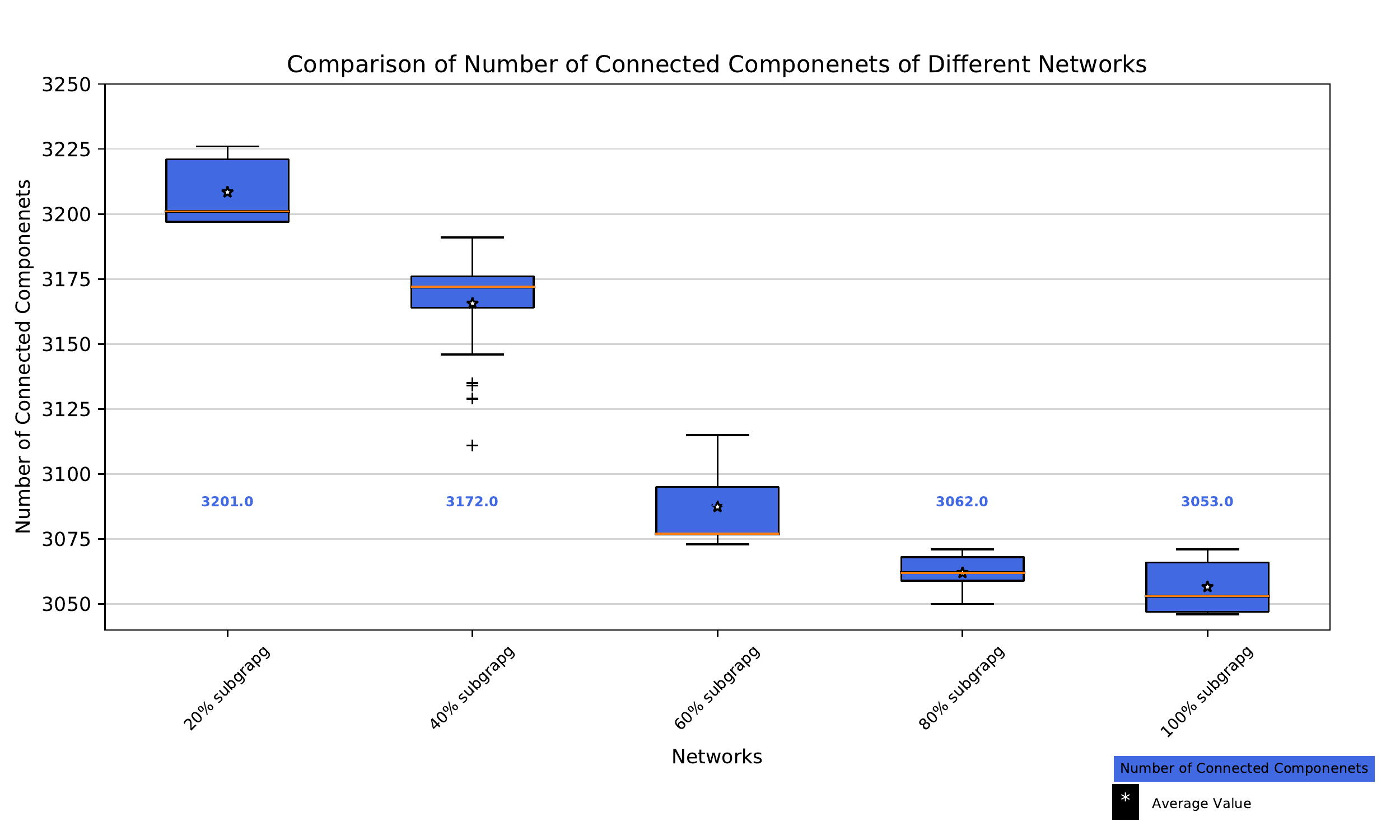}
\caption[\textbf{Box plot for  number of  connected components}]{Box plot representing number of  connected components, $Nc$, for each group of networks: a significant difference is observed  in $Nc$ between each group, $Nc$ is lower in subgraph of social networks.}
\label{snc}
\end{figure}
\begin{figure}[htp]
\centering
\includegraphics[width=.7\paperwidth,height=.3\paperheight]{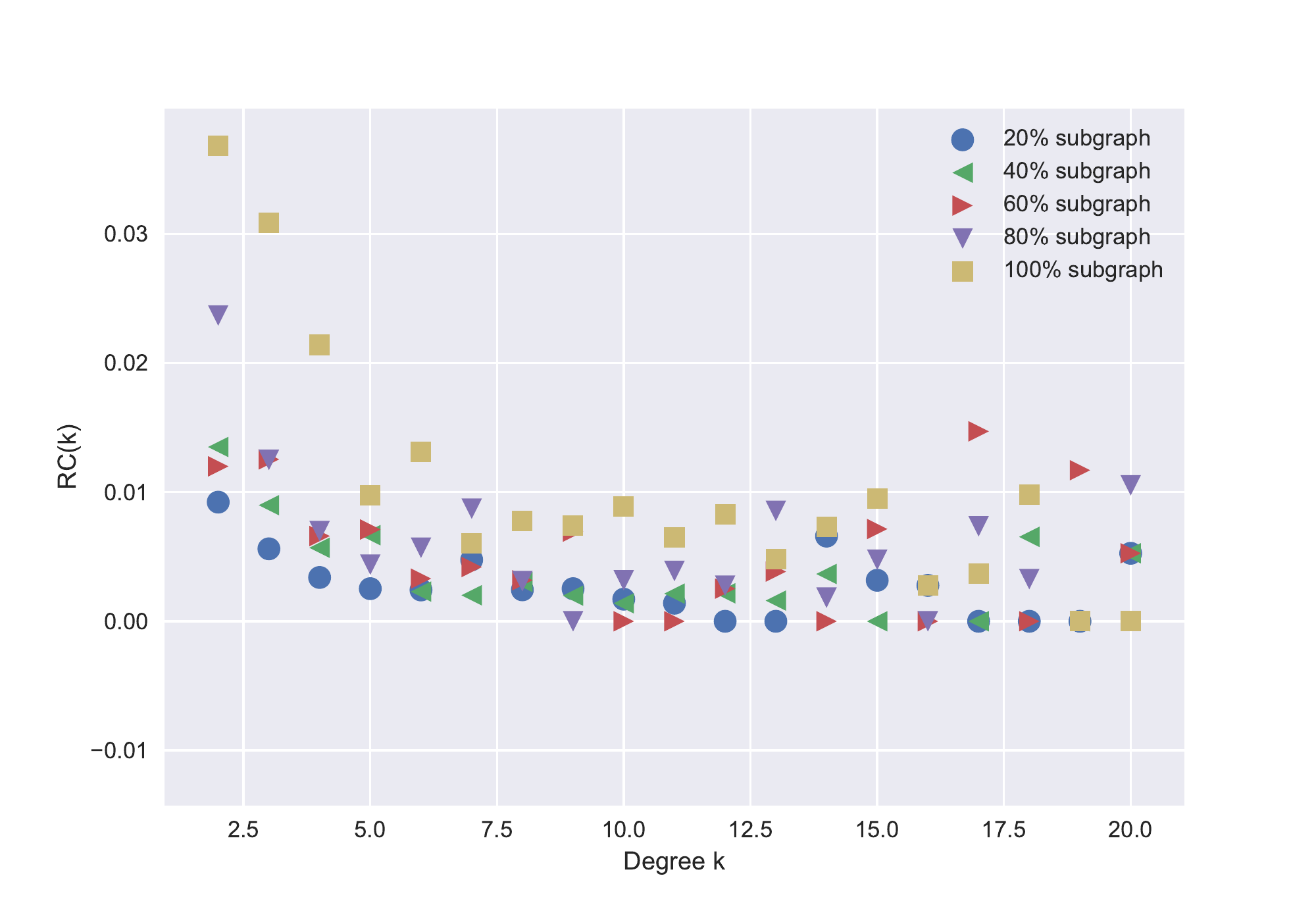}
\caption[\textbf{Redundancy coefficient vs degree}]{Scatter plot of $Rc$  versus degree for five different networks:  $Rc$ for sexual network  which is  strong subgraph of social network is higher, because clustering coefficient for social network is high and therefore, sexual network inherits this property by having higher $Rc$ than the ones which are weak subgraph of social network. }
\label{cl_rc}
\end{figure}

The number of connected components, $Nc$, is another measure characterizing  network toughness. This measure can be, but not necessarily, correlated to size of components of network.  The Figure (\ref{snc})  display descriptive statistics for $Nc$ in each network group. Note that  data distributions are approximately symmetrical, and measures of $Nc$ are similar across groups, but, they change  by changing the way of selecting sexual partners.

 We also compare $Rc$ for the networks in Figure (\ref{cl_rc}):  each data point $Rc(k)$ for degree $k$ is obtained by averaging redundancy coefficient over the group of  people with $k$ partners. In most of the networks these  values decrease with $k$ \cite{newman2010networks}.  Redundancy coefficient $Rc$ is affected by social network: when people select more partners from their social friends the value for $Rc$ increases which this fact is because of Phase. 1 of the algorithm in Appendix (\ref{ap2})- the extension of social network- when we connect friends of a person in social network it means we increase its clustering coefficient. Therefore, when sexual network is subgraph of social network it carries this property by increasing $Rc$ of the network.

\subsection{Age distribution of the network}
When generating sexual network, we tried to select the partners for each person related to his/her age such that we can capture age data reported in Figure (\ref{age_women}).  The left panel of Figure (\ref{age_age_net}) is the box plot of age for all social friends of men with different age and right panel is the box plot for the sexual partners. As it is observed while the age distribution for social friends is uniformly distributed, the age of sexual partners is correlated to the age of index men.

\begin{figure}
\centering
\includegraphics[width=.75\paperwidth,height=.3\paperheight]{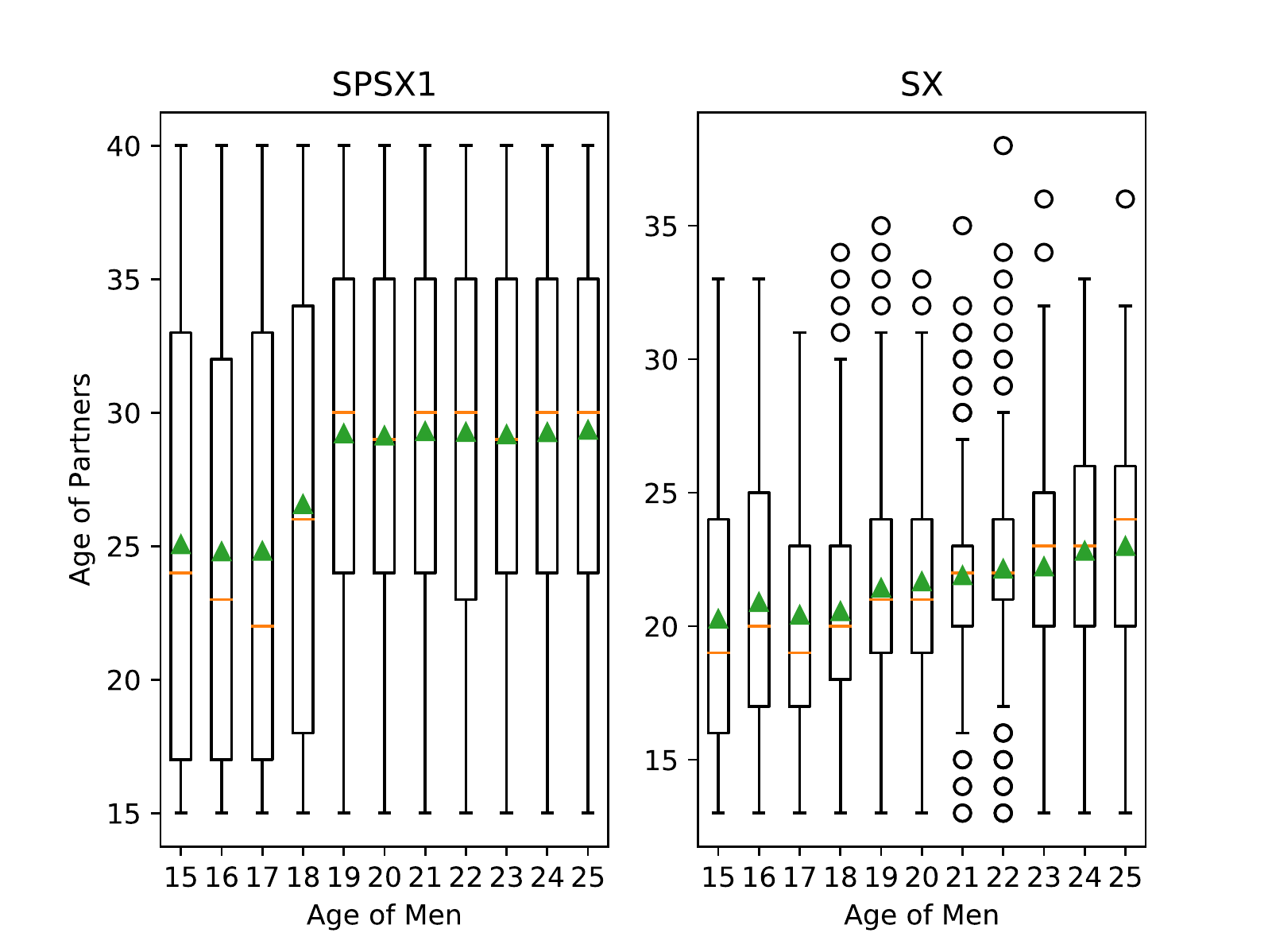}
\caption[\textbf{Box plot of age distribution for social friends and sexual partners of men}]{ Box plot of age for  all social friends (left panel) and sexual partners (right panel) of men with different age: though in social activities partner's age are distributed uniformly, in the generated sexual network $100\%$ embedded in social contact one  age of partners are correlated to age of men.}
\label{age_age_net}
\end{figure}

\section{Ct Transmission Model Overview}

In our heterosexual network model, each person is a distinct identity represented by a node in the network.
This model structure allows the sexual partnership dynamics, such as partner concurrency, sexual histories of each person, and complex sexual networks, to be governed at the individual level.

Each node \textbf{i} in the network represents a person, denoted by the index \textbf{i}, and each edge $\textbf{ij}$ represents sexual partnership between two nodes \textbf{i} and \textbf{j}.
The network is weighted, where the weight $0<w_\textbf{ij}\leq 1$ for edge $\textbf{ij}$ is the probability that there will be a sexual act between two partners \textbf{i} and \textbf{j} on an average day.  
In the model, each day and through a stochastic process, the edge  $\textbf{ij}$ will exist (turn on) with probability $w_\textbf{ij}$, the probability that two nodes \textbf{i} and \textbf{j} have sexual act in that day, or not exist (turn off)  with probability of $1-w_\textbf{ij}$, which is equivalent to not having an  edge (sexual act) between individuals \textbf{i} and \textbf{j} on that day. For the values for $w_\textbf{ij}$ we use the values in Table (\ref{act_value}) in which if \textbf{i} is woman then $w_\textbf{ij}=a(\textbf{deg(i)})$
otherwise $w_\textbf{ij}=a(\textbf{deg(j)})$

In our stochastic Susceptible--Infectious--Susceptible (SIS) model, a person \textbf{i} is either infected with Ct, $I_\textbf{i}(t)$, or susceptible to being infected, $S_\textbf{i}(t)$.
During the day $t$, an infected person, $I_\textbf{j}(t)$, can infect any of the susceptible partners, $S_\textbf{i}(t)$, they have sexual act with. 
We define $\lambda_{ij}$ as the  probability that $S_\textbf{i}(t)$ will be infected by $I_\textbf{j}(t)$ by the end of the day, $S_i(t) \overset{\lambda_{ij}}{\rightarrow} I_i(t+1)$.
  Similarly, we define $\gamma_j$  as the probability that an infected person, $I_\textbf{j}(t)$, will recover by the end of the day,   $I_j (t) \overset{\gamma_j}{\rightarrow} S_j(t+1)$.

\subsection{The force of infection }

The force of infection, $\lambda_{ij}(t)$, is the probability that a susceptible person $S_i$ is infected on day $t$ by infected partner $I_j$.  This depends on  probability of a sexual act between persons $i$ and $j$ on a typical day, as defined by edge weight, $w_{ij}$, in the model. 
We define $\beta_{nc}$ as the probability of transmission per act when a condom is not used, and $\beta_{c}$ as the reduced probability of transmission per act when a condom is used.
The forces of infection between \textbf{i} and \textbf{j} for when condom is not used, $\lambda^{nc}_{ij}$, and for when condom is used, $\lambda^{c}_{ij}$, are defined by
\begin{equation}
 \lambda^{nc}_{ij}=\begin{cases} 
      \beta_{nc} & \textnormal{with probability}~w_{ij} \\
      0 & \textnormal{with probability}~1-w_{ij} 
   \end{cases}
~~~,~~~
\lambda^c_{ij}=\begin{cases} 
      \beta_{c} &  \textnormal{with~probability}~w_{ij} \\
      0 & \textnormal{with~probability}~1-w_{ij}~~
   \end{cases}.
\end{equation}
We assume that couples use a condom in $\kappa$ fraction of their acts correctly and that  condom is $90\%$ effective in preventing the infection from being transmitted, that is, $\beta_c=0.1\beta_{nc}$.  
\subsection{Recovery from infection}

The model accounts for infected people recovering  through  natural recovery or after being treated with antibiotics. We assume that a fraction of the people treated for infection will return later to be retested for infection.

\underline{\textbf{Natural recovery}}:  We assume that all infected people will eventually recover and return to susceptible status, even if they are not treated for infection. 
In the model, the time for  natural (untreated) recovery has an exponential distribution with an average time of infection of $\frac{1}{\gamma^n}$ days, and the duration of infection for an individual is a random number from this distribution.

\underline{\textbf{Recovery through treatment}}:
We assume that the time for infected person to recover after treatment is a log-normal distribution with an average of $\frac{1}{\gamma^t}$ days. 
That is, the duration of infection for a treated infected person \textbf{k} would be set to a random number that follows log-normal distribution, $\log \mathcal{N}(\frac{1}{\gamma^t}, 0.25)$, rounded to the nearest day.  
In the model, if that number of days is smaller than the duration remaining for naturally clearing the disease, then the shorter time is use for the recovery period. 

Each year, a fraction of the population is tested for Ct infection, e.g. through a routine medical exam (random screening) or after being notified that one of their previous partners or social friends  was infected. Here we define all biomedical interventions implemented on the network.

\emph{\textbf{Random Screening}}: We define random screening as testing for infection when there are no compelling reasons to suspect a person is infected.  For example, random screening might be part of a routine physical exam and is an effective mitigation policy to identify  asymptomatic infections.
In our model, we assume that the fraction $\sigma_y\%$ of people are randomly screened each year, and an individual is screened with probability $\sigma_d=1-(1-\sigma_y)^{\frac{1}{365}}\%$ each day.

\emph{\textbf{ Partner Notification }}: 
We assume that an infected person notifies $\theta_n^p$ fraction of their partners about their exposure to infection.
Some of the notified partners will do nothing, some will test for Ct infection, and some will seek treatment for Ct without first testing because it is simpler. The notified partners of a tested and treated person are divided into three classes: 
\begin{enumerate}
\item  Partner treatment: $\theta_t^p$ fraction of the notified partners are treated, without first testing for infection.
\item Partner screening: $\theta_{s}^p$ fraction of the notified partners engage in screening test and then start treatment if infected.
\item Do nothing: ($1-\theta_t^p-\theta_{s}^p$) fraction are not tested or treated.
\end{enumerate}
For simplicity, we assume notified partners are the ones who are notified and do something, that is, we assume  $1-\theta_t^p-\theta_{s}^p=0$. Then we define  $\theta_n^p\theta_t^p$ as fraction of partners of an infected person who are treated without first testing for infection (Partner treatment), and $\theta_n^p\theta_{s}^p$ as fraction of the partners follow screening test (Partner screening).

Partner notification spreads across the same network that originally spread the disease.  The partner screening approach is more effective since every time the partner of an infected person is found to be infected, then the cycle repeats itself and their partners are notified.

\emph{\textbf{ Social Friends Notification \footnote{When we say social friends we mean friends who are not sexual partners that is, if a person is both partner and social friend of an index case, we consider him/her as sexual partner not social friend.}}}: 
We also assume that an infected person notifies $\theta_n^f$ fraction of their social friends about screening test.
Some of the notified friends will do nothing, and some will test for Ct infection. The notified  friends of a treated person are divided into two classes: 
\begin{enumerate}

\item Social Friend Screening: $\theta_{s}^f$ fraction of the notified friends engage in screening test and then start treatment if infected.
\item Do nothing: ($1-\theta_{s}^f$) fraction are not tested or treated.
\end{enumerate}
For simplicity, we assume notified friends are the ones who are notified and do screening test, that is, we assume  $1-\theta_{s}^f=0$. Then we define $\theta_n^f\theta_{s}^f$ as fraction of the friends follow screening test (Social friend screening).

The model includes a time-lag of $\tau_N$ days between the day a person is found to be infected and the day their partners or friends are notified.

\textit{\textbf{Rescreening}}:
A common practice in disease control is rescreening.
People found to be infected are given a treatment and asked to return after a short period to be tested again.
In our model, we assume that a fraction, $\sigma_r$, of the treated people return for retesting $\tau_R$ days after treatment.

The Figure (\ref{fig:intervention_flow}) shows the diagram of above explained biomedical interventions.

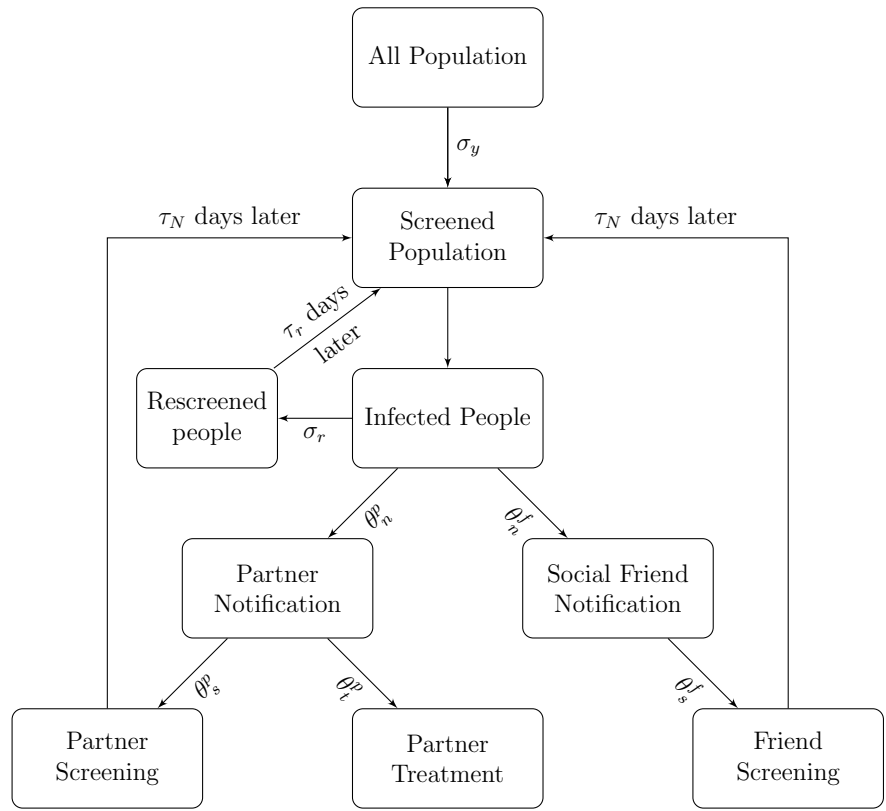
\begin{figure}[htp]
\centering
\tikzstyle{block} = [rectangle, draw, fill=white!20, 
    text width=7em, text centered, rounded corners, minimum height=4em,node distance=3cm]
\tikzstyle{line} = [draw, -latex']
 \begin{tikzpicture}[scale=0.8, every node/.style={scale=0.8}]
    \node [block] (all) {All Population};
    \node [block, below of=all] (scr) {Screened Population};
    \node [block, below of=scr] (ifct) {Infected People};
     \node [rectangle, draw, fill=white!20, 
    text width=5em, text centered, rounded corners, minimum height=4em,node distance=4cm, left of=ifct] (reifct) {Rescreened people};
     \node [block,node distance=4cm, below left of=ifct] (pn) {Partner Notification};
      \node [block,node distance=4cm, below right of=ifct ] (fn) {Social Friend Notification};
    \node [block,node distance=4cm, below right of=fn] (fs) {Friend Screening};
      \node [block,node distance=4cm, below right of=pn] (ept) {Partner Treatment};
      \node [block,node distance=4cm, below left of=pn] (eps) {Partner Screening};
   \path [line] (all) -- (scr);
    \path [line] (scr) -- node[right]{}(ifct);
   \path [line] (ifct) -- node [sloped, anchor=center, below] {$\theta_n^p$} (pn);
    \path [line] (ifct) -- node [sloped, anchor=center, below] {$\theta_n^f$} (fn);
    \path [line] (pn) -- node [sloped, anchor=center, below] {$\theta_t^p$} (ept);
    \path [line] (pn) -- node[sloped, anchor=center, below]{$\theta_s^p$} (eps);
     \path [line] (fn) -- node[sloped, anchor=center, below]{$\theta_s^f$} (fs);
     \path [line] (ifct) -- node[sloped, anchor=center, below]{$\sigma_r$} (reifct);
     \path [line] (reifct) -- node[ sloped,above]{$\tau_r$ days}node[ sloped,below]{ later} (scr);
    \path [line] (fs) |- node [near end, above]{$\tau_N$ days later} (scr);
    \path [line] (eps) |- node [near end, above]{$\tau_N$ days later} (scr);
   \path [line] (all) -- node[right] {$\sigma_y$}(scr);
\end{tikzpicture}
\caption[\textbf{Biomedical intervention flow diagram}]{Biomedical intervention flow diagram.}
\label{fig:intervention_flow}
\end{figure}

 \subsection{Model initialization}
Our goal is to model the current Ct epidemic in New Orleans with an initial prevalence of $i_0$.
Infected people are not dropped into an otherwise susceptible population, instead they are distributed as they would be as part of an emerging epidemic, one that started some time in the past.
We call these initial conditions \textit{balanced} because when the simulation starts the infected and susceptible populations, along with durations of infection, are in balance as an emerging epidemic would on average have.
When the initial conditions are not balanced, then there is usually a rapid (nonphysical) initial transient of infections that quickly dies out as the infected and susceptible populations relax to a realistic infection network.

To define the balanced initial conditions, we start an epidemic in the past by randomly infecting a few high degree individuals.
We then advance the simulation until the epidemic grows to the prevalence $i_0$.
We then reset the time clock to zero and use this distribution of infected people, complete with their current infection timetable, as our initial conditions. 
Because these are stochastic simulations, when doing an ensemble of runs we reinitialize each simulation by seeding different initial infected individuals. 
The Figure (\ref{forbeta}) illustrates the typical progression of the epidemic to reach the current  Ct prevalence of $9\%$ in men and $14\%$ women in the $15-25$ year-old New Orleans AA community. The numerical simulations comparing the different mitigation strategies all start at this endemic stochastic equilibrium. 

\begin{figure}[H]
\centering
\includegraphics[width=100mm,scale=0.5]{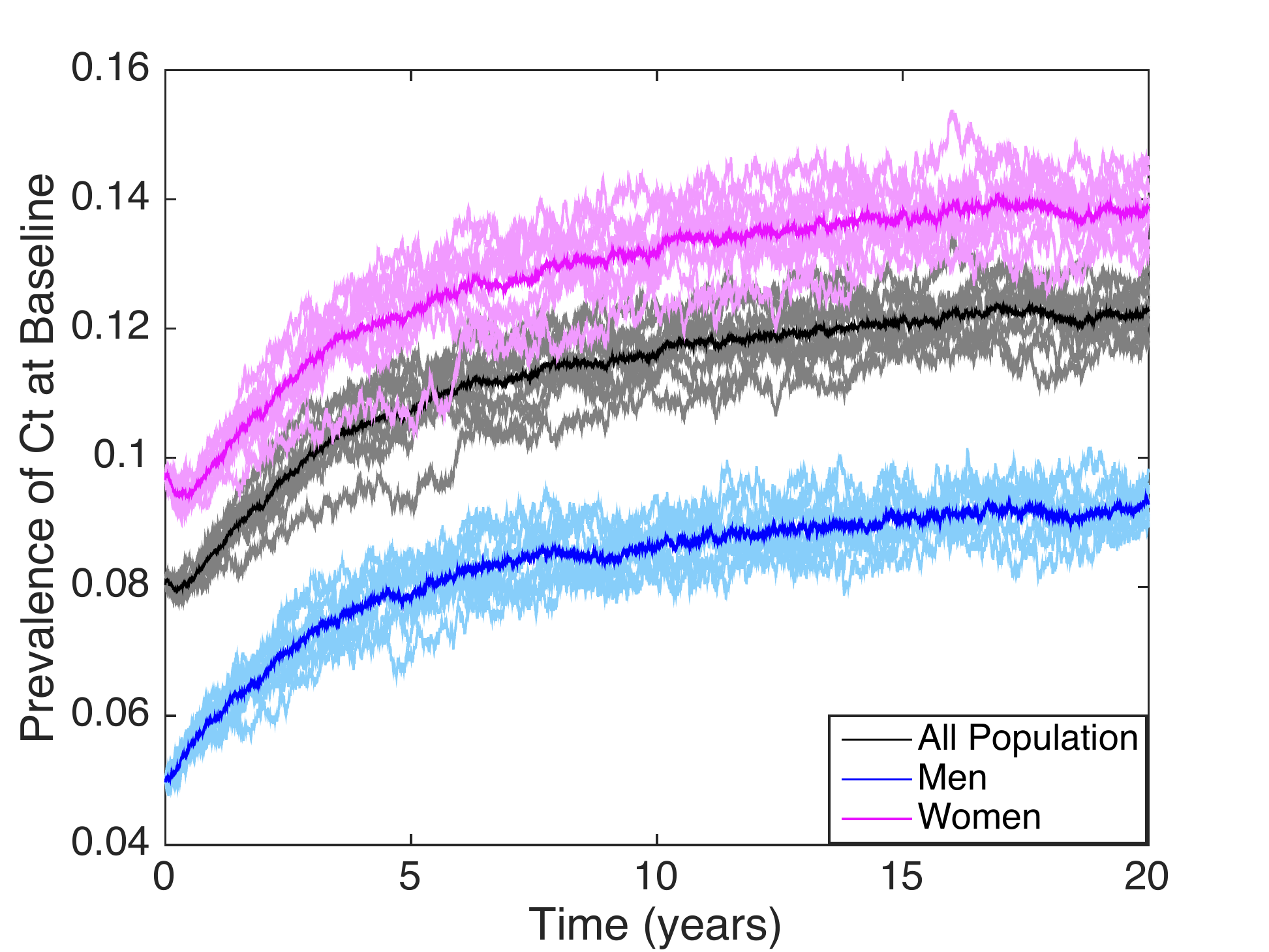}
\caption[\textbf{Prevalence of Ct vs time}]{Prevalence increases to reach the current quasi-steady state. About $9\%$ of men  (blue lower curve) and $14\%$ of women  (pink upper curve) are infected at the equilibrium for the baseline model parameters.  This is approximately the current prevalence in New Orleans 15-25 year-old AA population.
The light areas around the dark mean values show the range of the solutions after $10$ simulations. }
\label{forbeta}
\end{figure}

\section{Numerical Simulations}

We compare the model-projected impact of increased random screening, partner notification-  which includes partner screening and partner treatment- social friend notification and rescreening on the prevalence of Ct infection.
All of the simulations start at a balanced equilibrium obtained with the model baseline parameters in Table (\ref{parameters}). For probability of transmission per act, there is a wide range of reported values from $0.04$ to $0.16$ in different studies \cite{low2007epidemiological,welte2005costs,adams2007cost,andersen2006prediction,gillespie2012cost,roberts2007cost}. We calibrated this parameter to the current prevalence of Ct among adolescents and young adult AAs in New Orleans \cite{kissinger2014check}, and found out our estimated value is close to corresponding parameter in \cite{hui2013potential,low2007epidemiological}- in which probability of transmission from man to woman is $0.16$ and from woman to man is $0.12$ and prevalence is $12\%$.

\begin{table}[H]
\centering
\resizebox{\columnwidth}{!}{
\begin{tabular}{ llp{10cm}ll }
\toprule[1.5pt]
  & \multicolumn{4}{c}{\head{}}\\
  & \normal{\head{Parameter}} & \normal{\head{Description}}
  & \normal{\head{Unit}} & \head{Baseline}\\
  \cmidrule(lr){2-3}\cmidrule(l){3-5}
  \multirow{2}{2.5cm}{Network Parameters} &  P& Number of nodes in \textbf{SexNet}. & people & $15000$\\
 & $\alpha_m$ & Sexually active age range for men. & years &  $[15,25]$\\
  & $\alpha_w$ & Sexually active age range for women. & years &  $[15,40]$\\

  \cmidrule(lr){2-3}\cmidrule(lr){3-5}
  \multirow{2}{1.1cm}{Disease Parameters} & $\beta^{m2w}$ & Probability of transmission per act from men to women. &--
  & $0.16$ \\
 & $\beta^{w2m}$ & Probability of transmission per act from women to men. & -- & $0.16$\\
        &$1/\gamma^n$& Average time to recover without treatment.   & days&  $365$   \\
\cmidrule(lr){2-3}\cmidrule(lr){3-5}
\multirow{4}{1.1cm}{Intervention Parameters} 
&$1/\gamma^s$& Average time to recover with treatment.   & days&  $7$   \\ 
&$\kappa$& Fraction of times that condoms are  used during sex.  & --& $0.6$  \\
&$\epsilon$& Condom effectiveness.  & --& $0.90$    \\
& $\sigma^m_y(\sigma_y^w)$ & Fraction of men(women) randomly screened per year. & -- & $0.05(0.40)$ \\
        & $\sigma_r$ & Fraction of infected people return for rescreening. & - & $0.10$\\
        & $\theta_n^p(\theta_n^f)$ & Fraction of the partners (friends) of an infected person who are notified and do test or treated for infection. & - & $0.26 (0)$ \\
        & $\theta_t^p$ & Fraction of notified partners  of an infected person who are treated without testing. & - & $0.75(0)$\\
        & $\theta_s^p (\theta_s^f)$ & Fraction of notified partners (friends) of an infected person who are tested and treated for infection.  & - & $0.25(0)$\\
      & $\tau_N$ & Time lag of partner notification.  & days & $5$\\
        & $\tau_R$ & Time lag of rescreening. & days & $100$\\
\bottomrule[1.5pt]
\end{tabular}}
\caption[\textbf{Parameters of the model and their baseline values}]{Parameters and their baseline values: the model parameters  describing the transmission of Ct infection, as well as recovery associated with natural recovery, and interventions  were obtained from the literature \cite{kissinger2014check,althaus2010transmission,morre1998monitoring}, but other parameters are calibrated to biological, behavioral, and epidemiological data from general heterosexual population resides in New Orleans.}
\label{parameters}
\end{table}

\subsection{Dynamic of Ct on networks}
To determine the impact of social network in  sexual partner selection on spread of Ct, we compared the prevalence versus time for networks with different structures: networks in which people select their sexual partners from different sources. We want to observe  whether different network properties mentioned in Section \ref{NA2} causes a drastic difference on prevalence of Ct, and we hope not. Otherwise, it means that we cannot predict Ct prevalence, because there exist some  parameters in generating network that are not  preserved but has  key  role in the spread of infection. 
Thus, we introduced  infection over different networks and let them spread  till they converge to quasi-steady state\footnote{In stochastic SIS model the steady state is disease free equilibrium point i.e if we run the model we eventually reach zero infection point even if $\mathcal{R}_0\geq 1$, the amount of $\mathcal{R}_0$  affect the time  to reach 0 infection point. Here saying quasi-steady state we mean that we run the simulations for a big enough time and we stop it if the average of simulations for the next day after final time is relatively close to average of prevalence for the final day.}.  Figure \ref{fig:com_pre} is Ct spread over time for five different sexual networks:  $20\%, 40\%, 60\%, 80\%$ and $100\%$ subgraph of social network. Each curve correspondent to different network is the average of $50$  runs where each run initialized by seeding same initial balanced  infected individuals. The progression of Ct over time for these networks is not exactly the same, however they are slightly close to each other, the maximum difference of prevalence at quasi-steady state is less than one percent. Therefore, we can conclude that Ct prevalence has a mild  dependence on sexual partner source. For this work we can ignore this dependence, however to have a better estimation on network structure, the question about source of partner selection from participants is needed.
From now on for all the simulations, we use the sexual network which is $80\%$ subgraph of social network, unless stated otherwise.

To find probability of transmission per sexual act $\beta$, we calibrated it to current Ct prevalence  of $9\%$ in men and $14\%$ women in the $15-25$ year-old New Orleans AA community. The Figure \ref{fig:base} illustrates the typical progression of the epidemic to reach this current  Ct prevalence.

\begin{figure}[htp]
\centering
\subfloat[\textbf{Prevalence of Ct vs time for different networks}][]{\label{fig:com_pre}\includegraphics[width=0.5\textwidth]{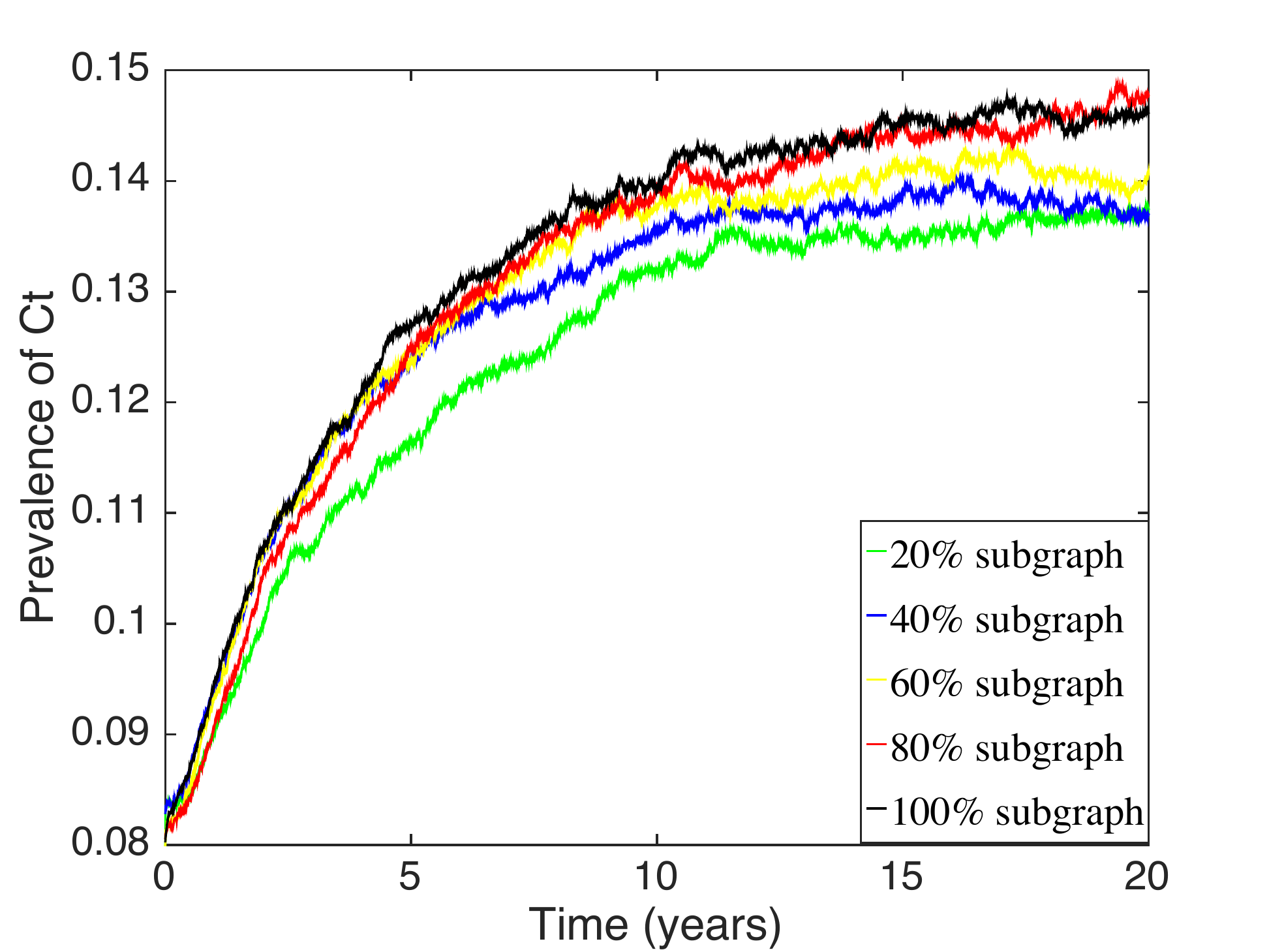}}
\subfloat[\textbf{Prevalence of Ct vs time at baseline}][]{\label{fig:base}\includegraphics[width=0.5\textwidth]{Figures/baseline.pdf}}\\
\begin{minipage}{.9\textwidth}
\caption[\textbf{Ct dynamic on network}]{(a) Comparison of effect of different networks on Ct prevalence: Ct at quasi-steady state is mildly dependent on sexual partner source. There is less than one percent difference on Ct prevalence on different networks. (b) Ct prevalence at baseline values: about $9\%$ of
men (blue lower curve) and $14\%$ of women (pink upper curve) are infected at the
converging point for the baseline model parameters.  The light areas around the
dark mean values (mean is average of $50$ different stochastic simulations) show the range of the solutions for only  $10$ simulations.}
\label{centrality}
\end{minipage}
\end{figure}

\subsection{ Analysis of infected population at quasi-steady state}
The structure of sexual network  plays an important role on the spread of infection and the status of infected people at quasi-steady state. To implement a proper intervention strategy, we look at the properties of infected people at quasi-steady state. These properties tell us about the highest risk people in network. We compare degree, betweenness and closeness centrality\footnote{These centralities are defined in Appendix (\ref{ap1})} of infected people at quasi-steady state.

 Each panel of  Figure (\ref{centrality}) compares different centrality score of all people and infected people at quasi-steady state.
The panels \ref{fig:deg_cen} and \ref{fig:bet_cen} suggest that people with higher degree score (ones with many partners), or with higher betweenness score (people who are in the shortest path between many other people) are not necessarily at higher risk of infection, because the distribution of degree and betweenness centrality for all people and infected people are the same. In panel \ref{fig:clo_cen} we observe  different distribution shape for all people and infected people at quasi-steady state: most of infected people at quasi-steady state are from people with higher closeness centrality scores (individuals  who are reachable for many other people in  network by a path).

\begin{figure}[htp]
\centering
\subfloat[\textbf{Degree Centrality}][Degree Centrality]{\label{fig:deg_cen}\includegraphics[width=0.5\textwidth]{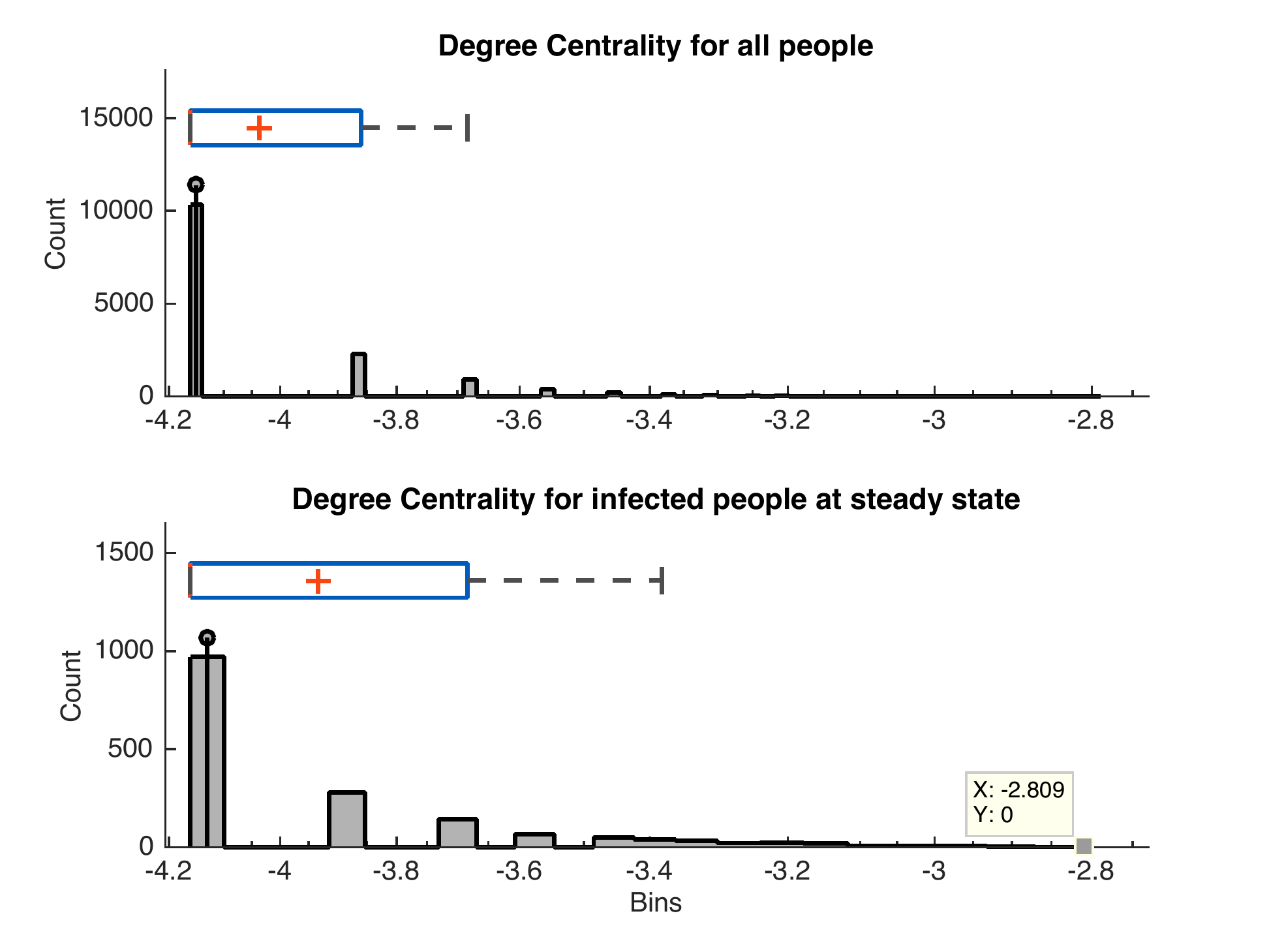}}
\subfloat[\textbf{Betweenness Centrality}][Betweenness Centrality]{\label{fig:bet_cen}\includegraphics[width=0.5\textwidth]{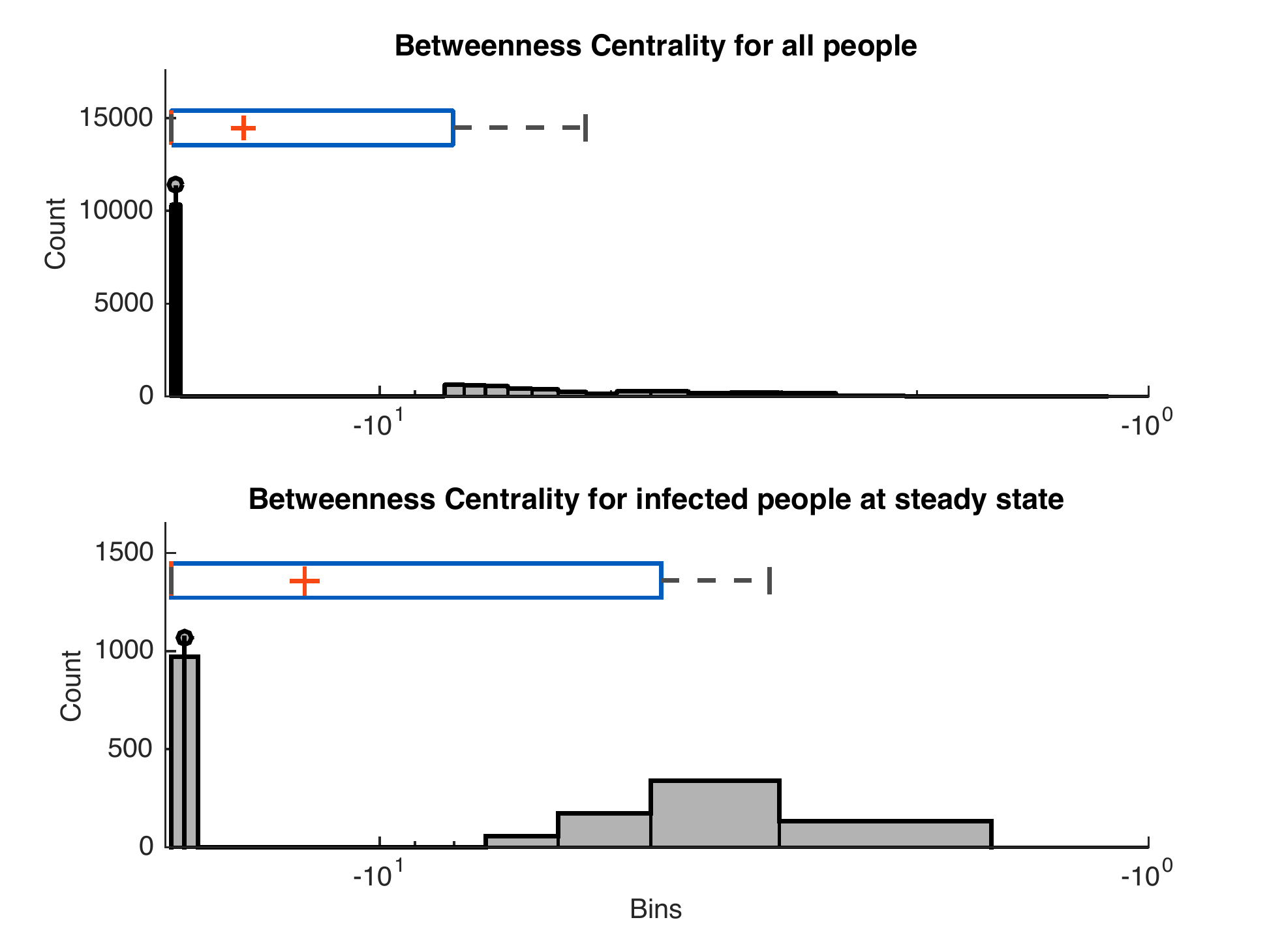}}\\
\subfloat[\textbf{Closeness Centrality}][Closeness Centrality]{\label{fig:clo_cen}\includegraphics[width=0.55\textwidth]{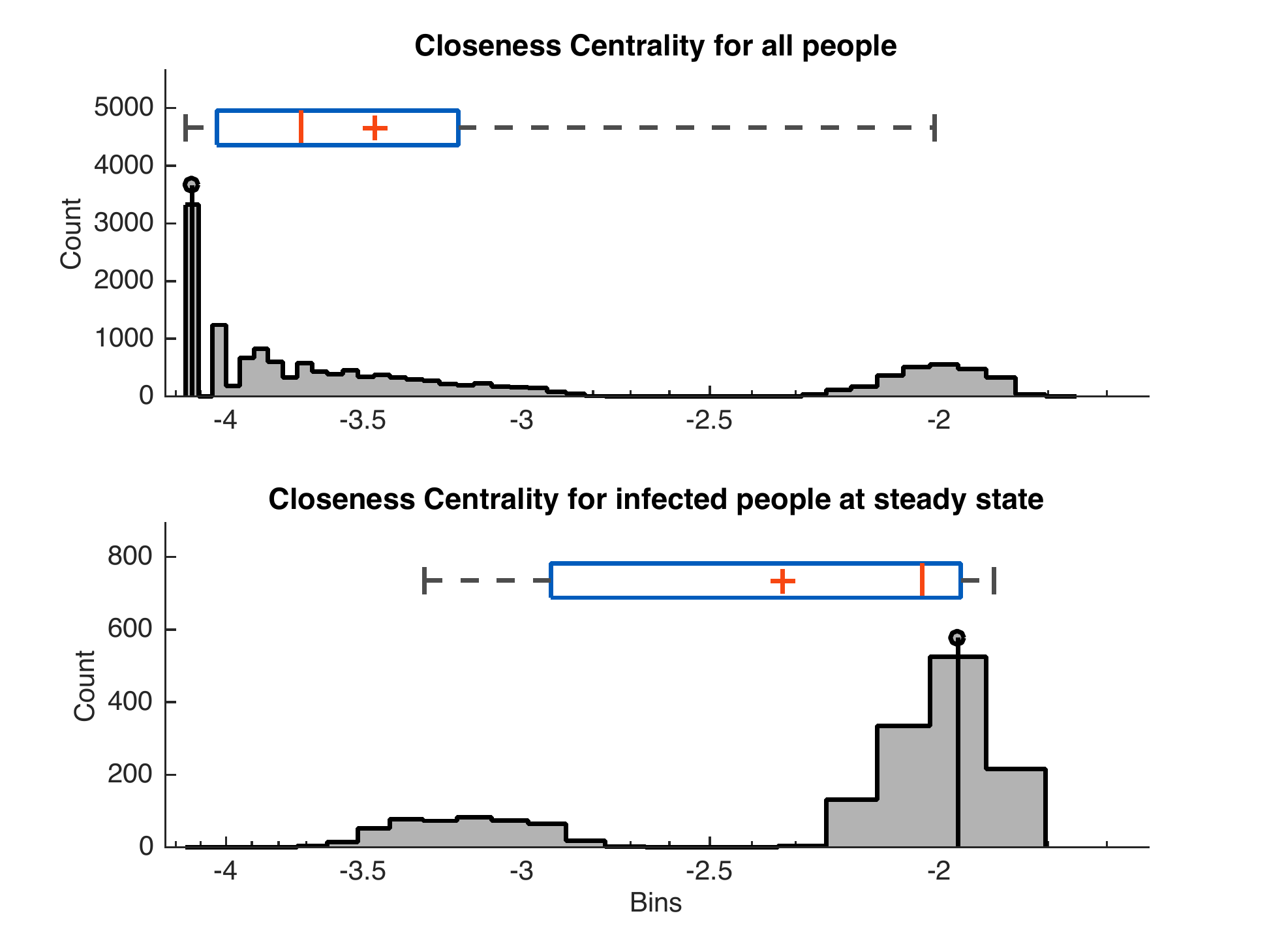}}
\begin{minipage}{.43\textwidth}
\vspace*{-5cm}
\caption[\textbf{Bar plot of distributions for }]{(a) \textbf{Degree Centrality:} the degree distribution for all people and  infected ones have the same trend, thus, degree centrality score is not  key parameter in infection spread. (b) \textbf{Betweenness Centrality:} the betweenness  distribution for all people and  infected ones have the same trend, thus, betweenness centrality score is not  key parameter in infection spread. (c) \textbf{Closeness Centrality:} most of infection at quasi-steady state is clustered on  people with high closeness score, that is, people with high reachability are at higher risk of infection. }
\label{centrality}
\end{minipage}
\end{figure}

 This information suggests us that people with higher and faster reachability are at higher risk of infection. We explain this result in schematic Figure (\ref{case_deg_clo}): in this network node \textbf{i} has $5$ neighbors and reaches total $5$ nodes in the network, but node \textbf{j}, in spite of having fewer neighbors, reaches $9$ other nodes in the network and therefore, is at higher risk for catching or transmitting the infection.

 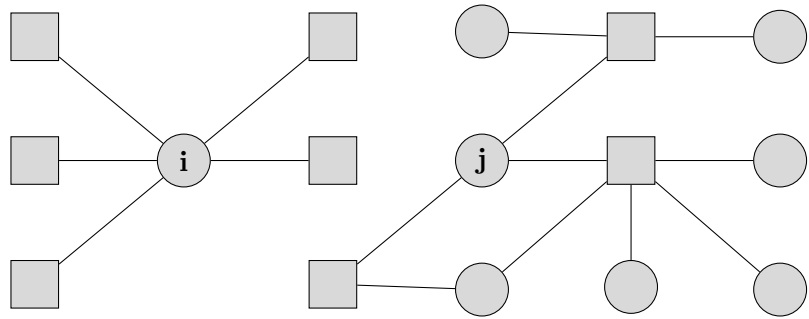
\begin{figure}[htp]
\centering
  \begin{tikzpicture}[
      mycircle/.style={
         circle,
         draw=black,
         fill=gray,
         fill opacity = 0.3,
         text opacity=1,
         inner sep=0pt,
         minimum size=20pt,
         font=\small},
         myrec/.style={
         rectangle,
         draw=black,
         fill=gray,
         fill opacity = 0.3,
         text opacity=1,
         inner sep=0pt,
         minimum size=18pt,
         font=\small},
      myarrow/.style={-},
      node distance=1.cm and 1.3cm
      ]
      \node[mycircle] (i) {$\textbf{i}$};
      \node[myrec, right=of i] (k1) {};
      \node[myrec,above =of k1] (k0) {};
       \node[myrec,below =of k1] (k2) {};
       \node[myrec, left=of i] (k3) {};
      \node[myrec,above =of k3] (k4) {};
       \node[myrec,below =of k3] (k5) {};
       \node[mycircle,right =of k1] (j) {\textbf{j}};
       \node[myrec, right=of j] (k6) {};
       \node[mycircle, right=of k6] (l) {};
      \node[myrec,above =of k6] (k7) {};
      \node[mycircle, right=of k7] (w) {};
       \node[mycircle,below =of k6] (k8) {};
      \node[mycircle,above =of j] (k10) {};
       \node[mycircle,below =of j] (k11) {};
      \node[mycircle,below =of l] (k) {};
      \draw [myarrow] (i) -- node[sloped,above] {} (k0);
     \draw [myarrow] (i) -- node[sloped,above] {} (k1);
     \draw [myarrow] (i) -- node[sloped,above] {} (k3);
     \draw [myarrow] (i) -- node[sloped,above] {} (k4);
     \draw [myarrow] (i) -- node[sloped,above] {} (k5);
     \draw [myarrow] (j) -- node[sloped,above] {} (k2);
     \draw [myarrow] (j) -- node[sloped,above] {} (k7);
     \draw [myarrow] (k6) -- node[sloped,above] {} (k8);
     \draw [myarrow] (k6) -- node[sloped,above] {} (l);
     \draw [myarrow] (k6) -- node[sloped,above] {} (k11);
     \draw [myarrow] (k2) -- node[sloped,above] {} (k11);
     \draw [myarrow] (k7) -- node[sloped,above] {} (k10);
     \draw [myarrow] (k7) -- node[sloped,above] {} (w);
     \draw [myarrow] (k) -- node[sloped,above] {} (k6);
     \draw [myarrow] (j) -- node[sloped,above] {} (k6);
    \end{tikzpicture}
    \caption[\textbf{Difference between degree and closeness scores}]{An example to compare degree and closeness scores for two typical nodes: node \textbf{i} has higher degree ($de(\textbf{i})=0.33$) than node \textbf{j} ($de(\textbf{j})=0.2$) but reaches fewer nodes ($5$) than \textbf{j} ($9$), therefore, \textbf{j} is at higher risk of infection than \textbf{i}.}
    \label{case_deg_clo}
\end{figure}

That means individuals who are close to too many other people in the network are good candidates to be tracked by public health staff. But there is no  clue about reachability of a typical person. The only information we can ask individuals is about their number of partners (their degree) or the number of partners of partners. Thus, the first question may come in mind is that if there is any correlation between degree and closeness score of a node. The Figure \ref{fig:cor} is the scatter plot of degree and closeness score of people in our sexual network. The Figure does not show any correlation between degree and closeness of people in sexual network. 

To  have a deeper look into relation between degree and closeness scores, we plot the distance-reachability cumulative distribution for each group of people with $k$ partners in Figure \ref{fig:dis_reach1}: the point $(x,y)$ in each curve shows the average probability of reaching $y$ fraction of population within at most $x$ steps. For example on average a person with $15$ partners reaches to $12.5\%$ of whole population within at most $20$ steps or through people who are at most in distance $20$ of him/her. For the networks  with homophily in the degree, it is obvious that this distance-reachability cumulative distribution for higher degree individuals should move faster and reaches more people. In sexual network, for most of the degree groups we observe this pattern, however, for some high degree groups we see opposite pattern. For example curve for  people with $15$ partners grows faster and reaches more  in compared with the curve for  people with $16$ partners. 

This plot is for just one sexual network, but we should not fool ourselves based on one network. 
We generated $30$ sexual networks and then for each of these networks we evaluated the fraction of population  in each degree group can reach within $20$ steps. The Figure \ref{fig:dis_reach2} shows the result: each point $(x,y)$ corresponds to one sexual network   and tells that on average a person with $x$ partners reaches $y$ fraction of population within at most $20$ steps. The red curve is smoothing spline fitted to the circle data. The large fluctuation  for large degree x  is because of paucity of people with high number of partners. The Figure tells that for lower risk people (people with few or moderate amount of partners)   there is a linear correlation between degree and reachability, that is, when people increase their  number of partners they can reach higher fraction of sexually active population. However, increasing the number of partners to   $10$ or more, the reachability fraction converges to a constant value $0.03$ meaning that no matter how many partners a person has, he/she cannot reach more than $3\%$ of whole population within at most $20$ steps. That is, to find proper person for screening- person with a high chance of carrying infection- the number of partner he/she may have matters if they have less than $10$ partners. In other words, for people with degrees in range $[1,10]$, more partner means reaching more people in the network which causes higher risk of catching  Ct infection. But  people with more than $10$ partners reach the same fraction of population, therefore, there would not be any screening liability  among them. 
This observation is similar to the result in Chapter (\ref{continuous}) that studied the impact of condom-use in controlling Ct. In that works, we observed that there is a threshold for the impact of number of partners on risk of catching infection: individuals with number of partners more than a threshold value have the same risk of infection no matter how many times they use condom.

\begin{figure}[htp]
\centering
\subfloat[\textbf{Degree vs Closeness Centrality}][Degree vs Closeness Centrality]{\label{fig:cor}\includegraphics[width=0.475\textwidth,height=.225\paperheight]{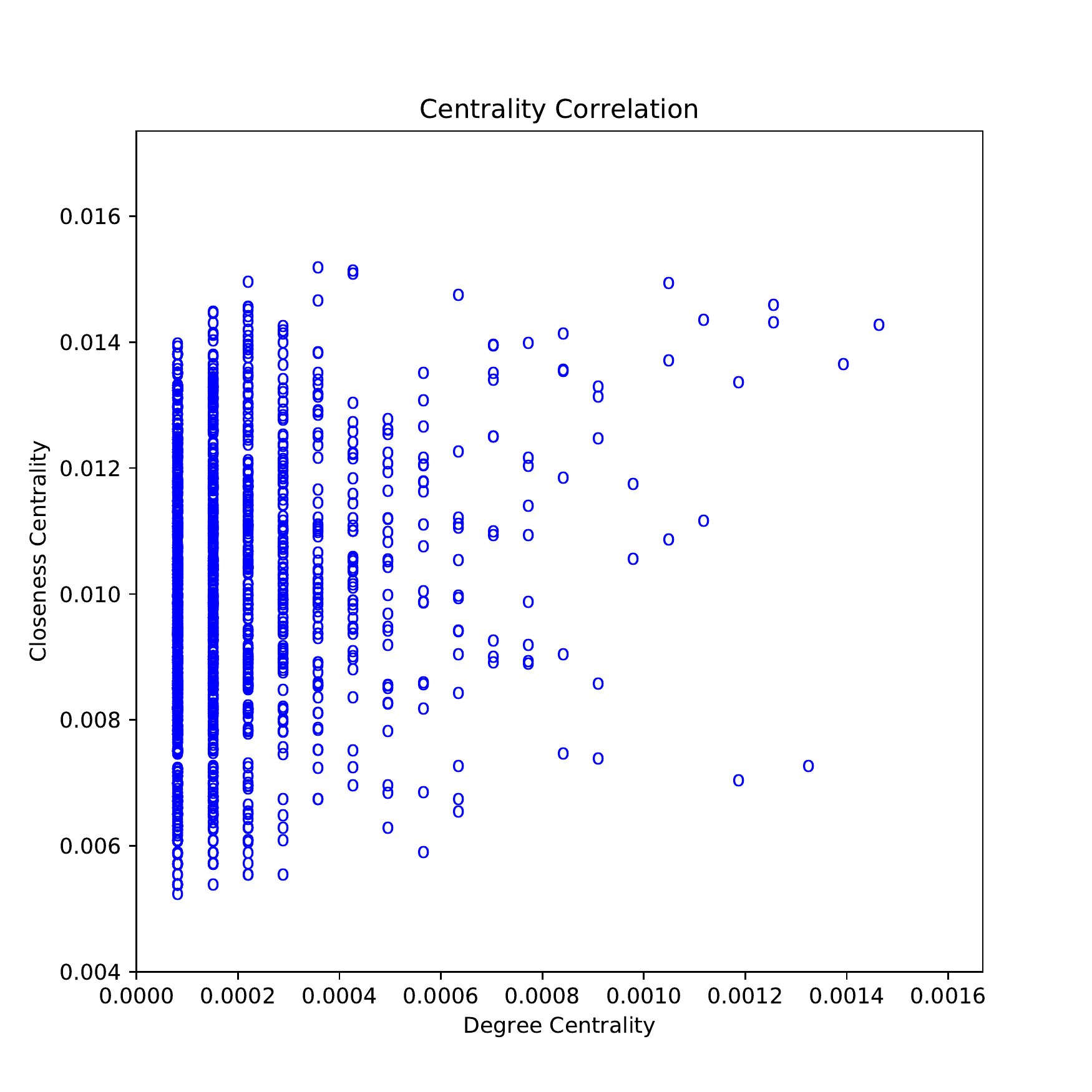}}
\subfloat[\textbf{Distance-Reach distribution}][Distance-Reach Distribution]{\label{fig:dis_reach1}\includegraphics[width=0.54\textwidth,height=.225\paperheight]{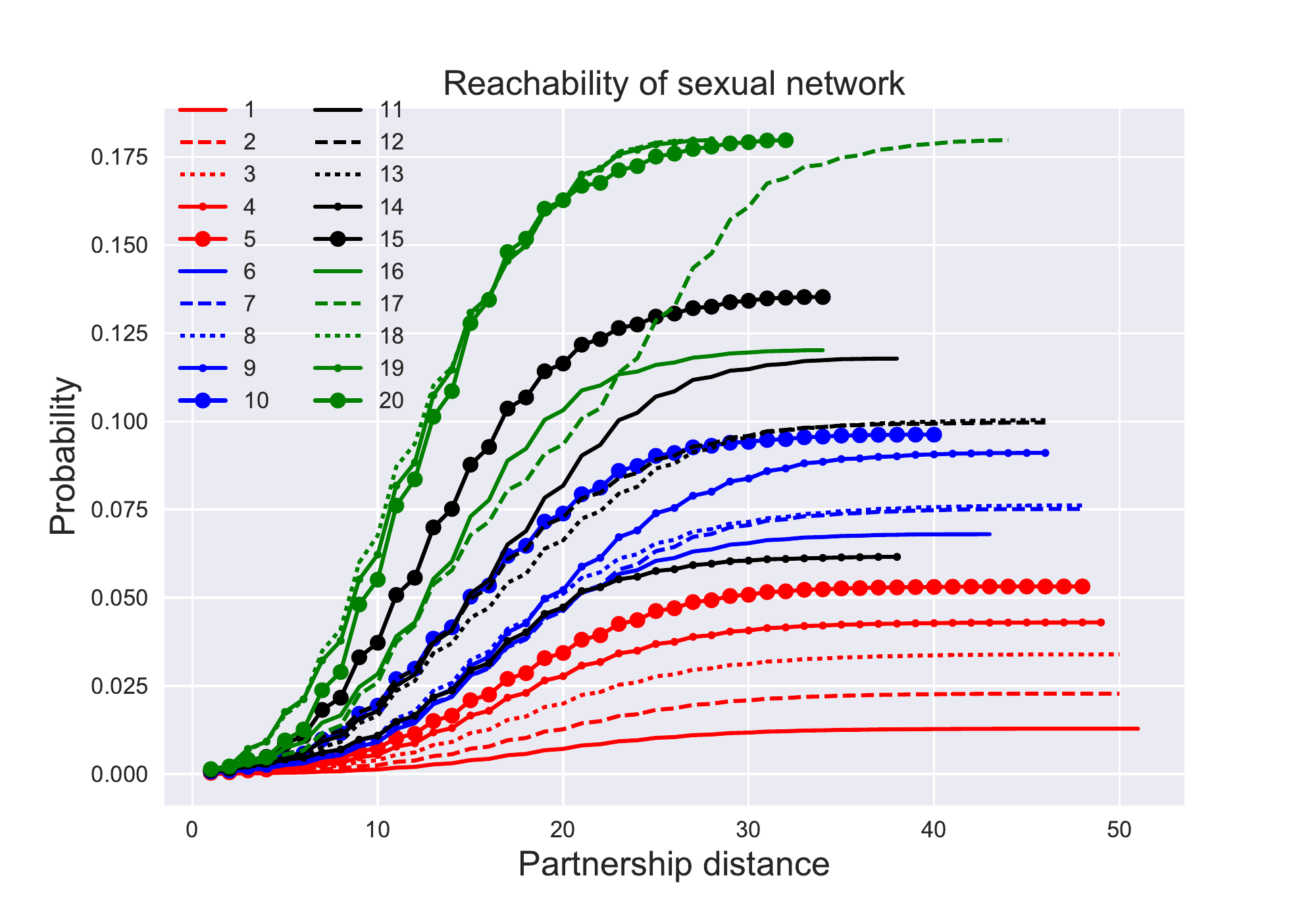}}\\
\subfloat[\textbf{Reachability Probability within 20 Distance}][{Reachability probability within 20 distance}]{\label{fig:dis_reach2}\includegraphics[width=0.42\textwidth,height=.225\paperheight]{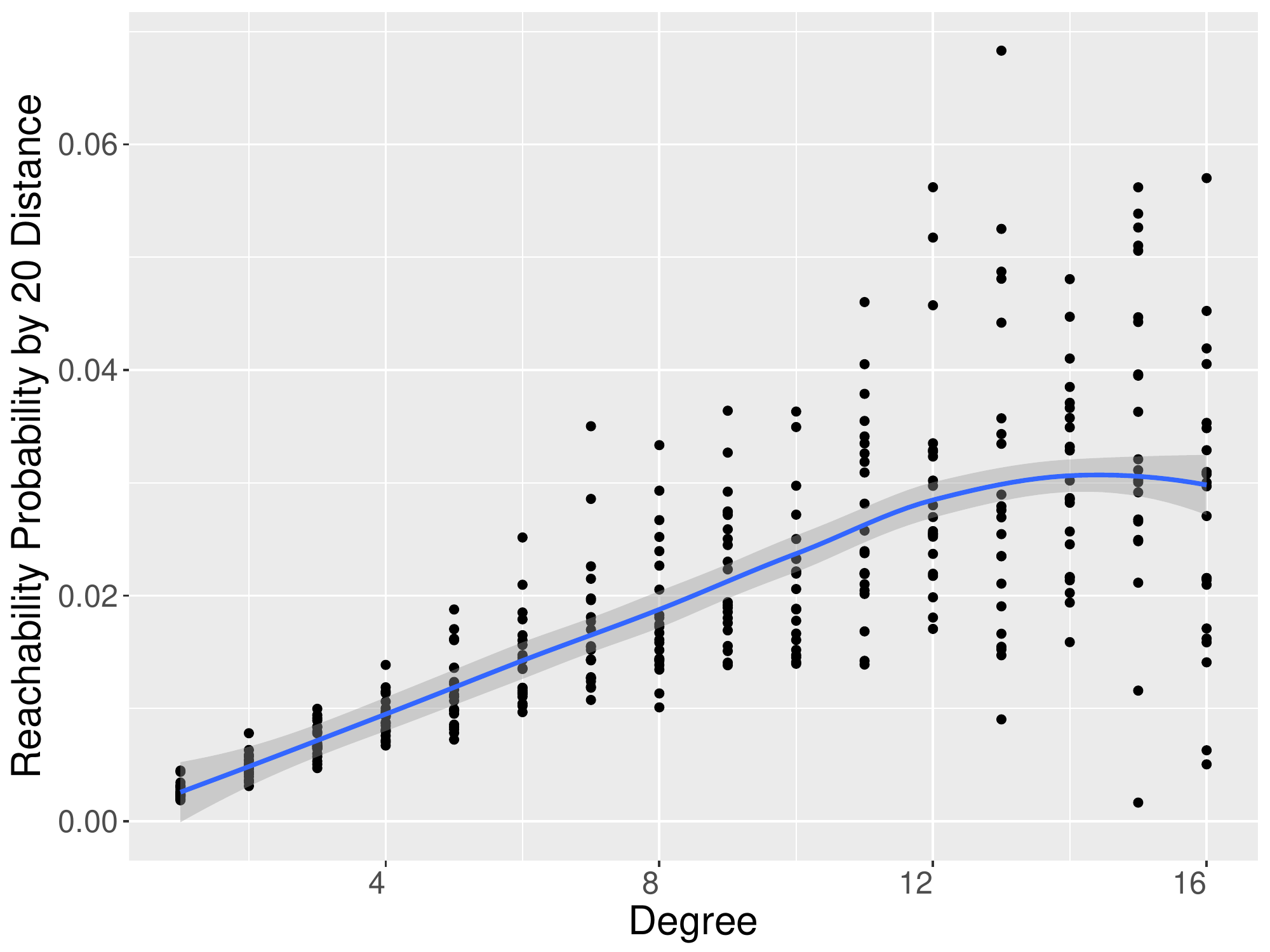}}
\hspace*{1cm}
\begin{minipage}{.5\textwidth}
\vspace*{-5cm}
\caption[\textbf{Distance reachability for nodes with different degree and closeness scores}]{(a) Degree versus closeness centrality:  the closeness score of people does not correlate with their degree. (b) \textbf{Distance-Reach distribution:} for degrees less than or equal to $10$ the graph move faster and reaches more people  as degree increases, for higher degrees we cannot see any trend which is because of paucity of people with that number of partners. (c) \textbf{Reachability probability within 20 distance:} for small and moderate degree individuals reach more people as they increase their number of partners, for high degree the fraction of reachable population tends to a constant value that is, when the degree is very high ($10$ or more for this set of data) that risk of Ct infection loses its dependence on the number of partners.}
\label{centrality}
\end{minipage}
\end{figure}

This result may not   provide an intervention strategy by  prioritizing  individuals  based on some criteria other number of partners for screening. However, it can be used in consulting individuals about their risky sexual behavior like number of partners: to reduce their risk of infection they have to keep their number of concurrent sexual partners less then some threshold value.

\subsection{ Mitigation efforts for controlling chlamydia}
We compare the model-projected impact of increased random screening, partner notification -- which includes partner screening and partner treatment -- social friend notification, and rescreening on the prevalence of Ct infection.
All of the simulations start at a balanced equilibrium obtained with the model baseline parameters in Table (\ref{parameters}).

\emph{\textbf{Random Screening}}:
the fundamental component of our bundled intervention is screening men for Ct to reduce infections in population. An expert panel, convened in $2006$
by the CDC, concluded that evidence is insufficient to
recommend routine screening for Ct in sexually active young men because of several factors such as  feasibility, efficacy, and cost-effectiveness, however, since then, evidence of the benefit of screening young
men for Ct in high prevalence areas has been mounting including that it can be cost-effective  and can make an impact on rates among women \cite{gift2008program,gopalappa2013cost}. Therefore, we consider  screening the other part of the sexual network (i.e. men).
Our model  can provide information on how much of the intervention for men is needed for impact on  Ct rates. To find men for screening we follow a  venue-based screening approach:
since most Ct infections are asymptomatic and young men are unlikely to seek traditional health care, a community rather than a clinic based approach is likely to reach more at risk AA men \cite{kissinger2014check,reagan2012randomized,stein2008screening}. 

To determine the effectiveness of increasing the number of men screened for Ct per year, we compare the quasi-stationary  state prevalence by varying the fraction of men who are screened randomly each year, $\sigma^m_y$.
The current screening rate for young men for Ct in high prevalence areas, like New Orleans, is low. This scenario can estimate the cost effectiveness of increased screening of young men on the Ct prevalence in women \cite{gift2008program, gopalappa2013cost}.
The Figure \ref{fig:scr1} shows a reduction in the overall Ct prevalence and the Figure. \ref{fig:scr2} shows a reduction in  Ct prevalence for different genders as the number of men randomly screened for Ct increases from $0$ to $50\%$, $0\leq \sigma^m_y\leq 0.5$.
The filled circles are the mean of $50$ different stochastic simulations and error bar are the $95\%$ confidence intervals. 
The least-square linear fit suggests that the steady-state Ct prevalence will decrease by $0.01$ for every additional $10\%$ of the men screened during a year.
Though a drop of five percent in prevalence is an admirable decrease, increased screening alone would not be sufficient to control Ct.

\begin{figure}[htp]
\centering
\subfloat[\textbf{Prevalence of Ct vs time for different networks}][Population prevalence]{\label{fig:scr1}\includegraphics[width=0.5\textwidth]{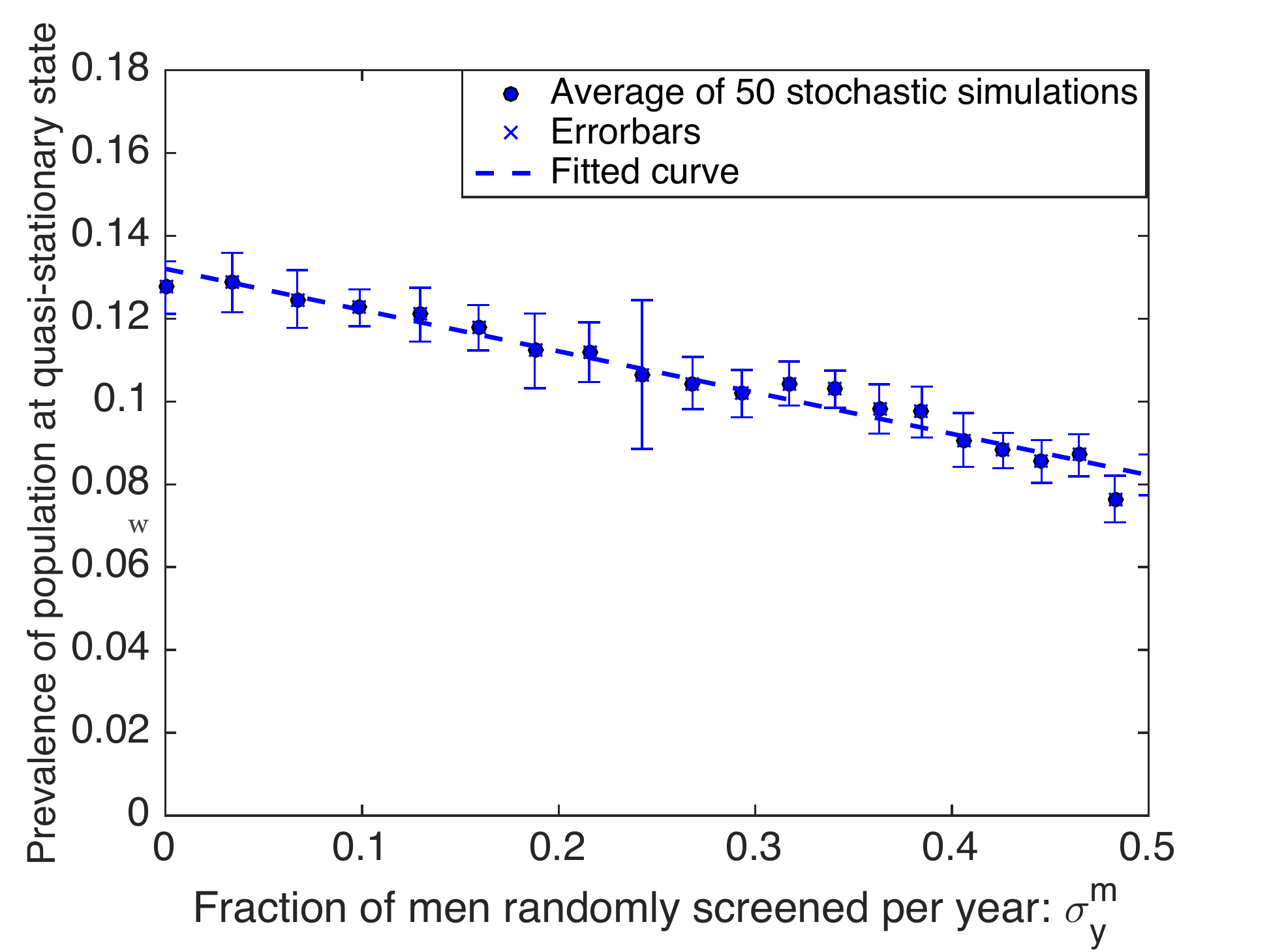}}
\subfloat[\textbf{Prevalence of Ct vs time at baseline}][Gender-based prevalence]{\label{fig:scr2}\includegraphics[width=0.5\textwidth]{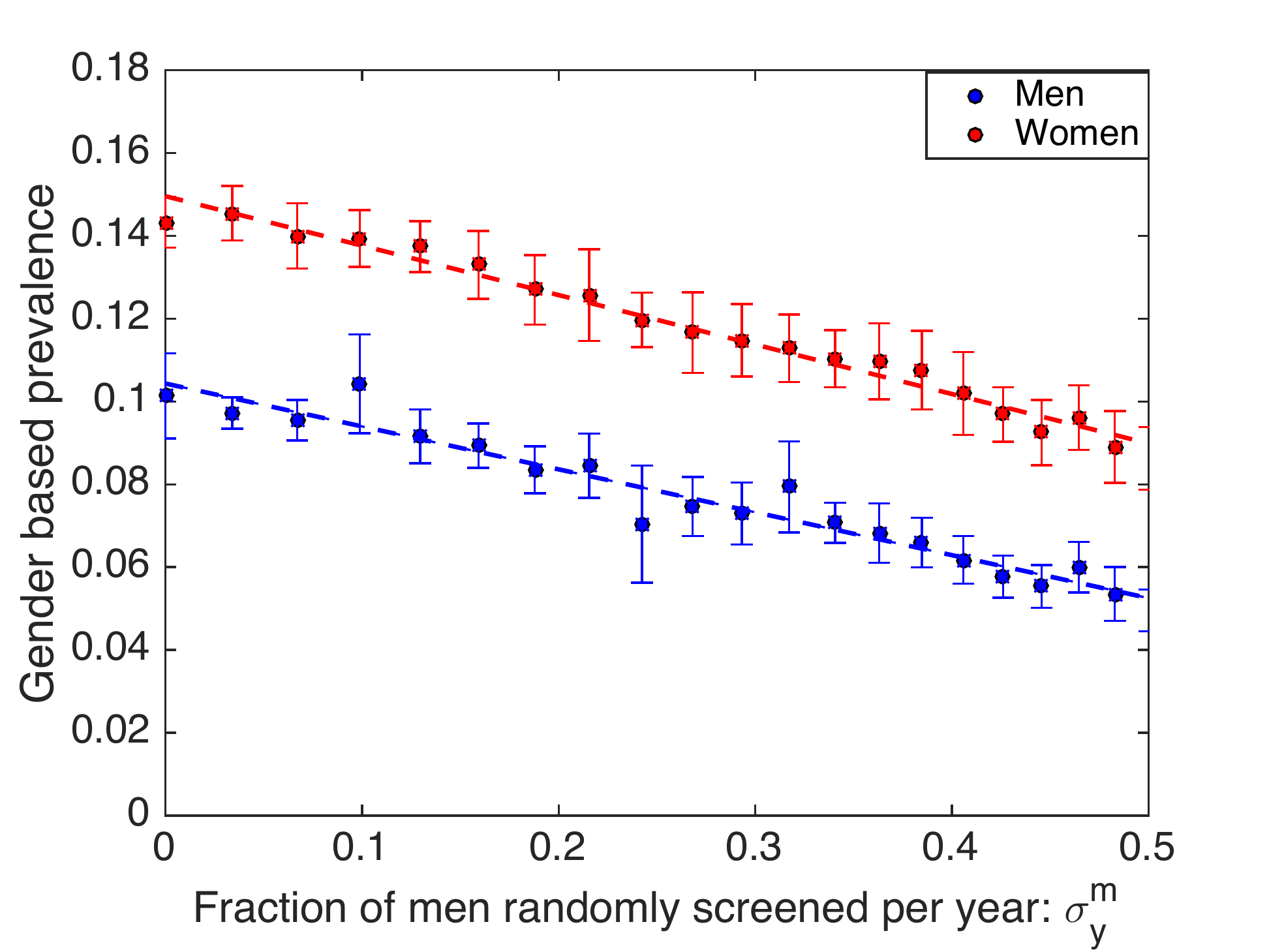}}\\
\begin{minipage}{.9\textwidth}
\caption[\textbf{Prevalence vs screening rate for men}]{Quasi-stationary state prevalence decreases as more men are screened each year. ${\sigma_y^m}$: a low negative correlation, screening men randomly by $50\%$ reduces prevalence by $5\%$ which is not effective enough to implement as a sole intervention. }
\label{just_screening}
\end{minipage}
\end{figure} 

\emph{\textbf{Rescreening}}:  
The rescreening scenario targets two goals: first finding the time-lag for resceening, and second quantifying the impact of rescreening on prevalence of Ct: 
\begin{enumerate}
\item \textbf{Interval for rescreening}\\
People who are found to be infected are more likely to be reinfected in the future.
Repeated Ct infection can be the result of treatment failure, sexual activity with a new partner, or being reinfected from an existing infected partner.
It makes sense to ask the infected people who were treated to return in a few months for retesting.
We will use the model to compare the rates of reinfection to help optimize the time, $\tau_r$, from treatment to rescreening.

The time $\tau_r$ for rescreening should be long enough so it is likely that the person will be reinfected, but not so long that such a reinfected person could infect significantly more people.
We start by identifying the rescreening time when the prevalence for the treated population exceeds the prevalence for the whole population.
That is, it becomes cost effective to rescreen when over $15\%$ of previously screened people are again infected.
To find this optimal time, we compute the time taken between screening time and next reinfection time for all individuals in the network assuming there is no rescreening i.e $\sigma_r=0$.  
On average at baseline $12\%$ of individuals are infected, therefore, through random screening, $12\%$ of infected individuals are found.
The current CDC guidelines recommend that people are rescreened for infection three months after treatment \cite{peterman2006high}. 
 
In our model, a person may be screened and be reinfected multiple times.  Therefore, we count the number of tests and reinfection events rather than the number of individuals with a test and infection. 

We plot cumulative distribution of time between screening and reinfection events in Figure (\ref{cdf-rescreening}).
Past studies have observed that about $25\%$ of the rescreened people are again found to be infected by three months.
The Figure demonstrates that in our model also predicts that about $25\%$ of treated individuals are again infected after $100$ days.
Although the model supports the CDC guideline as reasonable, the time between treatment and rescreening could be shortened to two months with an improved impact. 

 \begin{figure}[htp]
\centering
\includegraphics[scale=0.35]{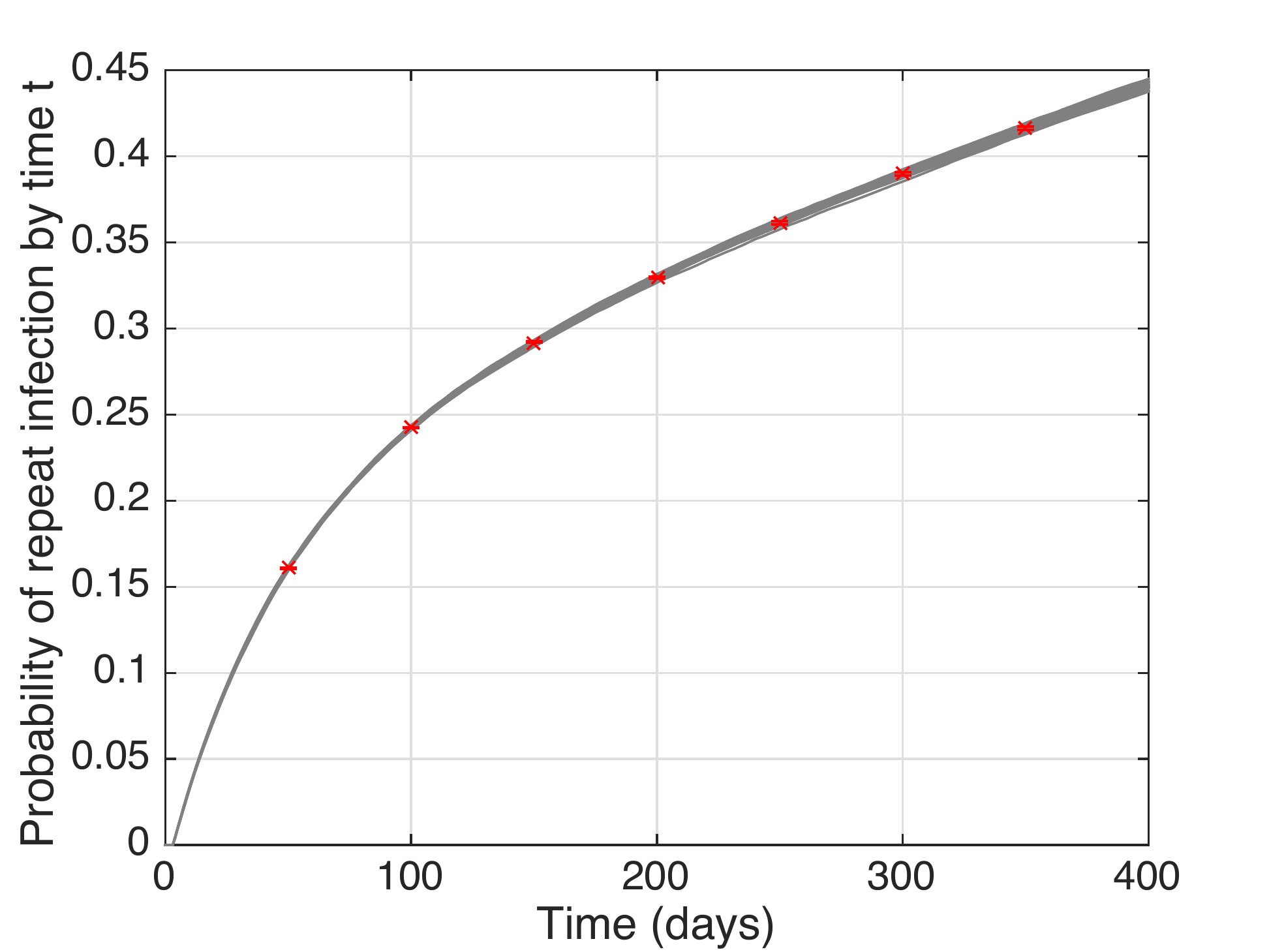}
\caption[\textbf{Interval for rescreening}]{Truncated cumulative probability distribution of time between treatment and reinfection with Ct: fifty different curves from $50$ stochastic simulations and $95\%$ confidence interval are shown in this figure.
About $25\%$ of the treated people are again infected after $100$ days.  This increases to about $45\%$ are reinfected after almost a year. }
\label{cdf-rescreening}
\end{figure}

\item \textbf{Rescreening rate}\\
The secondary goal of rescreening scenario is to determine if rescreening for Ct infection at a larger rate would be successful in reducing its prevalence.
At the baseline case only $10\%$ of screened individuals return for rescreening.
We assume a $0\leq \sigma_r\leq 1$ fraction of screened individuals participate in a rescreening plan $100$ days after their current screening day.
The Figure (\ref{rescreening}) quantifies the prevalence of Ct at quasi-stationary state dependent on rescreening rate $\sigma_r$: there is a  negative correlation between prevalence at quasi-stationary state and $\sigma_r$ when $\sigma_r$ fraction of screened individuals returns for rescreening,   if $\sigma_r$ fraction of screened individuals follow screening again then the prevalence reduces roughly by $0.02\sigma_r$.

\begin{figure}[htp]
\centering
\hspace{1cm}
\includegraphics[scale=0.3]{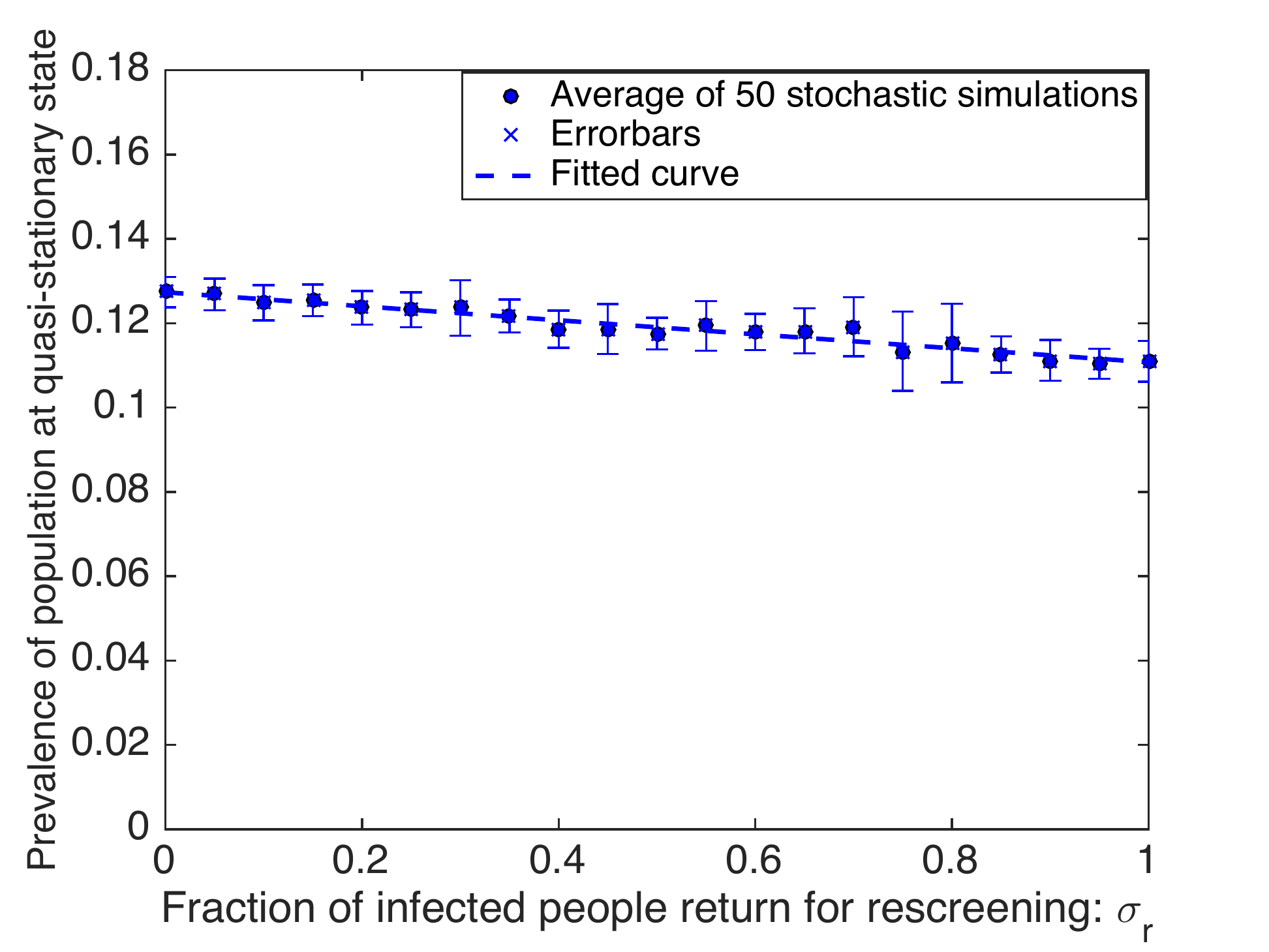}
\caption[\textbf{Quasi-steady state vs rescreening rate}]{Quasi-stationary state prevalence of population versus the fraction of treated people who come back for screening, $\pmb{\sigma_r}$: the circles are the mean of $50$ different stochastic simulations and error bars are $95\%$ confidence intervals. Rescreening all the infected people reduces prevalence by  $2\%$.}
\label{rescreening}
\end{figure} 
\end{enumerate}
\emph{\textbf{Partner Notification}}:
this scenario of partner notification quantifies the impact of giving an infected person's partners a chance to be tested and treated.
We define $\theta_n^p$ as the fraction of infected person's partners who are notified that they might be infected.
We then assume that only $\theta_t^p$ fraction of those notified partners are treated, without testing (partner treatment), and $\theta_s^p$ fraction are tested and if necessary, treated (partner screening).
Note that the fraction $1-\theta_{n}^p$ fraction of the partners are not notified.
\begin{enumerate}
\item \textbf{Partner treatment} \\
In partner treatment we assume when someone is found to be infected the fraction $\theta_n^p$ of their partners are notified and then all of the notified partners will seek treatment without testing, in other words, we define $\theta_t^p=1$.
The Figure (\ref{pt}) shows the impact of partner treatment ranging from no notified partners treated,  $\theta_n^p=0$, to all notified partners treated, $\theta_n^p=1$.
The filled circles are the mean of $50$ different stochastic simulations and error bars are $95\%$ confidence intervals. 
The least-square linear fit suggests that the quasi-stationary state Ct prevalence will decrease by $0.07$ for every $10\%$ increment in the fraction of notified partners seeking treatment.
This practice, although common in disease control today, is not as effective as partner screening as we will see.

\begin{figure}[htp]
\centering
\subfloat[\textbf{Prevalence of Ct at quasi-stationary state vs fraction of partners notified}][Population prevalence]{\label{fig:scr1}\includegraphics[width=0.5\textwidth]{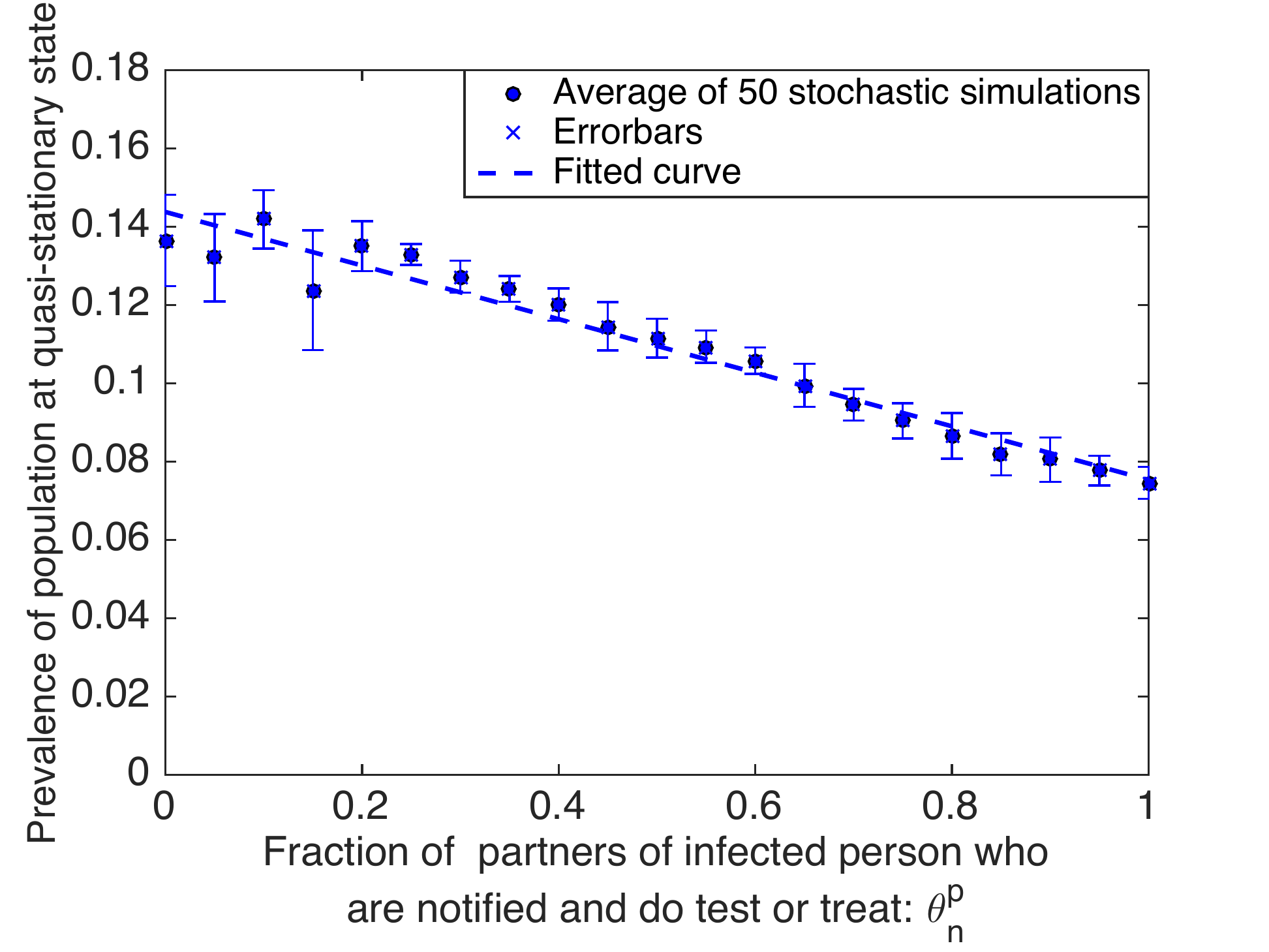}}
\subfloat[\textbf{Prevalence of Ct vs time at baseline}][Gender-based prevalence]{\label{fig:scr2}\includegraphics[width=0.5\textwidth]{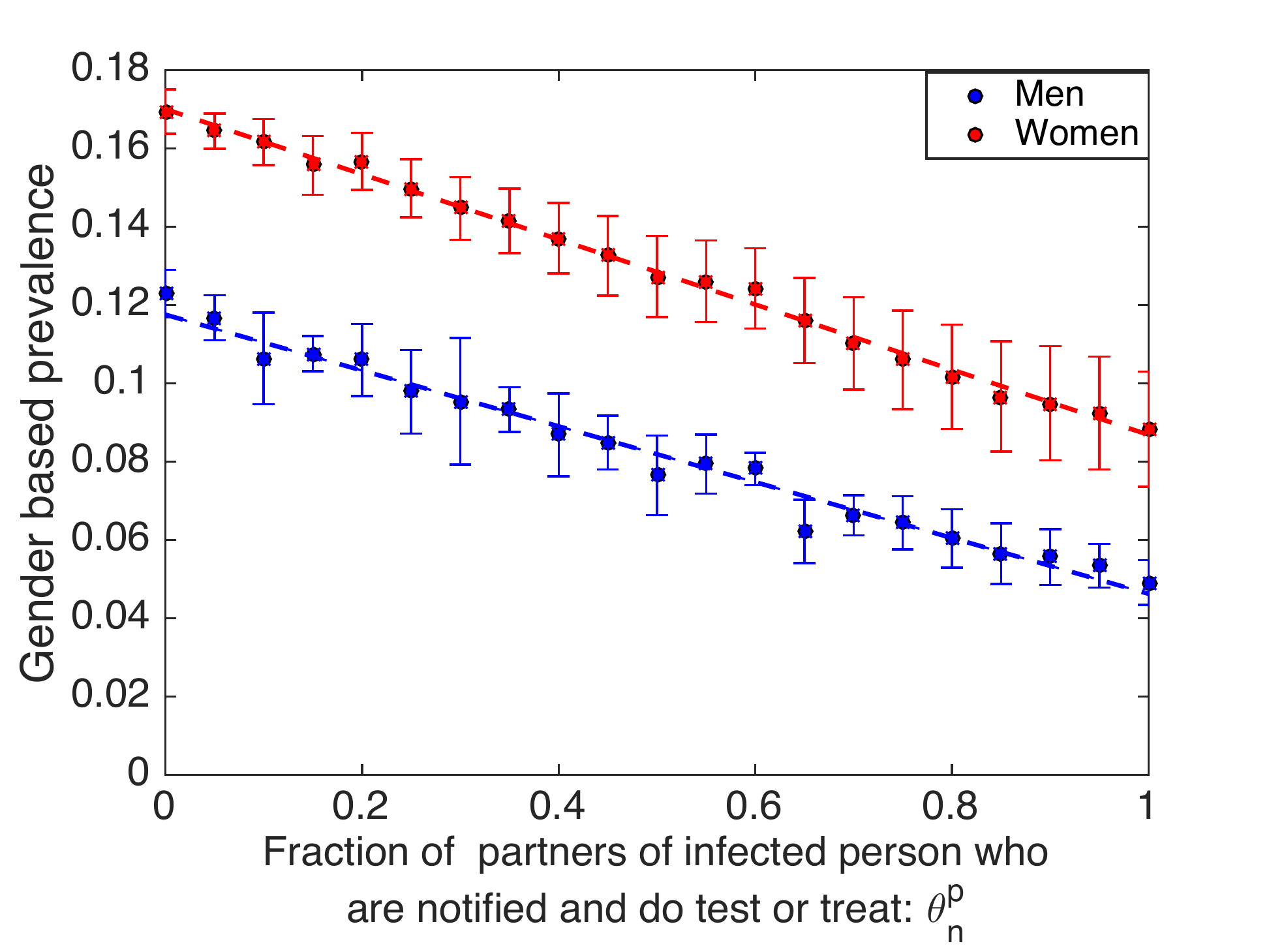}}\\
\begin{minipage}{.9\textwidth}
\caption[\textbf{Prevalence vs partner treatment}]{Prevalence decreases as the more partners are treated after being notified that they might be infected.  In these simulations, we assume that all the notified partners are treated, without testing $\theta_t^p=1$.  This approach is only mildly effective and the prevalence  remains high ($8\%$), even when all the partners of treated people are treated. }
\label{pt}
\end{minipage}
\end{figure} 

\item \textbf{Partner screening}\\
To quantify the impact of screening the partners of an infected person, where partners are tested then treated if found to be infected,  we assume that all notified partners of an infected person are screened, that is, we assume $\theta_s^p=1$.
The Figure (\ref{ptt}) shows the impact as  $\theta_{n}^p$ varies from $0$ to $1$. 
The filled circles are the mean of $50$ different stochastic simulations and error bars are $95\%$ confidence intervals.
The logistic curve fit suggests that there is a threshold effect (tipping point) at $\theta_{n}^p \approx 0.4$, when the approach becomes extremely effective. This happens when the partner screening percolates through the sexual network to identify the infected individuals.
Our model indicates that this is by far the most effective approach for bringing the epidemic under control.

\begin{figure}[htp]
\centering
\subfloat[\textbf{Prevalence of Ct at quasi-stationary state vs fraction of partners notified}][Population prevalence]{\label{fig:scr1}\includegraphics[width=0.5\textwidth]{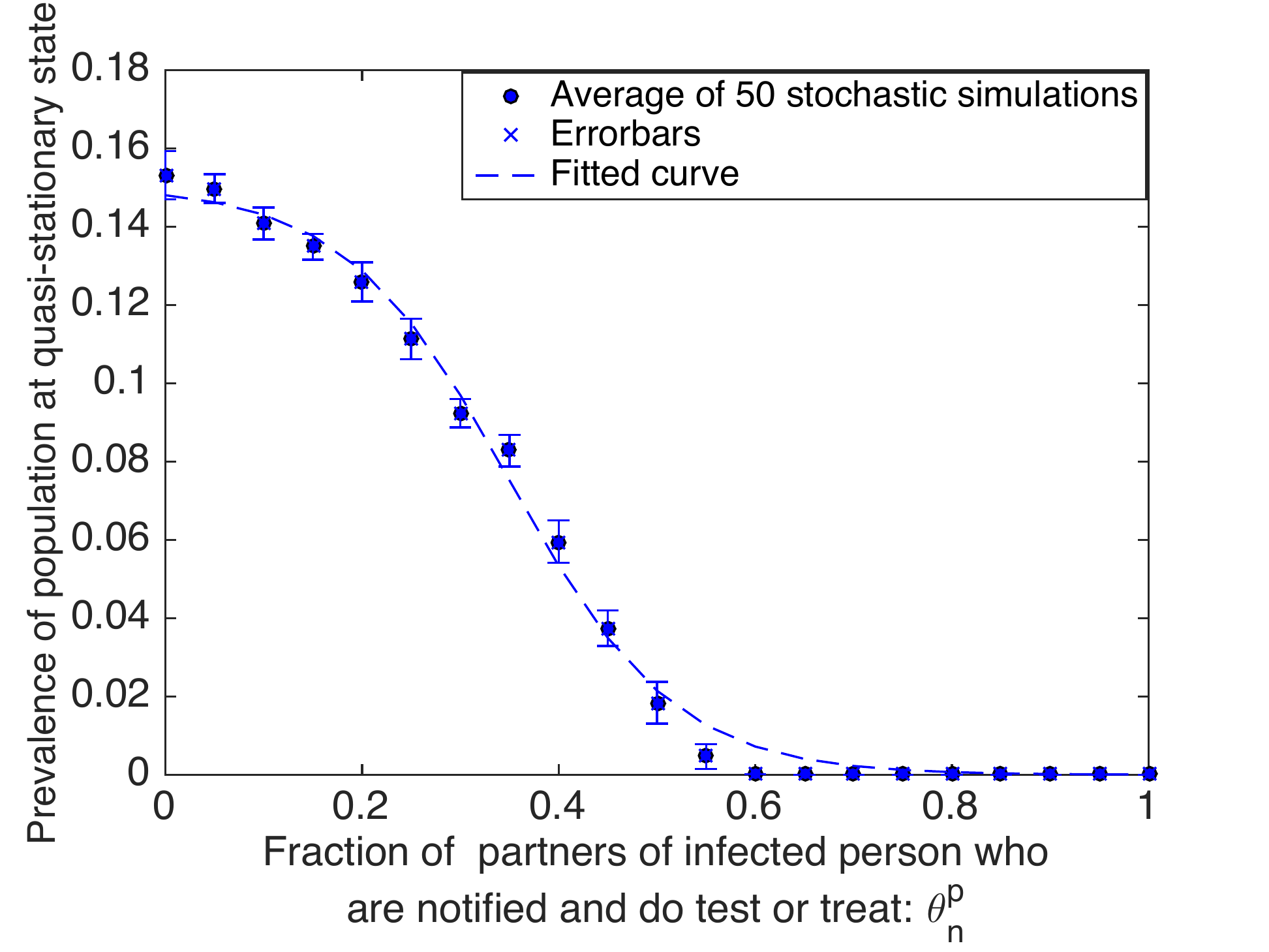}}
\subfloat[\textbf{Prevalence of Ct vs time at baseline}][Gender-based prevalence]{\label{fig:scr2}\includegraphics[width=0.5\textwidth]{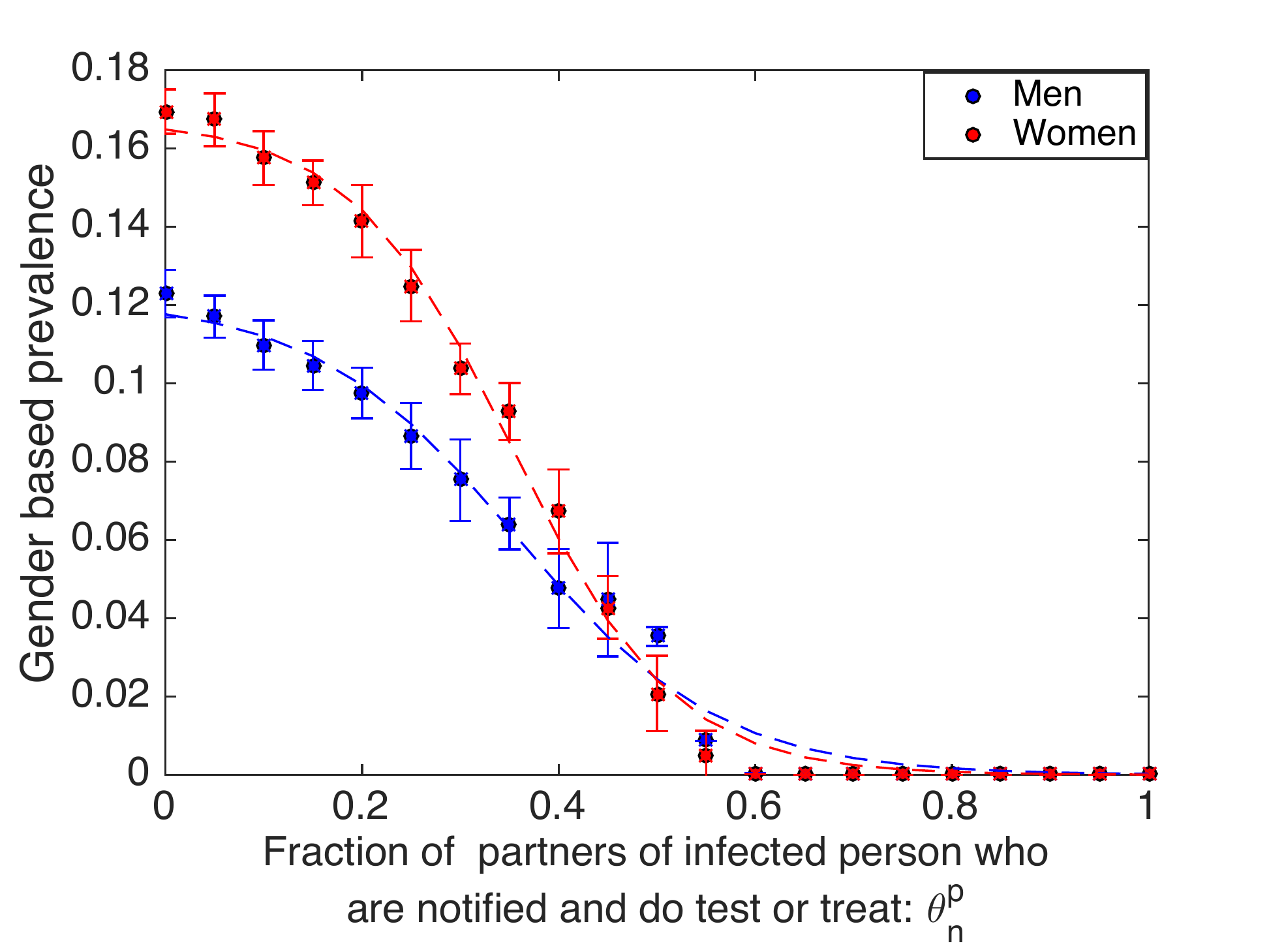}}\\
\begin{minipage}{.9\textwidth}
\caption[\textbf{Prevalence vs partner screening}]{Prevalence drops to zero as the fraction of the partners of treated people are tested before possible treatment. In these simulations we assume that all of the notified partners are tested for infection, ${\theta_s^p=1}$.  This partner screening approach is highly effective if the fraction of tested partners, $\theta_n^p$, exceeds its critical value $\theta_n^*=0.4$. That is, when $\theta_n^p\geq \theta_n^*$ and $\theta_{s}^p=1$, the Ct prevalence rapidly decays to zero.}
\label{ptt}
\end{minipage}
\end{figure} 

\item \textbf{Partner treatment and screening}\\
In reality, some of the notified partners will seek treatment without testing, and some will allow themselves to be tested before being treated.
We quantify the effectiveness of this mixture of the two  previous scenarios by varying fraction of the partners taking action ($\theta_n^p = 0.10, 0.20, 0.50, 0.65$ and $0.80$), along with the fraction of these notified partners that seek just treatment ($\theta_{t}^p$) and the fraction being screened for infection ($\theta_{s}^p=1-\theta_t^p$).

When few partners are notified and take action ($\theta_n^p$ is small), then partner treatment and partner screening have almost the same impact on controlling the prevalence.
For example, for $\theta_n^p=0.10$, the prevalence versus $\theta_t^p=1-\theta_s^p$ is almost flat, that is, there is no difference between cases if partners follow treatment without testing or first test and then treat if infected.

As $\theta_n^p$ increases the partner screening becomes a highly successful mitigation policy.
Consider the case when half of the partners are notified and take action,  $\theta_n^p =0.5$, and half of them are screened for infection, $\theta_{s}^p=0.5$, and the other half are treated without testing, $\theta_{t}^p=0.5$.
That is, half of an infected person's partners do nothing, the fraction $\theta_n^p\theta_t^p=0.5\times0.5=0.25$ are treated without testing for infection, and the fraction $\theta_n^p\theta_s^p=0.5\times0.5=0.25$ are tested and treated if found infected.
If any of the tested notified partners of the infected person are found to be infected, their partners are then notified and the cycle repeats to spread out and identify more infected people. 
This conditional percolation of screening through the sexual network is why this policy is so effective.
For this case, the prevalence reduction is $7\%$. Thus, compared with if all the notified partners follow treatment without testing, $\theta_{t}^p=1$, which reduces the prevalence by only $1\%$, it works better. But compared with if all the notified partners follow test and treat if necessary, $\theta_{s}^p=1$, which reduces the prevalence by  $11\%$, this combined scenario is not the one to select, Figure (\ref{pt_ptt}). 

\begin{figure}[H]
\centering
\includegraphics[width=.5\linewidth]{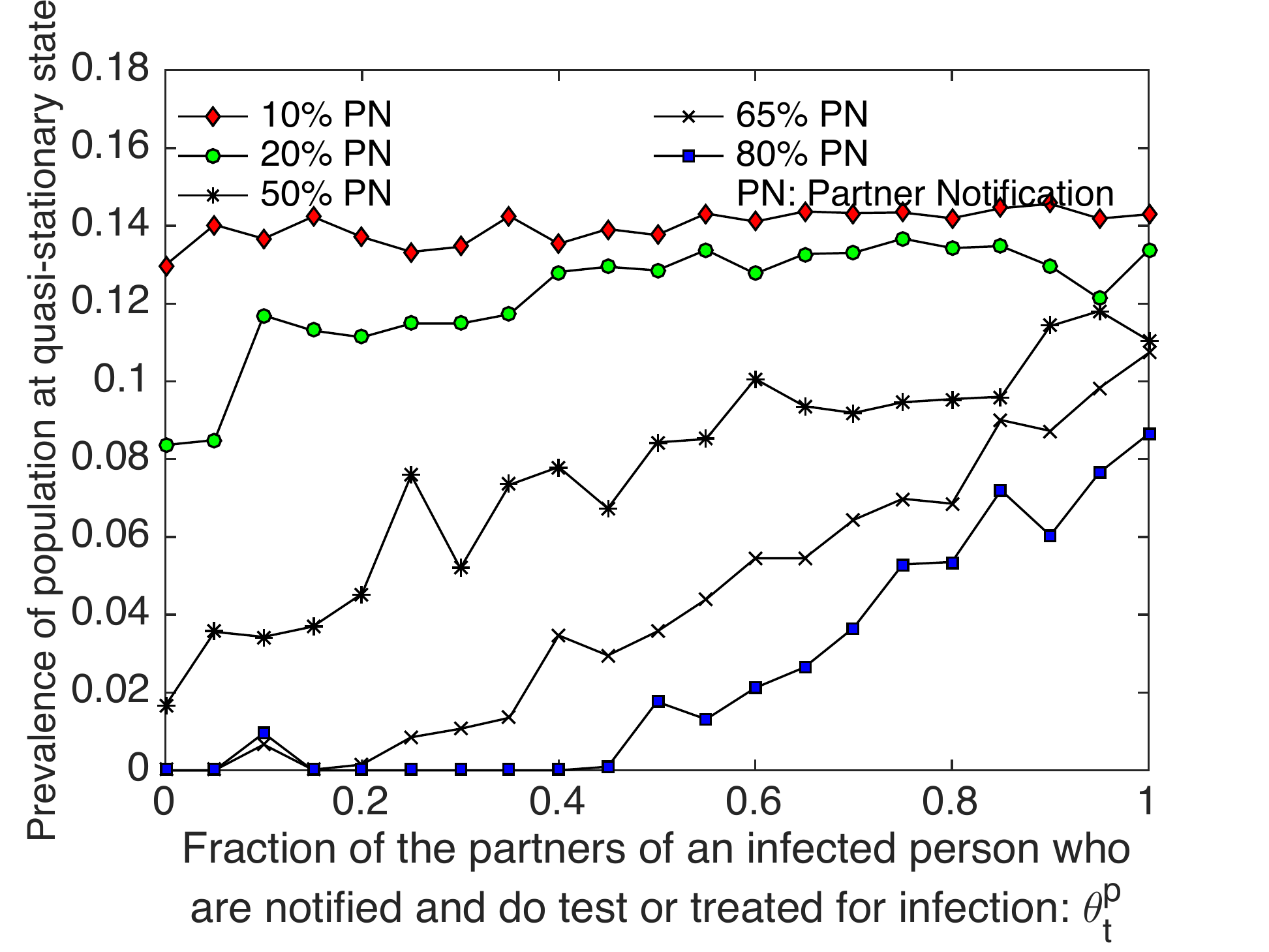}
\caption[\textbf{Prevalence vs partner treatment and screening}]{Prevalence at quasi-steady state increases when the fraction of partners notified are treated and not tested: 
when only a few partners of an infected person are notified, $\theta_n^p$ is small, then partner treatment and partner screening have similar small impact on Ct prevalence. When more partners of infected people take action,  $\theta_n^p$ increases, then the partner screening strategy is more effective in controlling the infection.}
\label{pt_ptt}
\end{figure}

It is important to note that partner screening is more expensive than partner treatment.
Given this last scenario, this suggests that when the fraction of partners we are able to notify, $\theta_n^p$, is small then partner screening may not be a good strategy compared to partner treatment. 
However, if a large enough fraction of partners are notified then it is better to test and treat (partner screening) to control the spread of Ct effectively.  
\end{enumerate}

\emph{\textbf{Social Friend Notification}}:
up to now we have implemented several different Ct interventions and concluded that partner screening along with random screening was the effective approach in controlling Ct epidemic.  Here we compare both their  partner screening and partner treatment with Social Friend Screening and a combination of partner and social friend notification. We have several scenarios explained in Table (\ref{scenario}).

\begin{table}[htbp]
\centering
\begin{tabular}{lp{10cm}}
\hline
\cline{1-2}
\textbf{Scenario} & \textbf{Description} \\
\hline
 \textbf{Scenario. 0}& There is no notification, i.e $\theta_n^p=\theta_n^f=0$.\\
 \hline
 \textbf{Scenario. 1}& $30\%$ of sexual partners are notified and follow screening, i.e  $\theta_n^p\times\theta_s^p=0.3$ and $\theta_t^p=0$.\\
 \hline
 \textbf{Scenario. 2}& $10\%$ of social friends are notified and follow screening, i.e  $\theta_n^f\times\theta_s^f=0.1$. \\
 \hline
  \textbf{Scenario. 3}& \textbf{Scenario. 1} and \textbf{Scenario. 2}. \\
  \hline
   \textbf{Scenario. 4}&  \textbf{Scenario. 1} and \textbf{Scenario. 2}, such that the social friends who are notified and follow screening are from  individuals  with high number of sexual partners. \\
\hline
  \textbf{Scenario. 5}&  $30\%$ of sexual partners are notified and follow treatment, i.e  $\theta_n^p\times\theta_t^p=0.3$ and $\theta_s^p=0$.\\
\hline
  \textbf{Scenario. 6}&  \textbf{Scenario. 5} and $10\%$ of social friends are notified and follow screening, i.e  $\theta_n^f\times\theta_s^f=0.1$. \\
\hline
  \textbf{Scenario. 7}&  \textbf{Scenario. 6}, such that the social friends who are notified and follow screening are from  individuals  with high number of sexual partners. \\
\hline
\cline{1-2}
\end{tabular}
\caption[\textbf{Notification scenarios}]{Different notification scenarios: in all these scenarios all other parameters are defined as in Table. (\ref{parameters}), unless stated otherwise. }
\label{scenario}
\end{table}

\begin{figure}[htp]
\centering
\subfloat[\textbf{Partner screening}][Partner screening]{\label{fig:not1}\includegraphics[width=1\textwidth,height=.33\paperheight]{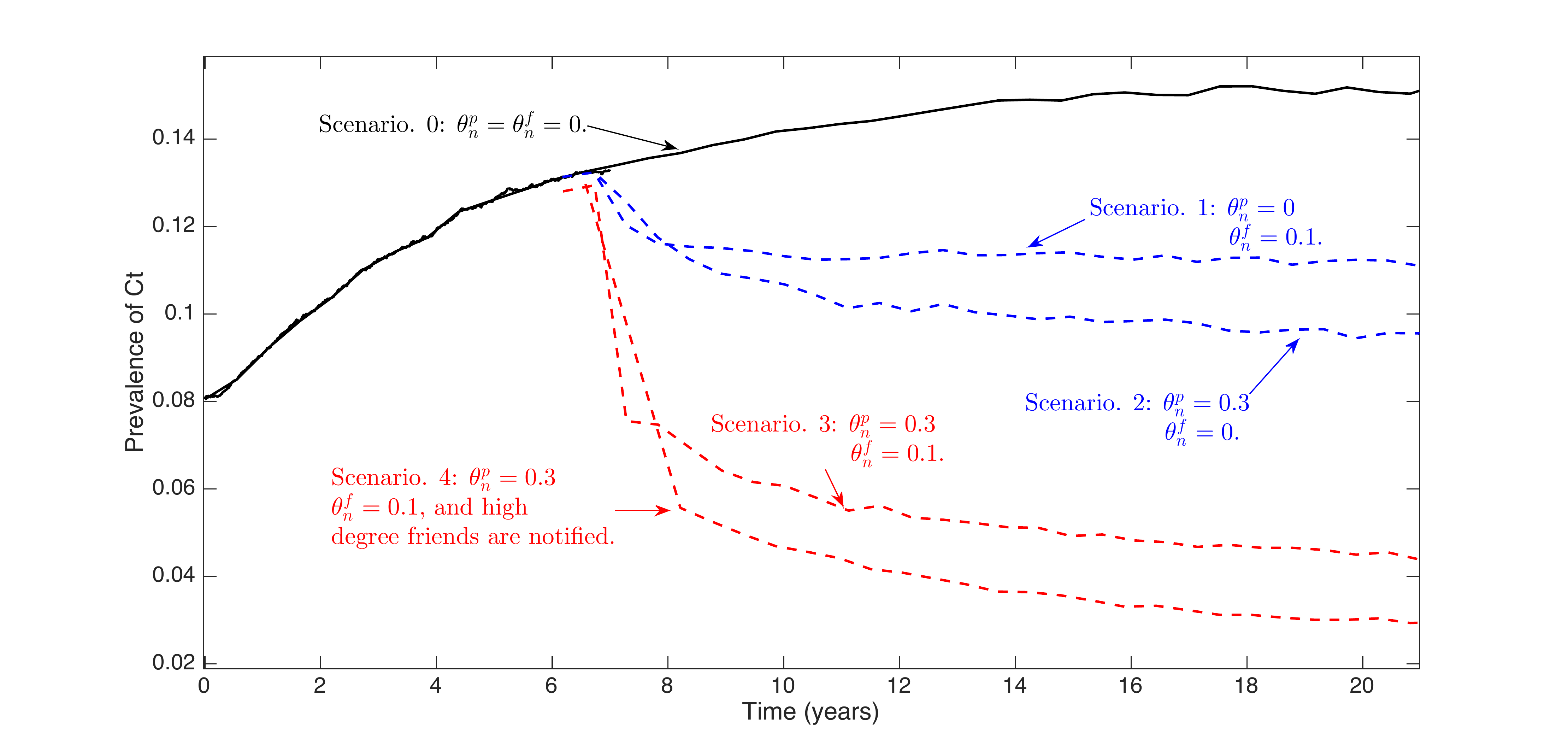}}\\
\subfloat[\textbf{Partner treatment}][Partner treatment]{\label{fig:not2}\includegraphics[width=1\textwidth,height=.33\paperheight]{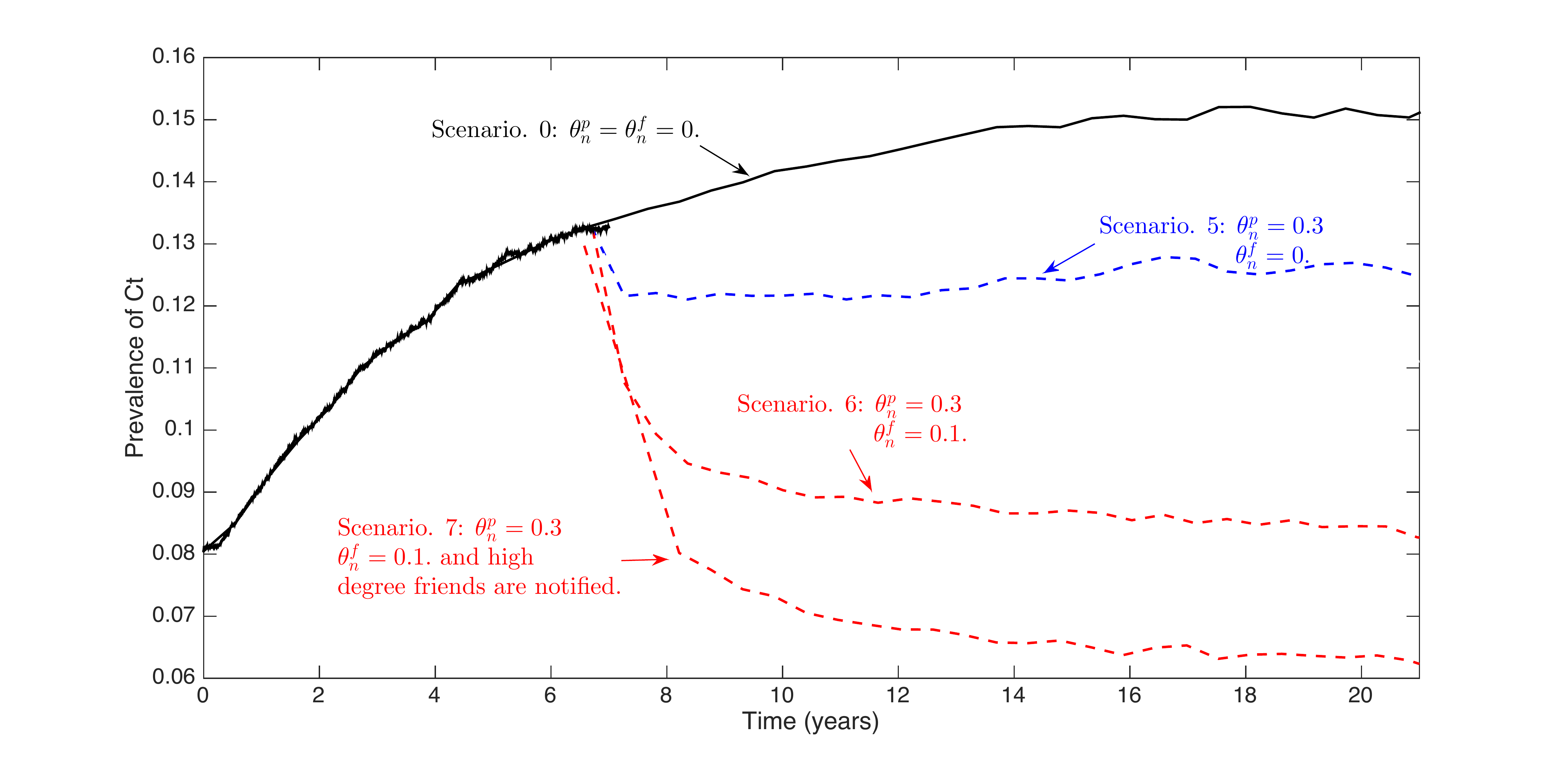}}\\
\begin{minipage}{.9\textwidth}
\caption[\textbf{Prevalence of Ct vs time for scenarios $1$-$7$}]{ Prevalence of Ct v.s time for scenarios $1$-$7$: (a)  the impact of combination of sexual partner notification  and social friend  notification along with screening is the summation of impact of partner notification and friends notification when notified individuals follow screening. 
Revised scenario  by notifying friends with high number of partners has $2\%$ improvement in compare with non-modified version (red  curves). (b)  The impact of notification along with combination of sexual partner treatment  and social friend  screening is also  the summation of impact of partner treatment and friends screening. 
Revised scenario  by notifying friends with high number of partners has $2\%$ improvement in compare with non-modified version (red  curve)s.}
\label{notify}
\end{minipage}
\end{figure} 
The  Figure \ref{fig:not1} shows prevalence of Ct in time for  scenarios $0$-$4$, and the  Figure \ref{fig:not2}  corresponds to scenarios $5$-$7$. Currently at baseline case $26\%$ of sexual partners of an screened infected person are notified, which some fraction of them do screening test and the rest treat themselves without testing \cite{kissinger2014check}. However, in our first scenario, we assume nobody is notified, neither sexual partners nor social friends that is,  $\theta_n^p=\theta_n^f=0$, the  black  curve is the prevalence in absence of notification which end up with around $16\%$ prevalence. Seven years after starting the infection epidemic, we implement different notification scenarios: first, infected individual found through screening  notifies only some fraction of their sexual partners, and these notified partners follow screening, that is $\theta_n^p\theta_s^p=0.3$. This scenario decreases  the prevalence to the half which is highly effective. What if infected individual found through screening notifies their social friends instead. We assume they notify some fraction of their social friends and these social friends follow screening, that is $\theta_n^f\theta_s^f=0.1$, which this reduces the prevalence to $10\%$. When we combine these two previous scenarios, that is, when infected individual found through screening notifies some of their sexual partners and some of social friends, the prevalence reduction becomes the summation of two sole notification scenarios, the red stared curve. The forth scenario is a modification of social friend notification: usually the sexual partners of infected person are more probable to follow screening than social friends, because their partner who had sex with is infected thus, they have a chance of carrying infection, the notified social friends may follow screening if they have too many sexual partners themselves. Therefore, in the forth scenario we assume notified social friends who follow screening are the ones with highest number of sexual partners. Doing so the prevalence reduces to $3\%$ which means there would be $13\%$ improvement in compare with no notification or $9\%$ improvement in compare with the current Ct prevalence in New Orleans.

Sexual partners of an infected individual are at a high risk of being infected, therefore, they may be  advised to follow treatment without being tested. This is the reason behind the scenarios $5$-$7$. In the fifth scenario  we assume infected individual found through screening  notifies only some fraction of their sexual partners, and these notified partners follow treatment that is, $\theta_n^p\theta_t^p=0.3$. By this intervention prevalence at steady state becomes $12\%$.
In scenario. 6 beside sexual partners, an infected person found through screening notifies his/her social friends i.e $\theta_n^f\theta_s^p=0.1$ which reduces prevalence to $8\%$, and even by revised social friend notification we can reduce it $2$ more percent, scenario. 7.

All these improvements in scenarios $1-7$ can be seen after almost five years of implementing them.

\section{Discussion and Conclusion}

In this chapter after reviewing sexual activity data from a pilot study, we used the algorithms in Appendices (\ref{ap1}) and (\ref{ap2}) to generate an ensemble of  heterosexual network with a prescribed degree and joint-degree distribution in and out of  social context. We used  heterosexual behavior survey and Ct prevalence data for adolescents and  young adult AA population in New Orleans to create a stochastic, Monte Carlo - Markov Chain, agent-based bipartite sexually transmitted disease-transmission network model.
In the model,  men and women are represented by the network nodes and sexual partners are characterized by edges between the nodes.   
The edges between partners in the network dynamically appear and disappear each day depending if the individuals have sexual act on that day.
The joint-degree distribution of the network captures the correlation of an individual's risk (their number of partners) with their partner's risk (number of partners of their partners).
Our network model is updated each day to account for sexual acts as a dynamic variable. We use this model to quantify the impact of increasing screening of men for infection, partners notification, social friend notification , and rescreening of treated individuals on reducing Ct prevalence.

In  analysis of  the properties of  ensemble of generated heterosexual networks, we observed a tight distribution in the number of connected components and size of giant component and bi-component for all the networks which have the same joint-degree distribution and have the same percentage of edges in social network. Preserving joint-degree distribution, when the property of being subgraph of social network becomes stronger the size of giant component increases, and consequently the number of connected components decreases, which it is because of reducing the mixing in generating sexual network: when people select their sexual partners from their social friends they stand in a tight group within social network. Redundancy coefficients for networks increases as dependence of sexual  network  on social one rises, which is because of high clustering coefficient of social network.
We studied these measures of networks because they may affect the spread of infection through network. However, none of the mentioned measures affects the spread of Ct over the network: the prevalence of Ct over networks in different groups are close to each other and therefore, ignorable.

Spreading infection over the network in absence of any type of intervention, we studied some properties of infected population at quasi-stationary state such as their degree,  betweenness, and closeness scores. Our result show that people who are closer to more  many other individuals in population are at higher risk of catching infection, even though he/she has a few number of partners. In network science terminology if there is a path between an individual  and many other people in the network then shorter the path, higher his/her closeness score, and therefore, higher risk of Ct infection he/she has. This information can help us to identify more qualified people for random screening by relating individual's degree and closeness score. At the first glance, we did not observe any correlation between degree and closeness score, that means a person with high number of partners may not be necessarily at high risk of Ct infection. But,  studying the relation between degree and reachability (how much they can reach other people in population) of individuals for an ensemble of heterosexual  networks, we observed that reachability probability increases to a fixed point as degree increases.
 That is, up to some degree value, increasing degree causes that reachability of a typical  person and therefore, its  closeness increases, which this puts him/her at higher risk of  Ct infection. However, when this correlation converges to its fixed point, the impact of degree as a risk  of Ct infection for a person disappears.

In intervention strategies the first approach was screening men. We observed that increasing Ct screening of men has  a modest impact on reducing Ct prevalence in the young adult AAs in New Orleans, the Figure  (\ref{just_screening}).
Starting at a baseline of $13\%$ prevalence under the assumption that $45\%$ of the women are being screened each year for Ct, then increasing the screening of men from $0\%$ to $50\%$ would only reduce the overall Ct prevalence to $8\%$.
Linking our result with \cite{clarke2012exploring} that found partner positivity is insensitive to screening, we found out that screening men alone cannot control epidemic in population and consequently among women drastically.


In evaluating the effectiveness of partner notification  we assumed that a fraction of the partners of an infected person will seek treatment (without testing) or be screened (tested and treated) for infection.
We observed that if most of the notified partners are treated, without testing, then this mitigation has only a modest impact on Ct prevalence.
This practice, although common in disease control today, is not as effective as partner screening.
When the partners of an infected person were tested before treatment, there was a tipping point where partner screening would bring the epidemic under control.
That is, when over $40\%$ of notified partners of all the infected people are screened for infection, then the Ct prevalence rapidly decreased to very low levels, the Figure (\ref{pt_ptt}).
This critical threshold represents the partner screening level where a contact tracing tree can spread through the heterosexual network to identify and to treat most of the infected people.  
Our model indicates that this is by far the most effective approach for bringing the epidemic under control.

However, partner screening is more expensive than partner treatment.
The partner treatment and screening  suggests that when the fraction of partners took action ($\theta_n$ is small), then partner screening may not be a good strategy compared to partner treatment. 
But if a large enough fraction of partners are notified then it is better to test and treat (partner screening) to control the spread of Ct effectively. 
These results of  impact of partner notification   are close to results from \cite{kretzschmar1996modeling} who found for
Ct, contact tracing is  less effective at lower
percentages when partners are treated, but with increasing levels of contact tracing it will be 
 a highly effective intervention strategy.
 
Using social network in generating sexual networks and in studying the spread of STIs not only enable us to construct a more realistic sexual network but also helps us to improve interventions by spreading information through social network, social friend notification.
Our result shows a combination of social friend notification and sexual partner notification has a significant reduction on prevalence of Ct compared with  when there is no notification, or only sexual partner are notified. We studied two different approaches for notification. In the first case, infected person notifies some fraction of their sexual partners and social friends, and we assume both notified partners and social friends who take action test and treat if infected, partner and social friend screening. In the second case, infected person notifies some fraction of their sexual partners and social friends, and we assume  notified  social friends who take action test and treat if infected, social friend screening, and notified sexual partners who take action treat themselves without testing, partner treatment. In both cases, there is a relative $40\%$
reduction on prevalence compared with the case only sexual partners receive notification and follow up test or treatment.


In rescreening, an infected individual returns for testing a few months after they are treated.
We used the model to estimate the probability that a treated person would be reinfected as a function of the time since they were treated. 
The CDC guidelines recommend that treated people return for screening three months after treatment. 
We observed that for the baseline case of $13\%$ infected population, about $25\%$ of the treated population were reinfected three months after treatment. 
We observed that although the rescreening is a cost effective approach to identify infected people, it has only a small impact on the overall Ct prevalence.
Even though there is a high chance of reinfection when the individual's behavior does not change, we do not observe an effective impact on prevalence of Ct by monitoring infected individuals.
Rescreening program has a trend similar to screening, and none of them are effective as sole intervention because they are not able to find the chain of infection like partner screening.
On the other hand sensitivity of prevalence to rescreening is less than that of screening, indicating the fact that for a limited budget the idea of finding more people to screen, random screening, is more effective than frequent screening for less people.

 The existence of  heterosexual network with a prescribed joint-degree distribution in the context of social network is the first concern when generating the sexual network. One of our limitation is that we cannot test if a sexual network with a particular joint-degree distribution within a social network exist or not, therefore, we have to find it by trial and error. On the other hand, source of partner selection can be correlated to degree of individuals, but, in this work we ignored this correlation.

The uncertainty in the model parameters will require an extended sensitivity analysis to quantify the robustness of the predictions in the presence of uncertainty. 
Our future work will focus on validating the model predictions and identifying which trends and quantities can, and cannot, be predicted within limits of the model uncertainty.

Although our model includes-condom use, it does not account for behavior changes, such as increased condom-use after being treated for infection or the differences in condom-use between primary and casual partners. When we assign the type of partners for each individual in the network, we will  change casual partners more frequently and will implement condom-use for contacting with casual partners. Also notification  of partner strategy may be affected when people are biased about notifying partners. Collecting more data regarding partnership level will help us improve our model by distinguishing between partners.
Our future research will improve the model so we can better quantify the impact of counseling and behavioral changes such as increasing condom-use or partner notification rates.   We are also expanding our data analysis to include a cost-benefit analysis and estimate the averted PID cases in women.

\appendix
\chapter{Generating Bipartite Networks with a Prescribed Joint Degree Distribution}\label{ap1}

Bipartite networks can provide an  insightful representation of the interactions between two disjoint groups, with applications ranging from  ecological networks \cite{dormann2009}, social interactions, the spread of sexually transmitted infections,  and  citation/collaboration networks \cite{cimini2014scientific}.
The accuracy of a mathematical model to understand the interactions of these networks depends on generating an ensemble of random graphs that faithfully captures the structure of the known real-world networks needed to  reproduce  the dynamics of the underlying problem.   

When simulating a real-world problem on a network, the graph properties, such as the degree and joint-degree distributions, for the generated random graphs must be analyzed to see if they are consistent with the original problem.  
If any of these properties affect the questions being asked of the model, such as how fast a disease will spread among a population, then these properties must be preserved in the mathematical model.  
For example, social networks often exhibit homophily where there is a tendency of individuals to associate with others having similar characteristics.

This homophily is captured in the network by the joint-degree distribution, sometimes called the degree correlation or degree-degree distribution. 
Although there are several methods, called 2K network generation algorithms,  for generating simple graphs that preserve both a given degree and joint-degree distributions \cite{boguna2003class,mahadevan2006systematic,newman2003mixing,eubank2004modelling,d2012robustness,williams2014degree}, there are few results for bipartite networks.

Typical  network generation algorithms that preserve the degree distributions are based on \textit{stochastic}, \textit{rewiring}, or \textit{reconfiguration} approaches.
The \ErdosRenyi random graph generation algorithm \cite{erdHos1959random} is an example of stochastic approach where every two nodes are connected with probability $p$ defined by the average degree of nodes in the network divided by their size. 
This approach is easily generalized to match a given degree \cite{chung2002connected} or joint-degree distribution \cite{boguna2003class}. 
The rewiring approach rewires two random edges to preserve the average degree or  degree distribution.
The rewiring approach converges, although there is little analysis on the convergence rate  \cite{mihail2003markov}.
The pseudograph reconfiguration algorithm \cite{molloy1995critical} reproduces the given degree distribution exactly, however it may end up with self-loop or multiple edges between two nodes. 

The joint-degree distribution is correlated with the structural and dynamical properties of networks \cite{newman2003mixing,eubank2004modelling,d2012robustness,williams2014degree}. 
This information is quantified in the symmetric joint-degree matrix (JDM) whose $(i,j)$ element is the number of edges between nodes of degree $i$ and nodes of degree $j$  \cite{stanton2012constructing}. 
The necessary and sufficient condition for a simple network to exist for a given $JDM$ is given by the \ErdosGallai type theorem \cite{amanatidis2015graphic,czabarka2015realizations,stanton2012constructing}:
\begin{theorem} 
(\ErdosGallai Type Theorem for JDM)
Consider a network where $M$ is the largest degree of the nodes in the network, then there is a simple network that has an $M\times M$ symmetric JDM if and only if
\begin{enumerate}
\item $n_i=\frac{1}{i} \sum_{j=1}^M JDM(i,j)$ is the number of nodes with degree $i$ for $i=1\cdots M.$
\item $JDM(i,i) \leq \binom {n_i}{2},$ for $i=1\cdots M.$
\item $JDM(i,j)\leq n_in_j,$ for $i\neq j.$
\end{enumerate} 
\end{theorem}
Mahavedan et al. \cite{mahadevan2006systematic} have extended the rewiring approach to generate random networks using joint-degree distribution.
They use the term $2K-$series to introduce joint-degree distribution, and they compare  stochastic, pseudograph, matching and rewiring and the extended pseudograph algorithms to construct networks using $2K-$series.
They compare the topology of networks made based on different algorithms and suggest that $2K-$ series or joint-degree distribution is enough to reproduce most metrics of interest for the network.
They then use a configuration model to generate a 2K-network with the prescribed $JDM$, however their network may end up with multiple edges between two nodes. 

A balanced degree invariant algorithm is provided by \cite{stanton2012constructing} for constructing simple networks from a given $JDM$, and a Monte Carlo Markov Chain method is used for sampling the networks.
Gjoka et al. \cite{gjoka2015construction} design a new algorithm for constructing simple networks with a target $JDM$.
Bassler et al. at \cite{bassler2015exact} use JDM and develop an exact algorithm to find all pairwise degree correlations and the degree sequences.

We extend these methodologies for generating bipartite networks using prescribed joint-degree distribution.
Note that the {bipartite joint degree} ($BJD$) \cite{boroojeni2017generating} matrix of a bipartite network can be nonsymmetric and is not even a square matrix if the maximum degree in two groups are not the same.
We find and prove a similar necessary and sufficient condition as \ErdosGallai Type Theorem on $BJD$ for constructing simple bipartite network and then use  $BJD$ matrix as an input to construct network.
We then describe new bipartite algorithms for generating these random networks and investigating how well they reproduce other properties, such as the bipartite clustering, observed in real-world networks.
This family of algorithms are called $B2K$ algorithms and they preserve a given degree and joint-degree distributions of the network.

\section{Bipartite Network}\label{NModel}

A {\it bipartite network}, sometimes called {\it two-mode network} or {\it affiliation network}, is a network whose nodes can be divided into two disjoint sets $\textbf{v}^u$ and $\textbf{v}^l$  such that every edge connects a node in $\textbf{v}^u$ to one in $\textbf{v}^l$, there is no edge between nodes in $\textbf{v}^u$, and no edges between nodes in $\textbf{v}^l$.
This network is shown like $\textbf{G}=(\textbf{v}^u,\textbf{v}^l,\textbf{E})$ consisting of a set of $P^u = |\textbf{v}^u|$ upper nodes, $\textbf{v}^u = \{\textbf{v}^u_i| i = 1, 2, 3,  \dots P^u\}$, a set of $P^l = |\textbf{v}^l|$ lower nodes, $\textbf{v}^l= \{\textbf{v}^l_i| i = 1, 2, 3,  \dots P^l\}$, together with a binary adjacency relation defining the set of edges $\textbf{E} = \{\textbf{v}^u_i\textbf{v}^l_j|  i\in \{1,2,3, \dots P^u\}, j \in \{ 1, 2, 3,  \dots P^l\}\}$, where $\textbf{v}^u_i\textbf{v}^l_j$ denotes the edge between node $\textbf{v}^u_i$ and node $\textbf{v}^l_j$.

The {\it degree}  of a node $\textbf{v}_i$, $\textbf{deg}(\textbf{v}_i)$, is defined as the number of neighboring nodes connected to the node by an edge.
The {\it degree distribution} $d_k$ defines the number of nodes with degree $k$.
The {\it joint-degree distribution} or sometimes called {\it degree-degree distribution} or {\it degree correlation} $(k,j)$ is the number of nodes with degree $j$ that are connected to nodes with degree $k$.
A bipartite network $\textbf{G}$ can be represented by the {\it Bipartite Joint Degree} or $BJD$ matrix:

\[BJD_G = \parenMatrixstack{
    e_{11} & e_{12} & e_{13} & \dots  & e_{1l} \\
    e_{21} & e_{22} & e_{23} & \dots  & e_{2l} \\
    \vdots & \vdots & \vdots & \ddots & \vdots \\
    e_{u1} & e_{u2} & e_{u3} & \dots  & e_{ul} },\]
where, $u$ is the maximum degree in upper nodes, and $l$ is the maximum degree in lower nodes, each element $e_{ij}$ is the number of edges between upper nodes with degree $i$ and lower nodes with degree $j$. 
The degree distribution of network \textbf{G} is defined by the number of upper nodes, $d^u_k$, and lower nodes, $d^l_k$, with degree $k$:
\begin{equation*}
d^u_k=\frac{\sum_{j=1}^l e_{kj}}{k}~,\text{~~and~~~}~~~d^l_k=\frac{\sum_{i=1}^u e_{ik}}{k}~.
\end{equation*}
The number of nodes with degree $k$ is $d_k=d^u_k+d^l_k.$

A $BJD$ matrix is consistent with a bipartite network if there exist at least one bipartite network with this joint degree distribution. 
For an example, consider the $BJD$ matrix \[BJD = \parenMatrixstack{
    2 & 2 \\
    4 & 0 },\]
each entry in the matrix, $(i,j)$, is an edge between an upper node with degree $i$ and lower node with degree $j$. There are $\frac{2+2}{1}=4$ upper nodes with degree $1$, and $\frac{4+0}{2}=2$ upper node with degree $2$, $\frac{2+4}{1}=6$ lower nodes with degree 1, and $\frac{2+0}{2}=1$ lower node with degree 2. On the other hand,  $e_{21}=4$ means that four of the edges of the  graph will connect an upper node of degree $2$ to a lower node of degree $1$, or $e_{22}=0$ means there is no edge between upper and lower nodes with degree $2$. 

\begin{theorem}\label{simple}
Let $BJD$ be a matrix, 
\[BJD = \parenMatrixstack{
    e_{11} & e_{12} & e_{13} & \dots  & e_{1l} \\
    e_{21} & e_{22} & e_{23} & \dots  & e_{2l} \\
    \vdots & \vdots & \vdots & \ddots & \vdots \\
    e_{u1} & e_{u2} & e_{u3} & \dots  & e_{ul} },\]
where 
$d^u_i=\frac{\sum_{j} e_{ij}}{i}$ is the number of upper nodes with degree $i$, and  $d^l_j=\frac{\sum_{i} e_{ij}}{j}$ is the number of lower nodes with degree $j$.
If we have $e_{ij}\leq d^u_id^l_j$ for $i=1,...,u$ and $j=1,...,l$, 
then, there exist at least one simple network, a network without self-loops or multiple edges, as defined by the $BJD$ matrix.
\end{theorem}
\begin{proof}
Suppose the $BJD$ matrix satisfies the assumption of theorem, and the network $\textbf{G}$ corresponding to $BJD$ matrix has at most one edge between all nodes, except for two nodes, where there are two edges between the upper node \textbf{u} of degree $i$ and the lower node \textbf{v} of degree $j$. There are $d^u_i$ upper nodes with degree $i$ and $d^l_j$ lower nodes with degree $j$.

For the trivial case when $d^u_i=1$ and $d^l_j=1$, we have $e_{ij}\leq 1$ and there can not be a multiple edge between the sole node \textbf{u} with degree $i$ and sole node \textbf{v} with degree $j$.

When either $d^u_i>1$  or $d^l_j>1$, then there is an upper node, \textbf{u'}, with degree $i$ and a lower node, \textbf{v'}, with degree $j$ that are not connected.
This follows from the contradiction argument where if all nodes with degree $i$ are connected to all nodes with degree $j$ and \textbf{u} is connected to \textbf{v} with two edges, then $e_{ij}=d^u_id^l_j+1\nleq d^u_id^l_j$, which contradicts the assumption of theorem. Therefore, such \textbf{u'} and \textbf{v'} exist.
These \textbf{u'} may be the same as \textbf{u} or \textbf{v'} may be the same as \textbf{v}, but both can not happen: 
\begin{itemize}
\item Case 1: we consider the case that \textbf{u'} is different from \textbf{u} and  \textbf{v'} is different from \textbf{v}.
Because \textbf{u} is connected to \textbf{v} within two edges, therefore, it  has $i-1$ disjoint neighbors, however, because all edges connected to \textbf{u'} are simple, therefore, \textbf{u'} has $i$ disjoint neighbors, thus \textbf{u'} has a neighbor, like \textbf{w}, which is not a neighbor of \textbf{u}.
Also, for the same reason \textbf{v} has $j-1$ disjoint neighbors and \textbf{v'} has $j$ disjoint neighbors, thus \textbf{v'} has a neighbor, like \textbf{w'}, which is not a neighbor of \textbf{v}.
When this happens, then we rewire the network by first removing one of the the double edges \textbf{u}\textbf{v}, as well as edges \textbf{u'}\textbf{w} and \textbf{v'}\textbf{w'}, then adding the edges \textbf{u}\textbf{w}, \textbf{v}\textbf{w'}, and \textbf{u'}\textbf{v'}.
Therefore, we have a simple network \textbf{G}. The Figure (\ref{fig:case1}) illustrates this rewiring process.
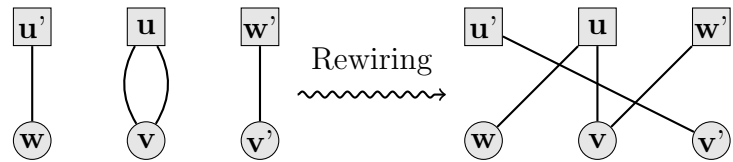
\begin{figure}[hpt]
\begin{center}
  \begin{tikzpicture}
    \begin{scope}
    \node[style={rectangle,fill=gray!20,draw,minimum size=.5cm,inner sep=0pt}] (u) {$\textbf{u}$};
    \node[style={rectangle,fill=gray!20,draw,minimum size=.5cm,inner sep=0pt}] (wp) [right = 1cm  of u]  {$\textbf{w}$'$$};
    \node[style={rectangle,fill=gray!20,draw,minimum size=.5cm,inner sep=0pt}] (up)[left = 1cm  of u]{$\textbf{u}$'$$};
    \node[main node,fill=gray!20] (v) [below = 1cm  of u] {$\textbf{v}$};
    \node[main node,fill=gray!20] (vp) [right = 1cm  of v] {$\textbf{v}$'$$};
    \node[main node,fill=gray!20] (w) [left = 1cm  of v] {$\textbf{w}$};
    \node[] (art1) [below = .5cm  of wp] {};
    
    \path[draw,thick]
    (u) edge[bend right] node {} (v)
    (up) edge node {} (w)
    (vp) edge node {} (wp)
     (v) edge[bend right] node [left] {} (u);
    ;
    \end{scope}
    
    \begin{scope}[xshift=6cm]
    \node[style={rectangle,fill=gray!20,draw,minimum size=.5cm,inner sep=0pt}] (ru) {$\textbf{u}$};
    \node[style={rectangle,fill=gray!20,draw,minimum size=.5cm,inner sep=0pt}] (rwp) [right = 1cm  of ru]  {$\textbf{w}$'$$};
    \node[style={rectangle,fill=gray!20,draw,minimum size=.5cm,inner sep=0pt}] (rup)[left = 1cm  of ru]{$\textbf{u}$'$$};
    \node[main node,fill=gray!20] (rv) [below = 1cm  of ru] {$\textbf{v}$};
    \node[main node,fill=gray!20] (rvp) [right = 1cm  of rv] {$\textbf{v}$'$$};
    \node[main node,fill=gray!20] (rw) [left = 1cm  of rv] {$\textbf{w}$};
    \node[] (art2) [below = .5cm  of up] {};
    
    \path[draw,thick]
    (rup) edge node {} (rvp)
    (ru) edge node {} (rw)
    (rv) edge node {} (rwp)
     (rv) edge node  {} (ru);
    ;
    \end{scope}
    
    \draw [-to,thick,snake=snake,segment amplitude=.4mm,
         segment length=2mm,line after snake=1mm]
    ([xshift=5mm]art1 -| vp) -- ([xshift=-5mm]art2 -| rup)
    node [above=1mm,midway,text width=3cm,text centered]
      { Rewiring };
\end{tikzpicture}
\end{center}
\caption[\textbf{Rewiring with 3 swaps}]{Rewiring with 3 swaps: node \textbf{u}  is connected to  node \textbf{v}  two times, there are nodes \textbf{u'} with the same degree as \textbf{u}, and \textbf{v'} with the same degree as \textbf{v} which are not connected. There are nodes \textbf{w} (neighbor of \textbf{u'} but not a neighbor of \textbf{u}), and \textbf{w'} (neighbor of \textbf{v'} but not a neighbor of \textbf{v}). We remove edges \textbf{u}\textbf{v}, \textbf{u'}\textbf{w} and \textbf{v'}\textbf{w'}, add edges \textbf{u}\textbf{w}, \textbf{v}\textbf{w'}, and \textbf{u'}\textbf{v'}.}
\label{fig:case1}
\end{figure}
\item Case 2: If \textbf{u'} is the same as \textbf{u}, then \textbf{v'} has to be different from \textbf{v}, in that case, because \textbf{v} has $j-1$ disjoint neighbors and \textbf{v'} has $j$ disjoint neighbors, then \textbf{v'} has a neighbor, \textbf{w'}, which is not a neighbor of \textbf{v}.
Therefore, we rewire the network by removing edges \textbf{u}\textbf{v} and \textbf{w'}\textbf{v'} and adding the edges \textbf{u}\textbf{v'} and \textbf{w'}\textbf{v}.
The Figure (\ref{fig:case2})  illustrates this process. For the case when \textbf{v'} is the same as \textbf{v} we have similar approach. 

\begin{figure}[htp]
\begin{center}
\begin{tikzpicture}
    \begin{scope}
    \node[style={rectangle,fill=gray!20,draw,minimum size=.5cm,inner sep=0pt}] (u) {$\textbf{u}=\textbf{u'}$};
     \node[style={rectangle,fill=gray!20,draw,minimum size=.5cm,inner sep=0pt}] (wp) [right = 1cm  of u]  {$\textbf{w}$'$$};
    \node[main node,fill=gray!20] (v) [below = 1cm  of u] {$\textbf{v}$};
    \node[main node,fill=gray!20] (vp) [right = 1.17cm  of v] {$\textbf{v}$'$$};
    \node[] (art1) [below = .5cm  of wp] {};
    
    \path[draw,thick]
    (u) edge[bend right] node {} (v)
     (vp) edge node {} (wp)
     (v) edge[bend right] node [left] {} (u);
    ;
    \end{scope}
    
    \begin{scope}[xshift=6cm]
    \node[style={rectangle,fill=gray!20,draw,minimum size=.5cm,inner sep=0pt}] (ru) {$\textbf{u}=\textbf{u'}$};
     \node[style={rectangle,fill=gray!20,draw,minimum size=.5cm,inner sep=0pt}] (rwp) [right = 1cm  of ru]  {$\textbf{w}$'$$};
    \node[main node,fill=gray!20] (rv) [below = 1cm  of ru] {$\textbf{v}$};
    \node[main node,fill=gray!20] (rvp) [right = 1cm  of rv] {$\textbf{v}$'$$};
    \node[] (art2) [below = .5cm  of up] {};
    
    \path[draw,thick]
    (ru) edge node {} (rvp)
    (ru) edge node {} (rv)
     (rv) edge node  {} (rwp);
    ;
    \end{scope}
    
    \draw [-to,thick,snake=snake,segment amplitude=.4mm,
         segment length=2mm,line after snake=1mm]
    ([xshift=5mm]art1 -| vp) -- ([xshift=-5mm]art2 -| ru)
    node [above=1mm,midway,text width=3cm,text centered]
      { Rewiring };
\end{tikzpicture}
\end{center}
\caption[\textbf{Rewiring with 2 swaps}]{Rewiring with 2 swaps: \textbf{v'} is not neighbor of \textbf{u} and has a neighbor like \textbf{w'} which is not neighbor of \textbf{v}, we remove edges \textbf{u}\textbf{v} and \textbf{w'}\textbf{v'} and add edges \textbf{u}\textbf{v'} and \textbf{w'}\textbf{v}.}
\label{fig:case2}
\end{figure}
\end{itemize}

\end{proof}
\section{Some Definitions Related to Bipartite  Network}
Before designing algorithms to generate bipartite networks  we define some properties for bipartite networks which are used in our network analysis.

\begin{definition} We define $Nc$ as  the number of connected components of the network.
\end{definition}
 
 \begin{definition} We define $Sg$ as the size of giant component (the biggest connected component) of the network.
\end{definition}
 \begin{definition} 
The clustering coefficient for a network is the average of clustering coefficient for all nodes. 
The bipartite clustering coefficient for a node  is a measure of local density of connections defined as \cite{latapy2008basic}
$$cl(\textbf{v})=\frac{\sum_{\textbf{u} \in N(N(\textbf{v}))} c_{\textbf{uv}}}{|N(N(\textbf{v}))|}$$
where $N(N(\textbf{v}))$ are the second order neighbors of \textbf{v} in network excluding \textbf{v}, and $c_{\textbf{uv}}$ is the pairwise clustering coefficient between nodes \textbf{u} and \textbf{v} and defined by
$$c_{\textbf{uv}}=\frac{|N(\textbf{u})\cap N(\textbf{v})|}{|N(\textbf{u})\cup N(\textbf{v})|}.$$
\end{definition}

\begin{definition} Redundancy coefficients $Rc$ are  measure of the degree to which nodes in a bipartite graph tend to cluster together. For a bipartite network,  redundancy for a node  is the ratio  of its overlap  to its maximum possible overlap  according to its degree.
The overlap of a node is the number of pairs of neighbors that have mutual neighbors themselves, other than that node \cite{latapy2008basic}. For a typical node \textbf{v}, the redundancy coefficient of \textbf{v} is defined as 

$$Rc(\textbf{v})=\frac{|\{\{\textbf{u,\textbf{w}}\}\subseteq N(\textbf{v}), \exists \textbf{v'}\neq \textbf{v}~s.t~\textbf{uv'}\in\textbf{E}, \textbf{wv'}\in \textbf{E} \}|}{\frac{|N(\textbf{v})|(|N(\textbf{v})|-1)}{2}},$$

where, $N(\textbf{v})$ is set of all neighbors of node \textbf{v}, and \textbf{E} is set of all edges in the network.
\end{definition}  
In graph theory and network analysis, we can  identify the most important nodes  within a network using centrality score of nodes. For example the most influential person(s) in a sexual network, such as  super-spreaders of disease. Through centrality scores we can seek to quantify the influence of every node in the network. We define some of the most important centralities which are used widely in network analysis: degree, betweenness and closeness centrality scores.
\begin{definition}{} The first and conceptually simplest centrality is degree centrality which is defined by the number of links to a  node: {degree centrality} of node \textbf{v} is given by $$de(\textbf{v}) = \frac{\textbf{deg}(\textbf{v})}{N-1},$$
where $\textbf{deg(v)}$ is the number of neighbors of \textbf{v} and $N$ is the total number of nodes in the network. This centrality  interprets  the immediate risk of a node for catching or transmitting whatever is flowing through the network. 
\end{definition}
\begin{definition}{Betweenness Centrality}
 of a node \textbf{v} is given by 
$$bet(\textbf{v})=\frac{2}{(N-1)(N-2)}\sum_{\textbf{u}\neq \textbf{w}\neq \textbf{v}}\frac{\sigma_{\textbf{uw}}(\textbf{v})}{\sigma_{\textbf{uw}}},$$
where $\sigma_{\textbf{uw}}$  is the total number of shortest paths from node \textbf{u} to node \textbf{w}, $\sigma_{\textbf{uw}}(\textbf{v})$ is the number of those paths that pass through \textbf{v}, and $N$ is the number of nodes in the graph. This centrality quantifies the number of times a node acts as a bridge along the shortest path between two other nodes. Therefore, the nodes  that have a high probability to occur on a randomly chosen shortest path between two randomly chosen node have a high betweenness  score \cite{brandes2008variants}.
\end{definition}
\begin{definition}{Closeness Centrality}  of a node \textbf{v} -in a not necessarily connected network- is sum of the reciprocal of   the shortest path distances from \textbf{v} to all $N-1$ other nodes. Since the sum of distances depends on the number of nodes in the graph, closeness is normalized by the sum of minimum possible distances $N-1$:
$$Clo(\textbf{v}) = \frac{\sum_{\textbf{u}=1}^{N-1} \frac{1}{d(\textbf{v},\textbf{u})}}{N-1},$$
where $d(\textbf{v},\textbf{u})$ is the shortest-path distance between \textbf{v} and \textbf{u}, and $N$ is the number of nodes in the graph \cite{freeman1978centrality}. In this concept the more central a node is, the closer it is to all other nodes.

\end{definition}
\section{Generating Bipartite Network}

We introduce five $B2K$ algorithms to construct simple bipartite networks for a given $BJD$ matrix satisfying the assumptions of Theorem \ref{simple}.
These algorithms are categorized as either an \emph{edge} or \emph{node} algorithm, depending on the network generation  process.
Both approaches can be used to generate an ensemble of B2K networks.
However, as we will show, the statistical properties of the networks for  different algorithms differ.
That is, the different approaches have different biases in sampling the space of all feasible networks.

The algorithms begin by grouping nodes into upper and lower nodes and assigning each node a desired degree based on the $BJD$ matrix.
In the edge algorithms, we first choose an entry in the BJD matrix - a tuple of the desired degrees of the upper node and lower node.
We find the list of pair of nodes that satisfies the conditions of the tuple.
We choose one pair of nodes randomly from this list and attach the edge if one does not exist.
If we cannot find a pair of nodes in the list that do not have an edge between them, we choose one at random and add a double-edge.
We repeat adding edges until all edges are placed.
If we attached a double-edge during the generation, we rewire the graph as described in subsection \ref{rewire1}.
In practice, we found that the node based algorithms were more computationally expensive than the edge based algorithms.

The edge algorithms start with the unconnected list of upper and lower nodes and iteratively add new edges guided by the current state of the edges in the network:
\begin{itemize}
\item \textbf{Random Edge (RE)}:  Choose one (i,j) randomly from the BJD.  We find a pair of upper and lower nodes with degrees i, j respectively and add an edge between them, then update $(i,j)\rightarrow (i,j)-1$.  This process continues for each edge until the $BJD$ becomes zero matrix.
\begin{algorithm}
 \SetAlgoLined
  \While{$BJD>0$}{
  Randomly select an element $(i,j)$\;
  Randomly select an upper node \textbf{v} with degree $i$ and $stub(\textbf{v})>0$\;
       Randomly select a lower node \textbf{v} with degree $j$ and $stub(\textbf{v})>0$\;
       Make edge \textbf{u}\textbf{v}.  $stub(\textbf{u})\leftarrow stub(\textbf{u})-1,$ and $stub(\textbf{v})\leftarrow stub(\textbf{v})-1$\;
      $(i,j) \leftarrow (i,j)-1.$
 }
 \caption{ Random Edge (\textbf{RE})}\label{RE}
\end{algorithm}
\item \textbf{Maximum Edge-Degree (\ED/)}: Of the remaining edges choose an edge to add from those with $max(d_{\text{max}})$ where $d_{\text{max}}(i,j)\defeq\max (i,j)$ at random until there are no remaining edges.
\begin{algorithm}
 \SetAlgoLined
  \While{$BJD>0$}{
  For remaining $(i,j) > 0$, find $m$ where $m \defeq \max\limits_{(i,j)>0} \left\{i,j\right\}$\;
    Randomly select an element $(i,j) > 0$ where $i=m$ or $j=m$\;
     Randomly select an upper node \textbf{u} with degree $i$ and $stub(\textbf{u})>0$\;
       Randomly select a lower node \textbf{v} with degree $j$ and $stub(\textbf{v})>0$\;
      Make edge \textbf{u}\textbf{v}, $stub(\textbf{u})\leftarrow stub(\textbf{u})-1,$ and $stub(\textbf{v})\leftarrow stub(\textbf{v})-1$\;
      $(i,j) \leftarrow (i,j)-1.$
 }
 \caption{ Maximum Edge Degree (\ED/)}\label{ED+}
\end{algorithm}
\item \textbf{Total Edge-Degree (TED)}: Of the remaining edges choose an edge to add from those with $max(d_{\text{total}})$ where $d_{\text{total}}(i,j)\defeq i + j$ at random until there are no remaining edges.  (\textit{Note:} we observed no statistical differences between the TED and \ED/ approaches and will only present  results for the \ED/ algorithm.)

\begin{algorithm}[H]
 \SetAlgoLined
  \While{$BJD>0$}{
   For remaining $(i,j) > 0$, find $m$ where $m\defeq \max\limits_{(i,j)} \left(i+j\right)$\;
     Randomly select an element $(i,j) > 0$ where $i+j=m$\;
   Randomly select an upper node \textbf{u} with degree $i$ and $stub(\textbf{u})>0$\;
      Randomly select a lower node \textbf{v} with degree $j$ and $stub(\textbf{v})>0$\;
     Make edge \textbf{u}\textbf{v}, $stub(\textbf{u})\leftarrow stub(\textbf{u})-1,$ and $stub(\textbf{v})\leftarrow stub(\textbf{v})-1$\;
       $(i,j) \leftarrow (i,j)-1.$
 }
  \caption{ Total Edge Degree (\textbf{TED})}\label{TED}
\end{algorithm}
\end{itemize}
The  node algorithms start with the unconnected list of upper and lower nodes and iteratively add new edges  
based on current state of the nodes in the network:
\begin{itemize}
\item \textbf{Maximum Node-Degree (\ND/)}: Choose from the nodes with the highest desired degree.
Choose possible edges that the chosen node could have and select appropriate neighbors.
Select all of the neighbors for the chosen node, then choose another node from those with the highest desired degree.
Repeat until all edges are added.
\begin{algorithm}[H]
 \SetAlgoLined
  \While{$BJD>0$}{
  From remaining nodes with positive stub, find the one with highest degree: $\textbf{u}= \max\limits_{deg(\textbf{v})} \{\textbf{v}\}$\;
  \While {$stub(\textbf{u})>0$}{
  \eIf{\textbf{u} is upper node with degree $i$ }{
    From the row i of $BJD$ matrix randomly select an element $(i,j)>0$\;
       Randomly select a lower node \textbf{v} with desired degree $j$    and $stub(\textbf{v}) >0$\;
       Make edge \textbf{u}\textbf{v}, $stub(\textbf{u})\leftarrow stub(\textbf{u})-1,$ and $stub(\textbf{v})\leftarrow stub(\textbf{v})-1,$ and $(i,j) \leftarrow (i,j)-1$\;
   }{
    From the column j of $BJD$ matrix randomly select an element $(i,j)>0$\;
      Randomly select an upper  node \textbf{v} with desired degree $i$    and $stub(\textbf{v}) >0$\;
   Make edge \textbf{u}\textbf{v}, $stub(\textbf{u})\leftarrow stub(\textbf{u})-1,$ and $stub(\textbf{v})\leftarrow stub(\textbf{v})-1,$ and $(i,j) \leftarrow (i,j)-1$\; 
  }}
  }
   \caption{ Maximum Node Degree (\textbf{ND+})}\label{ND+}
\end{algorithm}
\item \textbf{Maximum Stub Minimum Node-Degree (\SND/)}: Find the nodes with the fewest placed edges -- the most stubs -- and sort them by their desired degree, and choose the one with the maximum desired degree. Choose possible edges that the chosen node could have and select appropriate neighbors and make the edge.
Continue until all edges are added.

\begin{algorithm}[H]
 \SetAlgoLined
  \While{$BJD>0$}{
    From remaining nodes, find the node  \textbf{u} with $\textbf{u}=\min\limits_{deg(\textbf{v})} \{\max\limits_{stub(\textbf{v})} \{\textbf{v}\} \}$\;
  
  \eIf{\textbf{u} is upper node with degree $i$ }{
    From the row i of $BJD$ matrix randomly select an element $(i,j)>0$\;
       Randomly select a lower node \textbf{v} with desired degree $j$    and $stub(\textbf{v}) >0$\;
       Make edge \textbf{u}\textbf{v}, $stub(\textbf{u})\leftarrow stub(\textbf{u})-1,$ and $stub(\textbf{v})\leftarrow stub(\textbf{v})-1,$ and $(i,j) \leftarrow (i,j)-1$\;
   }{
    From the column j of $BJD$ matrix randomly select an element $(i,j)>0$\;
      Randomly select an upper  node \textbf{v} with desired degree $i$    and $stub(\textbf{v}) >0$\;
   Make edge \textbf{u}\textbf{v}, $stub(\textbf{u})\leftarrow stub(\textbf{u})-1,$ and $stub(\textbf{v})\leftarrow stub(\textbf{v})-1,$ and $(i,j) \leftarrow (i,j)-1$\; 
  }
  }
    \caption{ Maximum Stub Minimum Degree (\SND/)}\label{S+ND-}
\end{algorithm}
\end{itemize}
\section{Rewiring Approach}\label{rewire1}
During the construction of our network, it is possible that there is not two valid nodes, \textbf{u} and \textbf{v}, each of valid desired degrees, $i$ and $j$ that do not already have an edge between them.
Our options are to increase a node beyond its desired degree, or attach a multiple edge.
Because the final network satisfies a certain BJD, we choose to simplify the generation by allowing multiple edges between nodes \textbf{u} and \textbf{v} as long as the network still satisfies $(i,j)$.
Once all of the edges are attached, we use rewiring to remove multiple edges, which maintains the proper edge count for each $(i,j)$.

Our approach follows the proof of theorem \ref{simple}.
Suppose \textbf{G} is the bipartite network generated by one of the approaches defined in the last subsection, and suppose \textbf{G} has at least one multiple edge.
We randomly start with one of the multiple edges, say edge attached to the nodes \textbf{u} and \textbf{v} and then we follow the rewiring process explained in the theorem \ref{simple}. Here is the  algorithm of the rewiring process:
\vspace*{5mm}

\begin{algorithm}[H]
 \SetAlgoLined
  \While{Network \textbf{G} ~is not simple}{
    Select upper node \textbf{u} and lower node \textbf{v} with more than one edge between them\;
    \uIf{There is lower node \textbf{v'} with $deg(\textbf{v'})=deg(\textbf{v})$ not connected to \textbf{u}}{
     Find a neighbor of \textbf{v'} which is not neighbor of \textbf{v}: upper node \textbf{w'}\;
   Remove edges \textbf{u}\textbf{v} and \textbf{v'}\textbf{w'}\;
 Add edges \textbf{u}\textbf{v'}, \textbf{v}\textbf{w'}.
  }
    \uElseIf{There is lower node \textbf{u'} with $deg(\textbf{u'})=deg(\textbf{u})$ not connected to \textbf{v}}{
   Find a neighbor of \textbf{v'} which is not neighbor of \textbf{v}: upper node \textbf{w'}\;
   Remove edges \textbf{u}\textbf{v} and \textbf{v'}\textbf{w'}\;
 Add edges \textbf{u}\textbf{v'}, \textbf{v}\textbf{w'}.
  }
  \Else{
     Find upper node \textbf{u'} with $deg(\textbf{u'})=deg(\textbf{u})$, a
     lower node \textbf{v'} disconnected to \textbf{u'} with $deg(\textbf{v'})=deg(\textbf{v})$\;
    Find a neighbor of \textbf{u'} which is not neighbor of \textbf{u}: upper node \textbf{w}\;
    Find a neighbor of \textbf{v'} which is not neighbor of \textbf{v}: upper node \textbf{w'}\;
   Remove edges \textbf{u}\textbf{v}, \textbf{u'}\textbf{w}, and \textbf{w'}\textbf{v'}\;
   Add edges \textbf{u}\textbf{w}, \textbf{w'}\textbf{v}, and \textbf{u'}\textbf{v'}.
      }
    }
\caption{ Rewiring Process}\label{rewire}
\end{algorithm}
\section{Romance Network}  \label{numerical_result}
In order to verify the accuracy of presented $B2K$ algorithms, we conduct some simulations by generating random bipartite networks, using RE, \ED/, \ND/, and \SND/ algorithms.
We test the algorithms on an special
real-world network called Romance network by computing several properties of bipartite networks. We compare  $Nc$, $Sg$, $Cl$, and $Rc$.

The network of sexual contact depicted in Figure (\ref{fig:romanc})
describes the structure of the adolescent romantic and sexual network in a population of $573$ students at Jefferson High \cite{bearman2004chains}.
The original network is not a bipartite network: there are two edges that links two men and two women, representing homosexual relationships.
We remove these two edges so that we have a bipartite network. 
\begin{figure}[H]
\centering
\includegraphics[width=4in]{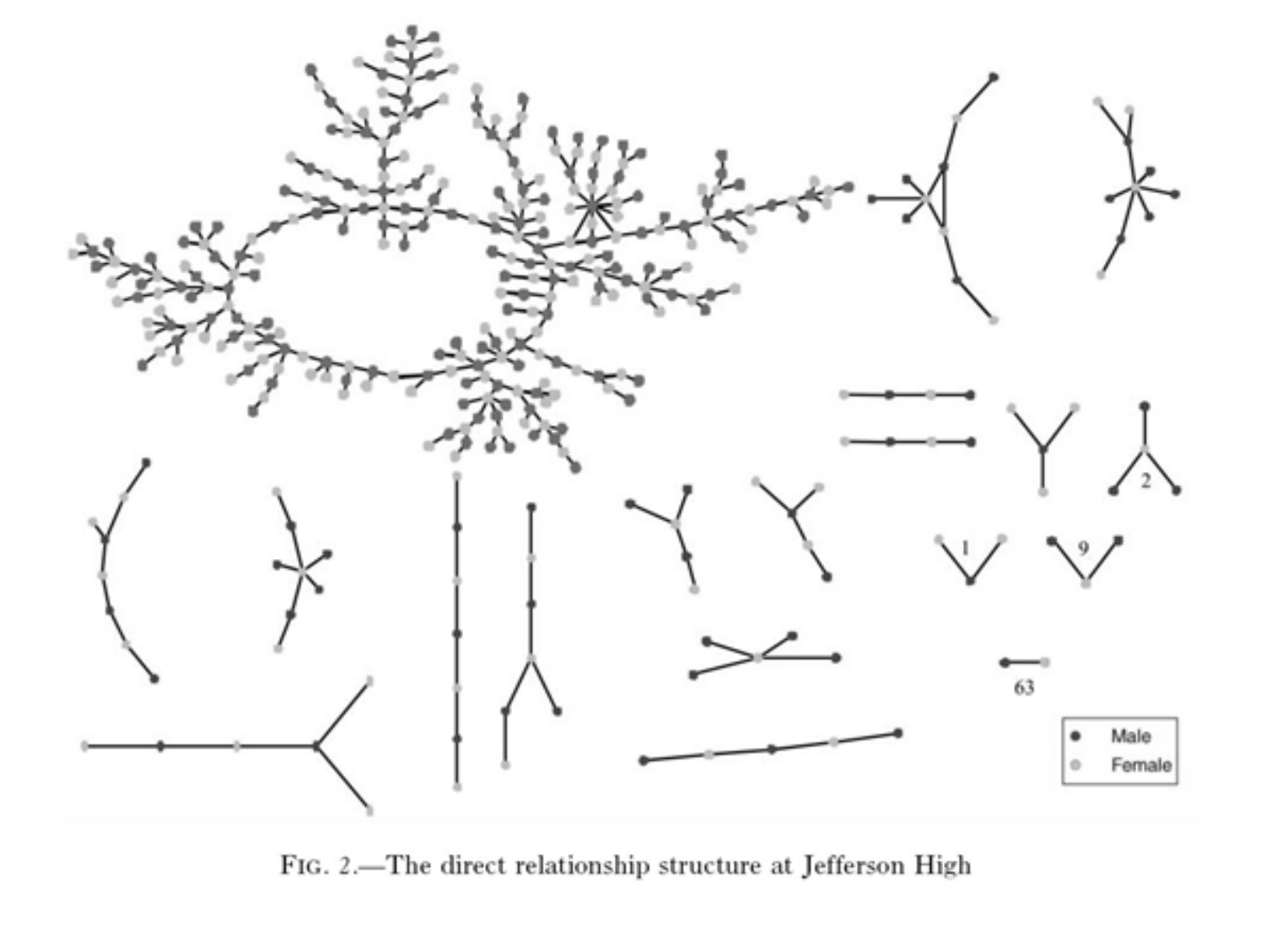}
\caption[\textbf{Romance contact network }]{The romance contact network at Jefferson High \cite{bearman2004chains} consists of a single large connected component and several smaller romance groups. }
\label{fig:romanc}
\end{figure}
We begin extracting the $BJD$ matrix of this network, called $BJD_{R}$.
The matrix has $k_w=6$ rows, and $k_m=8$ columns, where $k_w$ is the maximum degree for women  and $k_m$ would be maximum degree for men.

Each element $(i,j)$ is the number of edges between women with  $i$ partners and men with  $j$ partners. By $BJD_R$, we are able to find degree distribution for women and men: the number of women with  $i$ partners is summation of elements in $i-th$ row divided by $i$, and the number of men with  $j$ partners is summation of elements in $j-th$ column divided by $j$. 
Here is the matrix extracted from the network:
\[ BJD_{R}= \left( \begin{array}{cccccccc}
 $63$ & $56$ & $30$ & $16$ & $2$ & $0$ & $0$ & $2$ \\
 $46$ & $40$ & $25$ & $14$ & $1$ & $0$ & $0$ & $4$\\
 $23$ & $20$ & $18$ &  $8$ & $1$ & $0$ & $0$ & $2$\\
 $26$ & $24$ & $21$ &  $4$ & $1$ & $0$ & $0$ & $0$\\ 
  $8$ &  $9$ &  $1$ &  $2$ & $0$ & $0$ & $0$ & $0$ \\ 
  $4$ &  $1$ &  $1$ &  $0$ & $0$ & $0$ & $0$ & $0$ \\  \end{array} \right).\] 
As the matrix shows, we have big numbers for low degrees (upper left corner of matrix) and the rest are small or zero, making the average degree very low.
Using our algorithms we compare the properties of $BJD_{R}$ and the original romance network.
In our numerical tests of generating an ensemble of $10000$ networks for each algorithm, the algorithms succeeded in generating networks that preserved both the degree and joint-degree distributions in every simulation. 
We observed that the statistical properties of the networks for the $B2K$ algorithms were different, as shown in Table (\ref{compare_paras_RN}). This Table lists the $Sg$, $Nc$, $Cl$, and the $Rc$ for the ensemble of generated networks.  
Note that the average size of giant component for the real network is noticeably above the mean size of giant components in the randomly generated networks, especially \SND/. 

\begin{table}[htp]
  \centering
  \resizebox{.99\textwidth}{!}{
\begin{tabular}{@{}c d{4} d{4} d{4} d{4} d{4} d{4} @{}}
    \toprule
&  & \mlc{Real Network } & RE
    & $\ED/$ &  $\ND/$& $\SND/$\\
    \midrule
&<Cl>&$0.3300$& $0.3338$ & $ 0.3638$ &   $0.3617$ & $0.3009$  \\ [1ex] 
     & SD(Cl)
&--& $0.0052$ & $0.0053$ &  $0.0058$ & $0.0032$  \\ [1ex]
     & <Rc> 
&$0.0040$& $0.0039$ & $0.0071$ &   $0.0117$ & $ 0.0015$  \\ [1ex]
     & SD(Rc)
&--& $0.0057$ & $0.0074$ &   $0.0091$ & $0.0033$  \\ [1ex]
&<Sg>  
&$287$& $217.57$ & $256.35 $ &   $244.79$ & $87.0603$  \\ [1ex]
&SD(Sg)  
&--& $43.90 $ & $19.03$ &   $21.88$ & $31.4488$  \\ [1ex]
&<Nc>  
&$101$& $102.59$ & $110.09$ &   $109.47$ & $99.5267$  \\ [1ex]
&SD(Nc) 
&--& $ 1.7910$& $2.51$ &   $2.4992$ & $0.7099$  \\ [1ex]
   \bottomrule
\end{tabular}}
\caption[\textbf{Properties of the real and randomly generated Romance networks}]{Properties of the real and randomly generated Romance networks.  
Note that the  giant component size $Sg$ of the real network is larger than the average 
size of the giant component in the randomly generated networks.}
\label{compare_paras_RN}
\end{table}

The Figure (\ref{romance-bars}) plots the distribution of properties of the $10000$ simulated networks  from $B2K$ algorithms. For $Cl$, subfigure \ref{cl}, we observe all the algorithms have a normal trend with an small variance, however, for \SND/ algorithm the mean value is smaller than the others, depicting the fact that joint-degree distribution may not be enough to capture clustering coefficient of the network. The subfigure \ref{rc} is the distribution of $Rc$ for the networks which all are  right skewed.
An interesting result from $Sg$ is that \SND/ underestimates this value compared to other algorithms and has a weak right skew unless the others which are almost normal, subfigure \ref{sg}. Finally for $Nc$, subfigure \ref{nc}, all the algorithms but \SND/ follow a normal distribution,  \SND/.

\begin{figure}[htp]
\centering
\subfloat[\textbf{Clustering coefficient}][Clustering Coefficient]{\label{cl}\includegraphics[width=0.85\textwidth]
{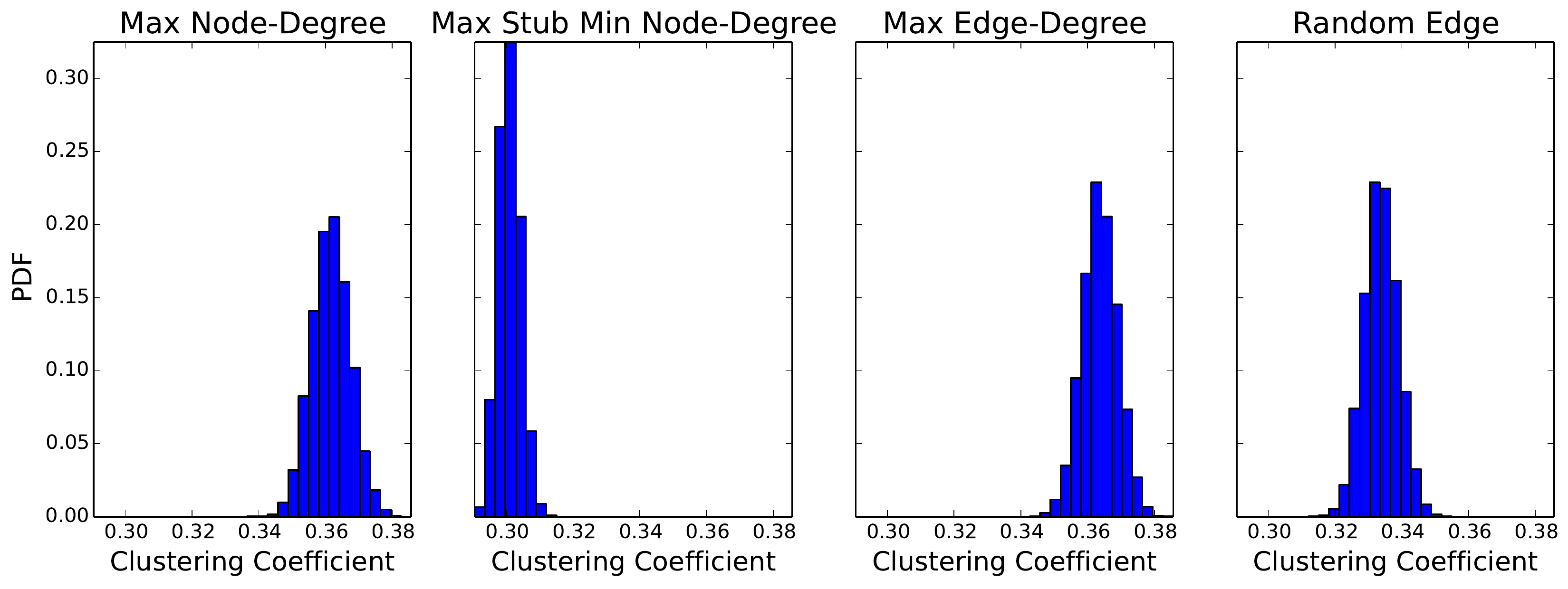}}\\
\subfloat[\textbf{Redundency coefficient}][Redundency Coefficient]{\label{rc}\includegraphics[width=0.85\textwidth]
{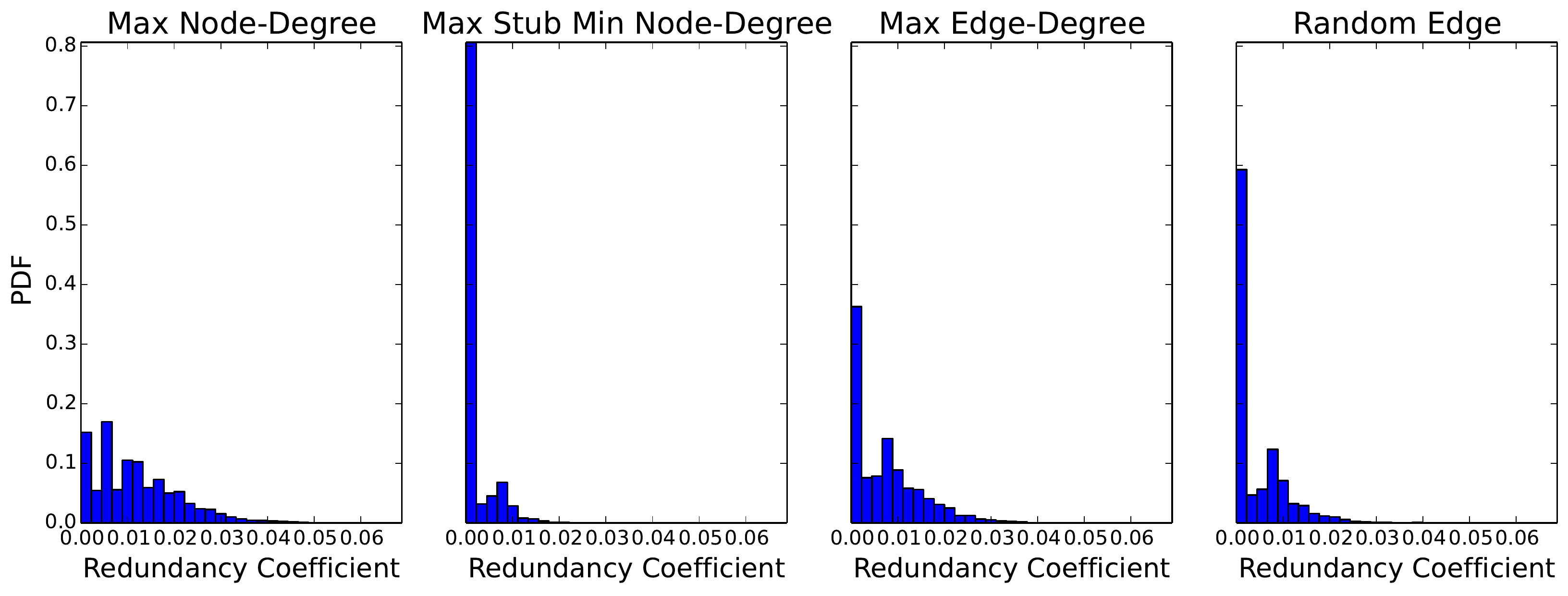}}\\
\subfloat[\textbf{Size of giant component}][Size of Giant Component]{\label{sg}\includegraphics[width=0.85\textwidth]{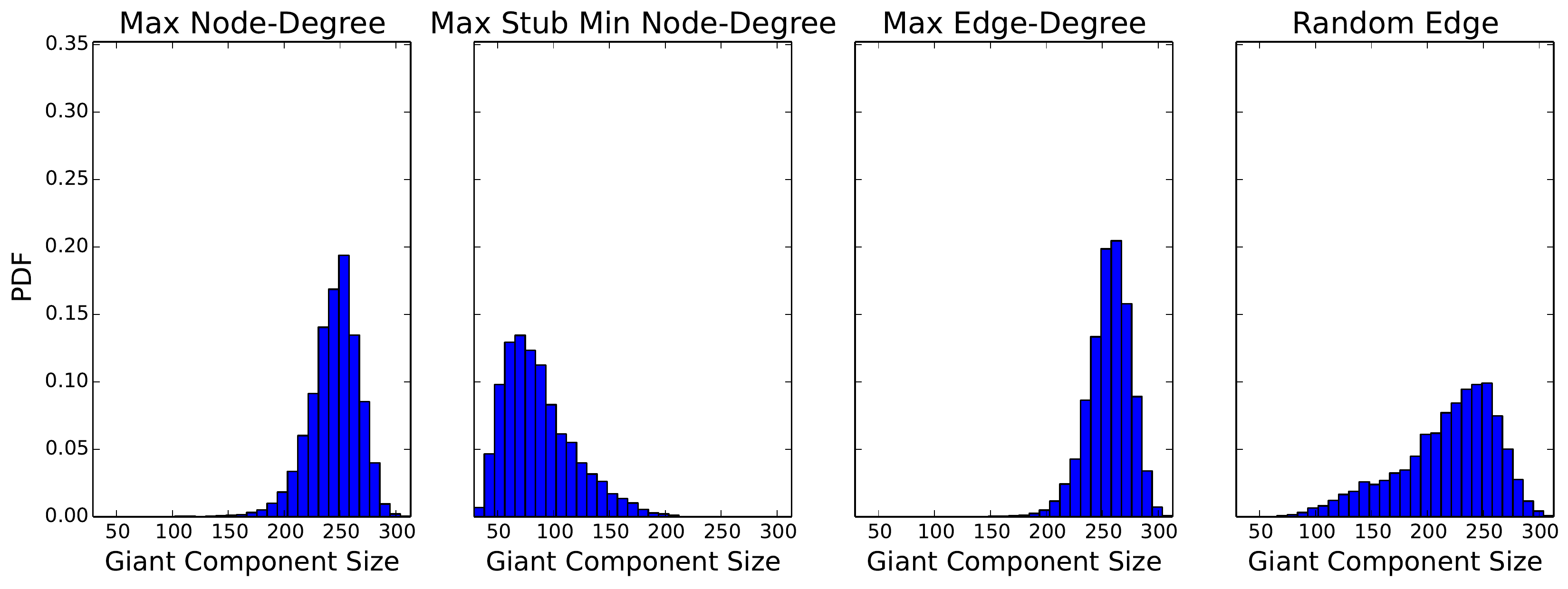}}\\
\subfloat[\textbf{Number of connected components}][Number of Connected Components]{\label{nc}\includegraphics[width=0.85\textwidth]{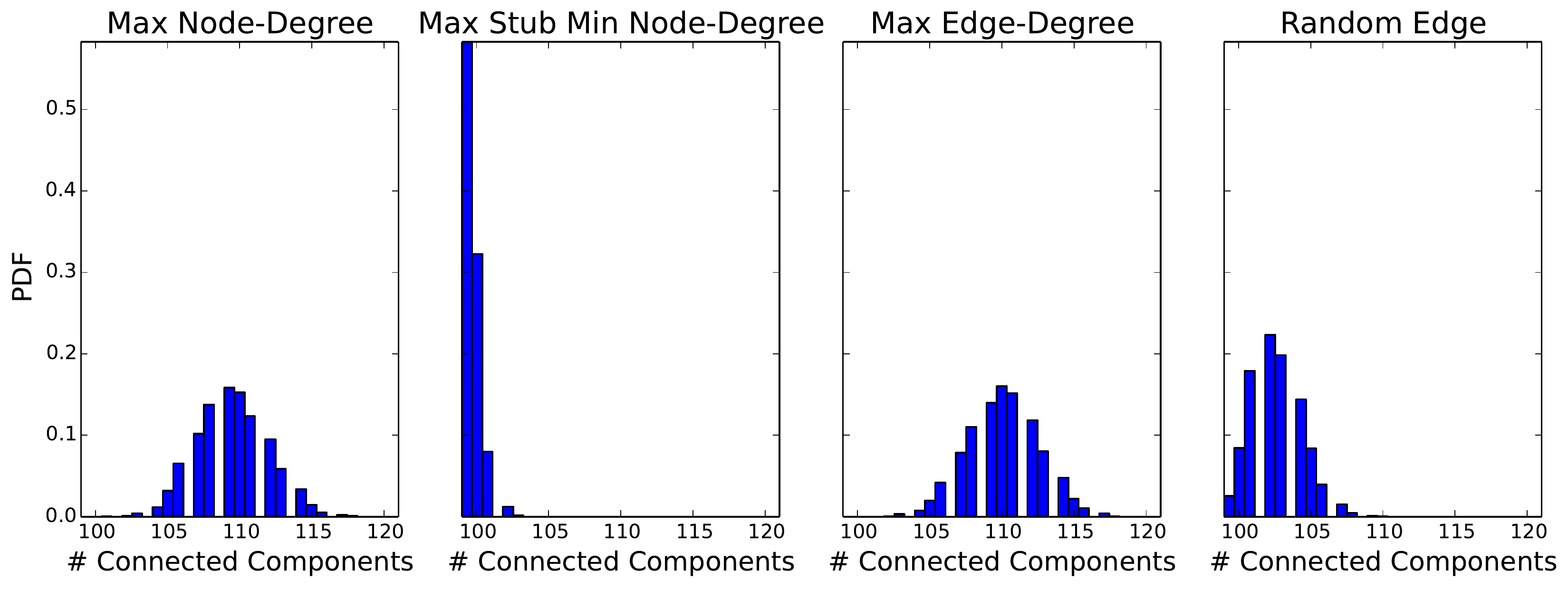}}
\caption[\textbf{Bar plot of distribution for properties of generated networks}]{Bar plot of distribution of properties of $10,000$ generated networks using $B2K$ algorithms.}
\label{romance-bars}
\end{figure}

We also compared \ED/ generating network algorithm with existing algorithms in NetworkX \cite{schult2008exploring} using the degree distribution based on the configuration model \cite{newman2002structure} and Havel Hakimi graph \cite{hakimi1962realizability}.
Currently, there are no NetworkX algorithms to use joint-degree distribution for bipartite networks.
The Table (\ref{T:network}) lists some properties of real Romance network, \ED/ network and all other algorithms in NetworkX generated using Romance data.

\begin{table}[htp]
  \centering
  \resizebox{.99\textwidth}{!}{
\begin{tabular}{@{}c d{4} d{4} d{4} d{4} d{4} d{4} @{}}
    \toprule
  Network Model& $Nc$ & $Sg$& $Cl$ &   $Rc$ \\ [0.5ex] 
    \midrule
 Real Romance Network & $101$& $287$  & $0.3395$ &   $0.0044 $\\ [1ex]
\ED/&$100$& $259$ & $ 0.3245$ &    $0.0049$ \\ [1ex] 
Configuration Model
&$101$& $139$ & $0.3693$ & $0.0000$  \\ [1ex]
Havel-Hakimi Network
&$178$& $82$ & $0.1827$ & $0.5046$  \\ [1ex]
Alternative Havel-Hakimi
&$120$& $50$ & $0.3834$ & $0.0683$   \\ [1ex]
Reverse Havel-Hakimi 
&$132$& $9$ & $0.8210$ & $0.5637$\\
 \bottomrule
\end{tabular}}
\caption[\textbf{Properties of real Romance network and of the networks generated by the \ED/ algorithm, configuration model and Havel Hakimi algorithm.}]{Properties of real Romance network and of the networks generated by the \ED/ algorithm, configuration model and Havel Hakimi algorithm. Note that the \ED/
accurately approximates the size of the giant component, the network clustering, and average redundancy coefficient in this low average degree network.}
\label{T:network}
\end{table}

As we see in the Table (\ref{T:network}), \ED/ algorithm is in better agreement with the real Romance network than other existing algorithms, though, it uses more information than these algorithms.
The giant component of generated network using these algorithms are shown in Figure (\ref{nets}). This Figure shows the network generated by \ED/ is in agreement with the real network more than other existing networks in NetworkX.
As expected, this example shows that the joint-degree distributions preserves more properties of the original network than the bipartite algorithms that just preserve the degree distribution.

\begin{figure}[htp]
\centering
\subfloat[\textbf{Romance Network} ][Romance Network ]{\label{fig:1a}\includegraphics[width=0.5\textwidth]
{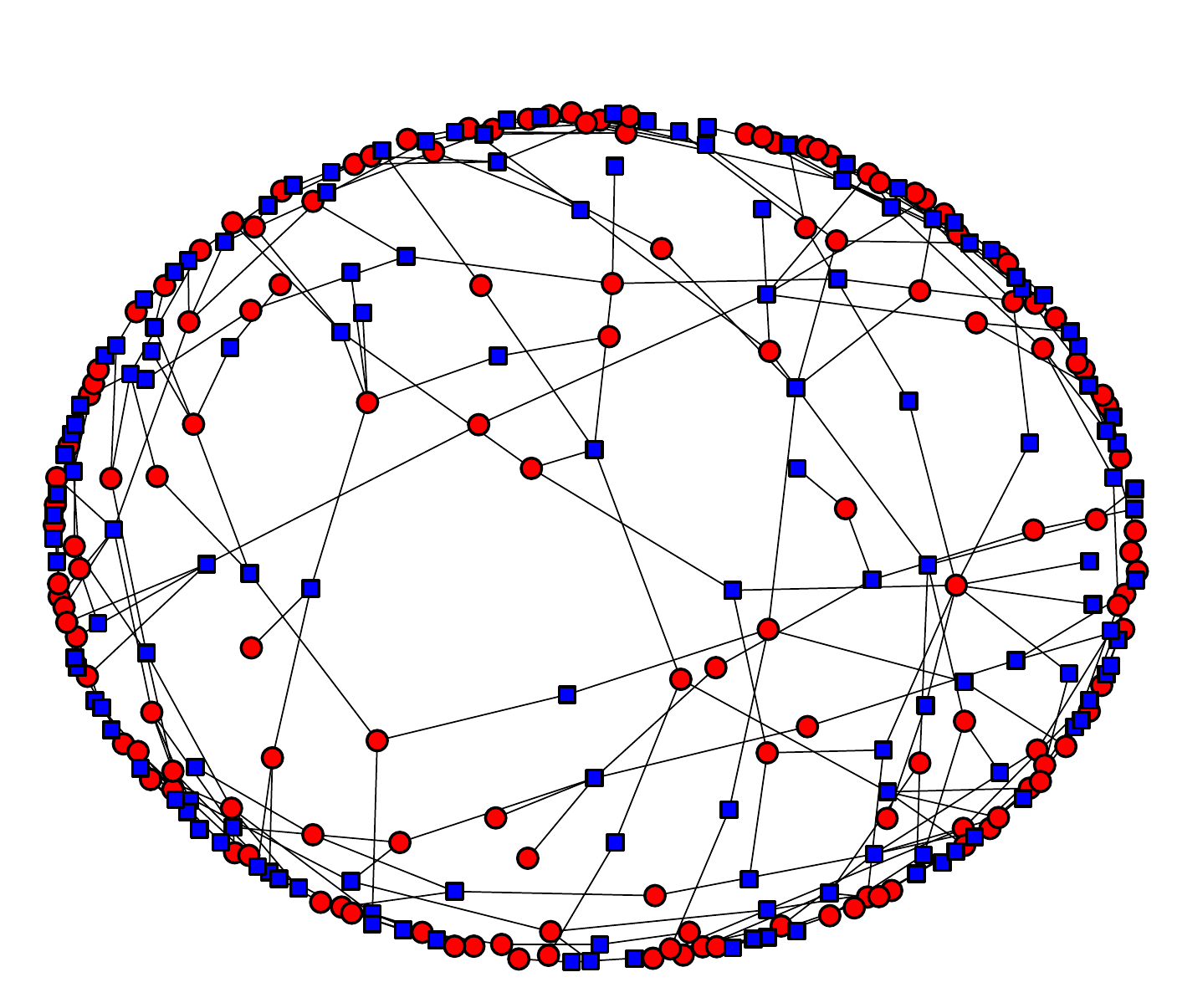}}
\subfloat[\textbf{\ED/ Algorithm}][\ED/ Algorithm]{\label{fig:1b}\includegraphics[width=0.5\textwidth]{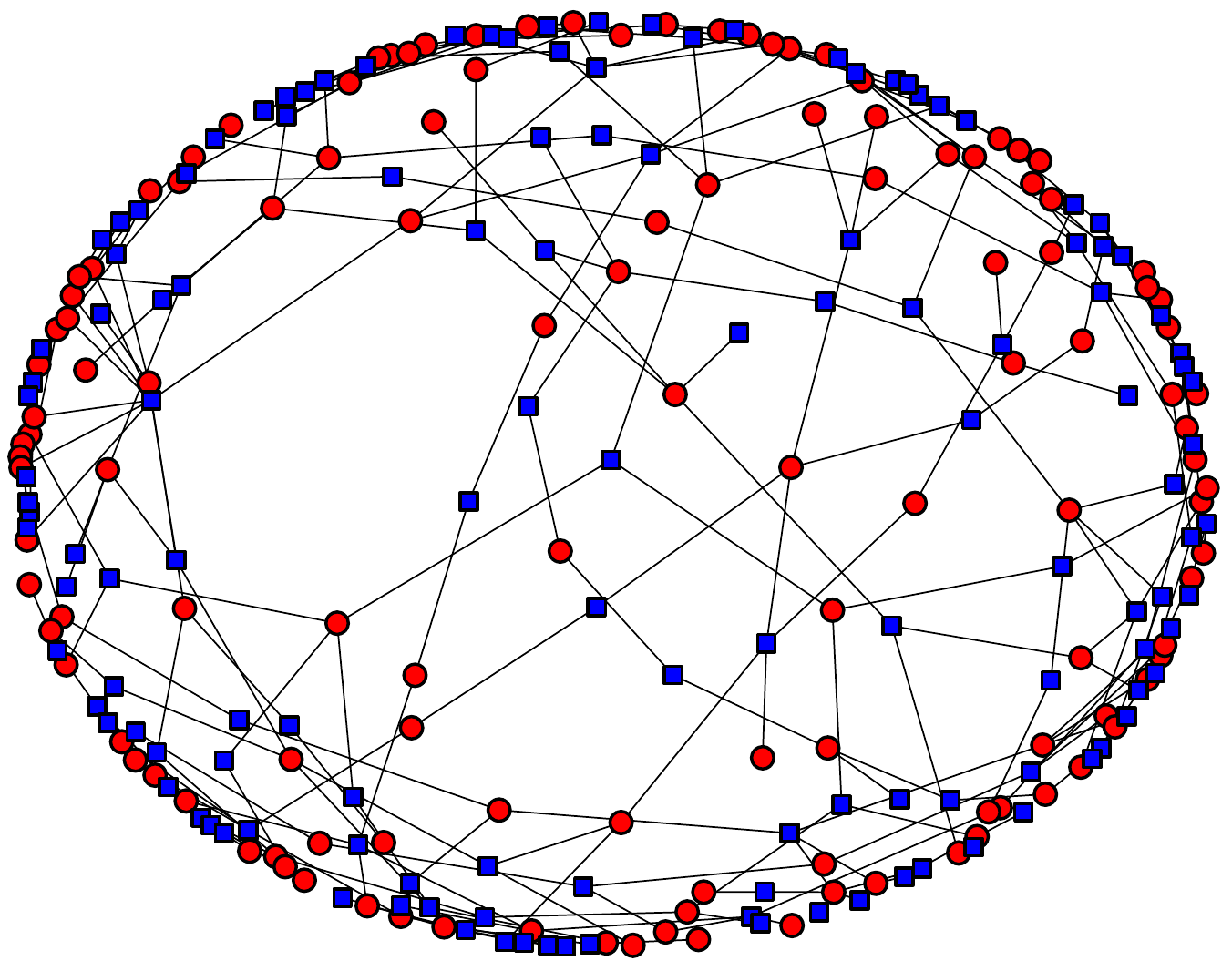}}\\
\subfloat[\textbf{Configuration Model}][Configuration Model]{\label{fig:1e}\includegraphics[width=0.5\textwidth]{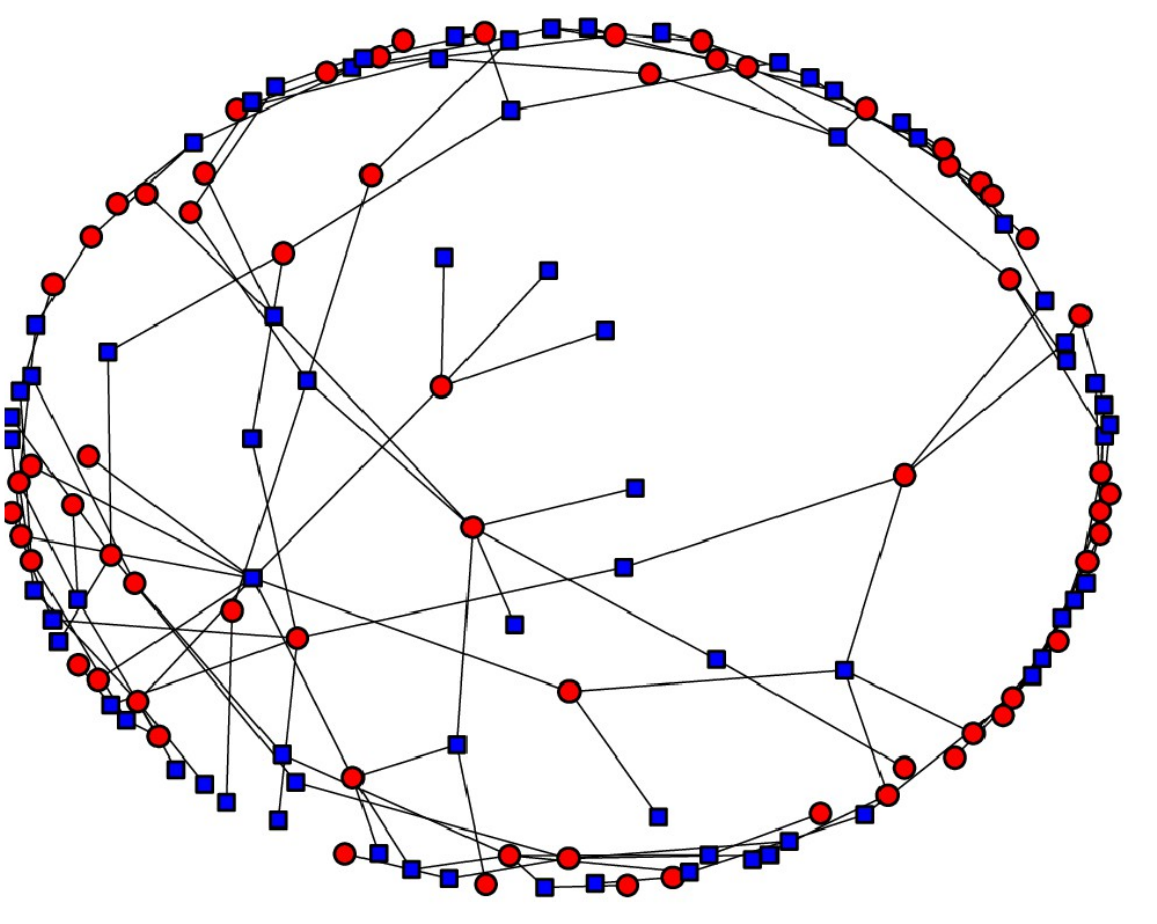}}
\subfloat[\textbf{Havel Hakimi model}][Havel Hakimi Model]{\label{fig:1f}\includegraphics[width=0.5\textwidth]{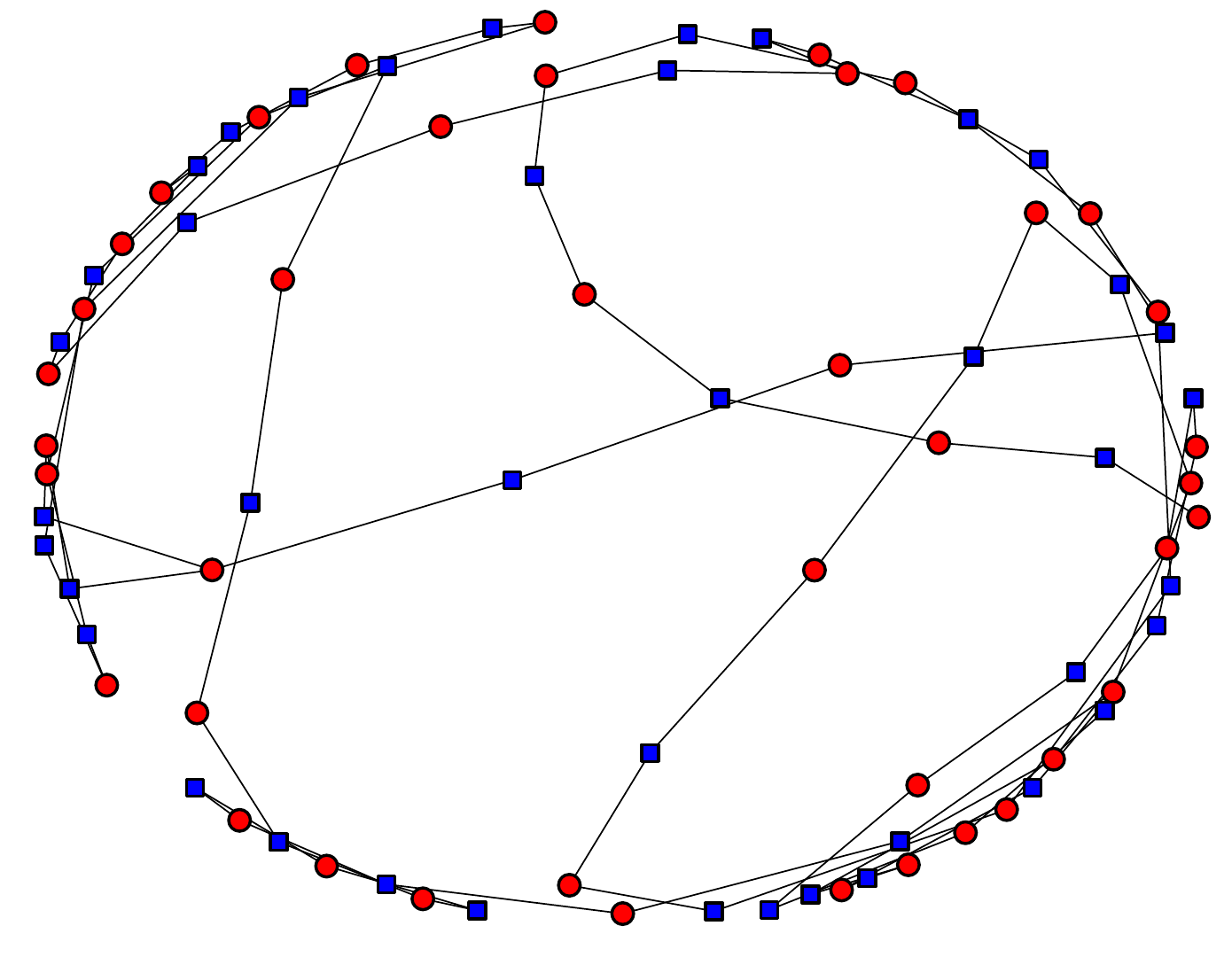}}
\caption[\textbf{Giant components of romance network}]{The structure of the giant components of romance network and the $B2K$ generated network are similar.
The configuration model and  Havel Hakimi  algorithms have the same degree distributions, but do not capture this property}\label{nets}
\end{figure}

\chapter{Sexual Activities Hidden in Social Organization: A Preferential Attachment Mechanism for Human Sexual Network Formation in  Social Network Context}\label{ap2}
The current appendix adds to our previous appendix on generating  heterosexual  networks with prescribed joint-degree distribution  considering impact of people's social behavior (i.e., non-sexual partners) on their sexual partner selection. We generate a sexual network which has two properties: first property is that our sexual network is a subgraph of social network with a structure defined in \cite{eubank2010detail}, second it follows a joint-degree (degree-degree) distribution  originated from sexual activity data.

\section{Degree and joint-degree distribution and $BJD$ matrix for a sexual network}
A conventional heterosexual network  $\textbf{G}$, which is from bipartite family networks, is a sexual network of men and women, in which men only have partnership with women and vise versa, there is no partnership between two men or two women. Every edge $\textbf{ij}$ in this network means that two persons \textbf{i} and \textbf{j} are sexual partners.
The \textit{degree}  of a person \textbf{i}, is defined as the number of his/her sexual partners.
The \textit{degree distribution} $d_k$ defines the number of people with  degree $k$.
The \textit {joint-degree distribution}  $(k,j)$ is the number of men with degree $j$ who are connected to women with  degree $k$.
This distribution  can be represented by the \textit{Bipartite Joint Degree} or $BJD$ matrix:

\[BJD_G = \parenMatrixstack{
    e_{11} & e_{12} & e_{13} & \dots  & e_{1m} \\
    e_{21} & e_{22} & e_{23} & \dots  & e_{2m} \\
    \vdots & \vdots & \vdots & \ddots & \vdots \\
    e_{w1} & e_{w2} & e_{w3} & \dots  & e_{wm} },\]

where, $w$ is the maximum degree in women nodes, and $m$ is the maximum degree in men nodes, each element $e_{ij}$ is the number of edges between women with $i$ partners and men with $j$ partners. 
The degree distribution of the number of women nodes, $d^w_k$, and men nodes, $d^m_k$, with  $k$ partners can be obtained from $BJD_G$:
\begin{eqnarray}
d^w_k=\frac{\sum_{j=1}^m e_{kj}}{k}~,\text{~~and~~~}~~~d^m_k=\frac{\sum_{i=1}^w e_{ik}}{k}~.
\end{eqnarray}

\section{Social network embedding sexual network}
The  Social Network  called  \so is a graph whose nodes are synthetic people, labeled by their demographics,
and whose edges represent contacts determined in which each synthetic person is deemed to have made contact with a subset of other synthetic people through some \textit{Activity} types. 
We have five different activity locations: Home($H$), Work($W$), School($SC$), Shopping($SH$), and Others($O$).
Each edge is labeled with one of these activity locations and is weighted by  the time spent on these contact per day. For example edge \textbf{i}\textbf{j} labeled by $W$ and weighted by  $T$   means two persons \textbf{i} and \textbf{j} have contact for $T$ fraction of their total time  spent at work  \cite{eubank2010detail}.  

\section{Generating Bipartite Sexual Network }\label{method}
Our goal is to introduce an algorithm for generating a heterosexual network called \se that is  a partial subgraph of social network \so and meets a particular joint-degree distribution represented by matrix $BJD$.
\subsection{The algorithm }
Now, we  provide an algorithm that uses the sexual activity data (BJD matrix) and social network to generate the sexual network of individual  called \se. This  graph is partially embedded in \so-depends on what percentage of sexual partners of a typical  person are selected from his/her social friends- and has a joint-degree distribution that meets BJD matrix. 
The algorithm includes three phases.
\subsubsection{Phase 1: extension and revision of SocNet }

This Phase  is a procedure to extend \so and  then to find its subgraph  in a way that  only consist the potensial edges for sexual activities:
\begin{enumerate}
\item The first step  is to condense \so by making friendship between friends of an index case. For two persons \textbf{i} and \textbf{j} who are not currently connected, suppose the probability of their meet through a common social friend like \textbf{k} within an activity $A$ is $p^A_{ij}$. If they have $k(i,j)$ common social friends through a particular activity $A$, therefore, with probability of $1-(1-p^A_{ij})^{k(i,j)}$ they meet each other and make friendship, that is, with probability of $1-(1-p^A_{ij})^{k(i,j)}$ we make a new social edges between \textbf{i} and \textbf{j}. We define $p^A_{ij}=p^A_ip^A_j$ where $p^A_k$ for a person \textbf{k} is the average fraction of time spent  per social friend for social friends within  activity $A$, that is, if $N_A(k)$ is set of all  social friends for person \textbf{k} through an activity location $A$, then 
\begin{eqnarray}
p^A_k=\frac{\sum_{l \in N_A(k) } T_{k,l}}{|N_A(k)|}.
\end{eqnarray}

 The Figure (\ref{frined_of_friend}) represents an schematic of this approach: for the persons \textbf{i} and \textbf{j} who are not currently  social friends but have three different common friends \textbf{k}$_0$,\textbf{k}$_1$, and \textbf{k}$_2$ though different activities. Suppose $A_{ik_0}=A_{jk_0}=A_{ik_1}=A_{jk_1}=A\neq A_{ik_2}\neq A_{ik_2}$, that is, \textbf{k}$_0$ meets \textbf{i} and  \textbf{j} at the same location, similarly \textbf{k}$_1$ meets \textbf{i} and  \textbf{j} at the same location to \textbf{k}$_0$'s, however, \textbf{k}$_2$ meets them in different places. To compute  $p^A_i$ and $p^A_j$ we only count the friends who meet them at the same location A, therefore, $p^A_i=\frac{T_{ik_0}+T_{ik_1}}{2}$, and $p^A_j=\frac{T_{jk_0}+T_{jk_1}}{2}$.
\begin{figure}[H]
\centering
  \begin{tikzpicture}[
      mycircle/.style={
         circle,
         draw=black,
         fill=gray,
         fill opacity = 0.3,
         text opacity=1,
         inner sep=0pt,
         minimum size=20pt,
         font=\small},
      myarrow/.style={-},
      node distance=1.5cm and 1.9cm
      ]
      \node[mycircle] (i) {$\textbf{i}$};
      \node[mycircle, right=of i] (k1) {\textbf{k}$_1$};
      \node[mycircle,above =of k1] (k0) {\textbf{k}$_0$};
       \node[mycircle,below =of k1] (k2) {\textbf{k}$_2$};
       \node[mycircle,right =of k1] (j) {\textbf{j}};
      
     \draw [myarrow] (i) -- node[sloped,above] {$A_{ik_0}$,~$T_{ik_0}$} (k0);
     \draw [myarrow] (i) -- node[sloped,above] {$A_{ik_1}$,~$T_{ik_1}$} (k1);
     \draw [myarrow] (i) -- node[sloped,above] {$A_{ik_2}$,~$T_{ik_2}$} (k2);
     \draw [myarrow] (j) -- node[sloped,above] {$A_{jk_0}$,~$T_{jk_0}$} (k0);
     \draw [myarrow] (j) -- node[sloped,above] {$A_{jk_1}$,~$T_{jk_1}$} (k1);
     \draw [myarrow] (j) -- node[sloped,above] {$A_{jk_2}$,~$T_{jk_2}$} (k2);
    \end{tikzpicture}
    \caption[\textbf{Schematic of step 1}]{An schematic of step 1, the only case that \textbf{i} and \textbf{j} have a chance to meet is when at least one of their common friends meet them within the same activity location.}
    \label{frined_of_friend}
\end{figure}
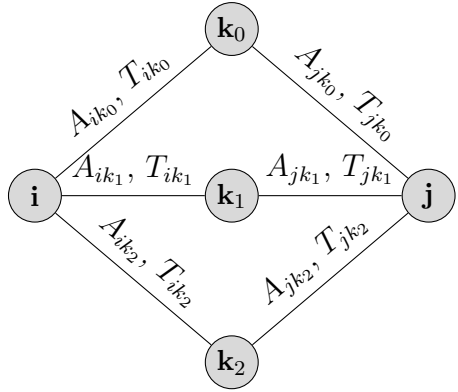

\item If sexually active population under study has an age range $\alpha=[\alpha_1,\alpha_2]$, therefore people in  \so but out of this age range cannot play a role in sexual network, for example a child with age six years old cannot be part of sexual network. Thus, the second step is to remove all people with ages $\notin\alpha$.

\item The sexual network we aim to extract from \so is heterosexual. That is, we do not have homosexual individuals in \textbf{SexNet}, therefore,  the edges between two men or two women in \so can not be a potensial sexual edge, so we remove this edge from \so.
\item Our last assumption is that sexual partners are not living together, therefore household edges in \so (edges with label H) cannot be the proper edge ending up with sexual activity, so we remove household edges.
\end{enumerate}

\subsubsection{Phase 2: main algorithm }

The main algorithm takes three inputs \textbf{SocNet}, $BJD$ matrix corresponding to  joint-degree distribution for \textbf{SexNet}, and $p\in[0,1]$ the embedding percentage of \textbf{SocNet}, in fact the value p tells what percentage of edges in \se is in \textbf{SocNet}. The output would be heterosexual network \se  which is a partially subgraph of \so and matrix $BJD$ represents its joint-degree distribution. 
Starting with an empty set of nodes for \textbf{SexNet}, we select nodes in \so to add to \textbf{SexNet}: we sort nodes in \so based on their degree and as the first node in \se add the man node in \so with maximum degree to set of nodes in \textbf{SexNet} and assign its desired degree in \se equal to column size of $BJD$, and then select his first sexual partner with probability $p$ from his social neighbors and with probability $1-p$ from closest nodes in \so but not a social friend.

At any stage, we first find the partners for nodes in \se from their social friend, if the degree of all nodes in \se meet their desired degree, we add a new node to \textbf{SexNet}. We repeat adding edges until all edges are placed. 
Algorithms \ref{main_alg} and \ref{help_alg}
explain this procedure, and table  \ref{notation} is the table of definition for symbols in the algorithms.

\begin{table}[htbp]
\begin{center}
\begin{tabular}{r c p{10cm} }
\toprule
\footnotesize
$G.n$ & $\defeq$ &set of nodes  in  $G$\\
$G.e$ & $\defeq$ &set of edges  in network $G$\\
$d_G(\textbf{i})$ & $\defeq$ &degree of node \textbf{i} in network $G$\\
$d_G$ & $\defeq$ &degree frequency list for network $G$\\
$G.N(\textbf{i})$ & $\defeq$ & set of neighbors of node \textbf{i} in Network G\\
$dis(G,\textbf{u},\textbf{v})$ & $\defeq$ & distance between two nodes \textbf{u} and \textbf{v} in Network G\\
$M.col(M.row)$ & $\defeq$ & column size (row size) of a matrix M\\
$M(i,:) (M(:,i))$ & $\defeq$ & $i^{th}$ row (column) of matrix M\\
$V( i)$ & $\defeq$ & $i^{th}$ element of vector V\\
$|S|$ & $\defeq$ & size (the number of elements) of  a set S\\ 
$S.remove(m)$ &$\defeq$ & remove member m from   a set S\\ 
$S.sample(P)$ &$\defeq$ & randomly select an element with property P (if P$=1$ there is no property) from set S \\ 
$urn$ &$\defeq$ & uniform random number in $[0,1]$\\
\bottomrule
\end{tabular}
\end{center}
\caption[\textbf{Notations table}]{Table of notation for a conventional network $G$ in algorithms.}
\label{notation}
\end{table}
To keep or remove an edge  we have to calculate the degree of nodes attached to it for each possible edge in the \textbf{SocNet}, thus,  the full set of experiments run in O($|E|P^mP^w$) time, where $|E|$ is the number of edges in \textbf{SocNet}, $P^m$ number of its men nodes and $P^w$ number of its women nodes. This method is feasible if average degree ($\frac{2|E|}{P^m+P^w}$) of the network is not high.

\newpage

\begin{algorithm}[H]
 \SetAlgoLined
 $\textbf{SexNet}.n=\emptyset,~\textbf{SexNet}.n \leftarrow \textbf{u}={max_{d_\so(\textbf{k})}}\{\textbf{k} \in \textbf{SocNet}.n \}$\;
  $d_{\se}(\textbf{u}) \triangleq BJD.col,~stub(\textbf{u}) \triangleq d_\se(\textbf{u})$\;
   $E\triangleq \sum_i\sum_jBJD(i,j)$\;
  \While{$|\textbf{SexNet}.e|\leq E$}{
 $NF\triangleq\{\textbf{k} \in \textbf{SexNet}.n~if~stub(\textbf{k})>0\}$\;
  \While {$|NF|\geq 1$}{
  $\textbf{u}={max_{stub(\textbf{k})}}\{\textbf{k} \in NF \}$\;
  $(d',\textbf{v})=\Call{FP}{\textbf{u}, \textbf{SocNet},            \textbf{SexNet}, BJD, NF, p}$\;
  \eIf{$(d',\textbf{v})\neq False$}{
    $NF.remove$(\textbf{u})\;
   }{
    Make edge (\textbf{u},\textbf{v}) in \textbf{SexNet}$,stub(\textbf{u})\leftarrow stub(\textbf{u})-1,$ $stub(\textbf{v})\leftarrow stub(\textbf{v})-1$\;
    \If{stub(\textbf{u})=0 [stub(\textbf{v})=0]}
    {$d_\se.remove(d_\se(\textbf{u}))$ [$d_\se.remove(d_\se(\textbf{v}))$]\;}
    \eIf{\textbf{u} is woman}{$BJD(d,d')-1$}{$BJD(d',d)-1$}
  }}
   \textbf{SexNet}$.n \leftarrow \textbf{u}={max_{d_\so(k)}}\{k \in \textbf{SocNet}.n-\textbf{SexNet}.n\}$\;
   $d_\se(\textbf{u})\triangleq max\{d_\se\},~stub(\textbf{u})\triangleq d_\se(\textbf{u}).$
  }
   \caption{ Extracting sexual network from social network (\textbf{Soc2sex})}\label{main_alg}
\end{algorithm}

\begin{minipage}{\linewidth}
\begin{algorithm}[H]
 \SetAlgoLined
   $d=d_{\se}(\textbf{u})$\;
      \lIf{\textbf{u} is woman}
       {$R\triangleq BJD(d,:)$, \textbf{else} $R\triangleq BJD(:,d)$}
\For{\texttt{iter} in range($|R|$)}
{\eIf{$R\neq0$}
    {$R.sample(d': R(d')\neq 0)$\;
    \eIf{$urn\leq$p}
          { $K1\triangleq\{\textbf{k}\in NF: \textbf{k}\in \textbf{SocNet}.N(\textbf{u})-\textbf{SexNet}.N(\textbf{u}), d_\se                      (k)=d'\}$\;
         $K2\triangleq\{\textbf{k}\in \textbf{SocNet}.N(\textbf{u})-\textbf{SexNet}.n:  d_\so(k)\geq d'\}$\;} 
         {$K1\triangleq \{\textbf{k}\in NF: \textbf{k}\notin \textbf{SocNet}.N(\textbf{u})\cup\textbf{SexNet}.N(\textbf{u}), d_\se (k)=d'\}$\;
        $K2\triangleq \{\textbf{k}\in \textbf{SocNet}.n-\textbf{SexNet}.n:  d_\so(k)\geq d'\}$\;}
          \uIf{$K1\neq \emptyset$}
           {\footnotesize{\textbf{v}$=K1.sample(\textbf{w}: dis(\textbf{SocNet},\textbf{u},\textbf{w})=min\{dis(\textbf{SocNet},\textbf{u},\textbf{k})~for~k\in K1\})$}\;
              Break\; }
              \uElseIf{$K2\neq \emptyset$}
              { \footnotesize{\textbf{v}$=K2.sample(\textbf{w}: dis(\textbf{SocNet},\textbf{u},\textbf{w})=min\{dis(\textbf{SocNet},\textbf{u},\textbf{k})~for~k\in K2\})$\;
       $d_\se(\textbf{v})\triangleq d'$, $stub(\textbf{v})\triangleq d'$}\;
          Break\;}
              \Else{$R(d')\triangleq 0$\;}
         }
         {$(d',\textbf{v})=False$\;
         break\; }
         }
\textbf{Return} $(d',\textbf{v})$\;
 
   \caption{  Finding partner with proper degree for a given node (\textbf{FP})}\label{help_alg}
\end{algorithm}
\end{minipage}
\subsubsection{Phase 3: correcting $BJD$ }
We intent to have  a generated sexual network  such that $BJD$ matrix represents its  joint-degree distribution. But, because in the main algorithm, we may not find a proper partner for some people, therefore, we may end up with a different joint-degree distribution for the network. The goal of third Phase is to correct joint-degree distribution of generated \se through some edge rewiring.  We define $\widetilde{BJD}$ as matrix representing joint-degree distribution of \se and $\mathcal{E}=BJD-\widetilde{BJD}$ as error matrix. Matrix  $\mathcal{E}$ may have nonzero elements:
\begin{enumerate}
\item If element $(i,j)$ of $\mathcal{E}$ for $i, j>1$ is a positive value $k$, it means that \se  needs $k$ more edges between women with $i$ partners and men with $j$ partners. To make these edges  we do the following process $k$ times: In \se we find a woman with $i$ partners like \textbf{w}$_i$ which has a degree-1 partner like \textbf{m}$_1$, then we find a man with $j$ partners  like \textbf{m}$_j$ which is social friend of \textbf{w}$_i$  but not her sexual partner and also has a degree-1 partner like \textbf{w}$_1$. Then we do a rewiring: remove edges (\textbf{w}$_i$,\textbf{m}$_1$) and (\textbf{m}$_j$,\textbf{w}$_i$) and add edge (\textbf{w}$_i$, \textbf{m}$_j$).
\item If element $(i,j)$ of $\mathcal{E}$ for $i, j>1$ is a negative value $k'$, it means that we have  extra $k'$  edges between women with $i$ partners and men with $j$ partners. To remove these extra edges  we do the following  process $k'$ times: we find a woman with $i$ partners like \textbf{w}$_i$ which has a degree-j partner like \textbf{m}$_j$, and  that at least one of their social friends are not in \se. Then we find a man  like \textbf{m}$_1$ which is social friend of \textbf{w}$_i$ that  not in \textbf{SexNet}$.n$, and also  we find a woman  like \textbf{w}$_1$ which is social friend of \textbf{m}$_j$ that is not in \textbf{SexNet}$.n$. Then we do a rewiring: remove edge (\textbf{w}$_i$,\textbf{m}$_j$) and add edges (\textbf{w}$_i$, \textbf{m}$_1$),(\textbf{w}$_1$,\textbf{m}$_j$).
\item In the previous steps, we pushed  back nonzero elements in $\mathcal{E}$ to its first row and column, which this causes new nonzero elements in the first row and column.
To remove these nonzero values, we have to add or remove small components. For example if  element (i,1) of $\mathcal{E}$ is a positive value $k$ it means  that we need a small component of a woman with $i$ partners whose all partners  are degree-1 men, therefore, we simply make this component from the people who are not currently in \textbf{SexNet}$.n$. If  element $(1,j)$ of $\mathcal{E}$ is a negative value $k$ it means  we have to remove a small component of a man with $j$ partners whose all partners  are degree-1 women, therefore, we simply look for such a component and remove it from \textbf{SexNet}.
\end{enumerate}
This Phase 3 may not make matrix $\mathcal{E}$ exactly equal to zero, because in its approach the proper nodes may not exist. However, it improves joint-degree distribution of \textbf{SexNet}. In the Result Section of Chapter (\ref{abm}), we observe that in practice Phase 3 would not be needed.

\newpage
\singlespacing
\addcontentsline{toc}{chapter}{References} 
\bibliographystyle{own_unsrt} 
\bibliography{thesis}
\thesisspacing

\chapter*{Biography}

The author was born in Boroojen in $1984$ year and graduated from Amir Kabir University with MSc in $2009$.
The author started the PhD program at the Tulane University mathematics department in August $2012$, eventually completing the program in May $2018$.

\end{document}